\documentclass[pra,onecolumn,superscriptaddress,nofootinbib,11pt,tightenlines]{revtex4}
\usepackage[a4paper, left=0.8in, right=0.8in, top=0.9in, bottom=0.9in]{geometry}
\usepackage{lmodern}

\usepackage{cmap} 
\usepackage[utf8]{inputenc}
\usepackage[english]{babel}
\usepackage[T1]{fontenc}
\usepackage{amsmath,amssymb,amsthm,mathrsfs,amsfonts}
\usepackage{booktabs,siunitx}
\usepackage{bm}
\usepackage{cancel}
\usepackage{xspace}
\usepackage{stmaryrd}
\usepackage{tikz-cd}
\usepackage{enumerate}

\usepackage{comment}

\usepackage{physics}
\usepackage{complexity}
\usepackage{appendix}
\usepackage{mathtools}
\usepackage[dvipsnames]{xcolor}
\definecolor{blueviolet}{rgb}{0.2, 0.2, 0.6}
\definecolor{webgreen}{rgb}{0,.5,0}
\definecolor{webbrown}{rgb}{.6,0,0}
\usepackage[
	bookmarks=false,
	colorlinks=true, 
	urlcolor=webbrown,
	linkcolor=blueviolet,
	citecolor=PineGreen,
	pdfstartpage=1,
	pdfstartview={FitH},  
	bookmarksopen=false
	]{hyperref}
\allowdisplaybreaks
\usepackage{tikz}
\usepackage{dsfont}
\usepackage{braket}
\usetikzlibrary{tikzmark}  
\usepackage{natbib}
\usepackage{listings}
\usepackage{makecell}
\usepackage{array}
\usepackage{longtable}
\usepackage{float}
\usepackage{needspace}
\usepackage{titlesec}
\usepackage{multirow}

\titleformat*{\paragraph}{\bfseries}

\newtheorem{corollary}{Corollary}
\newtheorem{definition}{Definition}
\newtheorem{lemma}{Lemma}

\newtheorem{proposition}{Proposition}
\newtheorem{remark}{Remark}

\newtheorem{theorem}{Theorem}

\usepackage{bbm}

\newcommand{\Lrep}[1]{L (#1)}
\newcommand{\Rrep}[1]{R (#1)}
\newcommand{\bin}{\mathtt{b}}
\renewcommand{\b}[1]{\mathtt b(#1)}

\newcommand{\ringR}{\mathcal{R}}

\newcommand{\coker}{\operatorname{coker}}

\newcommand{\rowspace}{\operatorname{rowspace}}

\newcommand{\quop}{\textup{\textsc{quop}}\xspace}
\newcommand{\quops}{\textup{\textsc{quop}s}\xspace}
\newcommand{\otimesR}{\otimes_\ringR}

\newcommand{\orbit}{\mathcal{O}}
\renewcommand{\order}[1]{\operatorname{ord}(#1)}
\newcommand{\wt}[1]{\operatorname{wt}(#1)}

\newcommand{\RepCode}{\operatorname{Rep}}
\newcommand{\LP}{\operatorname{LP}}
\definecolor{darkpurple}{HTML}{330A66}

\DeclareFixedFont{\ttb}{T1}{txtt}{bx}{n}{9} 
\DeclareFixedFont{\ttm}{T1}{txtt}{m}{n}{9}  

\DeclareMathOperator{\supp}{\mathrm{supp}}
\DeclareMathOperator{\Tor}{Tor}
\newcommand{\bvec}[1]{\b{\operatorname{vec}(#1)}}

\usetikzlibrary{positioning,arrows.meta,calc,backgrounds}

\definecolor{ink}{HTML}{253346}      
\definecolor{sub}{HTML}{5F6B7A}      
\definecolor{accA}{HTML}{6D72D8}     
\definecolor{accB}{HTML}{7C6EE6}
\definecolor{accC}{HTML}{8B7CF6}
\definecolor{accD}{HTML}{A78BFA}
\definecolor{accE}{HTML}{0077B6}     
\definecolor{accF}{HTML}{118AB2}
\definecolor{accG}{HTML}{D95F02}     
\definecolor{accH}{HTML}{E76F00}
\definecolor{accI}{HTML}{F28E00}
\colorlet{levI}{accB}
\colorlet{levII}{accE}
\colorlet{levIII}{accH}
\colorlet{levId}{accA!80!black}
\colorlet{levIId}{accE!80!black}
\colorlet{levIIId}{accG!85!black}

\newcommand{\icoFlag}[1]{%
  \draw[#1, line width=1.6pt, line cap=round] (-0.20,-0.42) -- (-0.20,0.42);
  \fill[#1] (-0.20,0.42) -- (0.32,0.27) -- (-0.20,0.12) -- cycle;
  \draw[#1, line width=1.2pt, line cap=round] (-0.34,-0.42) -- (-0.06,-0.42);}
\newcommand{\icoSliders}[1]{%
  \foreach \x in {-0.26,0,0.26}{
    \draw[#1!55, line width=1.3pt, line cap=round] (\x,-0.36) -- (\x,0.36);}
  \fill[#1] (-0.26,0.14) circle (0.08);
  \fill[#1] (0,-0.12) circle (0.08);
  \fill[#1] (0.26,0.24) circle (0.08);}
\newcommand{\icoGrid}[1]{%
  \foreach \i in {-1,0,1}{\foreach \j in {-1,0,1}{
    \fill[#1!45] (\i*0.30,\j*0.30) circle (0.06);}}
  \fill[#1] (0.30,0) circle (0.075);
  \draw[#1, line width=1.2pt] (0.30,0) circle (0.165);}
\newcommand{\icoSearch}[1]{%
  \fill[#1!45] (-0.20,0.18) circle (0.045);
  \fill[#1!45] (-0.02,0.06) circle (0.045);
  \fill[#1!45] (-0.05,0.26) circle (0.045);
  \draw[#1, line width=1.6pt] (-0.10,0.14) circle (0.27);
  \draw[#1, line width=2.4pt, line cap=round] (0.11,-0.07) -- (0.36,-0.36);}
\newcommand{\icoCircuit}[1]{%
  \draw[#1!55, line width=1.2pt] (-0.42,0.17) -- (0.42,0.17);
  \draw[#1!55, line width=1.2pt] (-0.42,-0.17) -- (0.42,-0.17);
  \draw[#1, line width=1.3pt] (-0.16,0.17) -- (-0.16,-0.17);
  \fill[#1] (-0.16,0.17) circle (0.06);
  \draw[#1, line width=1.3pt] (-0.16,-0.17) circle (0.10);
  \draw[#1, line width=1.3pt] (-0.16,-0.27) -- (-0.16,-0.07);
  \filldraw[fill=#1!18, draw=#1, line width=1.3pt]
    (0.13,0.06) rectangle (0.35,0.28);}
\newcommand{\icoPlot}[1]{%
  \draw[#1, line width=1.3pt, line cap=round]
    (-0.36,0.36) -- (-0.36,-0.30) -- (0.42,-0.30);
  \draw[#1, line width=1.5pt]
    (-0.27,0.25) .. controls (-0.05,0.00) and (0.10,-0.12) .. (0.34,-0.19);
  \fill[#1] (-0.27,0.25) circle (0.05);
  \fill[#1] (0.02,-0.045) circle (0.05);
  \fill[#1] (0.34,-0.19) circle (0.05);}
\newcommand{\icoGear}[1]{%
  \foreach \a in {0,45,...,315}{
    \draw[#1, line width=3.2pt, line cap=butt] (\a:0.235) -- (\a:0.40);}
  \draw[#1, line width=1.5pt] (0,0) circle (0.27);
  \draw[#1, line width=1.5pt] (0,0) circle (0.11);}
\newcommand{\icoMove}[1]{%
  \fill[#1!45] (-0.32,0.16) circle (0.06);
  \fill[#1!45] (-0.08,0.16) circle (0.06);
  \fill[#1!45] (-0.32,-0.20) circle (0.06);
  \fill[#1!45] (-0.08,-0.20) circle (0.06);
  \fill[#1] (0.16,0.16) circle (0.075);
  \draw[#1, line width=1.0pt, dash pattern=on 1.4pt off 1.2pt] (0.16,-0.20) circle (0.085);
  \draw[#1, line width=1.2pt, -{Stealth[length=1.7mm]}]
    (0.28,0.10) .. controls (0.40,0.02) and (0.40,-0.06) .. (0.29,-0.14);}
\newcommand{\icoChipAtom}[1]{%
  \foreach \p in {-0.15,0,0.15}{
    \draw[#1, line width=1.4pt, line cap=round] (\p,0.26)  -- (\p,0.40);
    \draw[#1, line width=1.4pt, line cap=round] (\p,-0.26) -- (\p,-0.40);
    \draw[#1, line width=1.4pt, line cap=round] (0.26,\p)  -- (0.40,\p);
    \draw[#1, line width=1.4pt, line cap=round] (-0.26,\p) -- (-0.40,\p);}
  \draw[#1, line width=1.5pt, rounded corners=1.2pt] (-0.26,-0.26) rectangle (0.26,0.26);
  \draw[#1, line width=0.9pt] (0,0) ellipse (0.16 and 0.06);
  \draw[#1, line width=0.9pt, rotate=60]  (0,0) ellipse (0.16 and 0.06);
  \draw[#1, line width=0.9pt, rotate=-60] (0,0) ellipse (0.16 and 0.06);
  \fill[#1] (0,0) circle (0.04);}

\tikzset{
  box/.style={
    draw=ink, line width=0.5pt, rounded corners=3pt,
    inner sep=6pt, anchor=center},
  io/.style={
    draw=ink, line width=0.5pt, rounded corners=3pt,
    inner sep=6pt, anchor=center},
  badge/.style={
    circle, draw=none, minimum size=0.48cm, inner sep=0pt,
    font=\bfseries\fontsize{13}{13}\selectfont, text=white},
  flow/.style={
    draw=ink, line width=1.2pt, rounded corners=3pt,
    -{Stealth[length=2.2mm,width=2.2mm]}},
  fb/.style={
    draw=sub, line width=1.0pt, dashed, dash pattern=on 3pt off 2.5pt,
    rounded corners=4pt, -{Stealth[length=2mm,width=2mm]}},
}

\newcommand{\card}[1]{%
  \begin{minipage}[c][1.30cm][c]{3.85cm}\raggedright
    \hspace*{1.05cm}%
    \begin{minipage}[c]{2.25cm}\raggedright\hyphenpenalty=10000\exhyphenpenalty=10000
      {\color{ink}\fontsize{9.8}{10.5}\selectfont\bfseries #1}%
    \end{minipage}%
  \end{minipage}}

\definecolor{deepblue}{rgb}{0,0,0.5}
\definecolor{deepred}{rgb}{0.6,0,0}
\definecolor{deepgreen}{rgb}{0,0.5,0}
\newcommand\pythonstyle{\lstset{
language=Python,
basicstyle=\ttm,
morekeywords={self},              
keywordstyle=\ttb\color{deepblue},
emph={MyClass,__init__},          
emphstyle=\ttb\color{deepred},    
stringstyle=\color{deepgreen},
frame=tb,                         
showstringspaces=false
}}
\lstnewenvironment{python}[1][]
{
\pythonstyle
\lstset{#1}
}
{}

\newcommand\pythoninline[1]{{\pythonstyle\lstinline!#1!}}

\renewcommand{\R}{\mathcal{R}}

\newcommand{\Z}{\mathbb{Z}}

\newcommand{\F}{\mathbb{F}}

\newcommand{\indicator}{\mathds{1}}

\newcommand{\MergeCode}{\text{Merge}}

\DeclareMathOperator{\im}{\mathrm{im}}
\DeclareMathOperator{\dist}{\mathrm{dist}}

\newtheorem{example}{Example}

\usepackage[noend]{algpseudocode}
\usepackage{algorithm,algorithmicx}
\algrenewcommand\alglinenumber[1]{\sf\scriptsize\color{blue}{#1}}
\algrenewcommand\algorithmicrequire{\textbf{Input:}}
\algrenewcommand\algorithmicensure{\textbf{Output:}}

\newcommand{\quis}{\mathcal{I}}
\newcommand{\bquis}{\mathcal{I}_B}
\newcommand{\hquis}{\mathcal{I}_H}
\newcommand{\fquis}{\mathcal{I}_F}

\renewcommand{\thesection}{\Roman{section}}
\renewcommand{\thesubsection}{\Alph{subsection}}

\makeatletter
\renewcommand{\p@subsection}{\thesection.}
\renewcommand{\p@subsubsection}{\thesection.\thesubsection.}
\makeatother

\newcommand{\Hx}{H_X}
\newcommand{\Hz}{H_Z}
\newcommand{\dx}{d_x}
\newcommand{\dz}{d_z}
\newcommand{\NSx}{\mathcal{N}_x}
\newcommand{\NSMx}{N_x}
\newcommand{\NSMxsk}{\tilde{N}_x}
\newcommand{\NSMxskcp}{\tilde{N}_x'}
\newcommand{\NSMxskcpreff}{\tilde{N}_{x,REFF}'}
\newcommand{\NSz}{\mathcal{N}_z}
\newcommand{\NSMz}{N_z}
\newcommand{\kerx}{\nu_x}
\newcommand{\kerz}{\nu_z}
\DeclareMathOperator{\rowspan}{rowspan}

\usepackage[capitalise]{cleveref}
\crefname{subsubsection}{Section}{Sections}
\Crefname{subsubsection}{Section}{Sections}

\makeatletter
\def\magicportmatrix#1{\begingroup\m@th
  \@tempdima 8.75\p@
  \setbox\z@\vbox{%
    \def\cr{\crcr\noalign{\kern2\p@\global\let\cr\endline}}%
    \ialign{\hfil$##$\hfil\kern2\p@\kern\@tempdima&\thinspace\hfil$##$\hfil
      &&\enspace\hfil$##$\hfil\crcr
      \omit\strut\hfil\crcr\noalign{\kern-\baselineskip}%
      #1\crcr\omit\strut\cr}}%
  \setbox\tw@\vbox{\unvcopy\z@\global\setbox\@ne\lastbox}%
  \setbox\tw@\hbox{\unhbox\@ne\unskip\global\setbox\@ne\lastbox}%
  \setbox\tw@\hbox{$\kern\wd\@ne\kern-\@tempdima\left[\kern-\wd\@ne
    \global\setbox\@ne\vbox{\box\@ne\kern2\p@}%
    \vcenter{\kern-\ht\@ne\unvbox\z@\kern-\baselineskip}\,\right]$}%
  \null\;\vbox{\kern\ht\@ne\box\tw@}\endgroup}
\makeatother

\begin{document}

\title{High-rate qLDPC processors}
\date{\today}
\author{Aditya Bhardwaj\equalcontrib}
\email{a7b@caltech.edu}
\affiliation{California Institute of Technology, Pasadena, CA 91125, USA}
\author{Muzhou Ma\equalcontrib}
\email{mma2@caltech.edu}
\affiliation{California Institute of Technology, Pasadena, CA 91125, USA}
\author{Nadine Meister}
\affiliation{California Institute of Technology, Pasadena, CA 91125, USA}
\author{Robbie~King}
\affiliation{Oratomic, Pasadena, CA 91125, USA}
\author{Dolev~Bluvstein}
\affiliation{Oratomic, Pasadena, CA 91125, USA}
\author{John Preskill}
\affiliation{California Institute of Technology, Pasadena, CA 91125, USA}
\affiliation{Oratomic, Pasadena, CA 91125, USA}
\author{Madelyn~Cain}
\affiliation{Oratomic, Pasadena, CA 91125, USA}
\author{Qian Xu}
\email{qxu@oratomic.com}
\affiliation{Oratomic, Pasadena, CA 91125, USA}
\author{Hsin-Yuan Huang}
\email{hhuang@oratomic.com}
\affiliation{Oratomic, Pasadena, CA 91125, USA}
\affiliation{California Institute of Technology, Pasadena, CA 91125, USA}

\makeatletter
\DeclareRobustCommand{\equalcontrib}{%
  \frontmatter@footnote{Equal contribution.}
}
\makeatother

\begin{abstract}
Despite significant progress on quantum low-density parity-check (qLDPC) codes, building qLDPC processors that are high-rate, high-throughput, hardware-friendly, and fast-to-decode remains a challenge.
In this work, we introduce \emph{mitten codes}, a family of qLDPC processor codes with encoding rate~20\% and check weight 9, constructed from non-abelian groups. The non-abelian structure evades stringent distance bounds suffered by their abelian counterparts, allowing mitten codes to reach distance~$18$ and beyond with only a few hundred data qubits.
All logical operators of a mitten code are related by the underlying group action, and this symmetry yields a modular, low-overhead logical toolkit: full Clifford operations follow from bridging just two reusable seed surgery gadgets of tens of qubits each, or from a single fixed extractor. 
Furthermore, qLDPC processors based on mitten codes support parallel magic-state injection into all logical qubits at once and high-rate surgery that executes many logical measurements in parallel.
Under circuit-level depolarizing noise, our fast decoder shows, without extrapolation, that the $\llbracket 300,60,14 \rrbracket$ code attains a block logical error rate of ${\sim} 10^{-11}$ per round at $0.1\%$ physical error rate (PER), while the $\llbracket 975,195,\leq 24 \rrbracket$ code reaches ${\sim}10^{-8}$ at $0.4\%$ PER.
Directly decoding $15$ billion surgery experiments on the $\llbracket 540,108,18 \rrbracket$ code at $0.1\%$ PER, we observe only two logical failures, thereby demonstrating a qLDPC processor capable of running ${\sim} 10^{10}$ logical operations.
Our decoder achieves this accuracy while being compatible with sub-millisecond average latency per logical cycle, sufficient for real-time decoding on neutral atom hardware.
Discovered by an end-to-end design pipeline built on \textsf{sQetch}, a distance estimator orders of magnitude faster than existing tools, and mapping efficiently onto near-term neutral atom and superconducting hardware, mitten codes open a practical path toward fault-tolerant quantum computation.
\end{abstract}

\maketitle

\section{Introduction}

Useful quantum algorithms call for operations with logical error rates orders of magnitude below the physical error rates that can be achieved natively in hardware~\cite{rivest1978method, koblitz1987elliptic, deutsch1992rapid, kitaev1995quantum, lloyd1996universal, simon1997power, shor1999polynomial, wang2008quantum, reiher2017elucidating, alexeev2021quantum, dalzell2023quantum, eisert2025mind, huang2026vast, babbush2026grand}, making fault tolerance essential for practical quantum computing~\cite{shor1995scheme, gottesman2013fault}. The surface code~\cite{bravyi1998quantum, dennis2002topological, kitaev2003fault, fowler2009high, fowler2012surface,Litinski_2019} has been widely studied as a potential route to fault-tolerance because it can tolerate relatively high physical error rates~\cite{Wang_2011, Stephens_2014} and requires only geometrically local processing in a two-dimensional layout~\cite{google2023suppressing, google2025quantum,vezvaee2025surfacecodescalingheavyhex}. However, despite significant theoretical progress on variants of the surface code~\cite{gidney2025yoked,low2026denserplanarsurfacecode}, its low encoding rate entails a daunting overhead cost~\cite{low2026denserplanarsurfacecode, gidney2021factor, gidney2025factor}. If instead geometrically nonlocal operations are permitted during error syndrome extraction~\cite{hong2024entangling, ransford202698, bluvstein2024logical, reichardt2024fault, HAL_paper, wang2026demonstration}, then high-rate quantum low-density parity-check (qLDPC) codes can substantially reduce this overhead cost in a quantum memory~\cite{Tillich_2014, kovalev2013quantum, breuckmann2016constructions, breuckmann2021balanced, panteleev2021degenerate, higgott2021subsystem, 10.1145/3519935.3520017, tremblay2022constant, leverrier2022quantum, wang2023abelian, higgott2024constructions, lin2024quantum, bravyi2024high, zhao2026towards, Cohen_2022, Pecorari_2025}. 

But a quantum processor is more than just a memory; its utility is determined by the number of logical operations it can execute per unit time and per physical qubit without logical error. Factors contributing to this include the footprint of the code and its logical gadgets, the parallelism of logical operations, the clock speed of syndrome extraction, and the achievable error suppression. An architecture that optimizes one factor in isolation might still be impractical if the others lag far behind. Thus, the designer of a fault-tolerant quantum processor seeks to satisfy the following four desiderata simultaneously:
\begin{enumerate}
    \item \textit{High encoding rate:} The number of physical qubits per logical qubit should be kept low.
    \item \textit{High throughput:} The processor should support flexible logical operations acting on programmable sets of logical qubits, so that many operations can be performed per unit time.
    \item \textit{Hardware compatibility:} The processor should map efficiently onto physical platforms.
    \item \textit{Fast and accurate decoding:} The processor should be decodable quickly and accurately, during both storage and logical operations.
\end{enumerate}

\begin{figure}[t]
    \centering
    \includegraphics[width=1.0\linewidth]{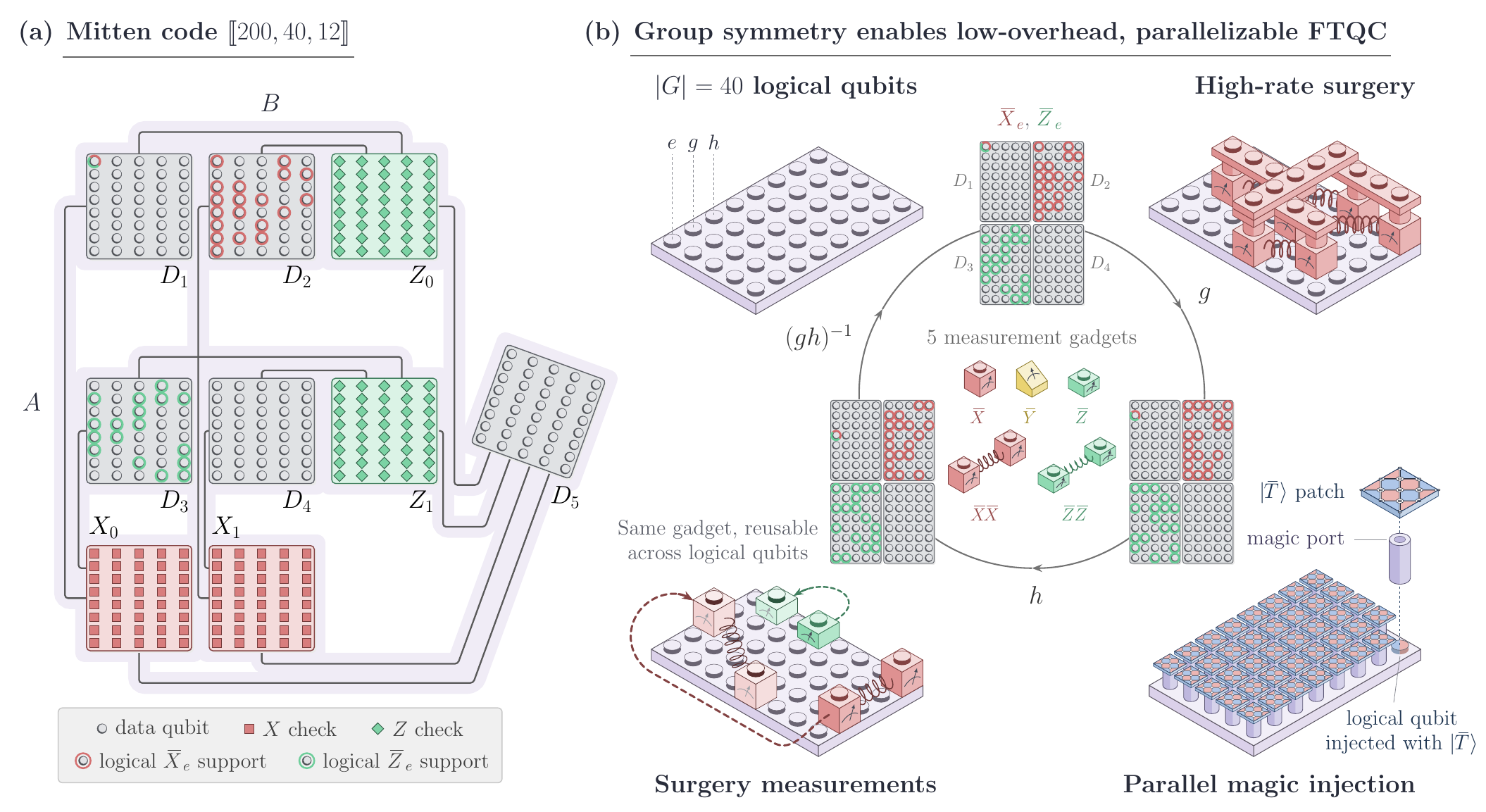}
    \caption{\textbf{Mitten codes as fault-tolerant qLDPC processors}
    \textbf{(a)} The $\llbracket 200,40,12 \rrbracket$ mitten code built from the lifted product of base matrices $A$ and $B$ defined over the group algebra $F_2[G]$ of the non-abelian group $G = C_4 \times D_{10}.$ The code consists of 
    five blocks $D_1,\ldots,D_5$ of $|G|=40$ physical data qubits and two blocks each of $X$-check and $Z$-check ancilla qubits. Grey lines show the block-wise
    connectivity, while the connectivity between individual qubits is
    determined by the group structure through the ring entries of the
    base matrices $A$ and $B$. Data qubits highlighted in red (green)
    mark the support of the canonical logical $\bar{X}_e$ ($\bar{Z}_e$) operator
    of weight $20$ ($18$), where $e$ is the identity element of $G$. The supports of the conjugate pair $(\bar{X}_e, \bar{Z}_e)$ intersect at a single data qubit in $D_1$.
    \textbf{(b)} The code encodes $|G|=40$ logical qubits, one
    per group element, and the group action carries the canonical
    representatives $\bar{X}_e$, $\bar{Z}_e$ to those of every other
    logical qubit (center). Consequently, five measurement
    gadgets---for $\bar{X}$, $\bar{Y}$, $\bar{Z}$, $\bar{X}\bar{X}$,
    and $\bar{Z}\bar{Z}$, all generated from just two seed
    gadgets for $\bar{X}$ and $\bar{Z}$---are reusable across all logical qubits and enable universal Clifford operations (bottom left). The same symmetry supports high-rate surgery, which measures many Pauli products in parallel (top right), and parallel magic state injection (bottom right): $|\bar{T}\rangle$ states are grown in
    $|G|$ distance $d_{\mathrm{rep}}$ surface-code patches which can be expressed as the lifted product $\LP(\RepCode(d_{\mathrm{rep}}),\RepCode(d_{\mathrm{rep}}))$ of repetition codes, and injected into all logical qubits at once through a \emph{magic
    port}, the intermediate code $\LP(\RepCode(d_{\mathrm{rep}}),B)$
    that interfaces the patches with the mitten code.
    } 
    \label{fig:mitten-code}
\end{figure}

In this work, we introduce \emph{mitten codes}\footnote{The name ``mitten codes'' derives from the block structure of the check matrices. The $X$ and $Z$ checks involve the left and right action of $G$ and each check has five columns: four similar ``fingers'' and a distinguished ``thumb.''}, a family of non-abelian lifted product codes\footnote{Lifted product codes constructed from classical base matrices over $\F_2[G]$ are a special case of balanced product codes~\cite{breuckmann2021balanced} that can be constructed from more generic chain complexes whose chain groups need not be free modules.}~\cite{panteleev2021degenerate,Panteleev_2022,breuckmann2021balanced} that furnish quantum processors satisfying all four desiderata (Fig.~\ref{fig:mitten-code}). Mitten codes are built from $1\times 2$ base matrices over the group algebra $\mathbb{F}_2[G]$ of a \emph{non-abelian} group~$G$. Comparing the number of physical qubits to the number of checks guarantees an encoding rate of at least $20\%$, and allowing $G$ to be non-abelian yields codes that surpass the distance-$6$ upper bound suffered by all abelian constructions with the same base-matrix shape (Section~\ref{sec:mitten_codes}). Remarkably, they do so at small block sizes, achieving distances $\gtrsim 20$ with just hundreds of data qubits. Their modest size not only facilitates experimental implementation and fast decoding, but also enables low-overhead and highly parallel logical operations that could otherwise be challenging to realize on larger codes~\cite{zheng2025highratesurgeryconstantoverheadlogical, cain2026shorsalgorithmpossible10000, zhao2026towards, kasai2026breakingorthogonalitybarrierquantum, okada2026highgirthregularquantumldpc}. A large processor can instead be assembled by tiling many such identical blocks, which can have advantages for fault-tolerant operation compared to a single monolithic quantum device. 

\begin{figure}[t]
    \centering
    \includegraphics[width=\textwidth]{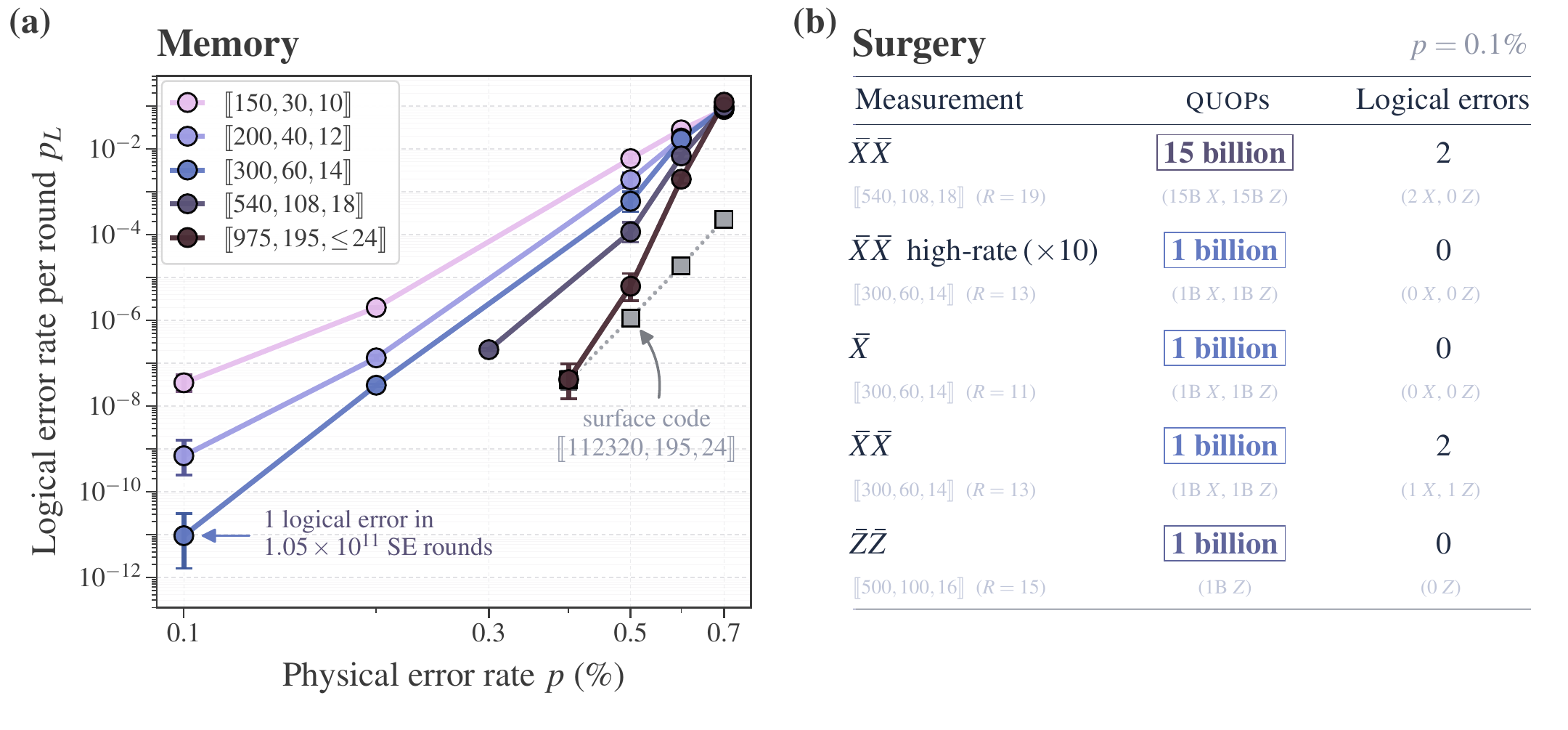}
    \caption{\textbf{Performance of mitten codes as qLDPC processors using our telescoping decoder.} All simulations are performed under circuit-level noise with uniform depolarizing noise of strength $p$ and no idling noise. \textbf{(a)} Memory performance of mitten codes. We plot the logical error rate per syndrome extraction round averaged over an equal number of $X$ and $Z$ basis experiments where error bars represent $68.27\%$ Clopper-Pearson confidence intervals; the mitten codes form a family consistent with a threshold of $p_c\approx0.7\%$. At $p=0.4\%$, the $\llbracket 975, 195, \leq 24\rrbracket$ mitten code matches the logical error rate of the $\llbracket 112320, 195, 24\rrbracket$ stack of rotated surface codes (grey) decoded using correlated matching despite using two orders of magnitude fewer physical qubits. \textbf{(b)} Surgery simulations at $p=0.1\%$, where $R$ is the number of rounds of syndrome extraction performed on the merged code to fault-tolerantly measure the target logical operator(s). For each experiment, the table reports the total number of logical operations (\quops) simulated and the number of logical errors observed; the code, the number of rounds $R$, and the per-basis breakdowns of shots and errors are given in grey beneath each entry. The high-rate $\bar{X}\bar{X}$ gadget measures ten logical $\bar{X}\bar{X}$ operators in parallel, hence the $\times 10$. The $\llbracket 540, 108, 18\rrbracket$ simulation is decoded using only a single basis of detectors throughout all stages of the telescoping decoder.}
    \label{fig:numerics}
\end{figure}

A central structural feature of the mitten codes is a highly symmetric canonical logical basis which enables low-overhead and parallelizable fault-tolerant quantum processing. Whenever the base matrices satisfy a square invertibility condition (Definition~\ref{def:square_inv_condition}), one can explicitly construct low-weight conjugate pairs of logical operators $\{(\bar{X}_g, \bar{Z}_g)\}_{g\in G}$ representing the $k = |G|$ encoded logical qubits, such that every logical $\bar X_g$ ($\bar Z_g$) is carried to any other logical $\bar X_h$ ($\bar Z_h$) through the group action of $G$. Because all logical representatives are related by group symmetry, we can build a
modular, low-overhead surgery toolkit: just five reusable gadgets, themselves built from only two seed gadgets, suffice for universal Clifford processing. Alternatively, we may
trade this modularity for expressiveness with a fixed full
extractor~\cite{he2025extractorsqldpcarchitecturesefficient} that can directly measure any logical Pauli product. The same symmetry enables high-throughput logic through parallel surgery~\cite{zheng2025highratesurgeryconstantoverheadlogical, Cowtan_2026} and parallel magic state injection, the latter of which provides the remaining non-Clifford resource needed for universal quantum computation. With our \emph{telescoping decoder}, each  window of $O(d)$ syndrome extraction rounds can be decoded with sub-millisecond average latency in both memory and surgery experiments (Section~\ref{sec:decoder}, Table~\ref{tab:rt-load}), where $d$ is the code distance, while attaining very low logical error rates (Fig.~\ref{fig:numerics}). Finally, the structure of the code aligns naturally with parallel block moves of atoms in neutral-atom arrays, and for superconducting platforms we show that mitten codes have a similar hardware complexity~\cite{HAL_paper} to the bivariate bicycle codes~\cite{bravyi2024high} of comparable block size while achieving much higher encoding rate.

Mitten codes, however, form a vast family, and building a practical device requires singling out explicit instances and equipping them with concrete gadgets, syndrome extraction schedules, and hardware layouts. Our second main contribution is a qLDPC processor discovery pipeline (Section~\ref{sec:knitter}) that builds on our theoretical results and turns a target specification of code parameters, logical error rate, gadget footprint, and hardware constraints into a complete qLDPC processor design (Fig.~\ref{fig:pipeline}). Running the pipeline end to end produced the mitten codes with $20\%$ encoding rate and check weight $9$ listed in Table~\ref{tab:1x2_processor_code_param}.
This was made tractable by \textsf{sQetch}~\cite{knitkitgithub}, our new GPU-based distance estimator, which is up to $800{,}000$ times faster than current state-of-the-art methods (Appendix~\ref{app:sqetch}).

The remainder of the text is organized as follows. We first formalize what we require of a qLDPC processor and the parameters by which we evaluate one (Section~\ref{sec:processor_metrics}). We then present mitten codes together with their canonical logical basis (Section~\ref{sec:mitten_codes}), and show how this structure enables low-overhead gadgets that turn them into qLDPC processors (Section~\ref{sec:universal_ftqc}). Next, we establish their performance under full circuit-level decoding experiments with our new telescoping decoder (Section~\ref{sec:decoder}) and demonstrate how they may be implemented on both neutral atom and superconducting qubit platforms (Section~\ref{sec:hardware_implementations}). Finally, we present our qLDPC processor discovery pipeline that produced the codes (Section~\ref{sec:knitter}) and the accoutrements needed to convert them into qLDPC processors.

\begin{figure}[t]
  \centering
  \resizebox{\linewidth}{!}{%
  \begin{tikzpicture}[font=\rmfamily]

\def\rA{0}   \def\rB{-3.6} \def\rC{-7.2}

\node[io,fill=accA!22]  (A) at (0,\rA)    {\card{Processor targets}};
\node[box,fill=accB!12]  (B) at (5,\rA)    {\card{\mbox{Code-parameter} filter}};
\node[box,fill=accC!12]  (C) at (10,\rA)   {\card{Addressability filter}};
\node[box,fill=accD!12]  (D) at (15,\rA)   {\card{Code search}};

\node[box,fill=accE!9] (E) at (5,\rB)    {\card{\mbox{Circuit-level} design}};
\node[box,fill=accF!9] (F) at (10,\rB)   {\card{Simulation \& decoding}};

\node[box,fill=accG!8] (G) at (2.5,\rC)  {\card{\mbox{Instruction-set} gadgets}};
\node[box,fill=accH!8] (H) at (7.5,\rC)  {\card{Hardware implementation}};
\node[io,fill=accI!25] (I) at (12.5,\rC) {\card{Processor design}};

\foreach \n/\lbl/\c in {A/1/accA,B/2/accB,C/3/accC,D/4/accD,E/5/accE,F/6/accF,G/7/accG,H/8/accH,I/9/accI}{
  \node[badge,fill=\c] at (\n.north west) {\lbl};
}

\foreach \n/\ico/\c in {A/\icoFlag/accA, B/\icoSliders/accB, C/\icoGrid/accC,
                        D/\icoSearch/accD, E/\icoCircuit/accE, F/\icoPlot/accF,
                        G/\icoGear/accG, H/\icoMove/accH, I/\icoChipAtom/accI}{
  \begin{scope}[shift={($(\n.west)+(0.73,0)$)}, scale=0.90]
    \ico{\c}
  \end{scope}
}

\draw[flow] (A) -- (B);
\draw[flow] (B) -- (C);
\draw[flow] (C) -- (D);
\draw[flow] (D.south) -- ++(0,-1.45) -| (E.north);
\draw[flow] (E) -- (F);
\draw[flow] (F.south) -- ++(0,-1.45) -| (G.north);
\draw[flow] (G) -- (H);
\draw[flow] (H) -- (I);

\draw[fb] (F.east) -- (17.6,\rB) -- (17.6,1.75) -- (6.4,1.75) -- ($(B.north)+(1.4,0)$);
\node[font=\itshape\fontsize{12}{13.5}\selectfont,color=sub,anchor=south] at (12,1.78) {retune group, base matrices};

\node[anchor=south west] (LI) at ($(A.north west)+(0,0.58)$)
      {{\fontsize{15}{14}\selectfont\bfseries\color{levId}Level I}~~{\fontsize{13}{14.5}\selectfont\bfseries\color{levId}design and algebraic search}};
\draw[levI,line width=1pt] ($(LI.south west)+(0,-1pt)$) -- ++(2.6,0);

\node[anchor=south west] (LII) at ($(E.north west)+(0,0.58)$)
      {{\fontsize{15}{14}\selectfont\bfseries\color{levIId}Level II}~~{\fontsize{13}{14.5}\selectfont\bfseries\color{levIId}fault-tolerance validation}};
\draw[levII,line width=1pt] ($(LII.south west)+(0,-1pt)$) -- ++(2.75,0);

\node[anchor=south west] (LIII) at ($(G.north west)+(0,0.58)$)
      {{\fontsize{15}{14}\selectfont\bfseries\color{levIIId}Level III}~~{\fontsize{13}{14.5}\selectfont\bfseries\color{levIIId}gadgets and hardware}};
\draw[levIII,line width=1pt] ($(LIII.south west)+(0,-1pt)$) -- ++(2.95,0);

\end{tikzpicture}}
  \caption{
    \textbf{High-rate qLDPC processor discovery pipeline.}
    Starting from the processor targets of Definition~\ref{def:processor_params}, the pipeline narrows the design space in three levels. The first level explores the space of groups and base matrices: theoretical distance bounds (Appendix~\ref{app:distance_bounds}) first
    constrain the search space, and the fast code-distance estimator
    \textsf{sQetch} (Appendix~\ref{app:sqetch}) then brute-force searches the remaining space to find codes with good parameters and addressable logical operations. The second level searches for fault-tolerant syndrome extraction schedules and verifies logical performance through circuit-level simulations decoded with
    our telescoping decoder (Appendix~\ref{app:decoding}). The third level constructs gadgets for instruction sets ($\bquis,\hquis,\fquis$), and optimizes for hardware implementation on atom array and superconducting qubit platforms to obtain the final processor design. Dashed arrows indicate feedback: candidates that fail downstream validation trigger
    retuning of the group or base matrices upstream. The toolkit for our pipeline is \href{https://github.com/a7b/yarn}{open-sourced}.
    }
    \label{fig:pipeline}
\end{figure}
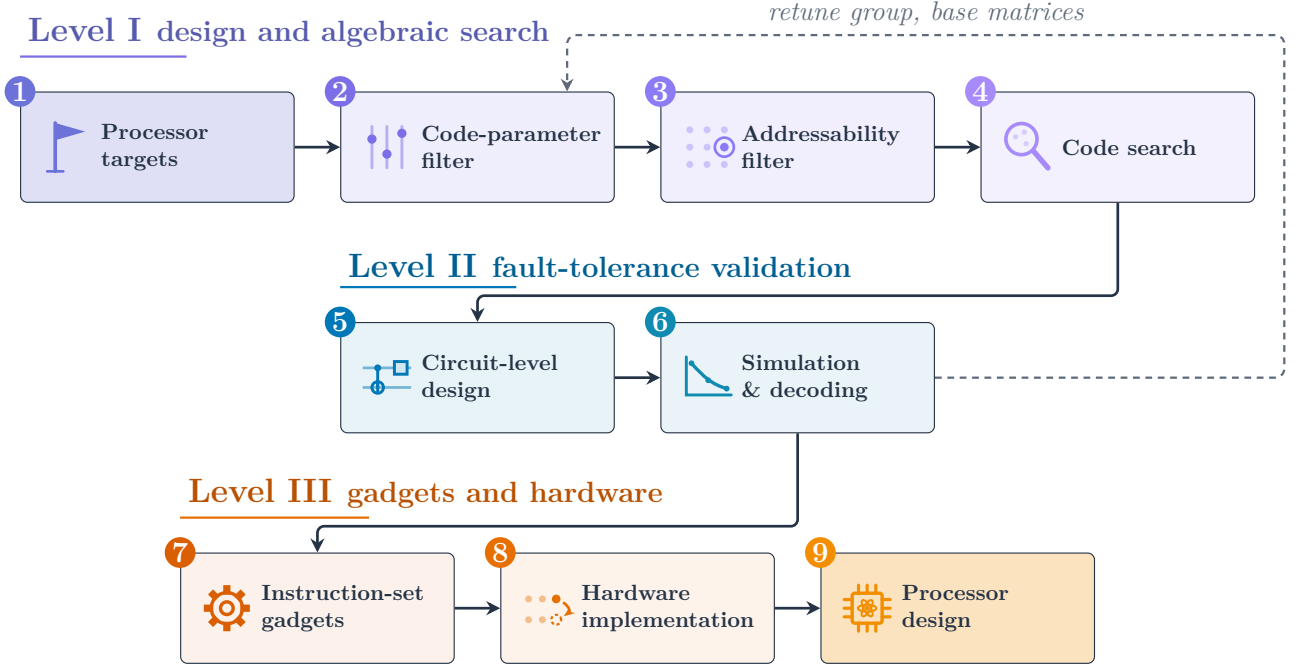

\section{Metrics for fault-tolerant qLDPC processors}
\label{sec:processor_metrics}
Fault tolerance is achieved by encoding logical information in quantum error-correcting codes, but codes alone only protect information. To process information, a code must be equipped with an instruction set, together with the supporting gadgets and a decoder. In this section, we formalize this notion of a fault-tolerant quantum processor and present three parameters by which we characterize the performance of a quantum processor.

\begin{definition}[Quantum instruction set (informal version of Definition~\ref{def:quis_formal})]
\label{def:quis}
    A quantum instruction set $\mathcal{I}$  of a fault-tolerant quantum processor is a set of instructions, each a logical operation together with its physical gadget realization, that supports universal fault-tolerant quantum computation.
\end{definition}

\noindent Quantum instruction sets are not unique. We consider three, each designed for a different objective. The \emph{basic instruction set} $\bquis$ (Definition~\ref{def:bquis_formal}) consists of fault-tolerant gadgets for single-qubit logical Pauli measurements, weight-two Pauli product measurements (PPMs), and noisy $T$-state injection; these compile any Clifford+$T$ circuit~\cite{Barenco_1995, boykin1999universalfaulttolerantquantumcomputing, Bravyi_2005, Bravyi_2016, fowler2019lowoverheadquantumcomputation, Litinski_2019}, and the noisy $T$-states are distilled by inter-block transversal\footnote{Since all CSS codes support transversal CNOT gates, any instruction set based on a CSS code naturally includes it.}~\cite{Bravyi_2005, Reichardt_2005} or intra-block~\cite{xu2026distillingmagicstatesbicycle} distillation. Since distillation uses only the fault-tolerant Pauli measurements, the magic state distillation process is also fault-tolerant.

The \emph{high-throughput instruction set}~$\hquis$ (Definition~\ref{def:hquis_formal}) generalizes the basic instruction set to also include high-rate gadgets for parallel Pauli product measurements and parallel $T$-state injection. This suits architectures where the time overhead is as important as the space overhead~\cite{Zhou_2025, Nguyen2025Quantum, cain2026shorsalgorithmpossible10000}. The \emph{fixed-gadget instruction set} $\fquis$  (Definition~\ref{def:fquis_formal}) consists of a single fixed extractor~\cite{he2025extractorsqldpcarchitecturesefficient, blue2026extractorslogicalprocessinghypergraph} that measures arbitrary Pauli products selected in software, plus noisy $T$-state injection. Because the connectivity of the extractor is fixed, it suits platforms with rigid connectivity.

The three instruction sets serve different purposes and can be combined; as we show in Section~\ref{sec:universal_ftqc}, the mitten codes support all three at low overhead. In particular, the entire basic instruction set can be realized with only $5$ reusable gadgets. The formal definitions of each instruction set are in Appendix~\ref{app:metrics}. A processor can be characterized in terms of its instruction set and hardware capabilities. 
\begin{definition}[\quop]
\label{def:quop}
    Given a quantum instruction set $\quis$, call an instruction \emph{primitive} if it is not a parallel composition of other instructions in $\quis$. A quantum operation (\quop) is the execution of a single primitive instruction; an instruction composed of $m$ primitive instructions executed in parallel counts as $m$ \quops. 
\end{definition}
\noindent Under this definition, every instruction in the basic instruction set $\bquis$ counts as a single \quop, while instructions in the high-throughput instruction set $\hquis$ that perform $m$ weight-1 or weight-two PPMs or inject $m$ magic states in parallel count as $m$ \quops. In the fixed-gadget instruction set $\fquis$ an arbitrary PPM performed by the extractor counts as a single \quop.\footnote{
    While the ability to perform arbitrary Pauli-product measurements can lead to more efficient compilations of algorithms, the aim of our definitions is to characterize a processor in a manner agnostic of any algorithm. We therefore count every arbitrary Pauli-product measurement in $\fquis$ as one \quop regardless of its weight; the compilation savings appear instead in the number of \quops an algorithm requires under different instruction sets.
}
With the notion of a \quop in hand, we now define three parameters that characterize a universal fault-tolerant qLDPC processor.
\begin{definition}[Parameters of a universal fault-tolerant qLDPC processor]
\label{def:processor_params}
Fix a qLDPC processor with a quantum instruction set $\quis$ and a decoder. We define:
\begin{enumerate}
    \item Processing capacity: The expected number of \quops that the processor can execute without any logical error; for a logical error rate $p_{\mathrm{L}}$ per \quop, the processing capacity is $1/p_{\mathrm{L}}$.

    \item Throughput: The maximum number of \quops that the processor can execute within a single logical cycle consisting of $R$ rounds of syndrome extraction, where the most standard protocols require $R=O(d)$.

    \item Cycle time: The time to complete one syndrome extraction (SE) cycle.
\end{enumerate}
\end{definition}

\noindent The three parameters are set by distinct design choices. Processing capacity demands resilience to noise during both storage and processing. This favors a low physical error rate, a small code block length, and low gadget overhead: fewer and less faulty components mean fewer errors to correct, and a smaller decoding problem that the decoder can solve quickly and accurately. Throughput is driven by the encoding rate and the parallelism of the logical layer. The cycle time is set by the time to complete one round of syndrome extraction, which can vary significantly not only across hardware platforms but also across syndrome extraction schedules; the latter is especially relevant when the schedules must be realized by non-local qubit movements such as on neutral atom arrays. Throughput and cycle time jointly set the \emph{logical operation rate}, i.e.\ the number of \quops executed per unit time, and can sometimes partially compensate for one another. Processing capacity, by contrast, caps the total number of operations executed reliably, and no gain in operation rate can remedy a processor with low processing capacity.

\section{Mitten codes}\label{sec:mitten_codes}

We present mitten codes, a family of qLDPC codes that furnish universal fault-tolerant qLDPC processors performing well on all three parameters in Definition~\ref{def:processor_params}. 

\begin{definition}[Mitten codes]\label{def:mitten_processor_code}
    Mitten codes are a family of lifted product codes $\LP(A,B)$ with check weight $9$ and encoding rate $20\%$ whose classical base matrices $A,B\in \F_2[G]^{1\times 2}$ have the canonical form
    \begin{equation}
        A = \begin{bmatrix} a_0 & a_1 \end{bmatrix}
          = \begin{bmatrix} g_1+g_2+g_3 & e+g_4+g_5 \end{bmatrix},
        \qquad
        B = \begin{bmatrix} b_0 & b_1 \end{bmatrix}
          = \begin{bmatrix} h_1+h_2+h_3 & e+h_4+h_5 \end{bmatrix},
    \end{equation}
    where $e$ is the identity and $g_i,h_i$ are group elements of the non-abelian group $G$. The $X$ and $Z$ parity check matrices of $\LP(A,B)$ are
    \begin{equation}\label{eq:mitten_HxHz}
    \begin{split}
        H_X &=
        \bordermatrix{
            & D_1 & D_2 & D_3 & D_4 & D_5 \cr
        X_0 & \Lrep{a_0} & 0 & \Lrep{a_1} & 0 & \Rrep{b_0^*} \cr
        X_1 & 0 & \Lrep{a_0} & 0 & \Lrep{a_1} & \Rrep{b_1^*} \cr
        } \\
        H_Z &=
        \bordermatrix{
            &   &   &   &   & \cr
        Z_0 & \Rrep{b_0} & \Rrep{b_1} & 0 & 0 & \Lrep{a_0^*} \cr
        Z_1 & 0 & 0 & \Rrep{b_0} & \Rrep{b_1} & \Lrep{a_1^*} \cr
        },
    \end{split}
    \end{equation}
    where the involution operator ${}^*$ maps $\sum_g \alpha_g\, g \mapsto \sum_g \alpha_g\, g^{-1}$, $\Lrep{\cdot}$ and $\Rrep{\cdot}$ are the left and right regular representations of $\ringR = F_2[G]$ (Definition~\ref{def:group-ring-and-regular-reps}), and we have labeled the blocks of data qubits, $X$-checks and $Z$-checks as in Figure~\ref{fig:mitten-code}(a). Additionally, we require the left regular representation $\Lrep{a_1}$ and the right-regular representation $\Rrep{b_1}$ to be full-rank matrices.
\end{definition}
The $1\times 2$ shape of the classical base matrices guarantees a $20\%$ encoding rate while the non-abelian group is essential for achieving high distance. Note that if $G$ were abelian, the minimum-weight codeword $(a_1, a_0)^T$ of $A$ produces a $\bar{Z}$ logical operator $(a_1, 0, a_0, 0, 0)^T$, thus capping the distance of $\LP(A,B)$ at the check weight of its classical base matrices. Taking $G$ non-abelian removes this cap and is what allows a small block to reach high distance. Using the discovery pipeline described in Section~\ref{sec:knitter}, we obtain eight instances of mitten codes with parameters ranging from $\llbracket 150,30,10\rrbracket$ through $\llbracket 540,108,18\rrbracket$ to $\llbracket 975,195,\leq 24\rrbracket$ (Table~\ref{tab:1x2_processor_code_param}).

\begin{table}[t]
\centering
\setlength{\tabcolsep}{4pt}
\renewcommand{\arraystretch}{1.1}
\begin{tabular}{@{}l c c c c c@{}}
\toprule
\textbf{$\llbracket n,k,d \rrbracket $} & group $G$ & $\wt{Lx}/\wt{Lz}$ & \makecell{$\bar{X}\bar{X}$-gadget\\ \footnotesize ($n_{\mathrm{anc}}$, $w_{\mathrm{merge}}$)}  & \makecell{Neutral-Atom SE cycle (ms)\\ \footnotesize (4-AOD, 2-AOD)} & \makecell{Superconducting \\ \footnotesize Complexity $C_{\mathrm{hw}}$}\\
\midrule
$\llbracket 150,30,10 \rrbracket$ & $C_5\times S_3$  & 18/10 & (78, 11) & $(5.49,\ 8.83)$ & 2.02\\
$\llbracket 200,40,12 \rrbracket$  & $C_4\times D_{10}$ & 20/18 & (88, 11)  & $(7.41,\ 10.79)$ & 2.09\\
$\llbracket 300,60,14 \rrbracket $   & $C_{10}\times S_3$ & 22/22 & (100, 12)  & $(7.22,\ 10.98)$ & 2.28\\
$\llbracket 500,100,16 \rrbracket $  & $C_5\rtimes C_{20}$ & 28/24 & (156, 12)  & $(10.61,\ 16.15)$ & 2.60\\
$\llbracket 540,108,18 \rrbracket $  & $C_9\rtimes C_{12}$ & 22/28 & (123, 10)  & $(11.74,\ 18.41)$ & 2.72\\
$\llbracket 630,126,\leq 20 \rrbracket $ & $C_7\rtimes C_{18}$  & 28/44 & (134, 10)  & $(11.47,\ 19.00)$ & 2.98\\
$\llbracket 780,156,\leq 22 \rrbracket $ & $C_{13}\times A_4$ & 74/84 & (435, 12)  & $(14.90,\ 23.24)$ & 3.08\\
$\llbracket 975,195,\leq 24 \rrbracket $ & $C_{13}\rtimes C_{15}$ & 102/92 & (622, 12)  & $(14.60,\ 22.70)$ & 3.40\\
\bottomrule
\end{tabular}
\caption{\textbf{Parameters of high-rate qLDPC processors from mitten codes.} Distances of the first five codes are exact; the last three are estimated from over $50$ million \textsf{sQetch} iterations together with over $50$ thousand iterations of a BP+OSD distance estimator~\cite{PhysRevResearch.2.043423}. $G$ indicates the group algebra $\F_2[G]$ over which the underlying base matrices $A,B$ of the mitten code $\LP(A,B)$ are defined on. 
$\wt{Lx}/\wt{Lz}$ are the weights of the canonical logical $X$/$Z$ operators. Explicit code construction data are provided in Table~\ref{tab:code_construction_full}. The $\bar{X}\bar{X}$-gadget column reports the overhead of the $\bar{X}\bar{X}$ surgery gadget, with $n_{\mathrm{anc}}$ the number of ancilla qubits and $w_{\mathrm{merge}}$ the merged-code check weight; a more complete breakdown is provided in Table~\ref{tab:processor_graph_surgery}.
We report the estimated syndrome-extraction cycle time on a neutral atom array using either $2$ sets (2-AOD) or $4$ pairs (4-AOD) of crossed AODs (further details in Table~\ref{table:atomarray-metrics}, Appendix~\ref{app:experimental_complexity}). 
$C_{\mathrm{hw}}$ is the hardware complexity of the code's multilayer superconducting layout from HAL~\cite{HAL_paper}, normalized so a surface-code layout yields $C_{\mathrm{hw}}=1$ regardless of its size; the full breakdown (routing tiers, coupler length, bump-bond transitions, through-silicon vias) is in Table~\ref{tab:paper_nonab_hw} (Appendix~\ref{app:experimental_complexity}). For comparison, the $\llbracket 144,12,12 \rrbracket$ gross code achieves $C_{\mathrm{hw}}=2.12$ and the $\llbracket 288,12,18 \rrbracket$ two-gross code achieves $C_{\mathrm{hw}}=2.24$.}
\label{tab:1x2_processor_code_param}
\end{table}

Mitten codes compare favorably in encoding rate and block size with other leading qLDPC code families. With a block size under a thousand, the $\llbracket 975,195,\leq 24\rrbracket$ instance encodes $195$ logical qubits, whereas a rotated surface code of distance $24$ would require over $100{,}000$ physical qubits to encode the same number. The bivariate bicycle codes~\cite{bravyi2024high} such as the $\llbracket 144,12,12\rrbracket$ gross code and the $\llbracket 288,12,18\rrbracket$ two-gross code, achieve encoding rates of only about $8\%$ and $4\%$, respectively, compared to the guaranteed $20\%$ of mitten codes. Mitten codes also have smaller block sizes than high-rate abelian LP codes~\cite{cain2026shorsalgorithmpossible10000} and Kasai codes~\cite{kasai2026breakingorthogonalitybarrierquantum, zhao2026towards}, simplifying decoding and facilitating low-overhead and highly parallelizable logical operations. 

The low overhead and high parallelism of logical operations on mitten codes stem from the full-rank requirement in Definition~\ref{def:mitten_processor_code} which is a special case of the square invertibility condition of Definition~\ref{def:square_inv_condition}. This condition equips mitten codes with a canonical logical basis whose logical operators are related by group action.
\begin{theorem}[Canonical logical basis for mitten codes {[special
case of Theorem~\ref{thm:canonical-logicals}]}]
\label{thm:mitten_canonical_basis}
    Let $\{\indicator_g\}_{g\in G}$ denote the standard basis of
    $\F_2^{|G|}$, with coordinates indexed by the elements of $G$.
    A mitten code as defined in
    Definition~\ref{def:mitten_processor_code} has a canonical
    logical basis $\{\bar{X}_g\}_{g\in G},\{\bar{Z}_g\}_{g\in G}$,
    such that $[\bar{X}_g,\bar{Z}_h] =
    2\bar{X}_g\bar{Z}_h\,\delta_{g,h}$ for $g,h\in G$, and each of
    the two sets is a single orbit of the group action of $G$ on
    the identity-element representative. Concretely, the logical
    representatives have the form
    \begin{equation}
        \bar{X}_g =
        \left(\indicator_g,\,u_g,\,0_{|G|},\,0_{|G|},\,0_{|G|}\right)^T,
        \qquad
        \bar{Z}_g =
        \left(\indicator_g,\,0_{|G|},\,v_g,\,0_{|G|},\,0_{|G|}\right)^T,
    \end{equation}
    where $u_g, v_g\in\F_2^{|G|}$ are the unique solutions of
    \begin{equation}\label{eq:mitten_logical_solutions}
        \Lrep{a_0}\indicator_g + \Lrep{a_1}v_g = 0,\qquad
        \Rrep{b_0}\indicator_g + \Rrep{b_1}u_g = 0.
    \end{equation}
\end{theorem}

\noindent

This canonical logical basis has three particularly nice properties. First, it is constructed directly from the codewords $(\indicator_g, u_g)^T$ and $(\indicator_g, v_g)^T$ of the classical base matrices $\Rrep{B}$ and $\Lrep{A}$ respectively. Secondly, all logical representatives can be obtained from the identity-element representatives $\bar{X}_e = (\indicator_e, u_e, 0_{|G|}, 0_{|G|}, 0_{|G|})^T$ and $\bar{Z}_e = (\indicator_e, 0_{|G|}, v_e, 0_{|G|}, 0_{|G|})^T$ by group multiplication. To see this, identify $\F_2^{|G|}$ with the group algebra $\ringR=\F_2[G]$ by mapping each standard basis vector $\indicator_g$ to the group element $g$, so that $u_g$ and $v_g$ correspond to ring elements (Definition~\ref{def:group-ring-and-regular-reps}). Under this identification, $\bar{X}_e=\left(e,\, u_e,\, 0,\, 0,\, 0\right)^T$. Left-multiplying each block by $g\in G$ yields $\left(g,\, g u_e,\, 0,\, 0,\, 0\right)^T$, which still satisfies the defining condition in~\eqref{eq:mitten_logical_solutions} because left multiplication commutes with the right-regular representations $\Rrep{b_0}$ and $\Rrep{b_1}$. By the uniqueness of $u_g$, this operator is precisely $\bar{X}_g$. Similarly, multiplying on the right by $g$ generates $\bar{Z}_g$ from $\bar{Z}_e$. Finally, the canonical basis has a favorable low-weight support on the data qubits. Each representative is supported on only two of the five data blocks, with weight at most $|G|+1$ although in practice much lower (Table~\ref{tab:1x2_processor_code_param}), and the supports of each conjugate pair $\bar{X}_g$ and $\bar{Z}_g$ intersect on exactly one data qubit---the qubit labeled $g$ in block $D_1$. These three properties of the canonical logical basis enable us to construct gadgets that support all three instruction sets $\bquis$, $\hquis$, and $\fquis$.

\section{Universal fault-tolerant logic}
\label{sec:universal_ftqc}
We equip mitten codes for universal fault-tolerant quantum processing by utilizing the properties of the canonical basis to construct gadgets that realize the three instruction sets $\bquis$, $\hquis$, and $\fquis$.

\textit{Reusable graph surgery for $\bquis$.}
A graph surgery gadget~\cite{Horsman_2012, fowler2019lowoverheadquantumcomputation, Vuillot_2019, cross2025improvedqldpcsurgerylogical, Ide_2025, Williamson_2026} measures a target logical Pauli operator by temporarily attaching an ancilla system to the code block (Fig.~\ref{fig:mitten-code}(b)). The ancilla system consists of ancilla qubits and new $X$ and $Z$ stabilizer checks that couple the ancilla checks to the support of the target logical. The new stabilizers are chosen so that a product of a subset of them equals the target logical operator. Measuring these stabilizers reads out the target logical measurement outcome while leaving the remaining encoded information intact. Because the canonical logical basis lies in a single group orbit, one seed gadget for $\bar{X}_e$ and one for $\bar{Z}_e$, rewired by the group action, can measure any $\bar{X}_g$ and $\bar{Z}_g$~\cite{webster2025explicit} for $g\in G$. Furthermore, bridging seed gadgets~\cite{cross2025improvedqldpcsurgerylogical} yields equally reusable gadgets for $\bar{X}_g\bar{X}_h$, $\bar{Z}_g\bar{Z}_h$, and $\bar{Y}_g$ for any $g,h\in G$. A mitten code therefore needs only five reusable gadgets for universal Clifford processing, and we construct all five explicitly and report their overheads in Table~\ref{tab:processor_graph_surgery}. All of the gadgets are distance-preserving and increase the maximum check weight by at most three. Further details regarding gadget construction can be found in Appendix~\ref{subsec:graph_surgery}

\textit{Parallel surgery gadgets and parallel magic state injection for $\hquis$.}
High-rate surgery gadgets can measure up to $k/2$ disjoint weight-two logical Pauli products~\cite{zheng2025highratesurgeryconstantoverheadlogical}, covering all logical qubits at once. We explicitly construct such gadgets measuring up to $k/2$ disjoint random Pauli products of the form $\bar{X}_g\bar{X}_h$ on the $\llbracket 300,60,14 \rrbracket$ mitten code. The low-weight canonical basis keeps the gadget overheads modest, as reported in Table~\ref{tab:processor_c300_high_rate_surgery}. Detailed explanations of their construction can be found in Appendix~\ref{subsec:parallel_surgery}. Next, we present a scheme for injecting $k$ magic states in parallel into a mitten code. The magic states from the surface codes of distance $d_{\mathrm{rep}}$ can be teleported to all logical qubits of the mitten code in parallel utilizing an ancillary LP code and parallel surgeries that preserve the LP structures. Notably, the extra space cost is dominated by the surface-code magic factory, and even this cost is typically subleading over the course of a computation,
since the factory's ancilla qubits can be substantially reused (Appendix~\ref{subsec:magic_injection_complexity},
Table~\ref{tab:parallel_magic_complexity}). The time overhead is $\Theta(d_{\mathrm{rep}})$, independent of $k$.
In Appendix~\ref{app:parallel_magic} we provide detailed analysis of the parallel magic injection protocol on mitten codes and prove its end-to-end distance-preserving property.

\textit{Full extractors for $\fquis$.}
A full extractor~\cite{he2025extractorsqldpcarchitecturesefficient} is a single fixed gadget that measures any logical Pauli product by activating or deactivating individual connections to the code block. Full extractors are in general costly in space, but the single-orbit canonical basis lets mitten codes reduce this cost. Using techniques similar to those in Refs.~\cite{blue2026extractorslogicalprocessinghypergraph, zheng2026canonical}, we construct single-sided $X$ and $Z$ extractors and join them with a bridge into a full extractor. Across the mitten code family, attaching the extractor keeps every merged check weight at most $14$; full overheads are reported in Table~\ref{tab:full_extractor_overhead} and further details of how we constructed the extractors can be found in Appendix~\ref{app:surgery}.

\section{Decoding performance and processing capacity of mitten codes}
\label{sec:decoder}

We now turn to the processing capacity of our mitten code qLDPC processors: the expected number of \quops they execute before the first logical error (Definition~\ref{def:processor_params}). Probing a processing capacity of $10^{10}$ \quops requires simulating and decoding tens of billions of shots, each an independent noise realization sampled from a circuit-level noise model with Stim~\cite{Gidney_2021}. Our decoding infrastructure is therefore designed to meet two requirements at once: it must have enough throughput to probe this low logical error rate regime, and its estimates should reflect a realistic decoding stack that can support real-time decoding.\footnote{By real-time decoding we mean meeting the \emph{reaction-time} requirement, not merely avoiding the backlog problem~\cite{Terhal_2015}. The backlog problem can always be avoided by adding compute to each stage of the pipeline via parallel window decoding~\cite{Skoric_2023}. However, a feedforward operation conditioned on a decoded logical measurement must wait for its window to be decoded, so what matters is the reaction time: the latency to decode a single window (Appendix~\ref{app:decoding}).}

To satisfy both, we develop a \emph{telescoping decoder} that applies many stages of belief propagation (BP) and Relay-BP~\cite{mueller2025improvedbeliefpropagationsufficient}, followed by a final integer-programming stage that solves the most-likely-error decoding problem exactly (Appendix~\ref{app:decoding}). The design builds on hierarchical decoding strategies that route harder shots to progressively more accurate decoders~\cite{delfosse2020hierarchicaldecodingreducehardware, ravi2022betterworstcasedecodingquantum, toshio2025decoderswitchingbreakingspeedaccuracy, zhao2026towards}.\footnote{We call our design \emph{telescoping} rather than the more common \emph{hierarchical} to emphasize its many nested stages of BP and Relay-BP: each stage winnows away the shots it can confidently decode, so that progressively more expensive decoders act on a progressively smaller residual of harder shots. This narrowing allows us to ``telescope'' in to the low logical error rate regime.} Custom CUDA kernels decode the vast majority of shots at high throughput on the GPU, and custom C kernels on the CPU handle the harder shots the GPU stages defer~\cite{knitkitgithub}. In most of our simulations, later stages also switch to a finer-grained but equivalent representation of the decoding problem given by the GARI transform~\cite{gari_paper} of the full
correlated detector error model. The GARI transform rewires the decoding graph
into a form more amenable to BP.

Benchmarked on memory experiments of the $\llbracket 144,12,12 \rrbracket$ gross code~\cite{bravyi2024high}, our telescoping decoder matches or improves upon the logical error rates of current state-of-the-art decoders, including most-likely-error approximators like Tesseract~\cite{beni2025tesseractsearchbaseddecoderquantum} and neural decoders like Cascade~\cite{gu2026scalableneuraldecoderspractical}, while delivering nearly twice the throughput of the neural decoder closest to it in accuracy (Figure~\ref{fig:bb144_dec_comp}, Table~\ref{tab:gross144_gari}, Appendix~\ref{app:decoding}).

Figure~\ref{fig:numerics} summarizes the performance of the mitten codes under our telescoping decoder. In all simulations, we employ a circuit-level noise model with uniform depolarizing noise of strength $p$ and no idling noise (Appendix~\ref{app:decdetails}). As a memory, the mitten codes behave as a family with numerics consistent with a threshold\footnote{At these finite block sizes, the common crossing of the logical error rate curves should be understood as an effective, finite-size threshold rather than an asymptotic one.} of $p_c \approx 0.7\%$. At $p=0.1\%$, the $\llbracket 300,60,14 \rrbracket$ code reaches a block logical error rate of $9.5^{+21.8}_{-7.9}\times 10^{-12}$ per syndrome extraction round, corresponding to only one logical error observed in over $100$ billion syndrome extraction rounds. Additionally, at $p=0.4\%$, the $\llbracket 975,195,\leq 24\rrbracket$ code attains a block logical error rate of $4.2^{+5.5}_{-2.7}\times 10^{-8}$ per round which matches that of the $\llbracket 112320, 195, 24\rrbracket$ stack of rotated surface codes decoded with correlated matching~\cite{Higgott2025sparseblossom} despite using two orders of magnitude fewer physical qubits. For all memory experiments in Figure~\ref{fig:numerics}(a) we use SE circuits that preserve the block-wise group structure, reflecting the experimental realization on atom-array platforms. Within this block-wise family, we used our new distance estimator \textsf{sQetch}~\cite{knitkitgithub} to search for schedules that likely preserve the circuit-level distance. The surgery experiments in Figure~\ref{fig:numerics}(b) instead use random coloration SE circuits, for which we verified with \textsf{sQetch} that the circuit-level distance is likewise preserved (Appendix~\ref{app:sqetch}, Appendix~\ref{app:experimental_complexity}).

Figure~\ref{fig:numerics}(b) shows the performance of fault-tolerant logical operations at $p=0.1\%$. For the $X$-type surgeries we simulate both the $X$ basis, which contains the measured logical operator, and the $Z$ basis, confirming that the unmeasured $Z$ logicals are not disturbed by the surgery. On the $\llbracket 300,60,14 \rrbracket$ code we simulate single-logical $\bar{X}$ measurements, joint $\bar{X}\bar{X}$ measurements, and high-rate surgery gadgets measuring ten $\bar{X}\bar{X}$ products in parallel, observing two logical errors across one billion $\bar{X}\bar{X}$ operations and none in the remaining experiments. This supports the conclusion that the $\llbracket 300, 60, 14 \rrbracket$ mitten code can function as a giga-\quop processor at $p = 0.1\%.$ The $\llbracket 540,108,18 \rrbracket$ code incurs only two logical errors across 15 billion $\bar{X}\bar{X}$ operations corresponding to a logical error rate of $1.33^{+1.76}_{-0.86}\times 10^{-10}$ per \quop{}. This is consistent with the conclusion that it can function as a ten-billion-\quop{} processor. Notably, every logical error we observe in the $X$ basis is a timelike error in the measured operator itself so these error rates can likely be suppressed further by simply increasing the number of syndrome extraction rounds. Finally, although the pipeline is tuned to maximize Monte Carlo throughput, its staged structure extrapolates to real-time operation: mapping each stage onto existing FPGA implementations, we estimate that the average latency of our decoder keeps pace with the incoming syndrome stream for memory and surgery experiments on representative mitten codes assuming a \SI{1}{\milli\second} syndrome extraction time (Appendix~\ref{app:decoding}, Table~\ref{tab:rt-load}).

\section{Hardware implementations}
\label{sec:hardware_implementations}

\begin{figure}[t]
    \centering
    \includegraphics[width=1\linewidth]{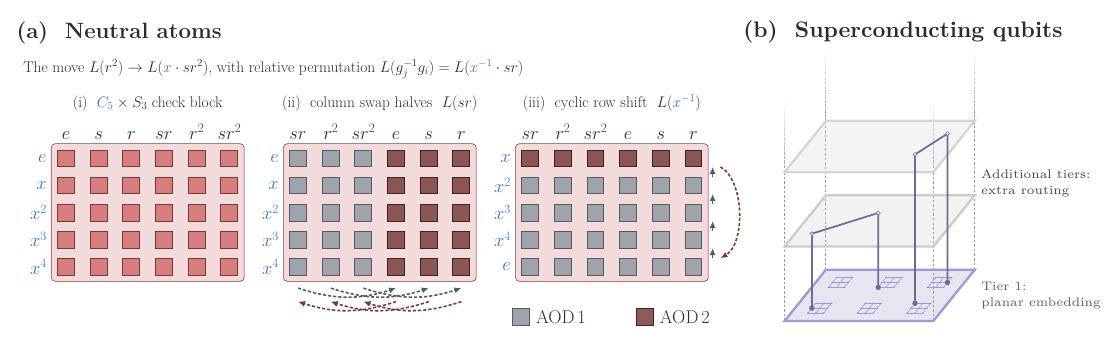}
    \caption{\textbf{Hardware implementation.} \textbf{(a)} In neutral atom hardware, each gate layer of syndrome extraction is implemented by entangling blocks of checks with corresponding blocks of data qubits, and permuting check qubits between layers. Here we show an example $X$-check qubit permutation for the $\llbracket 150, 30, 10 \rrbracket$ code ($G = C_5 \times S_3$), for which the check qubit permutations between layers $L(r^2)$ and $L(x\cdot sr^2)$ correspond to the group action $L(x^{-1}\cdot sr)$.
    \textbf{(b)} Mitten codes can also be implemented in multilayer superconducting hardware by placing a large planar subgraph on the first tier, then routing all remaining edges that cannot be placed on the first tier through higher tiers. 
    }
    \label{fig:hardware}
\end{figure}

To realize a mitten code processor, the nonlocal connectivity of the code must also be implemented on physical hardware. We focus on neutral atom and superconducting qubit platforms, which face different challenges. On neutral atom arrays, nonlocal connectivity is naturally realized by coherent atom transport~\cite{bluvstein2022quantum, bluvstein2024logical, manetsch2025tweezer, reichardt2024fault, zhang2026logical, sales2025experimental}, and the challenge is temporal: movement schedules must be optimized so that transport does not inflate the cycle time of the processor. On superconducting chips, the cycle time is already low (microsecond scale)~\cite{google2025quantum}, and the challenge is instead spatial: the nonlocal connectivity must be laid out as fixed couplers routed on a multilayer chip~\cite{wang2026demonstration,Yost_2020,Rosenberg_2017,PhysRevApplied.21.054063}. We address both challenges and demonstrate how mitten codes can be realized on both hardware platforms with competitive cycle times on neutral atoms and hardware complexities comparable to bivariate bicycle codes on superconducting hardware.

\subsection{Neutral atoms}
Our neutral atom SE protocol establishes the nonlocal connectivity by moving the ancilla atoms while keeping the data atoms stationary. For each entangling gate, the ancilla atom must be brought adjacent to the corresponding data atom. The exact movement is determined by the check matrices $H_X$ and $H_Z$ in \cref{eq:mitten_HxHz}, whose ring 
entries $(a_0, a_1, b_0, b_1)$ are sums of group elements $g \in G$. Each group element $g$ defines a gate layer between one check block ($X_0, X_1, Z_0, $ or $Z_1$) and one data block ($D_1 \dots, D_5$).

Once the ordering is chosen, moving between consecutive gate layers requires permuting the ancilla atoms within each check block according to multiplication by a group element.
On current neutral atom platforms, atom transport is typically implemented using a crossed pair of acousto-optic deflectors (AODs), which move entire rows and columns of atoms in parallel \cite{bluvstein2022quantum, bluvstein2026fault, tian2023parallel, reichardt2024fault}. The group structure of the mitten codes makes the required permutations highly structured and compatible with these AOD restrictions. When the group $G$ can be decomposed as a direct product $G = G_1 \times G_2$, the atoms can be arranged on a 2D grid with $G_1$ on one axis and $G_2$ on the other. Then, multiplication by a group element decomposes into row and column permutations of the $|G|$-atom check block. Thus, these permutations directly match the movements allowed under crossed AODs. Semidirect products admit a similar but twisted movement structure, described in more detail in Appendix~\ref{app:connectivity-atomarrays}.

To explicitly specify the movement rules, we first label the $|G|$ atoms in each check and data block by the elements of $G$. Recall that the first ring element is $a_0 = g_1 + g_2 + g_3$. In $H_X$, $\Lrep{a_0} = \Lrep{g_1} + \Lrep{g_2} + \Lrep{g_3}$ couples $X_0$ and $X_1$ to $D_1$ and $D_2$, respectively, with each term in the sum defining a separate gate layer. Consider two consecutive layers $\Lrep{g_i}$ and $\Lrep{g_j}$. During layer $\Lrep{g_i}$, check atom $k\in G$ is adjacent to data atom $h=g_i^{-1}k$, since a 1 in the binary permutation matrix $\Lrep{g_i}_{k,h} = 1$ corresponds to $k = g_ih.$ To then perform layer $\Lrep{g_j}$, the same check atom must move to $h' = g_j^{-1} k$. The required rearrangement of check blocks $X_0$ and $X_1$ is therefore the relative permutation $\Lrep{g_j^{-1}g_i}$, because left multiplication by $g_j^{-1}g_i$ sends $h = g_i^{-1} k$ to $h' = g_j^{-1} k$. Analogously, a transition from $\Rrep{g_i}$ to $\Rrep{g_j}$ requires the permutation $\Rrep{g_j^{-1}g_i}$. 
Repeating this process for all group elements in the check matrices completes one syndrome extraction cycle.

We show an example move in \cref{fig:hardware}(a) for the  $\llbracket 150, 30, 10 \rrbracket$ code, whose SE cycle video is also available at Ref.~\cite{knitkitgithub}. 
Here, the group is $G = C_5 \times S_3$, so we arrange each $|G|=30$-atom block on a $5\times6$ grid, with rows indexed by $C_5$ elements and columns by $S_3$ elements. Thus, each of the $|G| = 30$ atoms are labeled by a group element $(x^a, \sigma) \in C_5 \times S_3$ where $x$ is the generator of $C_5$, $a \in \{0,1,2,3,4\}$, and $\sigma \in S_3$. We show one transition from the hook-free SE schedule, which takes check blocks $X_0$ and $X_1$ from layer $\Lrep{g_i}=\Lrep{r^2}$ to layer $\Lrep{g_j}=\Lrep{x\cdot sr^2}$. Here, the reflection $s$ and the rotation $r$ are the generators of $S_3$.
Thus, the relative permutation to perform is $\Lrep{g_j^{-1} g_i} = \Lrep{x^{-1} \cdot sr}.$ This factorizes into the $S_3$ column permutation $\Lrep{sr}$ (\cref{fig:hardware}(a)ii) and the $C_5$ row shift $\Lrep{x^{-1}}$ (\cref{fig:hardware}(a)iii). $\Lrep{sr}$ simply swaps the 3 left and right columns, and $\Lrep{x^{-1}}$ is just a cyclic shift of the rows by 1. Each move can be implemented in a single transport step using two pairs of crossed AODs.

For the mitten codes in \cref{tab:1x2_processor_code_param}, we estimate syndrome extraction (SE) cycle times using present day AOD capabilities and experimentally demonstrated transport data from Ref.~\cite{bluvstein2022quantum}. 
This hardware model follows that of \cite{zhao2026towards}. For the codes listed in \cref{tab:1x2_processor_code_param}, we find that a single SE cycle takes about $8-24$ ms with 2 pairs of crossed AODs, and $5-15$ ms with 4 pairs of crossed AODs.

To isolate the overhead imposed by the AOD movement restrictions, we also consider a futuristic architecture in which every atom can follow an independently controlled trajectory. Such control could potentially be implemented using new and developing optical instruments other than AODs, such as a fast spatial light modulator (SLM) \cite{wei202610, bytyqi2026device}. 
As a comparison, we use the same transport parameters as in the AOD estimates, and find that every code admits an SE cycle time between $2.9-4.6$ ms (see \cref{table:atomarray-metrics}).
As this is a forward looking device, the raw SE cycle times are not to be emphasized, but rather its comparison with the AOD-constrained cycle times.

Finally, we want to emphasize movement speed is a tunable parameter depending on the desired performance.
In the idealized regime where finite trap depth, rather than AOD technical limitations, sets the maximum acceleration, a move which begins and ends at rest, takes time approximately $t \propto I^{-1/2}$, where I is the intensity of the tweezer.
(See Appendix~\ref{sec:transport_assumption_discussion} for more details.)
As such, to reduce movement times by a factor of $n$, one can, in theory, increase the power-per-trap by a factor of $n^2$. However, if laser power limits the number of qubits, this reduces the number of atoms that can be transported simultaneously by a factor of $n^2$. Conversely, parallel protocols, such as Shor-style syndrome extraction~\cite{shor1997faulttolerantquantumcomputation}, could also reduce SE times.
So, while SE times can readily be reduced with neutral atoms, both through using more laser power per-atom or through Shor-style SE, true hardware implementations will require a more holistic analysis in terms of space and time.

\subsection{Superconducting qubits}
Our smaller mitten codes are also well suited for superconducting qubit architectures, where the code's connectivity can be fabricated as couplers routed across the tiers of a multilayer chip. The simplest proxy for this fabrication cost is the planar thickness of the code's Tanner graph i.e. the minimum number of planar layers into which its edges can be decomposed. We prove that the thickness of every mitten code is exactly three (Theorem~\ref{thm:thickness}). This is one layer more than the thickness-$2$ bivariate bicycle codes~\cite{bravyi2024high}.

Thickness, however, is only a coarse proxy for superconducting hardware complexity; a finer assessment requires fixing a concrete hardware architecture. We choose a multi-chip stackup~\cite{HAL_paper} with flip-chip geometry~\cite{Yost_2020, PhysRevApplied.21.054063, Rosenberg_2017, norris2026performance} in which the qubits occupy the first tier, and each additional chip creates a higher routing tier that provides more space to route nonlocal couplers without crossings. 
The cost of a layout~\cite{HAL_paper} is then set by the number of tiers it uses, the length of each routed coupler, and the amount of vertical routing required through bump-bond transitions~\cite{Rosenberg_2017, norris2026performance} and through-silicon vias~\cite{Yost_2020, mallek2021fabricationsuperconductingthroughsiliconvias,hazard2023characterizationsuperconductingthroughsiliconvias}. 

As a first step towards evaluating the hardware feasibility of the mitten codes beyond their thickness, we lay out each code with HAL~\cite{HAL_paper}, a heuristic placement-and-routing algorithm for qLDPC codes on multilayer superconducting hardware. Although HAL can choose the first-tier qubit placement automatically, we find that supplying a custom layout is crucial. We place all the qubits on the first tier as a grid of $3\times 3$ modules from our thickness-3 decomposition (Figure~\ref{fig:thickness3_layout}), and optimize the module placement as a quadratic assignment problem. HAL then routes the remaining couplers through higher tiers while avoiding collisions (Appendix~\ref{app:superconducting_hw}). 

We summarize the resulting layouts by the hardware complexity $C_{\mathrm{hw}}$ defined in Ref.~\cite{HAL_paper}, which is a weighted average of the four  metrics defined above. These metrics are linearly rescaled so that a surface-code layout yields $C_{\mathrm{hw}}=1$ and a layout saturating the optimistic hardware targets of Ref.~\cite{HAL_paper}---five tiers, coupler lengths of ten times the nearest-neighbor distance, four bump-bond transitions and three TSVs per coupler---yields $C_{\mathrm{hw}}=2$. Table~\ref{tab:1x2_processor_code_param} lists $C_{\mathrm{hw}}$ for each mitten code, and Table~\ref{tab:paper_nonab_hw} gives the full breakdown into the individual metrics. Despite their larger thickness, the smaller mitten codes achieve hardware complexity similar to or lower than bivariate bicycle codes of comparable block size while encoding substantially more logical qubits.

\section{QLDPC processor discovery pipeline}
\label{sec:knitter}
We have shown that mitten codes enable qLDPC processors with high processing capacity and throughput. However, building a practical device requires singling out explicit instances and equipping them with concrete gadgets, syndrome extraction schedules, and hardware layouts. In this section, we describe our pipeline for qLDPC processor discovery. 
The pipeline relies on two ingredients, theoretical bounds that constrain the search space and \textsf{sQetch}~\cite{knitkitgithub}, a fast GPU-based distance estimator that runs up to $800{,}000 \times$ faster than current state-of-the-art methods~\cite{Pryadko_2022} and allows us to brute-force search through millions of candidate codes per hour on a single GPU.

As shown in Figure~\ref{fig:pipeline}, the first stage of the pipeline constrains the search space with filters that avoid constructing the quantum code altogether. This is enabled by the following new theorem.

\begin{theorem}[{[informal
version of Theorem~\ref{thm:dc_ub_dq}]}]
\label{thm:dc_ub_dq_informal}
    Let $A$ and $B$ be wide base matrices over the group algebra $\F_2[G]$ of any group $G$ such that the binary matrices $\Lrep{A}$ and $\Rrep{B}$ have full row rank. Then the minimum distance of the two classical base codes $A$ and $B$ upper bounds the distance of the lifted product code $\LP(A,B)$.
\end{theorem}
\noindent We also found a counterexample, presented in Appendix~\ref{app:distance_bounds}, that demonstrates that one cannot in general get rid of the full row rank condition in Theorem~\ref{thm:dc_ub_dq_informal}. Since every base matrix we search over satisfies the full row rank condition (Remark~\ref{rem:applicability}), Theorem~\ref{thm:dc_ub_dq_informal} allows us to just check the distance of the significantly smaller classical base matrices before combining them to make a quantum code. These classical distances are in turn upper bounded by a collection of purely algebraic bounds---the commutator subgroup bound and the element order bound among many others (Appendix~\ref{app:distance_bounds})---which are evaluated directly from the group $G$ and the base-matrix entries and allow us to rule out candidates that cannot have high distance. Finally, we use \textsf{sQetch} to brute force search through the remaining search space to pair classical codes up and construct mitten codes with high distance. We equip surviving codes with hook-error-free syndrome extraction schedules by searching via \textsf{sQetch} for schedules that likely preserve the circuit-level distance. Afterwards, the surviving codes are dressed with their surgery, extractor, and magic state injection gadgets, compiled to concrete atom-array movement schedules and multilayer superconducting layouts, and finally simulated under circuit-level depolarizing noise with the high-throughput telescoping decoder. A more detailed discussion of the design of our qLDPC processor discovery pipeline is provided in Appendix~\ref{app:knitter_pipeline}.

\section{Conclusion and outlook}
We have introduced mitten codes, a family of $1/5$-rate, check-weight-$9$ lifted product codes over non-abelian groups with instances from $\llbracket 150,30,10\rrbracket$ through $\llbracket 540,108,18\rrbracket$ to $\llbracket 975,195,\leq 24\rrbracket$, and shown that they enable universal fault-tolerant qLDPC processors that perform well by all three metrics in Definition~\ref{def:processor_params}: processing capacity, throughput, and cycle time. This high performance derives from the structured canonical logical basis admitted by these codes, which yields a basic instruction set based on five reusable graph-surgery gadgets, a high-throughput instruction set with parallel surgery gadgets and parallel magic state injection, and a fixed-gadget instruction set with a low-overhead extractor. Decoded by our telescoping decoder, mitten codes attain a logical error rate of $10^{-11}$ per syndrome extraction cycle in memory experiments, and achieve a logical processing capacity consistent with $10^{10}$ \quops, for a physical gate error rate of $p=0.1\%$.

These codes were found by our qLDPC processor discovery pipeline, powered by the GPU-based distance estimator \textsf{sQetch}. Because every design requirement enters the pipeline as an adjustable input, we anticipate its use well beyond the scope of this work, for co-designing processors suited for other hardware platforms, rate targets, and instruction sets. Indeed, this same pipeline also produced some of the quasi-cyclic $\LP$ codes that are reported in Ref.~\cite{cain2026shorsalgorithmpossible10000}, as well as other abelian $\LP$ codes described in Appendix~\ref{app:distance_bounds}.

Our telescoping decoder matches or improves upon the logical error rates of current state-of-the-art decoders. Even so, our analysis of logical error rates has been limited by the classical computing resources available for our surgery simulations; moreover, in the case of the $\llbracket 540,108,18\rrbracket$ code, the observed logical failures were exclusively due to timelike errors, indicating that increasing the number of syndrome measurement rounds could further suppress the logical error rates. Higher-distance mitten codes may well reach teraquop-scale processing capacity, a regime we have not yet been able to certify by Monte Carlo sampling. 

Mitten codes and the qLDPC processor discovery pipeline open many opportunities for future exploration. Further improvements in processor design may reduce the spacetime cost of surgery gadgets. Decoders that exploit code structure more effectively may run faster and perform even better. Better on-chip routing in superconducting processors and enhanced movement strategies in atom arrays may reduce the complexity of syndrome extraction and logical operations. Finally, the small block size, modest check weight, and rich symmetry of mitten codes may enable experimental demonstrations of high-rate qLDPC processors on near-term devices.

\begin{acknowledgments}
We direct the reader to the concurrent work Ref.~\cite{Hong2026RateOneFifthLDPC}, which also develops non-abelian lifted product codes with $20\%$ encoding rate, and Ref.~\cite{zheng2026canonical}, which constructs highly symmetric logical bases and efficient logic for abelian lifted product codes. We thank Yifan Hong for the flexibility in coordinating the paper release. We also thank Chris Cama\~{n}o, Margarita Davydova, Ryan Liu, Zachary Mann, Nathaniel Selub, Lucas Tecot, Victor Wei, and Han Zheng for many valuable and insightful discussions. We also thank Manuel Endres' group for helpful discussions and inspiring ideas on feasible hardware implementations. A.B. is supported by the Kortschak Scholars Program. N.M. is supported by the Air Force Office of Scientific Research under award number FA9550-23-F-0014. M.M., J.P., H.H. acknowledge support from the Institute for Quantum Information and Matter, an NSF Physics Frontiers Center (PHY-2317110). H.H. acknowledges support from the Broadcom Innovation Fund.
\end{acknowledgments}

\bibliographystyle{unsrt}
\bibliography{refs}

\clearpage
{\noindent \fontsize{22}{20} 
\textbf{Appendices}}
\tableofcontents

\appendix

\newpage
\section{Preliminaries}
\label{app:prelim}
We review the lifted product (LP) code construction~\cite{Panteleev_2022, breuckmann2021balanced}. We begin with hypergraph product (HGP) codes~\cite{Tillich_2014} and then show that the lifted product is a natural generalization of the HGP construction obtained by replacing scalar binary matrices with matrices whose entries lie in a group ring. We then recast both constructions from the perspective of homological algebra, where hypergraph product and lifted product codes can be described as tensor product chain complexes and their logical operators  associated with the homology of the chain complex. Throughout, we are interested in error-correcting codes on qubits and hence always work over $\mathbb{F}_2$.

\subsection{Hypergraph product codes}
    A classical linear code over $\F_2$ can be defined by a parity-check matrix $H \in \mathbb{F}_2^{r \times c}$, where $c$ is the number of bits and $r$ is the number of parity checks.  The codewords are precisely the vectors in $\ker H$.  A CSS code can be specified by two parity-check matrices $H_X$ and $H_Z$ that satisfy the commutation condition $H_X H_Z^T=0$.  The hypergraph product code construction gives a systematic way to construct such pairs of parity-check matrices from two classical codes.
    \begin{definition}[Hypergraph product code]
        Let $A \in \mathbb{F}_2^{r_1 \times c_1}$ and
        $B \in \mathbb{F}_2^{r_2 \times c_2}$.  The \emph{hypergraph product} (HGP)
        code~\cite{Tillich_2014} associated with $(A,B)$ is the CSS code on $n = c_1 c_2 + r_1 r_2$ physical qubits with parity-check matrices
        \begin{align}
            H_X & = \bigl[A \otimes I_{c_2} \;\big|\; I_{r_1} \otimes B^T\bigr], \label{eq:HX} \\
            H_Z & = \bigl[I_{c_1} \otimes B \;\big|\; A^T \otimes I_{r_2}\bigr]. \label{eq:HZ}
        \end{align}
    \end{definition}
    This definition ensures that the CSS commutation condition holds:
    \begin{align}
        H_X H_Z^T
        & = (A \otimes I_{c_2})(I_{c_1} \otimes B)^T
        + (I_{r_1} \otimes B^T)(A^T \otimes I_{r_2})^T \\
        & = A \otimes B^T + A \otimes B^T             \\
        & = 0.
    \end{align}
    where the last equality follows since we are working over $\mathbb{F}_2$.

    \begin{definition}[Kernel dimensions]
        For a binary matrix $M$, define
        \begin{equation}
            h(M) = \dim\ker M,
            \qquad
            h^\perp(M) = \dim\ker M^T .
        \end{equation}
    \end{definition}
    The hypergraph product code has
    \begin{equation}
        k = h(A)h(B) + h^\perp(A)h^\perp(B)
    \end{equation}
    logical qubits.  If $A$ and $B$ are full-row-rank parity-check matrices for classical codes of dimensions $h(A)$ and $h(B)$, then $h^\perp(A)=h^\perp(B)=0$ and the formula reduces to $k=h(A)h(B)$. We can also obtain a simple lower bound on $k$ in terms of the dimensions of $A$ and $B$ assuming they have more columns than rows:
    \begin{equation}
        k \ge (c_a - r_a)(c_b - r_b).
    \end{equation}

    \begin{example}
        The unrotated surface code is an example of a hypergraph product between two repetition codes.

        Let $\RepCode(d)$ denote the length-$d$ repetition code with parity-check matrix
        \[H_{\mathrm{rep}}(d)=
\begin{pmatrix}
1 & 1 & 0 & \cdots & 0\\
0 & 1 & 1 & \ddots & \vdots\\
\vdots & \ddots & \ddots & \ddots & 0\\
0 & \cdots & 0 & 1 & 1
\end{pmatrix}
\in \F_2^{(d-1)\times d}.
\]
Taking $A=B=H_{\mathrm{rep}}(d)$ in the hypergraph product construction gives the distance-$d$ unrotated surface code with parameters $ \llbracket d^2+(d-1)^2,\,1,\,d\rrbracket. $

\end{example}

LP codes generalize HGP codes by replacing binary matrix entries with elements of a group algebra. Thus, before introducing lifted product codes, we first fix some group conventions we use to describe our code families, and then review the group theory background needed for the construction.

\subsection{Group conventions and group algebras}

    Many of the groups used in our code constructions can be written as direct products or semidirect products. As the definition for a direct product is unambiguous, we just specify the convention used for semidirect product. 
    
\begin{definition}[Semidirect product]
\label{def:semidirect} 
Let $N$ and $H$ be two groups, let $\varphi: H \to \mathrm{Aut}(N)$ be a group homomorphism, where $\mathrm{Aut}(N)$ denotes the group of automorphisms of $N$ under composition, and let $\varphi_h = \varphi(h)$. Then, the semidirect product $N \rtimes_{\varphi} H$ is defined as the set of pairs $(n, h) \in N \times H$ with multiplication
\begin{equation}
    (n, h) \cdot (n', h') = \left(n \varphi_{h}(n'), h h' \right) \qquad \forall n, n' \in N, \quad h, h' \in H. 
    \label{eq:semidirect_mult}
\end{equation}
We abbreviate $(n,e_H)$ by $n$ and $(e_N,h)$ by $h$. Then 
$nh = (n,h)$, so every element is written uniquely with its $N$-factor
first. Written with this convention, a direct calculation with the multiplication law \cref{eq:semidirect_mult} gives
\[ hn = (e_N, h) (n, e_H) = (\varphi_h(n), h)= \varphi_h(n) h,\]
or equivalently,
\[ \varphi_h(n)=hnh^{-1}.\] 

\end{definition}

As the lifted product construction applies to any finite group $G$, we next review the necessary background on group algebras and their representations.

    \begin{definition}[Group ring and regular representations]\label{def:group-ring-and-regular-reps}
        Let $G$ be a finite group. The group algebra
        \begin{equation}
            \R
            =
            \mathbb{F}_2[G]
            =
            \left\{
            \sum_{g\in G} \alpha_g g \;:\; \alpha_g \in \mathbb{F}_2
            \right\}
        \end{equation}
        is the $\mathbb{F}_2$-vector space with basis $G$ and multiplication induced by the group law.  We write $\b{g}\in \mathbb{F}_2^{|G|}$ for the standard basis vector corresponding to $g\in G$.  By linearity, a ring element
        $a=\sum_g \alpha_g g\in \R$ corresponds to the binary vector
        \begin{equation}
            \b{a}=\sum_{g\in G}\alpha_g \b{g} .
        \end{equation}
        The \emph{left regular representation} and \emph{right regular representation} are first defined on group elements by
        \begin{equation}
            \Lrep{g}\b{h}=\b{gh},
            \qquad
            \Rrep{g}\b{h}=\b{h g^{-1}},
            \qquad g,h\in G,
        \end{equation}
        and then extended linearly to all $a\in \R$.  When the context is clear, we suppress the $\b{\cdot}$ notation and write, for example, $\Rrep{g}h=h g^{-1}$, meaning $\Rrep{g}\b{h}=\b{h g^{-1}}$.  In coordinates, for $a=\sum_g\alpha_g g$, these matrices are
        \begin{equation}
            \Lrep{a}_{y,x} = \alpha_{y x^{-1}},
            \qquad
            \Rrep{a}_{y,x} = \alpha_{y^{-1} x},
            \qquad x,y\in G.
        \end{equation}
    \end{definition}

    The left and right regular representations commute:
    \begin{equation}
        \Lrep{a}\Rrep{b} = \Rrep{b}\Lrep{a}
        \qquad \forall\, a,b\in \R.
    \end{equation}
    When we have a matrix $M$ defined over $\R$, we can expand it to a binary matrix by replacing each entry with its corresponding left or right regular representation. We will overload notation and denote the binary matrix obtained by replacing each entry of $M$ with its left (resp.~right) regular representation as $\Lrep{M}$ (resp. $\Rrep{M}$). Similarly, if we have a vector $v$ where each entry is a ring element, we write $\b{v}$ for the binary vector obtained by replacing each entry of $v$ with its corresponding binary vector. 
    \begin{definition}[Involution and conjugate transpose]\label{def:involution}
        For $a=\sum_g \alpha_g g \in \R$, its \emph{involution} is defined as
        \begin{equation}
            a^*=\sum_{g\in G}\alpha_g g^{-1}.
        \end{equation}
        For a matrix $M$ over $\R$, its \emph{conjugate transpose} $M^*$ is defined entrywise as
        \begin{equation}
            (M^*)_{ij} = (M_{ji})^* .
        \end{equation}
    \end{definition}
    This definition of conjugate transpose over the ring is compatible with the ordinary binary transpose after expansion:
    \begin{equation}
        \Lrep{M^*} = \Lrep{M}^T
        \quad \text{and} \quad
        \Rrep{M^*} = \Rrep{M}^T.
    \end{equation}
    To see this, first take a single group element $g\in G$.  The matrices $\Lrep{g}$ and
    $\Rrep{g}$ are permutation matrices, and transposing a permutation matrix gives the
    inverse permutation.  Hence $\Lrep{g}^T=\Lrep{g^{-1}}$ and
    $\Rrep{g}^T=\Rrep{g^{-1}}$.  Extending linearly gives, for $a=\sum_g\alpha_g g$,
    \begin{equation}
        \Lrep{a}^T
        =
        \sum_g \alpha_g \Lrep{g}^T
        =
        \sum_g \alpha_g \Lrep{g^{-1}}
        =
        \Lrep{a^*},
    \end{equation}
    and the same argument applies to $\Rrep{\cdot}$.  For matrices over $\R$, the ordinary binary transpose also swaps block positions, which is exactly the swap in $(M^*)_{ij}=(M_{ji})^*$.

    We now have all the necessary tools to define the lifted product code construction.

\subsection{Lifted product codes}
    Lifted product codes generalize hypergraph product codes through a richer internal structure afforded by the regular representations of a group algebra. Instead of taking the product of two ordinary binary classical codes, we allow the parity-check matrices to have entries in a group algebra, and each group algebra entry expands to a structured binary block via its left and right regular representations. The left and right regular representations are the two commuting actions that make this expansion compatible with the CSS commutation condition. 

    \begin{definition}[Lifted product code] \label{def:LP-code}
        Let $\R=\mathbb{F}_2[G]$ for a finite group $G$, and let
        $A\in \R^{r_1\times c_1}$ and $B\in \R^{r_2\times c_2}$.  The \emph{lifted product} code $\textrm{LP}(A,B)$ is the CSS code on 
        $n = |G|\,(c_1c_2+r_1r_2)$ physical qubits with parity-check matrices given by
        \begin{align}
            H_X & =
            \bigl[
                \Lrep{A\otimes I_{c_2}}
                \;\big|\;
                \Rrep{I_{r_1}\otimes B^*}
            \bigr], \label{eq:LP-HX} \\
            H_Z & =
            \bigl[
                \Rrep{I_{c_1}\otimes B}
                \;\big|\;
                \Lrep{A^*\otimes I_{r_2}}
                \bigr]. \label{eq:LP-HZ}
        \end{align}
    \end{definition}
    The lifted product construction can be viewed as a generalization of the hypergraph product construction \eqref{eq:HX}--\eqref{eq:HZ} where the tensor products and transposes are interpreted over the group algebra $\R=\mathbb{F}_2[G]$ and then expanded to binary matrices using the regular representations of $\R$. We can verify that the CSS commutation condition holds for LP codes.
    \begin{remark}[Relation to balanced product codes]
    \label{rem:balanced_product_code}
    Lifted product codes are a special case of the balanced product construction of~\cite{breuckmann2021balanced}. The tensor product over $\R=\mathbb{F}_2[G]$ in \eqref{eq:chain-complex-tensor-prod} is the balanced product of two classical chain complexes, in which the diagonal action of $G$ is quotiented out via the balancing relation. When the chain groups are free modules, as in this case, the balanced and lifted products coincide. 
    \end{remark}

    \begin{proposition}[CSS condition for LP codes]
        The parity-check matrices \(H_X\) and \(H_Z\) defined in
        \eqref{eq:LP-HX}--\eqref{eq:LP-HZ} satisfy
        \[
            H_XH_Z^T=0.
        \]
    \end{proposition}
    \begin{proof}
        From Definition \ref{def:LP-code}, we have
        \begin{equation}\label{eq:expand}
            H_X H_Z^T
            =
            \Lrep{A \otimes I_{c_2}}\,\Rrep{I_{c_1} \otimes B}^T
            \;+\;
            \Rrep{I_{r_1} \otimes B^*}\,\Lrep{A^* \otimes I_{r_2}}^T.
        \end{equation}
        Using the identities
        \[
            \Lrep{M}^T=\Lrep{M^*},
            \qquad
            \Rrep{M}^T=\Rrep{M^*}
        \]
        followed by the identities
        \[(M \otimes I)^* = M^* \otimes I, \qquad (I \otimes N)^* = I \otimes N^*,\]
        \eqref{eq:expand} simplifies to
        \begin{align}\label{eq:simplified}
            H_X H_Z^T
            & =
            \underbrace{\Lrep{A \otimes I_{c_2}}\,\Rrep{I_{c_1} \otimes B^*}}_{\displaystyle T_1}
            \;+\;
            \underbrace{\Rrep{I_{r_1} \otimes B^*}\,\Lrep{A \otimes I_{r_2}}}_{\displaystyle T_2} \\
            & = 0
        \end{align}
        where we have obtained the last equality by using the identity $(A \otimes B)(C \otimes D ) = AC \otimes BD$ and the fact that the left and right regular representations commute to show that $T_1 = T_2$ and hence $T_1 + T_2 = 0$ over $\mathbb{F}_2$.
        
    \end{proof}

    The number of physical qubits in the lifted product code is
    \begin{equation}
        n = |G|\,(c_1c_2+r_1r_2),
    \end{equation}
    since the two qubit blocks have sizes \(c_1c_2\) and \(r_1r_2\) over the
    group algebra, and each group-ring coordinate expands to \(|G|\) binary
    coordinates.  The number of logical qubits can always be computed from the
    expanded binary check matrices as
    \begin{equation}
        k
        =
        n-\operatorname{rank}_{\F_2}(H_X)-\operatorname{rank}_{\F_2}(H_Z)
    \end{equation}
    which leads to the elementary lower bound
    \begin{equation}
    \label{eq:k_of_LP}
        k
        \ge
        |G|\bigl(c_1c_2+r_1r_2-r_1c_2-c_1r_2\bigr)
        =
        |G|\,(c_1-r_1)(c_2-r_2).
    \end{equation}
In later sections when we prove distance bounds and construct a canonical logical basis for lifted product codes, it will be convenient to have a dictionary that lets us translate between working over the ring $\R$ and working over $\F_2$. We first fix our conventions. Over $\R$, a matrix acts on column vectors from the left, as in $Mu$, and on row vectors from the right, as in $vM$; matrix products over $\R$ multiply entries in the order written and are associative. In binary coordinates, we will always treat $\b{v}$ as a column vector, regardless of whether $v$ is a row or a column over $\R.$

Furthermore, because the blocks of the parity-check matrices in Definition~\ref{def:LP-code} are Kronecker products over $\R$, the vectors they act on are naturally indexed by pairs, and it will be convenient to view them as flattened matrices over $\R$. For $V\in\R^{m\times m'}$, let $\operatorname{vec}(V)\in\R^{mm'}$ denote the \emph{row-major} flattening of $V$, which lists the entries of the first row, then those of the second row, and so on:
\begin{equation}
\label{eq:vec-def}
    {\operatorname{vec}(V)}_{(i-1)m'+j}=V_{ij},
    \qquad
    i\in\{1,\dots,m\},\quad j\in\{1,\dots,m'\}.
\end{equation}
The following proposition collects all the compatibility properties between the ring and binary descriptions that we will need.
    \begin{proposition}\label{prop:ring_bin_compat}
        Let $M \in \R^{r \times c}$ and $N \in \R^{r' \times c'}$ be matrices over $\R$. The map $\bin : \R^c \to \F_2^{c|G|}$ is a bijection that satisfies the following properties:
        \begin{enumerate}
            \item[(i)] \textup{(Matrix-vector products)} For a column vector $u \in \R^{c}$ and a row vector $v \in \R^{1 \times c}$,
            \begin{equation}
            \label{eq:LR-actions}
                \Lrep{M}\,\b{u}=\b{Mu}
                \quad \text{and} \quad
                \Rrep{M}\,\b{v}=\b{vM^{*}}.
            \end{equation}
            \item[(ii)] \textup{(Kronecker products with an identity factor)} For $V\in\R^{c\times m}$ and $W\in\R^{m\times c'}$,
            \begin{equation}
            \label{eq:master-L}
                \Lrep{M\otimes I_m}\,\bvec{V}=\bvec{MV}
                \quad \text{and} \quad
                \Rrep{I_m\otimes N}\,\bvec{W}=\bvec{WN^{*}}.
            \end{equation}
        \end{enumerate}
    \end{proposition}
    \begin{proof}
        Using \cref{def:group-ring-and-regular-reps,def:involution}, it is straightforward to verify that $\bin$ is a bijection that satisfies the stated properties.
    \end{proof}
\subsection{Mitten codes and structured mitten codes}
\label{sec:app_structured-mittencodes}
    In Definition~\ref{def:mitten_processor_code}, we defined mitten codes, which are a class of $\LP$ codes with $1\times 2$ base matrices and non-abelian $G$ with further condition that $\Lrep{a_1}$ and $\Rrep{b_1}$ being full-rank. Here we single out a structured subfamily whose extra algebraic symmetry lets the syndrome-extraction (SE) cycle run ideally twice as fast with 2 AODs (Appendix~\ref{app:hookfree}). Recall the involution $a^*=\sum_g\alpha_g g^{-1}$ of \cref{def:group-ring-and-regular-reps}. A ring element is \emph{self-adjoint} if $a^*=a$. The structured mitten codes are mitten codes in which every base matrix entry is self-adjoint,
        \begin{equation}
            a_0^*=a_0,\quad a_1^*=a_1,\quad b_0^*=b_0,\quad b_1^*=b_1.
        \end{equation}
    For specific parameters of both mitten and structured mitten codes studied in this paper, see \cref{app:construction_data} and  \cref{tab:code_construction_full}.

\subsection{Homological algebra}

    We will reformulate the lifted product construction in the language of homological algebra. This way of constructing lifted product codes is akin to the machinery of the more general balanced product construction of Ref.~\cite{breuckmann2021balanced}. This perspective provides a systematic route to construct quantum codes from classical codes via tensor products of chain complexes and will later allow us to better understand the structure of the lifted product codes. 

    The starting observation is to note that a classical linear code can be described in terms of the homology of a two-term chain complex:
    \begin{equation}\label{eq:two-term}
        C_\bullet:
        \quad
        C_1 \xrightarrow{\partial = H} C_0,
    \end{equation}
    where $C_1=\mathbb{F}_2^c$ is the vector space of bits, $C_0=\mathbb{F}_2^r$ is the vector space of parity checks, and the boundary map $\partial:C_1\to C_0$ is given by matrix multiplication by the parity-check matrix $H \in \mathbb{F}_2^{r \times c}$.  The classical linear code with parity-check matrix $H$ is precisely the first homology group $H_1(\mathcal{C}(\partial))= \ker\partial = \ker H$, i.e.\ the set of vectors that satisfy all parity checks.  The zeroth homology group $H_0(\mathcal{C}(\partial))= C_0 / \im \partial = \mathbb{F}_2^r/\im H$ is nontrivial whenever the parity check matrix $H$ has linearly dependent rows.

    One can view a CSS code as two separate classical codes with parity-check matrices $H_X$ and $H_Z$ that must be compatible in the sense that they satisfy $H_X H_Z^T=0$.  The natural algebraic object that captures this compatibility is a three-term chain complex
    \begin{equation}
        C_2 \xrightarrow{\partial_2 = H_Z^T} C_1 \xrightarrow{\partial_1 = H_X} C_0 \label{eq:three-term}
    \end{equation}
    where $C_2$, $C_1$, and $C_0$ are $\F_2$-vector spaces whose basis elements label the $Z$ checks, qubits, and $X$ checks respectively.  Qubits are placed on $C_1$, and the parity-check matrices are $H_X = \partial_1$ and $H_Z = \partial_2^T$.
    The CSS commutation condition $H_X H_Z^T=0$ is then simply the statement that
    $\partial_1\partial_2=0$, i.e.\ that \eqref{eq:three-term} is indeed a chain complex.
    X-stabilizers are rows of $\partial_1$, while
    Z-stabilizers are columns of $\partial_2$ (equivalently, rows of $\partial_2^T$).
    A Z-type logical operator commutes with all the $X$-stabilizers but is not itself a $Z$-stabilizer. Such an operator can be represented as a vector $z\in C_1$ that is in the kernel of $\partial_1 = H_X$ but is not in the image of $\partial_2 = H_Z^T$. Hence Z-type logical operators are represented by nontrivial classes
    in the first homology group
    \begin{equation}
        H_1 = \ker\partial_1/\im\partial_2.
    \end{equation}
    Dually, X-type logical operators are represented by the first cohomology group
    $H^1=\ker\partial_2^T/\im\partial_1^T$.

    The question, then, is how to produce a three-term chain complex that corresponds to a CSS code from classical input data in the form of two-term chain complexes. There is a standard algebraic construction that does exactly this: the tensor product of chain complexes.
    
    \begin{definition}[Tensor product of chain complexes]
        Given two two-term chain complexes over $\F_2$-valued vector spaces,
    \begin{equation}
        C_\bullet^A:\; C_1^A \xrightarrow{\;\partial_A\;} C_0^A,
        \qquad
        C_\bullet^B:\; C_1^B \xrightarrow{\;\partial_B\;} C_0^B,
    \end{equation}
    their \emph{tensor-product complex} $C_\bullet^A\otimes C_\bullet^B$ is the three-term
    complex
    \begin{equation}\label{eq:F2-prod-complex}
        C_1^A\otimes C_1^B
        \;\xrightarrow{\;\partial_2\;}\;
        \bigl(C_1^A\otimes C_0^B\bigr)\oplus\bigl(C_0^A\otimes C_1^B\bigr)
        \;\xrightarrow{\;\partial_1\;}\;
        C_0^A\otimes C_0^B,
    \end{equation}
    with boundary maps
    \begin{equation}\label{eq:tensor-bdry}
        \partial_2
        =
        \begin{bmatrix}
            I\otimes\partial_B \\
            \partial_A\otimes I
        \end{bmatrix},
        \qquad
        \partial_1
        =
        \bigl[\,\partial_A\otimes I
            \;\big|\;
            I\otimes\partial_B\,\bigr].
    \end{equation}
    \end{definition}
    
    One can verify directly that $\partial_1\partial_2=0$:
    \begin{equation}
        \partial_1\partial_2
        =
        (\partial_A\otimes I)(I\otimes\partial_B)
        + (I\otimes\partial_B)(\partial_A\otimes I)
        = \partial_A\otimes\partial_B + \partial_A\otimes\partial_B
        = 0,
    \end{equation}
    so the tensor-product complex is indeed a chain complex, and the resulting CSS code automatically satisfies the commutation condition. To construct the hypergraph product code from classical codes with parity-check matrices $A$ and $B$, we take the tensor product of the chain complex $C_\bullet^A$ with the \emph{dual} of $C_\bullet^B$.

    \begin{definition}[Dual of a chain complex]\label{def:dual-F2}
        Let $C_\bullet : \F_2^{c} \xrightarrow{\;H\;} \F_2^{r}$ be a two-term chain complex, where the chain groups are equipped with their standard bases.  We define its \emph{dual complex} as
        \begin{equation}\label{eq:dual-F2}
            \hat{C}_\bullet : \F_2^{r} \xrightarrow{\;H^T\;} \F_2^{c}.
        \end{equation}
    \end{definition}

    The tensor product $C_\bullet^A \otimes \hat{C}_\bullet^B$ then has middle chain group
    $(\mathbb{F}_2^{c_1}\otimes\mathbb{F}_2^{c_2})\oplus(\mathbb{F}_2^{r_1}\otimes\mathbb{F}_2^{r_2})$,
    and the maps $\partial_1$ and $\partial_2^T$ reproduce \eqref{eq:HX}--\eqref{eq:HZ}.
    The dimension of the code follows from the K\"{u}nneth theorem.

    \begin{theorem}[K\"unneth formula \cite{weibel1994introduction}]
        Let $C_\bullet$ and $D_\bullet$ be chain complexes of finite-dimensional $\F_2$-vector
        spaces. Then the homology of their tensor product satisfies
        \begin{equation}
            H_n(C_\bullet \otimes D_\bullet)
            \cong
            \bigoplus_{i+j=n} H_i(C_\bullet)\otimes H_j(D_\bullet).
        \end{equation}
    \end{theorem}

    For the hypergraph product, we apply this theorem to $C_\bullet^A \otimes \hat{C}_\bullet^B$. Since the logical qubits are represented by the middle homology group, we obtain
    \begin{equation}
        H_1(C_\bullet^A \otimes \hat{C}_\bullet^B)
        \cong
        \bigl(H_1(C_\bullet^A)\otimes H_0(\hat{C}_\bullet^B)\bigr)
        \oplus
        \bigl(H_0(C_\bullet^A)\otimes H_1(\hat{C}_\bullet^B)\bigr).
    \end{equation}
    Taking dimensions gives the usual formula
    \begin{equation}
        k = h(A)h(B) + h^\perp(A)h^\perp(B),
    \end{equation}
    where
    \begin{equation}
        h(A) = \dim \ker A,
        \qquad
        h^\perp(A) = \dim \ker A^T,
    \end{equation}
    and similarly for $B$.

    To pass from the hypergraph product to the lifted product, we replace $\mathbb{F}_2$-vector spaces by modules over the ring $\R=\mathbb{F}_2[G]$. A module can be thought of as generalizing the notion of a vector space by allowing scalars to come from a ring instead of a field.

    \begin{definition}[$\R$-modules]
        Let $\R$ be a ring. A \emph{left $\R$-module} is an abelian group $M$ equipped with a scalar multiplication
        \[
            \R \times M \to M,
            \qquad
            (a,m)\mapsto am,
        \]
        satisfying the following axioms for all $a,b\in \R$ and all $m,n\in M$:
        \begin{align*}
            a(m+n) &= am+an, & &\text{(left distributivity over module addition)}\\
            (a+b)m &= am+bm, & &\text{(left distributivity over ring addition)}\\
            (ab)m &= a(bm),  & &\text{(associativity)}\\
            1m &= m.         & &\text{(identity)}
        \end{align*}
        A \emph{right $\R$-module} is defined similarly, except that scalar multiplication acts on the right:
        \[
            M \times \R \to M,
            \qquad
            (m,a)\mapsto ma.
        \]
        The corresponding axioms are
        \begin{align*}
            (m+n)a &= ma+na, & &\text{(right distributivity over module addition)}\\
            m(a+b) &= ma+mb, & &\text{(right distributivity over ring addition)}\\
            (ma)b &= m(ab),  & &\text{(associativity)}\\
            m1 &= m.         & &\text{(identity)}
        \end{align*}
        When $\R$ is commutative, left and right $\R$-modules are equivalent.
    \end{definition}

    \begin{remark}
        The distinction between left and right modules is not really about which side of $m$ we write the scalar. We could just as easily write a right action on the left, as $(a,m)\mapsto am$, and define the associativity axiom by $(ab)m = b(am)$ instead of $(ab)m = a(bm)$. This difference in the associativity axiom is the only true distinguishing feature between left and right modules. 
    \end{remark}

    \begin{example} The basic example is the free module $\R^n$. Its elements are vectors \[ v= \begin{pmatrix} r_1\\ \vdots\\ r_n \end{pmatrix}, \qquad r_i\in \R. \] When $\R^n$ is regarded as a left $\R$-module, we use the notation ${}_{\R}\R^n$  and scalars $a \in \R$ act entrywise on the left: \[ a v = \begin{pmatrix} a r_1\\ \vdots\\ a r_n \end{pmatrix}. \] When $\R^n$ is regarded as a right $\R$-module, we use the notation $\R^n_{\R}$ and scalars $a \in \R$ act entrywise on the right: \[ v a = \begin{pmatrix} r_1 a\\ \vdots\\ r_n a \end{pmatrix}.\] Sometimes it will also be convenient for us to view the elements of ${}_{\R}\R^n$ as row vectors. 
    \end{example}

We constructed the hypergraph product code by taking the tensor product of two
chain complexes defined over $\mathbb F_2$-vector spaces. This was defined
degree-by-degree from tensor products of the underlying vector spaces, as
in~\eqref{eq:F2-prod-complex}. For lifted product codes, we use the same
product-complex construction, but with vector spaces replaced by modules over
the group algebra $\R=\mathbb F_2[G]$.

To do this, we need the tensor product of modules over $\R$. Let $K$ be a right
$\R$-module and let $L$ be a left $\R$-module. Tensoring over $\R$ means that
scalars from $\R$ should act as scalars in the combined object, rather than as
data belonging separately to one factor or the other. Therefore, for
$k\in K$, $\ell\in L$, and $a\in\R$, multiplying $k$ by $a$ on the right should
be identified with multiplying $\ell$ by $a$ on the left. This is enforced by
the \emph{balancing relation}
\[
    (ka)\otimes \ell
    =
    k\otimes (a\ell).
\]
This relation is the module analogue of the scalar-linearity relation in the
tensor product of vector spaces that lets us move scalars across the tensor product. It only makes sense when the first module has a right $\R$-action and the second module has a left $\R$-action.

\begin{definition}[Tensor product of $\R$-modules]
    Let $K$ be a right $\R$-module and let $L$ be a left $\R$-module.
    The tensor product $K\otimes_{\R}L$ is the abelian group generated by
    formal symbols
    \[
        k\otimes \ell,
        \qquad k\in K,\ \ell\in L,
    \]
    subject to the relations
    \begin{align*}
        (k+k')\otimes \ell &= k\otimes \ell + k'\otimes \ell,\\
        k\otimes(\ell+\ell') &= k\otimes \ell + k\otimes \ell',\\
        (ka)\otimes \ell &= k\otimes(a\ell),
    \end{align*}
    for all $k,k'\in K$, $\ell,\ell'\in L$, and $a\in \R$.
    The symbols $k\otimes \ell$ are called \emph{pure tensors} and are said to generate $K\otimes_{\R}L$ since any element of $K\otimes_{\R} L$ can be written as a finite sum of pure tensors.
\end{definition}

Just as the tensor product is only well-defined between a right module and a left module, the tensor product of chain complexes is only well-defined between a complex of right modules and a complex of left modules.

\begin{definition}[Tensor product of a right-module complex and a left-module complex]
Let
\[
    K_\bullet
    =
    \left(
    \cdots
    \longrightarrow K_{p+1}
    \xrightarrow{\partial^{K}_{p+1}}
    K_p
    \xrightarrow{\partial^{K}_{p}}
    K_{p-1}
    \longrightarrow \cdots
    \right)
\]
be a chain complex of right $\R$-modules, and let
\[
    L_\bullet
    =
    \left(
    \cdots
    \longrightarrow L_{q+1}
    \xrightarrow{\partial^{L}_{q+1}}
    L_q
    \xrightarrow{\partial^{L}_{q}}
    L_{q-1}
    \longrightarrow \cdots
    \right)
\]
be a chain complex of left $\R$-modules. Their tensor product over
$\R$ is the chain complex of abelian groups
\[
    K_\bullet \otimes_{\R}L_\bullet
\]
whose degree-\(n\) term is
\[
    \left(K_\bullet \otimes_{\R}L_\bullet\right)_n
    =
    \bigoplus_{p+q=n} K_{p}\otimes_{\R}L_{q} .
\]
The boundary map
\[
    \partial_n:
    \left(K_\bullet \otimes_{\R}L_\bullet\right)_n
    \longrightarrow
    \left(K_\bullet \otimes_{\R}L_\bullet\right)_{n-1}
\]
is defined on pure tensors by
\begin{equation} \label{eq:module-tensor-bdry}
    \partial_n(k\otimes \ell)
    =
    \partial^{K}_{p}(k)\otimes \ell
    +
    (-1)^p k\otimes \partial^{L}_{q}(\ell),
    \qquad
    k\in K_p,\ \ell\in L_q,\ p+q=n.
\end{equation}
Since we will always be working over $\F_2$, we may neglect the $(-1)^p$ factor.
\end{definition}

\subsection{Lifted product codes: homological algebra viewpoint}

We now explain how the lifted product code construction arises from tensor products of $\R$-module chain complexes. Let $A\in \R^{r_1\times c_1}$ and
$B\in \R^{r_2\times c_2}$ be the classical base matrices associated with
our lifted product code. We associate to $A$ the two-term complex of right
$\R$-modules
\begin{equation}
    C_\bullet^A:
    \qquad
    C^A_1=\R^{c_1}
    \xrightarrow{\;\partial_A=A\;}
    C^A_0=\R^{r_1}.
\end{equation}
We view the elements of $C^A_1$ and $C^A_0$ as column vectors and the boundary map is defined by left multiplication by $A$. This is compatible with the right $\R$-module structure and defines a valid homomorphism between right $\R$-modules: for every
$x\in \R^{c_1}$ and $r\in \R$,
\begin{equation}
    \partial_A(xr)
    =
    A(xr)
    =
    (Ax)r
    =
    \partial_A(x)r.
\end{equation}

We may similarly associate a two-term complex to $B$. However, since the tensor product of chain complexes is only well-defined between a right-module complex and a left-module complex, we will instead consider the left-module complex given by the dual of $C_\bullet^B$: 
\begin{equation} \label{eq:Bcomplex}
    \hat{C}_\bullet^B:
    \qquad
    \hat{C}_1^B={}_{\mathcal R}\mathcal R^{r_2}
    \xrightarrow{\;\hat\partial_B = B^*\;}
    \hat{C}_0^B={}_{\mathcal R}\mathcal R^{c_2} 
\end{equation}
Here the matrix $B^*\in \R^{c_2\times r_2}$ defines a homomorphism of free left $\R$-modules via the coordinate formula 
\begin{equation} \label{eq;B*-action}
    (B^*y)_i = \sum_{j=1}^{r_2} y_j (B^*)_{ij}, \qquad y\in {}_{\R}\R^{r_2}.
\end{equation}
This produces a valid homomorphism since for every $y\in {}_{\R}\R^{r_2}$ and $r\in \R$, we have 
\begin{equation}
    (B^*(ry))_i
    =
    \sum_{j=1}^{r_2} (ry_j)(B^*)_{ij}
    =
    r\sum_{j=1}^{r_2} y_j(B^*)_{ij}
    =
    r(B^*y)_i.
\end{equation}

We can now form the tensor product complex
\begin{equation}
    Q_\bullet
    \coloneqq 
    C_\bullet^A\otimes_{\R}\hat C_\bullet^B
    \label{eq:chain-complex-tensor-prod}
\end{equation}
associated to the lifted product code $\textrm{LP}(A,B)$. Since both
$C_\bullet^A$ and $\hat C_\bullet^B$ are two-term complexes, $Q_\bullet$
is the three-term complex given as follows:
\begin{equation}
    Q_2
    \xrightarrow{\;\partial_2\;}
    Q_1
    \xrightarrow{\;\partial_1\;}
    Q_0.
\end{equation}
Its three nonzero chain groups are
\begin{align}
    Q_2
    &=
    C_1^A\otimes_{\R}\hat C_1^B
    =
    \R^{c_1}\otimes_{\R}\R^{r_2}
    \cong
    \R^{c_1r_2},
    \label{eq:LP-Q2}\\[1.5em]
    Q_1
    &=
    \left(C_1^A\otimes_{\R}\hat C_0^B\right)
    \oplus
    \left(C_0^A\otimes_{\R}\hat C_1^B\right)
    \nonumber\\
    &=
    \left(\R^{c_1}\otimes_{\R}\R^{c_2}\right)
    \oplus
    \left(\R^{r_1}\otimes_{\R}\R^{r_2}\right)
    \cong
    \R^{c_1c_2}\oplus \R^{r_1r_2},
    \label{eq:LP-Q1}\\[1.5em]
    Q_0
    &=
    C_0^A\otimes_{\R}\hat C_0^B
    =
    \R^{r_1}\otimes_{\R}\R^{c_2}
    \cong
    \R^{r_1c_2}.
    \label{eq:LP-Q0}
\end{align}
After expanding each copy of $\R$ as a binary vector space of dimension $|G|$, the middle term $Q_1$ has
\[
    |G|\,(c_1c_2+r_1r_2)
\]
binary coordinates. These correspond to the physical qubits of the lifted product code. 

We now compute the two boundary maps of $Q_\bullet$ which correspond to the $X$ and $Z$ parity check matrices. First consider
\[
    \partial_1:Q_1\longrightarrow Q_0.
\]
Using the definition of the tensor product boundary map
in~\eqref{eq:module-tensor-bdry}, for
$x\in C_1^A$, $z\in \hat C_0^B$, $u\in C_0^A$, and
$y\in \hat C_1^B$, we have
\begin{equation}
    \partial_1\bigl((x\otimes z)\oplus(u\otimes y)\bigr)
    =
    Ax\otimes z
    +
    u\otimes B^*y.
\end{equation}
The first term applies $A$ to the $C^A_1$ factor and leaves the $\hat C_0^B$ factor unchanged. This corresponds to the block matrix \[A\otimes I_{c_2}.\] Since $A$ acts by left multiplication on the group-ring entries, its binary
expansion is \[\Lrep{A\otimes I_{c_2}}.\]
The second term leaves the $C_0^A$ factor unchanged and applies the left-module map corresponding to $B^*$ to the $\hat C_1^B$ factor.
By the coordinate convention for left modules in~\eqref{eq;B*-action}, for $y \in C_1^B \cong \R^{r_2}$
\[
    (B^*y)_i
    =
    \sum_{j=1}^{r_2} y_j(B^*)_{ij}.
\]
This corresponds to usual matrix multiplication except with the entries of $B^*$ acting on the right. After binary expansion, this map is represented by $\Rrep{B^*}$. 
Including the unchanged $C_0^A$ factor, the second block of $\partial_1$ is
therefore
\[
    \Rrep{I_{r_1}\otimes B^*}.
\]
Hence the binary matrix representing $\partial_1$ is
\[
    H_X
    =
    \bigl[
        \Lrep{A\otimes I_{c_2}}
        \;\big|\;
        \Rrep{I_{r_1}\otimes B^*}
    \bigr],
\]
which agrees with the $X$-check matrix in Definition~\ref{def:LP-code}.

Next consider
\[
    \partial_2:Q_2\longrightarrow Q_1.
\]
With the ordering of $Q_1$ used in~\eqref{eq:LP-Q1}, namely
\[
    Q_1
    =
    \left(C_1^A\otimes_{\R}\hat C_0^B\right)
    \oplus
    \left(C_0^A\otimes_{\R}\hat C_1^B\right),
\]
the definition for the tensor product boundary map in~\eqref{eq:module-tensor-bdry} gives 
\[
    \partial_2(x\otimes y)
    =
    \left(
        x\otimes B^*y,\;
        Ax\otimes y
    \right)
\]
for $x \in C_1^A$ and $y\in \hat C_1^B$.
The first component applies the left-module map $y\mapsto B^*y$ to the
second factor, so its binary expansion is
\[
    \Rrep{I_{c_1}\otimes B^*}.
\]
The second component applies $A$ to the first factor, so its binary expansion
is
\[
    \Lrep{A\otimes I_{r_2}}.
\]
Hence, as a binary matrix, $\partial_2$ has the block form
\[
    \partial_2
    =
    \begin{bmatrix}
        \Rrep{I_{c_1}\otimes B^*}\\
        \Lrep{A\otimes I_{r_2}}
    \end{bmatrix}
\]
and so 
\[
    H_Z
    =
    \partial_2^T
    =
    \bigl[
        \Rrep{I_{c_1}\otimes B^*}^{\,T}
        \;\big|\;
        \Lrep{A\otimes I_{r_2}}^{\,T}
    \bigr]
    =  \bigl[
        \Rrep{I_{c_1}\otimes B}
        \;\big|\;
        \Lrep{A^*\otimes I_{r_2}}
    \bigr]
\]
which agrees with the $Z$-check matrix in Definition~\ref{def:LP-code}.
Therefore, the lifted product code associated with $(A,B)$ is exactly the
CSS code corresponding to the tensor product complex
\[
    Q_\bullet = C_\bullet^A\otimes_{\R}\hat C_\bullet^B:
    \qquad
    Q_2
    \xrightarrow{\;\partial_2\;}
    Q_1
    \xrightarrow{\;\partial_1\;}
    Q_0.
\]
The physical qubits correspond to $Q_1$ and the $X$ and $Z$ stabilizers correspond to $Q_0$ and $Q_2$ respectively.

\section{Canonical basis for lifted product codes}\label{app:canonical-basis}
In this appendix, we explain the general procedure to construct explicit low-weight representatives for the single-qubit $X$ and $Z$ logical operators of our codes. By this we mean that the representatives satisfy the canonical commutation relation $[X_i,Z_j] = 2X_iZ_j\delta_{ij}$ or equivalently $X_i \cdot Z_j = \delta_{ij}$ where the dot product denotes the binary symplectic inner product. We refer to this collection of representatives as a \emph{canonical basis} for the logical operators.

\subsection{Setup}
We begin by characterizing sufficient conditions for a collection of $X$-type and $Z$-type operators to form a canonical basis and then define the square-invertibility condition which is needed for our construction of a canonical basis for lifted product codes.
\subsubsection{Canonical logical bases}

We begin by establishing a simple lemma that provides a set of sufficient conditions for identifying a collection of $X$-type and $Z$-type operators as a canonical basis for a CSS code. 
\begin{lemma}\label{lem:symplectic-pairing-logicals}
Let $H_X,H_Z$ be the parity-check matrices of a CSS code with $k$ logical qubits. Identify $X$-type and $Z$-type Pauli operators with their binary support vectors in $\mathbb F_2^n$. Suppose we have a collection $\{X_i\}_{i=1}^k$ and $\{Z_i\}_{i=1}^k$ of such operators satisfying
\begin{equation} 
    X_i \in \ker H_Z,
    \qquad
    Z_i \in \ker H_X  \quad \text{for every } i=1,\ldots,k,
\end{equation}
and 
\begin{equation}\label{eq:canonical-commutation}
    X_i \cdot Z_j = \delta_{ij}
\end{equation}
where the dot denotes the $\F_2$ inner product. Then $X_i$ and $Z_i$ are single-qubit $X$ and $Z$ logical operators for the $i$-th logical qubit.
\end{lemma}

\begin{proof}
For a CSS code, an $X$-type operator commutes with all $Z$-stabilizers
precisely when its binary support lies in $\ker H_Z$.
To show that $X_i$ is a nontrivial logical operator, it remains to show
that it is not an $X$-stabilizer, i.e. that
$X_i\notin \operatorname{row}(H_X)$.

Suppose, for contradiction, that $X_i\in \operatorname{row}(H_X)$. Let
$h_1,\ldots,h_m$ be the rows of $H_X$. Then there exist coefficients
$c_1,\ldots,c_m\in \F_2$ such that
\[
    X_i = \sum_{a=1}^m c_a h_a .
\]
Since $Z_j\in \ker H_X$, we have 
\[
    h_a \cdot Z_j = 0
    \qquad
    \text{for every } a=1,\ldots,m.
\]
Hence we have 
\[
    X_i \cdot Z_j
    =
    \sum_{a=1}^m c_a (h_a\cdot Z_j)
    =
    0
\]
for all $Z_j$. This contradicts the assumption that $X_i\cdot Z_j=\delta_{ij}$, so we conclude that $X_i\notin \operatorname{row}(H_X)$, and hence $X_i$ is a nontrivial logical operator. A similar argument shows that all the $Z_i$ are nontrivial logical operators as well. Since they satisfy the canonical commutation relation \eqref{eq:canonical-commutation}, they are the single-qubit $X$ and $Z$ logical operators for the $k$ logical qubits.
\end{proof}

We take the most general definition of a \emph{canonical logical basis} to be a collection of $X$-type and $Z$-type representatives satisfying the conditions of Lemma~\ref{lem:symplectic-pairing-logicals}. The canonical basis we will construct for non-abelian lifted product codes enjoys three additional properties beyond this pairing: its representatives are built directly from classical codewords of the base matrices, the logical operators within a basis are related by group action, and the supports of conjugate pairs intersect on exactly one data qubit.

\subsubsection{Square-invertibility condition}
Let $\R=\F_2[G]$, and let
\[
    A\in \R^{r_1\times c_1},
    \qquad
    B\in \R^{r_2\times c_2},
\]
with \(r_1<c_1\) and \(r_2<c_2\), be the classical base matrices defining the
lifted product code \(\textrm{LP}(A,B)\). The following definition captures the assumptions we will make about the structure of $A$ and $B$.

\begin{definition}[Square invertibility condition]
\label{def:square_inv_condition}
Let \(M\in \R^{r\times c}\) with \(r<c\). We say that \(M\) satisfies the
\emph{square invertibility condition} with respect to a chosen binary
expansion of \(\R\) if, after possibly reordering its columns, \(M\) can be
written as
\[
    M=
    \bigl[
        M_{\mathrm{free}}
        \mid
        M_{\mathrm{pivot}}
    \bigr],
\]
where
\[
    M_{\mathrm{free}}\in \R^{r\times (c-r)},
    \qquad
    M_{\mathrm{pivot}}\in \R^{r\times r},
\]
and the binary expansion of $M_{\mathrm{pivot}}$ is an invertible matrix
over $\F_2$. We call the first $c-r$ columns the \emph{free columns} and
the last $r$ columns the \emph{pivot columns}.
\end{definition}

After possibly reordering columns, we can fix
decompositions
\begin{equation}\label{eq:ABdecomp}
    A=
    \bigl[
        A_{\mathrm{free}}
        \mid
        A_{\mathrm{pivot}}
    \bigr],
    \qquad
    B=
    \bigl[
        B_{\mathrm{free}}
        \mid
        B_{\mathrm{pivot}}
    \bigr],
\end{equation}
where
\[
    A_{\mathrm{pivot}}\in \R^{r_1\times r_1},
    \qquad
    B_{\mathrm{pivot}}\in \R^{r_2\times r_2},
\]
and $\Lrep{A_{\mathrm{pivot}}}$ and $\Rrep{B_{\mathrm{pivot}}}$ are invertible binary matrices.
Letting
\[
    f_A\coloneqq c_1-r_1,
    \qquad
    f_B\coloneqq c_2-r_2,
\]
the decompositions in~\eqref{eq:ABdecomp} induce coordinate splittings
\[
    \R^{c_1}
    =
    \R^{f_A}\oplus \R^{r_1},
    \qquad
    \R^{c_2}
    =
    \R^{f_B}\oplus \R^{r_2}.
\]
We refer to the first summand as the free coordinates and to the second
summand as the pivot coordinates. Throughout this section, we will assume that $A$ satisfies the square invertibility condition (Definition~\ref{def:square_inv_condition}) with respect to the left regular representation, and that $B$ satisfies the square invertibility condition with respect to the right regular representation.

\subsection{Main result}

We now state the main result of this section: an explicit canonical logical basis for the lifted product code $\LP(A,B)$. The construction of this basis exploits the product structure of the code; each logical representative is derived from a codeword of one of the classical base codes. Logical $X$ operators arise from codewords of $B$ and logical $Z$ operators arise from codewords of $A$. The key step is choosing suitable bases for the two classical codes.

The square invertibility condition (Definition~\ref{def:square_inv_condition}) ensures that $A$ and $B$ admit bases in which each basis codeword is supported on exactly one of the free coordinates, a distinct one for each codeword. After padding with zeros, these classical basis vectors become logical $X$ and $Z$ representatives that pair up canonically -- within each pair, the supports of $X$ and $Z$ overlap on exactly one physical qubit (hence the operators anticommute), while operators from different pairs have disjoint support (and therefore commute).

The basis is built from distinguished codewords of the two classical factor codes, which we now define. For $1\le \alpha\le f_A$ and $g\in G$, let $\epsilon_A^{\alpha}(g)\in \R^{f_A}$ denote the vector with $g$ in the $\alpha$-th coordinate and zero elsewhere. Similarly, for $1\le \beta\le f_B$, let $\epsilon_B^{\beta}(g)\in \R^{f_B}$ denote the vector with $g$ in the $\beta$-th coordinate and zero elsewhere. We first define the kernel vectors. Let
\[
    v_A^{\alpha,g}\in \R^{c_1} \quad \left(\text{resp. } v_B^{\beta,g}\in \R^{c_2}\right)
\]
be the unique element of $\ker \Lrep{A}$ (resp. $\ker \Rrep{B}$) whose free part is $\epsilon_A^{\alpha}(g)$ (resp. $\epsilon_B^{\beta}(g)$); its pivot part is uniquely determined because $\Lrep{A_{\mathrm{pivot}}}$ (resp. $\Rrep{B_{\mathrm{pivot}}}$) is invertible (Lemma~\ref{lem:kernel-cokernel-normal-forms}). Likewise, let
\[
    w_A^{\alpha,g}\in \R^{c_1} \quad \left(\text{resp. } w_B^{\beta,g}\in \R^{c_2}\right)
\]
be the vector whose free part is $\epsilon_A^{\alpha}(g)$ (resp. $\epsilon_B^{\beta}(g)$) and whose pivot part is zero; it is the unique representative of its class in $\coker \Lrep{A^*}$ (resp. $\coker \Rrep{B^*}$) supported only on the free coordinates (Lemma~\ref{lem:kernel-cokernel-normal-forms}). The canonical logical basis for $\LP(A,B)$ is characterized by the following theorem.

\begin{theorem}[Canonical basis of logical operators]\label{thm:canonical-logicals}
Let $A\in \R^{r_1\times c_1}$ and $B\in \R^{r_2\times c_2}$ satisfy the square invertibility condition (Definition~\ref{def:square_inv_condition}) with respect to the left and right regular representations, respectively. For $1\le \alpha\le f_A$, $1\le \beta\le f_B$, and $g\in G$, define
\begin{equation}\label{eq:LP-logical-Z}
    Z_{\alpha,\beta}^{g}
    \coloneqq
    \left(
        \b{v_A^{\alpha,g}\otimes_{\R}w_B^{\beta,e_G}},\,0
    \right),
\end{equation}
and
\begin{equation}\label{eq:LP-logical-X}
    X_{\alpha,\beta}^{g}
    \coloneqq
    \left(
        \b{w_A^{\alpha,g}\otimes_{\R}v_B^{\beta,e_G}},\,0
    \right).
\end{equation}
These operators form a canonical logical basis for the $k=f_A f_B|G|$ logical qubits of $\LP(A,B)$.
\end{theorem}

\begin{remark}
\label{rem:group_logic}
    For our mitten codes we have $f_A = f_B = 1$, so we may drop the $\alpha$ and $\beta$ subscripts and simply label the $X$ and $Z$ logical operators by the group elements of $G$. The logical operators of one basis are all related to each other via group action. For example, for $X$, we have 
    \begin{equation}
    \begin{aligned}
        X^g
        &= \left(
            \b{w_A^{g} \otimes_{\R} v_B^{e}},
            0
        \right) \\
        &= \left(
            \b{g\, w_A^{e} \otimes_{\R} v_B^{e}},
            0
        \right) \\
        &= g \cdot X^e .
    \end{aligned}
    \end{equation}
    where the action $g \cdot$ is defined via the last equality. A similar relationship holds for the $Z$ logicals, with the group element acting by right multiplication instead: $Z^g = Z^e \cdot g$. 
\end{remark}

The rest of this section is devoted to proving Theorem~\ref{thm:canonical-logicals}. We first introduce the four binary vector spaces from which the logical
operators will be constructed. For the remainder of this section, we will identify
$\R^m$ with its binary expansion $\F_2^{|G|m}$.
Define the kernels 
\begin{equation}\label{eq:ker_space_V}
    V_A
    \coloneqq 
    \ker \Lrep{A}
    \subseteq \F_2^{|G|c_1},
    \qquad
    V_B
    \coloneqq 
    \ker \Rrep{B}
    \subseteq \F_2^{|G|c_2}
\end{equation}
and the cokernels
\begin{equation}
\label{eq:coker_space_W}
    W_A
    \coloneqq 
    \coker \Lrep{A^*}
    =
    \F_2^{|G|c_1}/\im \Lrep{A^*},
    \qquad
    W_B
    \coloneqq 
    \coker \Rrep{B^*}
    =
    \F_2^{|G|c_2}/\im \Rrep{B^*}.
\end{equation}
Under the square invertibility condition, the elements of these four spaces can be expressed in a particularly simple form. 
\begin{lemma}
\label{lem:kernel-cokernel-normal-forms}
Assume \(A\) and \(B\) satisfy the square invertibility condition. Then:
\begin{enumerate}
    \item every element of $V_A$ and $V_B$ is uniquely determined by its free coordinates;

    \item every equivalence class in
    $W_A$ and $W_B$ has a unique representative
    supported only on the free coordinates.
\end{enumerate}
\end{lemma}

\begin{proof}
We prove the two statements for $A$. The proof for $B$ is identical, with left representations replaced by right representations.

First, consider the kernel. Using the decomposition
\[
    \F_2^{|G|c_1}
    =
    \F_2^{|G|f_A}
    \oplus
    \F_2^{|G|r_1}
\]
into free and pivot coordinates, write
\[
    v=
    \begin{pmatrix}
        x\\ y
    \end{pmatrix},
    \qquad
    x\in \F_2^{|G|f_A},
    \qquad
    y\in \F_2^{|G|r_1}.
\]
Then $v\in V_A=\ker\Lrep{A}$ means
\[
    \Lrep{A_{\mathrm{free}}}x
    +
    \Lrep{A_{\mathrm{pivot}}}y
    =
    0.
\]
Since $\Lrep{A_{\mathrm{pivot}}}$ is invertible over $\F_2$, the pivot
part $y$ is uniquely determined by the free part $x$. Hence, every element
of $V_A$ is uniquely determined by its free coordinates.

Next, consider the cokernel
\[
    W_A
    =
    \F_2^{|G|c_1}/\im\Lrep{A^*}.
\]
Let
\[
    w=
    \begin{pmatrix}
        x\\ y
    \end{pmatrix}
    \in \F_2^{|G|c_1},
    \qquad
    x\in \F_2^{|G|f_A},
    \qquad
    y\in \F_2^{|G|r_1}.
\]
We want to replace $w$ by an equivalent representative whose pivot part is
zero. Since representatives in the same equivalence class may differ by an element of
$\im\Lrep{A^*}$, we may add a vector of the form
\[
    \Lrep{A^*}u,
    \qquad
    u\in \F_2^{|G|r_1}.
\]
Hence,
\[
    w+\Lrep{A^*}u
    =
    \begin{pmatrix}
        x+\Lrep{A_{\mathrm{free}}^*}u\\
        y+\Lrep{A_{\mathrm{pivot}}^*}u
    \end{pmatrix}.
\]
Because $\Lrep{A_{\mathrm{pivot}}}$ is invertible, so is
$\Lrep{A_{\mathrm{pivot}}^*}$, and so there is a unique
$u\in \F_2^{|G|r_1}$ such that
\[
    y+\Lrep{A_{\mathrm{pivot}}^*}u=0.
\]
For this choice of $u$, the representative
$w+\Lrep{A^*}u$ is supported only on the free coordinates. Hence, every equivalence class
in $W_A$ has a free-supported representative.

It remains to show this representative is unique. Suppose two free-supported vectors belong to the same equivalence class in $W_A$. Their difference is again free-supported and lies
in $\im\Lrep{A^*}$, so is equal to $\Lrep{A^*}u$ for some
$u\in \F_2^{|G|r_1}$. Looking at the pivot coordinates gives
\[
    \Lrep{A_{\mathrm{pivot}}^*}u=0
\]
which implies $u=0$ since $\Lrep{A_{\mathrm{pivot}}^*}$ is invertible. This proves uniqueness.
\end{proof}

By Lemma~\ref{lem:kernel-cokernel-normal-forms}, the vectors $v_A^{\alpha,g}$, $w_A^{\alpha,g}$, $v_B^{\beta,g}$, and $w_B^{\beta,g}$ defined before Theorem~\ref{thm:canonical-logicals} are well defined. Moreover, the binary expansions $\b{v_A^{\alpha,g}}$ and $\b{v_B^{\beta,g}}$ form bases of $V_A$ and $V_B$, and the equivalence classes of $\b{w_A^{\alpha,g}}$ and $\b{w_B^{\beta,g}}$ form bases of $W_A$ and $W_B$, respectively. The logical operators of Theorem~\ref{thm:canonical-logicals} are tensor products of these vectors, lifted to their binary representation.
Since we are taking tensor products over $\R$, we recall that they satisfy the balancing relation
\[
    ug\otimes_{\R}v
    =   u\otimes_{\R}gv \quad \text{for all } u\in \R^{c_1}, v\in \R^{c_2}, g\in G
\]
by construction.
Hence, a tensor product of the form $v_A^{\alpha,g}\otimes_{\R}w_B^{\beta,h}$ is uniquely specified by the group element $gh$ and so we will work with the convention that the group label is always placed on the $A$-side and the identity element $e_G$ is placed on the $B$-side. We also adopt the convention that $A\otimes I_{c_2}$ acts on the first tensor factor by the left action, $(A\otimes I_{c_2})(u\otimes v)=(Au)\otimes v$, while $I_{c_1}\otimes B$ acts on the second factor by the right action, $(I_{c_1}\otimes B)(u\otimes v)=u\otimes (vB)$.

We can now prove Theorem~\ref{thm:canonical-logicals}.

\begin{proof}[Proof of Theorem~\ref{thm:canonical-logicals}]
Our strategy will be to show that the $X$-type and $Z$-type representatives defined above satisfy the hypotheses of Lemma~\ref{lem:symplectic-pairing-logicals}.
By Definition~\ref{def:LP-code}, the physical qubits of
\(\mathrm{LP}(A,B)\) are indexed by the binary expansion of
\[
    \left(\R^{c_1}\otimesR \R^{c_2}\right)
    \oplus
    \left(\R^{r_1}\otimesR \R^{r_2}\right).
\]
The representatives defined in Eqs.~\eqref{eq:LP-logical-Z} and~\eqref{eq:LP-logical-X} have support only on the first block. We first show that the $Z$-type representatives commute with all $X$-stabilizers and the $X$-type representatives commute with all $Z$-stabilizers.
Using the block form of \(H_X\) from~\eqref{eq:LP-HX}, we get
\[
    H_X
    \left(
        \b{v_A^{\alpha,g}\otimesR w_B^{\beta,e_G}},\,0_{r_1r_2|G|}
    \right)
    =
    \Lrep{A\otimes I_{c_2}}\,
    \b{v_A^{\alpha,g}\otimesR w_B^{\beta,e_G}}
    =
    \b{Av_A^{\alpha,g}\otimesR w_B^{\beta,e_G}}
    =
    0,
\]
because \(\b{v_A^{\alpha,g}}\in V_A=\ker\Lrep{A}\). Hence,
$ Z_{\alpha,\beta}^{g}\in \ker H_X.$ Similarly, using the block form of $H_Z$ from~\eqref{eq:LP-HZ}, we get
\[
    H_Z
    \left(
        \b{w_A^{\alpha,g}\otimesR v_B^{\beta,e_G}},\,0_{r_1r_2|G|}
    \right)
    =
    \Rrep{I_{c_1}\otimes B}\,
    \b{w_A^{\alpha,g}\otimesR v_B^{\beta,e_G}}
    =
    \b{w_A^{\alpha,g}\otimesR (v_B^{\beta,e_G}B)}
    =
    0,
\]
because \(\b{v_B^{\beta,e_G}}\in V_B=\ker\Rrep{B}\). Hence, $X_{\alpha,\beta}^{g}\in \ker H_Z.$

It remains to compute the symplectic pairing. The free part of $v_A^{\alpha,g}$ is
$\epsilon_A^{\alpha}(g)$, while $w_A^{\alpha',g'}$ is supported only on
the free coordinates with free part $\epsilon_A^{\alpha'}(g')$. Therefore,
\[
    \b{v_A^{\alpha,g}} \cdot \b{w_A^{\alpha',g'}} 
    =
    \delta_{\alpha,\alpha'}\delta_{g,g'}.
\]
Similarly, the free part of $v_B^{\beta',e_G}$ is
$\epsilon_B^{\beta'}(e_G)$, while $w_B^{\beta,e_G}$ is supported only on
the free coordinates with free part $\epsilon_B^{\beta}(e_G)$. Hence,
\[
    \b{w_B^{\beta,e_G}} \cdot \b{v_B^{\beta',e_G}} 
    =
    \delta_{\beta,\beta'}.
\]
Consequently,
\[
     X_{\alpha',\beta'}^{g'}
    \cdot
    Z_{\alpha,\beta}^{g}
    =
    \delta_{\alpha,\alpha'}
    \delta_{\beta,\beta'}
    \delta_{g,g'}.
\]
Hence, the hypotheses of Lemma~\ref{lem:symplectic-pairing-logicals} are satisfied and so the operators
\[
    \left\{
        X_{\alpha,\beta}^{g},
        Z_{\alpha,\beta}^{g}
    \right\}_{1\le \alpha\le f_A,\;1\le \beta\le f_B,\;g\in G}
\]
form a canonical basis for the $k = f_A f_B |G|$ single-qubit $X$ and $Z$  logical operators
of $\mathrm{LP}(A,B)$.
\end{proof}

\subsection{Aside on connection to  K\"unneth formula}

One might notice that our canonical logical basis and its product structure is suspiciously reminiscent of the K\"{u}nneth formula. This is not a coincidence, but rather a direct consequence of the square invertibility condition. We now explain this connection.

For a general group algebra $\R = \F_2[G]$, the homology of the tensor product
complex $C_\bullet^A \otimes_\R \hat{C}_\bullet^B$ usually does not reduce to a direct sum of the tensor products of the homology of the two factors like in the standard K\"{u}nneth formula. It is instead computed by the K\"unneth spectral 
sequence~\cite{weibel1994introduction}:
\begin{align}
    E^2_{p,q}
    =
    \bigoplus_{i+j=q}
    \Tor_p^{\R}\!\left( H_i(C_\bullet^A),\, H_j(\hat{C}_\bullet^B) \right)
    \;\Longrightarrow\;
    H_{p+q}\!\left( C_\bullet^A \otimes_{\R} \hat{C}_\bullet^B \right).
\end{align}
The symbol $\Longrightarrow$ indicates that $E^2$ is only the starting page; successive pages are obtained by taking homology with respect to differentials:
\[
d^r : E^r_{p,q} \longrightarrow E^r_{p-r,q+r-1},
\qquad
E^{r+1}_{p,q}
=
\frac{\ker\!\left(d^r : E^r_{p,q} \to E^r_{p-r,q+r-1}\right)}
{\operatorname{im}\!\left(d^r : E^r_{p+r,q-r+1} \to E^r_{p,q}\right)}.
\]
The limiting page $E^\infty$ then assembles along the diagonals $p+q=n$ into $H_n(C_\bullet^A \otimes_{\R} \hat{C}_\bullet^B)$. Since
$\Tor_0^\R(M,N) = M \otimes_\R N$, the $p=0$ column contains the ordinary
tensor-product terms of the form $H_i(C_\bullet^A) \otimes_\R H_j(\hat{C}_\bullet^B)$.
We will show that the square invertibility condition makes the spectral sequence collapse immediately at the $E^2$ page so that these ordinary tensor-product terms are the only ones that contribute to the homology of the product complex. 

First, we identify the homologies of the two factors. Recall that for the complex 
\[
    C_\bullet^A:\quad \R^{c_1} \xrightarrow{\;A\;} \R^{r_1}
\]
the square invertibility condition tells us that $A$ can be decomposed as 
$A = [\,A_{\mathrm{free}} \mid A_{\mathrm{pivot}}\,]$ where $\Lrep{A_{\mathrm{pivot}}}$ is invertible over
$\F_2$. Hence, the equation \[A_{\mathrm{free}} x + A_{\mathrm{pivot}} y = 0\] can be solved uniquely for $y$ given any choice of $x$, yielding
\[
    H_1(C_\bullet^A) = \ker A \cong \R^{f_A},
    \qquad
    H_0(C_\bullet^A) = \coker A = 0.
\]
Similarly, since the square invertibility condition also holds for $B$, for the dual complex $\hat{C}_\bullet^B$ defined in~\eqref{eq:Bcomplex}
we have 
\[
    H_1(\hat{C}_\bullet^B) = 0,
    \qquad
    H_0(\hat{C}_\bullet^B) \cong \R^{f_B}.
\]
Note that the two nonzero homology modules $H_1(C_\bullet^A) $ and $H_0(\hat{C}_\bullet^B)$ are free.

Plugging these into the spectral sequence, only one entry survives. Every
summand involving $H_0(C_\bullet^A)$ or $H_1(\hat{C}_\bullet^B)$ vanishes,
and freeness of $H_1(C_\bullet^A)$ and $H_0(\hat{C}_\bullet^B)$ forces all higher Tor to vanish as
well:
\[
    \Tor_p^{\R}\!\left(\R^{f_A},\, \R^{f_B}\right) = 0
    \qquad\text{for all } p \geq 1.
\]
Hence, the unique nonzero entry on the $E^2$ page is
\[
    E^2_{0,1}
    \;=\;
    H_1(C_\bullet^A) \otimes_\R H_0(\hat{C}_\bullet^B)
    \;\cong\;
    \R^{f_A f_B},
\]
so schematically the $E^2$ page looks like 
\[
\begin{array}{r|cccc}
q=2 & 0 & 0 & 0 & \cdots \\
q=1 & H_1(C_\bullet^A) \otimes_\R H_0(\hat{C}_\bullet^B) & 0 & 0 & \cdots \\
q=0 & 0 & 0 & 0 & \cdots \\ \hline
    & p=0 & p=1 & p=2 & \cdots 
\end{array} \hspace{2em} .
\]
The differentials touching $E^r_{0,1}$ have the form
\[
    E^r_{r,\,2-r}
    \;\xrightarrow{\;d^r\;}\;
    E^r_{0,1}
    \;\xrightarrow{\;d^r\;}\;
    E^r_{-r,\,r}.
\]
The entry on the right sits in column $p = -r < 0$ and is zero by convention.
The entry on the left is also zero on the $E^2$-page: at $r=2$ it occupies
the lattice point $(2,0)$, empty in the diagram above; for $r \geq 3$ it has $q = 2-r < 0$. With both neighbors zero, the outgoing $d^r$ has kernel all of $E^r_{0,1}$ and the incoming $d^r$ has image zero, so the recursion
\[
    E^{r+1}_{0,1}
    \;=\;
    \frac{\ker\!\left(d^r\colon E^r_{0,1} \to 0\right)}
         {\im\!\left(d^r\colon 0 \to E^r_{0,1}\right)}
    \;=\;
    \frac{E^r_{0,1}}{0}
    \;=\;
    E^r_{0,1}
\]
holds for every $r \geq 2$. The entry persists unchanged to $E^\infty$, and
we conclude
\[
    H_1\!\left( C_\bullet^A \otimes_\R \hat{C}_\bullet^B \right)
    = H_1(C_\bullet^A) \otimes_\R H_0(\hat{C}_\bullet^B).
\]
Hence, we may construct the $Z$ logical operators of our code by constructing a basis for $H_1(C_\bullet^A)$ and $H_0(\hat{C}_\bullet^B)$. This is exactly what we did in the construction of our canonical basis. The construction of the $X$ logical operators follows similarly by looking at the first cohomology group of the tensor product complex.

\section{Universal fault-tolerant quantum instruction sets}\label{app:metrics}

In this appendix, we provide a formal definition of the notion of a universal fault-tolerant quantum instruction set. We then more formally define the three instruction sets we considered in the main text and explain how they enable universal quantum computation. 

Consider a qLDPC processor based on a CSS code with parameters $\llbracket n,k,d \rrbracket$ together with a fixed decoder. An \emph{instruction} is a logical operation on the processor's logical qubits together with a physical gadget that realizes the operation.

\begin{definition}[Universal fault-tolerant quantum instruction set (formal version of Definition~\ref{def:quis})]
\label{def:quis_formal}
Consider a qLDPC processor with $k$ logical qubits. A set of fault-tolerant instructions $\quis$ is a \emph{universal fault-tolerant quantum instruction set} for the processor if every unitary $U$ on the $k$ logical qubits can be realized to any precision $\epsilon>0$ by a finite sequence of instructions from $\quis$~\cite{A_Yu_Kitaev_1997}, where
\begin{enumerate}
    \item each instruction may be chosen adaptively, conditioned on the classical record of earlier measurement outcomes, and
    \item the combined action of the sequence is required to equal $U$ only up to a Pauli frame determined by that record.
\end{enumerate}
\end{definition}

An example of a universal instruction set is the collection of two-qubit CNOTs together with the single-qubit Hadamard and $T$
gates~\cite{boykin1999universalfaulttolerantquantumcomputing}. Universality, however, does not require unitary gadgets; measurements and state injection suffice, as the constructions below show.

\begin{definition}[Basic instruction set $\bquis$]
\label{def:bquis_formal}
The basic instruction set consists of (i) the single-qubit logical Pauli measurements $\{\bar{X}_i,\bar{Y}_i,\bar{Z}_i\}_{i=1}^{k}$, (ii) the weight-two logical Pauli-product measurements $\{\bar{X}_i\bar{X}_j,\bar{Z}_i\bar{Z}_j\}_{1\leq i<j\leq k}$, and (iii) noisy $T$-state injection i.e. the preparation of a designated logical qubit in the state $T\ket{+}$ up to a known Pauli correction, with an injection error rate set by the physical noise.
\end{definition}

\noindent
The Clifford group together with the $T$ gate forms a universal gate set~\cite{boykin1999universalfaulttolerantquantumcomputing,Bravyi_2005}. Moreover, any Clifford+$T$ circuit can be compiled into a Pauli-based computation consisting of stabilizer-state preparations, adaptive logical Pauli-product measurements, and the consumption of logical $\ket{T}$ states~\cite{Bravyi_2016}.
Noisy logical $\ket{T}$ states are injected using instruction (iii) and are then distilled. Provided the input error lies below the distillation threshold and the logical stabilizer operations are sufficiently reliable, magic-state distillation can reduce the output error below any target $\epsilon>0$ with overhead polylogarithmic in $1/\epsilon$~\cite{Bravyi_2005,Reichardt_2005,Bravyi_2012}. The distillation circuit may be implemented either using transversal CNOT gates between separate encoded blocks or, when a single block encodes many logical qubits, as a sequence of logical Pauli-product measurements within that block~\cite{xu2026distillingmagicstatesbicycle}.
A distilled $\ket{T}$ state is then consumed via gate teleportation to implement the logical $T$ gate~\cite{Bravyi_2005,Litinski_2019}. These ingredients are standard; what is specific to the mitten codes is that every instruction in $\bquis$ is implemented using only five reusable gadgets (Section~\ref{sec:universal_ftqc}, Table~\ref{tab:processor_graph_surgery}).

\begin{definition}[High-throughput instruction set $\hquis$]
\label{def:hquis_formal}
The high-throughput instruction set consists of (i) parallel logical Pauli-product measurement i.e. the joint fault-tolerant measurement of a set of commuting logical Pauli products in a single gadget, and (ii) parallel $T$-state injection i.e. the simultaneous injection of noisy $T$ states into all $k$ logical qubits of a block.
\end{definition}

\begin{definition}[Fixed-gadget instruction set $\fquis$]
\label{def:fquis_formal}
The fixed-gadget instruction set consists of (i) the extractor measurement, i.e., the fault-tolerant measurement of an arbitrary logical Pauli product, and (ii) noisy $T$-state injection. The extractor is a single fixed gadget: its qubit layout and connectivity are fixed in hardware, and the Pauli product to be measured is selected in software by switching the appropriate connections on or off.
\end{definition}

\noindent Every instruction of the basic instruction set $\bquis$ can be viewed as a special case of an instruction of the other two sets. Therefore, both $\hquis$ and $\fquis$ are universal fault-tolerant quantum instruction sets.

\section{Surgery}
\label{app:surgery}

\begin{figure}[t]
    \centering
    \includegraphics[width=1\linewidth]{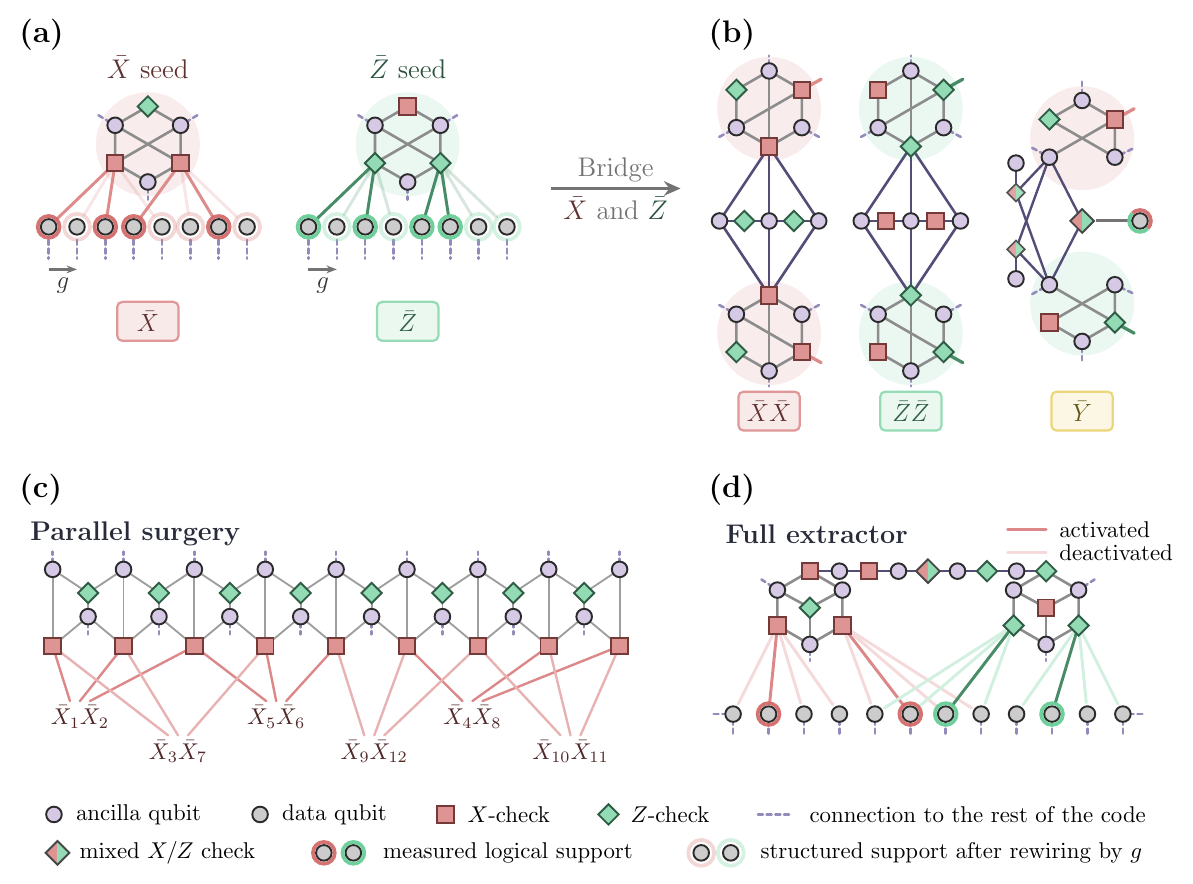}
    \caption{Surgery gadgets. \textbf{(a)} The seed graph surgery gadgets for measuring logical $\bar{X}$/$\bar{Z}$ operators constructed based on the group-structured canonical logical basis (Appendix~\ref{app:canonical-basis}) where different logical operators are related by group action $g$ and the seed gadgets can be used to measure all $|G|$ of them within an orbit by rewiring. \textbf{(b)} $\bar{X}\bar{X}$, $\bar{Z}\bar{Z}$ and $\bar{Y}$ graph surgery gadgets constructed by bridging the seed gadgets from (a). Together, the graph surgery gadgets in (a) and (b) can generate the entire Clifford group. \textbf{(c)} Parallel surgery gadget for measuring Pauli products in parallel. \textbf{(d)} Full extractor for measuring arbitrary weight Pauli product operators using a fixed gadget.}
    \label{fig:surgery_gadgets}
\end{figure}
Universal quantum computation can be achieved with Clifford gates and non-Clifford magic state inputs~\cite{Bravyi_2005,Bravyi_2016}.  On qLDPC codes, Clifford computation can be implemented through logical Pauli-product measurements (PPMs) generated via \textit{code surgery}~\cite{Horsman_2012, Cohen_2022, cross2025improvedqldpcsurgerylogical, Williamson_2026, Ide_2025}. Code surgery realizes PPMs on general qLDPC codes by attaching ancillary qubits and checks that couple to the targeted logical operators. This temporarily enlarges the original code into a merged code in which the targeted logical operators become stabilizers. Their measurement outcomes are inferred from syndrome extraction on the merged code.

More concretely, the merged code is a stabilizer code on the original data qubits and the ancilla qubits of the surgery gadget. Its checks comprise the original checks, some extended onto the ancilla qubits, and additional checks introduced by the gadget. We design these extensions and additional checks so that the merged code preserves the distance of the original code while keeping both the number of qubits in each check and the number of checks acting on each qubit as small as possible. For the $X$-type graph surgery described below, the protocol proceeds in three steps, which we first describe in the absence of faults. First, the gadget ancilla qubits are initialized in $\ket{0}$, preserving the eigenvalues of the original $Z$-checks when they are extended onto the ancillas. Second, the merged-code checks are measured for a number of syndrome extraction rounds proportional to the original code distance $d$. Although the individual new $X$-check outcomes can be random, a specified product of these checks equals the targeted logical operator, so the corresponding product of outcomes gives the logical measurement result. Third, the ancilla qubits are measured in the $Z$ basis and syndrome extraction resumes on the original code. Each restored $Z$-check eigenvalue is the product of its deformed-check eigenvalue and the final $Z$-measurement outcomes on the ancillas to which it was extended. With noise, the full measurement record, including the preparation and readout boundaries, is decoded to infer the logical measurement result and the corrections to track in the Pauli frame. The $Z$-type protocol exchanges $X$ and $Z$, using $\ket{+}$ preparation and $X$-basis readout~\cite{cross2025improvedqldpcsurgerylogical, Williamson_2026}. Preparation and readout for mixed-type gadgets depend on their construction.

The remainder of this appendix is organized as follows. In Appendix~\ref{subsec:graph_surgery} we construct the graph surgery gadgets $\mathbf{S}_X$, $\mathbf{S}_Y$, $\mathbf{S}_Z$, $\mathbf{S}_{XX}$  and $\mathbf{S}_{ZZ}$, measuring logical $\bar{X}$, $\bar{Y}$, $\bar{Z}$, $\bar{X}\bar{X}$ and $\bar{Z}\bar{Z}$, respectively for mitten codes (Figure~\ref{fig:surgery_gadgets}(a,b)). An unstructured canonical logical basis would require one dedicated gadget per measured operator, which means $2k$ seed gadgets for the basis single qubit $\bar{X}$ and $\bar{Z}$ operators and, from them, $3k+k(k-1)$ distinct gadgets to support full Clifford processing. The group structure of the canonical basis of mitten codes (Appendix~\ref{app:canonical-basis}) reduces this to $5$ reusable gadgets generated from only $2$ seed gadgets, $\mathbf{S}_X$ and $\mathbf{S}_Z$. Rewiring a seed gadget by the group action allows it to measure any other logical qubit, as illustrated in Figure~\ref{fig:surgery_gadgets}(a). Bridging seed gadgets produces the remaining three gadgets. In Appendix~\ref{subsec:parallel_surgery} we construct parallel surgery gadgets (Figure~\ref{fig:surgery_gadgets}(c)), which measure many logical Pauli products on a single merged code to support high-throughput Clifford processing. Because the canonical basis arises from the two classical factor codes, all canonical $X$ ($Z$) representatives, and hence the $X$-side ($Z$-side) parallel gadgets, are supported on only $2$ of the $5$ blocks of data qubits, $D_1\cup D_2$ ($D_1\cup D_3$) of mitten codes, reducing the overhead of the parallel surgery gadget. In Appendix~\ref{subsec:extractor} we construct and bridge an $X$-side and a $Z$-side extractor into a full extractor (Figure~\ref{fig:surgery_gadgets}(d)), which is a single fixed gadget measuring arbitrary logical Pauli products~\cite{he2025extractorsqldpcarchitecturesefficient, blue2026extractorslogicalprocessinghypergraph}; this supports the fixed-gadget instruction set. Finally, Appendix~\ref{subsec:distance_preserving_surgery} collects the expansion conditions under which all these gadgets preserve the code distance, and the methods we use to certify them. Throughout, every construction is presented concretely for the mitten codes, whose canonical basis $\{\bar{X}_g,\bar{Z}_g\}_{g\in G}$ consists of a single orbit pair (Theorem~\ref{thm:mitten_canonical_basis}). We also explain at the end of each subsection how the constructions generalize to LP codes of other base-matrix shapes.

\subsection{Graph surgery}
\label{subsec:graph_surgery}
We construct graph surgery gadgets for fault-tolerant measurements of individually selected logical $\bar{X}$, $\bar{Y}$, $\bar{Z}$, $\bar{X}\bar{X}$, and $\bar{Z}\bar{Z}$ operators. A separate gadget for each operator would require $3k+k(k-1)$ gadgets for a code with $k$ logical qubits. Their sizes also depend on the weights of the measured representatives~\cite{Cohen_2022, Williamson_2026}, which can increase when a low-weight spanning set is converted into a canonical basis.

For mitten codes, the canonical basis of Theorem~\ref{thm:mitten_canonical_basis},
\[
\{\bar{X}_g,\bar{Z}_g\}_{g\in G},
\qquad
\bar{X}_g=\left(g,u_g,0,0,0\right)^T,
\quad
\bar{Z}_g=\left(g,0,v_g,0,0\right)^T,
\]
offers two advantages. The representatives are supported on only two data blocks: $D_1\cup D_2$ for $\bar{X}_g$ and $D_1\cup D_3$ for $\bar{Z}_g$, with weights controlled by the classical factor codewords (Table~\ref{tab:1x2_processor_code_param}). Each family also forms a single group orbit. We therefore construct two seed gadgets, $\mathbf{S}_X$ and $\mathbf{S}_Z$, reuse them by rewiring, and bridge them to obtain the other three measurement types (Figure~\ref{fig:surgery_gadgets}(a,b)).

Write $\mathbf{S}(P)$, or $\mathbf{S}_P$, for a gadget measuring a logical representative $P$. Let $\mathbf{X}(P)$ and $\mathbf{Z}(P)$ denote the original $X$ and $Z$ checks, respectively, whose supports intersect $\supp(P)$. We describe the $X$-type construction; exchanging $X$ and $Z$ gives the $Z$-type construction. For mitten codes, $\mathbf{Z}(\bar{X}_g)$ lies in check block $Z_0$, and $\mathbf{X}(\bar{Z}_g)$ in $X_0$ (Equation~\eqref{eq:mitten_HxHz}).

The construction starts with a \emph{skeleton} determined by the Tanner graph between $\supp(\bar{X}_g)$ and $\mathbf{Z}(\bar{X}_g)$. For each data qubit $q\in\supp(\bar{X}_g)$, add an $X$-check $A_q$ acting on $q$. For each original $Z$-check $B_z\in\mathbf{Z}(\bar{X}_g)$, add an ancilla $a_z$ and extend $B_z$ to act with $Z$ on $a_z$. The new check $A_q$ also acts with $X$ on $a_z$ if and only if the original check $B_z$ acts on $q$. Thus, whenever $A_q$ and the extended check $B_z$ overlap on the data qubit $q$, they also overlap on $a_z$. Otherwise they are disjoint. Their overlap is therefore always even, so the new $X$-checks commute with the extended $Z$-checks.

Consider the product of all the new checks $A_q$. On the data qubits, each $A_q$ contributes $X$ on its own qubit $q$, so the product acts exactly as $\bar{X}_g$. It remains to show that the ancilla factors cancel. Each original $Z$-check $B_z$ overlaps $\bar{X}_g$ on an even number of data qubits, because the two operators commute. For each of these shared qubits $q$, the new check $A_q$ contributes an $X$ on the ancilla $a_z$. Thus each ancilla $a_z$ receives an even number of $X$ factors which cancels to the identity. This holds for every ancilla, so the product of all the new checks is precisely $\bar{X}_g$.

The skeleton therefore makes the target logical operator accessible through the new check outcomes. To complete the gadget, we must also ensure that it measures no additional independent logical observables, introduces no extra encoded qubits, and preserves the code distance. We address these requirements by completing the skeleton as a connected graph, adding enough ancilla connections to protect the distance, and adding checks that constrain the ancilla system. We first explain these steps, then give the check matrices and show how the gadget can be rewired to address other logical qubits and bridged with other gadgets to measure two-body logical Pauli products.

Consider the gadget graph in which a vertex represents a new $X$-check $A_q$, and an ancilla qubit connects all the vertices whose checks act on it. Thus the skeleton ancilla $a_z$ is a hyperedge since it can belong to more than two new $X$-checks. An ordinary edge represents an ancilla acted on by exactly the two $X$-checks at its endpoints. If a skeleton ancilla already belongs to exactly two new $X$-checks, it is already an ordinary edge and requires no replacement. If this holds for every skeleton ancilla, the skeleton is already a graph.

For an ancilla shared by more than two new $X$-checks, a graph construction requires a replacement by edge ancillas and corresponding changes to the checks. One way to do this is as follows. Suppose, for example, $B_z$ meets the target support on four data qubits, whose new checks are $A_1,A_2,A_3,A_4$. Its skeleton ancilla is shared by all four checks. One possible replacement uses an ancilla $a_{12}$ shared by $A_1,A_2$ and an ancilla $a_{34}$ shared by $A_3,A_4$, with $B_z$ extended onto both. Each $A_i$ still overlaps the extended $B_z$ on its data qubit and one ancilla, so the two checks commute. Each replacement ancilla now occurs in exactly two vertex checks, so its $X$ factors cancel in the product of all vertex checks. Such a pairing is possible because every skeleton hyperedge has an even number of incident vertices.
Finally, we ensure that the graph is connected so that only the product of all vertex checks cancels every ancilla, giving the selected logical operator $\bar{X}_g$.

Next, we also want to ensure that the merged code created by attaching the surgery gadget preserves the distance of the original code. A product of vertex checks can move part of a logical operator from the data onto the ancillas on the edges leaving the selected vertex subset. If too few such edges leave that subset, this can produce a lower-weight logical representative in the merged code. We therefore add edges to increase the graph's \emph{expansion}, which measures how many edges leave a vertex subset relative to its size, and design towards the sufficient distance-preserving condition of Appendix~\ref{subsec:distance_preserving_surgery}.

After completing the graph, we constrain the degrees of freedom introduced by the ancillas to ensure that attaching the surgery gadget does not introduce extra logical qubits. A cycle is a closed loop of edges in the gadget graph. The operator $Z_C=\prod_{e\in C}Z_{a_e}$ on the ancillas of a cycle $C$ commutes with every vertex $X$-check, because each vertex meets the cycle on zero or two edges. Unless already generated by the stabilizers, such an operator is an additional logical operator supported entirely on the gadget. We make these operators stabilizers by adding a $Z$-check for each cycle in an independent cycle basis of the completed graph. A connected graph with $|\mathcal{V}|$ vertices and $|\mathcal{E}|$ edges has $|\mathcal{E}|-|\mathcal{V}|+1$ independent cycles. Hence, including a full cycle basis removes the extra ancilla logical degrees of freedom. The merged code then encodes $k-1$ logical qubits, with the measured $\bar{X}_g$ fixed as a stabilizer. During the construction we keep vertex degrees, cycle-check weights, and original-check extensions small to retain low check weights and qubit degrees.

Ordering data qubits before ancillas, the merged code has binary parity-check matrices
\begin{equation}\label{eq:merged_HxHz}
\begin{split}
    H_X^{\mathrm{merged}} &=
    \bordermatrix{
        & [n] & \mathcal{E} \cr
    \text{original} & H_X & 0 \cr
    \mathcal{V} & \Pi & \Gamma^T \cr
    } \\
    H_Z^{\mathrm{merged}} &=
    \bordermatrix{
        & [n] & \mathcal{E} \cr
    \text{original} & H_Z & J \cr
    \mathcal{C} & 0 & \Lambda \cr
    }.
\end{split}
\end{equation}
Here $H_X$ and $H_Z$ are the original code's parity-check matrices, $\mathcal{V}$ indexes the vertex $X$-checks, $\mathcal{E}$ the edge ancillas, and $\mathcal{C}$ the chosen cycle basis. The individual block matrices are defined as follows:
\begin{itemize}
    \item The attachment matrix $\Pi\in\F_2^{\mathcal{V}\times n}$ records which data qubit each vertex check acts on. A skeleton row has one $1$; a vertex without a data attachment has a zero row.
    \item The incidence matrix $\Gamma\in\F_2^{\mathcal{E}\times\mathcal{V}}$ records which vertex checks act on each ancilla. Each row has two $1$s, one for each endpoint of the edge.
    \item The extension matrix $J$ has one row per original $Z$-check and one column per edge ancilla. Its entry $J_{s,e}$ is $1$ exactly when check $s$ is extended onto ancilla $a_e$.
    \item The cycle matrix $\Lambda\in\F_2^{\mathcal{C}\times\mathcal{E}}$ records the ancilla support of each new cycle $Z$-check.
\end{itemize}
Thus the original $X$-checks are unchanged, the original $Z$-checks acquire their ancilla extensions, and the lower row blocks are the new vertex and cycle checks. Given $H_XH_Z^T=0$ for the original code, the remaining commutation conditions are
\[
    \Lambda\Gamma=0,
    \qquad
    \Pi H_Z^T=\Gamma^TJ^T.
\]
The first equation expresses the constraint that every cycle check overlaps each vertex check on an even number of ancillas. The second equation says that the data overlap of each vertex check with each original $Z$-check has the same parity as their ancilla overlap; therefore, their total overlap is even. These are the cycle and extension rules described above.

The same notation also identifies the measured observable. If $x\in\F_2^{\mathcal{V}}$ selects a set of vertex checks, their product has data support $\Pi^Tx$ and ancilla support $\Gamma x$. The ancilla factors cancel exactly when $\Gamma x=0$, that is, when $x\in\ker\Gamma$. For a connected  graph, this occurs only for $x=0$ or $x=\mathbf{1}_{\mathcal{V}}$. Hence, the only nontrivial logical operator that can be measured as a product of the vertex checks is indeed only target logical operator $\bar{X}_g$.

Having constructed one gadget, we now explain why it can be reused for every logical qubit. Within each data and check block, qubits and checks are labeled by group elements. A $Z$-check labeled $s$ in $Z_0$ touches a data qubit labeled $q$ in $D_i$, for $i\in\{1,2\}$, exactly when $s^{-1}q$ is one of the group elements appearing in $b_{i-1}$ (Equation~\eqref{eq:mitten_HxHz}). Thus the connection depends only on the relative group label $s^{-1}q$. Translating the qubit label by $q\mapsto gq$ and the check label by $s\mapsto gs$ preserves this connection, since
\[
    (gs)^{-1}(gq)=s^{-1}g^{-1}gq=s^{-1}q.
\]
Every check therefore touches the same pattern of qubits after the relabeling.

The support of $\bar{X}_g$ is obtained from that of $\bar{X}_e$ by precisely this translation, so its attached $Z$-checks and their connections are translated in the same way. Consequently, the seed gadget $\mathbf{S}_X=\mathbf{S}(\bar{X}_e)$ can measure $\bar{X}_g$ by rewiring its attachments. A vertex check formerly attached to data qubit $q$ is attached to $gq$ in the same data block, and the ancilla extension of original check $s$ is transferred to check $gs$ (Figure~\ref{fig:surgery_gadgets}(a)). All internal ancillas and checks of the gadget remain fixed. The same idea lets us reuse the seed gadget $\mathbf{S}_Z=\mathbf{S}(\bar{Z}_e)$ to measure any $\bar{Z}_g$, but by relabeling qubits and checks by right multiplication instead.

 To measure a product $\bar{X}_g\bar{X}_h$   for distinct logical qubits, we combine two rewired  $\mathbf{S}_X$   gadgets using the bridging construction of Refs.~\mbox{
\cite{Cohen_2022, cross2025improvedqldpcsurgerylogical}}\hskip0pt
. We join their vertex sets by additional edge ancillas and add  $Z$  -checks for the new independent cycles. Multiplying all vertex checks on both sides cancels the $X$ factors on every gadget ancilla and bridge ancilla, since each belongs to two vertex checks, leaving $\bar{X}_g\bar{X}_h$ and giving the joint measurement without separately measuring either factor.

The bridge must also preserve the code distance. Since the product of the vertex checks on one side is a stabilizer of the merged code, that side's logical factor is equivalent, up to this stabilizer, to $X$ on the bridge ancillas alone. Fewer than $d$ bridge ancillas would therefore give a logical representative of weight below $d$. At least $d$ bridge edges are necessary; the completed gadget is also subject to the distance checks of Appendix~\ref{subsec:distance_preserving_surgery}. Exchanging $X$ and $Z$ gives $\mathbf{S}_{ZZ}$. Combining  $\mathbf{S}_X$   and  $\mathbf{S}_Z$   for the same logical qubit uses the mixed-type bridge of Ref.~\mbox{
\cite{cross2025improvedqldpcsurgerylogical} }\hskip0pt
to measure $\bar{Y}_g$. Its checks act  as $X$ towards the $\mathbf{S}_X$ side and as $Z$ towards the $\mathbf{S}_Z$ side, so the merged code is not CSS. Equation~\eqref{eq:merged_HxHz} describes the same-type gadgets, with $X$  and   $Z$ exchanged where appropriate.

  The two seeds therefore generate five reusable gadget types,  $\mathbf{S}_X$,  $\mathbf{S}_Z$,   $\mathbf{S}_{XX}$, $\mathbf{S}_{ZZ}$,  and   $\mathbf{S}_Y$, implementing the logical Pauli measurements of the basic instruction set $\bquis$. Each seed's interface can be rewired independently to select the logical qubits. Table~\ref{tab:processor_graph_surgery} reports the ancilla and check counts for the eight mitten codes, together with the maximum check weight of the merged code, which is  at most $12$   compared with  $9$   for the bare codes.

 \begin{remark}[Generalization of graph surgery gadgets to LP codes of general shape]
\label{rem:graph_surgery_general_shape}
 For any LP code satisfying the square invertibility condition (Definition~\ref{def:square_inv_condition}), the canonical basis $\{X^g_{\alpha,\beta},Z^g_{\alpha,\beta}\}$ of Theorem~\ref{thm:canonical-logicals}   has  $f_Af_B$ conjugate pairs of   group orbits. Within each orbit, the same relabeling argument lets one seed address all representatives. The construction therefore uses one pair of  $X$   and  $Z$   seeds per orbit pair, rather than the single pair needed for mitten codes. The skeleton, graph completion,  and   bridging steps apply as above. Including the bridged gadgets gives  $3f_Af_B+2f_A^2f_B^2$ reusable   gadgets for the five measurement types considered here; additional row or column symmetries can reduce this number further.
\end{remark}

\begin{table}[t]\centering
\footnotesize
\setlength{\tabcolsep}{2.6pt}
\begin{tabular}{l cccc cccc ccccc cccc cccc}
\toprule
Processor code
& \multicolumn{4}{c}{$\mathbf{S}_{X}$} & \multicolumn{4}{c}{$\mathbf{S}_{Z}$}
& \multicolumn{5}{c}{$\mathbf{S}_{Y}$}
& \multicolumn{4}{c}{$\mathbf{S}_{XX}$} & \multicolumn{4}{c}{$\mathbf{S}_{ZZ}$} \\
\cmidrule(lr){2-5}\cmidrule(lr){6-9}\cmidrule(lr){10-14}\cmidrule(lr){15-18}\cmidrule(lr){19-22}
$\llbracket n,k,d \rrbracket $
 & $n_{\mathrm{anc}}$ & $m_X$ & $m_Z$ & $w$
 & $n_{\mathrm{anc}}$ & $m_X$ & $m_Z$ & $w$
 & $n_{\mathrm{anc}}$ & $m_X$ & $m_Z$ & $m_{XZ}$ & $w$
 & $n_{\mathrm{anc}}$ & $m_X$ & $m_Z$ & $w$
 & $n_{\mathrm{anc}}$ & $m_X$ & $m_Z$ & $w$ \\
\midrule
$\llbracket 150,30,10 \rrbracket $      & 34  & 18  & 17  & 10 & 17  & 8   & 10  & 10 & 61  & 26  & 17  & 19 & 10 & 78  & 36  & 43  & 11 & 45  & 26  & 20  & 11 \\
$\llbracket 200,40,12 \rrbracket $      & 38  & 20  & 19  & 10 & 33  & 16  & 18  & 10 & 83  & 35  & 24  & 25 & 10 & 88  & 40  & 49  & 11 & 78  & 43  & 36  & 12 \\
$\llbracket 300,60,14 \rrbracket $      & 43  & 22  & 22  & 10 & 43  & 22  & 22  & 10 & 100 & 43  & 29  & 29 & 10 & 100 & 44  & 57  & 12 & 100 & 57  & 44  & 11 \\
$\llbracket 500,100,16 \rrbracket $     & 57  & 28  & 30  & 10 & 46  & 23  & 24  & 10 & 119 & 50  & 37  & 33 & 10 & 156 & 56  & 101 & 12 & 134 & 87  & 48  & 12 \\
$\llbracket 540,108,18 \rrbracket $     & 42  & 22  & 21  & 10 & 57  & 30  & 28  & 10 & 117 & 51  & 30  & 37 & 10 & 123 & 44  & 80  & 10 & 156 & 101 & 56  & 12 \\
$\llbracket 630,126,\leq20 \rrbracket $ & 57  & 28  & 30  & 10 & 97  & 54  & 44  & 10 & 174 & 81  & 53  & 41 & 10 & 134 & 56  & 79  & 10 & 255 & 168 & 88  & 12 \\
$\llbracket 780,156,\leq22 \rrbracket $ & 204 & 74  & 131 & 11 & 239 & 156 & 84  & 10 & 466 & 230 & 192 & 45 & 11 & 435 & 148 & 288 & 12 & 504 & 337 & 168 & 12 \\
$\llbracket 975,195,\leq24 \rrbracket $ & 296 & 102 & 195 & 10 & 261 & 170 & 92  & 11 & 581 & 271 & 262 & 49 & 12 & 622 & 204 & 419 & 12 & 549 & 366 & 184 & 12 \\
\bottomrule
\end{tabular}
\caption{
Distance-preserving canonical logical surgery gadgets for mitten  codes. $\mathbf{S}_{X}$, $\mathbf{S}_{Y}$, and $\mathbf{S}_{Z}$ are the single-qubit logical Pauli measurement gadgets, and $\mathbf{S}_{XX},\mathbf{S}_{ZZ}$ the weight-two logical Pauli products. For each gadget, $n_{\mathrm{anc}}$ is the number of ancilla qubits, $m_{X}$ and $m_{Z}$ are the numbers of new $X$- and $Z$-type checks added to the code (the vertex and cycle checks of the surgery graph), and $w$ is the final (maximal) check weight of the merged code. For $\mathbf{S}_{Y}$, $m_{XZ}$ counts the additional mixed-type checks, supported on both $X$ and $Z$, which fuse the $\mathbf{S}_{X}$ and $\mathbf{S}_{Z}$ graphs.
}
\label{tab:processor_graph_surgery}
\end{table}

\subsection{Parallel surgery}
\label{subsec:parallel_surgery}

As illustrated in Figure~\ref{fig:surgery_gadgets}(c), parallel surgery\footnote{Also known as high-rate surgery.} measures $n_p$ independent logical Pauli products on a single merged code. We follow the framework of Ref.~\cite{zheng2025highratesurgeryconstantoverheadlogical} to support the high-throughput instruction set $\hquis$. For mitten codes, we measure up to $k/2$ disjoint weight-two products $\bar{X}_g\bar{X}_h$, covering all $k=|G|$ logical qubits in one parallel measurement.

As in graph surgery, each target logical operator is obtained as a product of vertex checks whose combined action on the ancilla qubits is the identity. For an ordinary edge, cancellation requires selecting both endpoint checks or neither. In a connected graph this forces the selection of all vertex checks (or none), allowing only one independent logical measurement. A hyperedge instead only needs an even number of its incident checks to be selected, so several independent products can cancel the ancillas within one connected gadget. We therefore retain the skeleton's hyperedges.

Given the $n_p$ target products, we form the union of their supports and build the skeleton as in Appendix~\ref{subsec:graph_surgery}. For each data qubit $q$ in this union, add an $X$-check $A_q$ acting on $q$. For each original $Z$-check $B_z$ meeting the union, add an ancilla $a_z$, extend $B_z$ onto it, and let $A_q$ act on $a_z$ whenever $B_z$ acts on $q$. Each target commutes with every $B_z$, so the product of the $A_q$ over its support cancels every ancilla and yields that target. For mitten codes, the canonical $X$ representatives lie in $D_1\cup D_2$, so only checks in $Z_0$ are extended. Exchanging $X$ and $Z$ gives the $Z$-type construction.

The merged check matrices have the block form of Equation~\eqref{eq:merged_HxHz}. In the skeleton, $\Gamma_{z,q}=1$ exactly when $B_z$ acts on $q$, and $J$ extends each $B_z$ onto its own ancilla. Thus $\Pi H_Z^T=\Gamma^TJ^T$ holds by construction. Recall from Appendix~\ref{subsec:graph_surgery} that selecting vertex checks with a vector $x$ gives a product with ancilla support $\Gamma x$ and support $\Pi^Tx$ on the original qubits. For each target $P_i$, let $x_i$ select the vertices attached to its support. Since $P_i$ commutes with every original $Z$-check, $\Gamma x_i=0$, while $\Pi^Tx_i$ is the support indicator of $P_i$. Thus each target corresponds to a vector in $\ker\Gamma$, and multiplying the selected check outcomes gives its logical measurement result.

The skeleton makes the target products measurable, but it need not preserve the code distance. We therefore enlarge the gadget by adaptive thickening, following the iterative construction of Ref.~\cite{zheng2025highratesurgeryconstantoverheadlogical}. Starting with the skeleton as the first layer, we search for low-weight logical operators of the merged code and add ancilla qubits with $Z$-checks chosen to detect as many of these operators as possible. The detected operators must then acquire additional support to remain undetected. When no useful addition fits within the limits on check weights and qubit degrees, we add a new layer of $X$-checks, each joined to its counterpart in the previous layer by an ancilla qubit. This provides room for further additions within and between layers while keeping the checks sparse. The search is repeated after each addition, so the gadget grows according to the low-weight logical operators that remain. Only the first layer attaches to the data, and the target measurements are preserved.

In our construction, we use \textsf{sQetch} to estimate the merged-code distance at each step and stop when the estimated distance is at most two below the original distance (Appendix~\ref{subsec:distance_preserving_surgery}). The final merged codes encode $k-n_p$ logical qubits: the $n_p$ targets become stabilizers, and no extra logical qubits remain. The gadgets in Table~\ref{tab:processor_c300_high_rate_surgery} add exactly $n_p$ more independent checks than ancilla qubits.

For the $\llbracket 300,60,14 \rrbracket$ code, we construct gadgets measuring $n_p\in\{10,20,30\}$ randomly chosen, independent weight-two products of canonical $X$ logical operators. Table~\ref{tab:processor_c300_high_rate_surgery} reports their overheads and estimated merged-code distances.

\begin{remark}[Generalization of parallel surgery gadget to LP codes of general shape]
\label{rem:parallel_surgery_general_shape}
For any LP code satisfying the square invertibility condition, the same construction measures an arbitrary collection of independent same-type products of the canonical operators $X^g_{\alpha,\beta}$ (or, symmetrically, $Z^g_{\alpha,\beta}$). The canonical $X$ basis occupies only the row blocks $\alpha\le f_A$ of the first block of qubits, i.e.\ at most $f_Ac_2$ of the $c_1c_2+r_1r_2$ ring blocks, and the $Z$-checks acting on them lie in the corresponding $f_Ar_2$ ring blocks of checks; the union of the measured supports, and hence the parallel gadget, attaches only there. The skeleton, the adaptive thickening, and the \textsf{sQetch}-based certification carry over unchanged.
\end{remark}

\begin{table}[t]
\centering
\begin{tabular}{@{}c l r r r c c c@{}}
\toprule
$n_p$ & merged $\llbracket n,k,d \rrbracket $ & $+\text{anc}$ & $+X_{\text{chk}}$ & $+Z_{\text{chk}}$ & $d_x$ & $d_z$ & $w$ \\
\midrule
$10$ & $\llbracket \,901,\,50,\,\leq13 \rrbracket $ & $601$ & $318$ & $293$ & $13$ & $14$ & $12$ \\
$20$ & $\llbracket 1181,\,40,\,\leq14 \rrbracket $ & $881$ & $490$ & $411$ & $14$ & $14$ & $12$ \\
$30$ & $\llbracket 1574,\,30,\,\leq13 \rrbracket $ & $1274$ & $718$ & $586$ & $13$ & $17$ & $12$ \\
\bottomrule
\end{tabular}
\caption{High-rate surgery gadgets on the $\llbracket 300,60,14 \rrbracket $ code, each jointly measuring $n_p$ weight-two $X$-logical products. Each gadget adds the listed number of ancilla qubits, $X$-checks, and $Z$-checks. Distances are certified by two independent estimators on the saved merged code: \textsf{sQetch} ($5\times10^{6}$ trials, with $k_{\mathrm{sub}} = 512$ for $n_p\in\{10,20\}$ and $k_{\mathrm{sub}} = 384$ for $n_p=30$, detailed in Appendix~\ref{app:sqetch}) and the BP+OSD decoder based distance estimator~\cite{PhysRevResearch.2.043423} ($5\times10^{4}$ trials). $w$ is the maximum merged check weight.}
\label{tab:processor_c300_high_rate_surgery}
\end{table}

\subsection{Extractor}
\label{subsec:extractor}

An extractor~\cite{he2025extractorsqldpcarchitecturesefficient} as illustrated in Figure~\ref{fig:surgery_gadgets}(d) is a fixed ancilla system, wired to the code block once, in which individual connections can be activated or deactivated in software. Activating the connections that address a chosen logical Pauli product turns the extractor into a surgery gadget measuring exactly that product. A single fixed gadget thereby measures arbitrary logical Pauli products thereby supporting the fixed-gadget instruction set $\fquis$. Generally, extractors are costly in space, but it was shown for hypergraph product codes that the structure of the logical basis can suppress the extractor overhead~\cite{blue2026extractorslogicalprocessinghypergraph}. We follow the same principle where our extractors are organized around the low-weight canonical logical basis of Appendix~\ref{app:canonical-basis}.

We first construct an $X$-side and a $Z$-side extractor, then join them into a full extractor. For mitten codes, the canonical $X$ representatives lie in $D_1\cup D_2$, and the $Z$-checks acting on these blocks are precisely those in $Z_0$ (Equation~\eqref{eq:mitten_HxHz}). Each single-side extractor therefore attaches to only $2$ of the $5$ data blocks, which helps keep its overhead low.

We build the $X$-side skeleton from the union of the supports of the canonical operators $\bar{X}_g$, as in the previous two subsections. This union is contained in $D_1\cup D_2$. For each qubit $q$ in the union, add a vertex $X$-check $A_q$ with an available attachment to $q$. For each original $Z$-check $B_z$ meeting the union, add an ancilla $a_z$ and extend $B_z$ onto it. Let $A_q$ act with $X$ on $a_z$ exactly when $B_z$ acts on $q$. Since the relevant $B_z$ lie in $Z_0$, the skeleton has at most $2|G|$ vertex checks and at most $|G|$ ancillas.

To measure an arbitrary $X$-type logical Pauli product $P$, we must configure the extractor so that the appropriate connections are activated. In particular, we choose which vertex checks $A_q$ are attached to the data qubits of the original code and how the original checks $B_z$ are extended onto the ancillas. We activate the data attachment of a vertex check $A_q$ exactly when $q\in\supp(P)$; vertex checks without a data attachment keep their original ancilla factors. With this choice, multiplying all vertex checks cancels every edge ancilla contribution and leaves precisely the desired logical operator $P$. To ensure that the resulting checks all commute, we must also modify the way the original checks $B_z$ act on the ancillas.

We achieve this as follows. For each ancilla $a_z$, arrange the vertex $X$-checks $A_q$ acting on it in a cyclic order. Remove $a_z$ and replace it by one new ancilla $a_e$ for each neighboring pair of vertex checks, including the last and first. Each new ancilla is acted on with $X$ by its two neighboring checks. Thus the checks are the vertices of a loop, and the replacement ancillas are its edges. For example, if $a_z$ is shared by four checks, written as $A_1,A_2,A_3,A_4$, the replacement uses $a_{12},a_{23},a_{34},a_{41}$, with $a_{ij}$ shared by $A_i$ and $A_j$. The original check $B_z$ has an available connection to every replacement ancilla $a_e$; the target logical Pauli product $P$ determines which of these connections are used.

With the data attachments of the vertex checks $A_q$ chosen to match $\supp(P)$, we now choose the ancilla connections of each original check $B_z$ so that it commutes with these configured vertex checks. For each $B_z$, mark the vertex $A_q$ for every data qubit $q$ on which both $B_z$ and $P$ act. At each marked vertex, $A_q$ and $B_z$ anticommute on the data qubit, so their ancilla overlap must supply a second anticommutation. The number of marked vertices is even because $B_z$ commutes with the full product $P$. Pair consecutive marked vertices around the loop and extend $B_z$ with $Z$ on the ancillas along the paths between each pair. In the four-check example, if only $A_1$ and $A_3$ have active data attachments, the path $A_1$--$A_2$--$A_3$ gives the extension $B_z Z_{a_{12}}Z_{a_{23}}$. At each path endpoint, the extended $B_z$ overlaps the corresponding $A_q$ on one data qubit and one ancilla. At each interior vertex, it overlaps $A_q$ on two ancillas and no data qubit. In both cases, the two checks commute because they overlap on an even number of qubits.

The loops therefore provide the paths needed to extend the original checks for any target. Before using the extractor for measurements, we complete its construction by adding edges to ensure connectedness and the expansion required to preserve the code distance. We also add cycle checks $Z_C$ for a full cycle basis to ensure that the ancillas introduce no additional logical qubits.

Exchanging $X$ and $Z$ gives the $Z$-side extractor on $D_1\cup D_3$, with original checks in $X_0$.

Finally, we bridge the $X$- and $Z$-side extractors to obtain a single fixed extractor for arbitrary products of the canonical logical Pauli operators. Following Ref.~\cite{blue2026extractorslogicalprocessinghypergraph}, we connect the two vertices associated with each shared data qubit by a new edge ancilla.

For a chosen target $P$, each data qubit $q\in\supp(P)$ is attached to exactly one vertex. If $q$ is shared by the two sides, we choose one of its two vertices; the other has no data attachment. The selected vertex check acts on $q$ with the corresponding factor of $P$: $X_q$, $Y_q$, or $Z_q$. Data qubits outside $\supp(P)$ have no active attachment.

We first specify how the checks act on the ancillas. Vertex checks act with $X$ on existing $X$-side ancillas and with $Z$ on existing $Z$-side ancillas. Both endpoints of each bridge act with $X$ on its ancilla. Thus a $Z$-side vertex incident to a bridge has $Z$ on its original edges and $X$ on the bridge.

Combining these ancilla actions with the chosen data Pauli gives the full vertex check. For example, if an $X$-side vertex acts with $X_aX_b$ on its two incident ancillas and the target contains $Z_q$ on its attached data qubit, its check is
\[
    A_q=Z_qX_aX_b.
\]
Choosing $Y_q$ instead gives $A_q=Y_qX_aX_b$. Thus the joined extractor uses non-CSS checks, both at $Z$-side bridge endpoints and when the chosen data Pauli differs from the ancilla Paulis.

The original code checks must also be extended onto ancillas so that they commute with the vertex checks. To make these extensions possible for every target, we retain each original check's available connections to its loop ancillas. We also give each original check an available connection to the bridge ancilla at every shared data qubit in its support. These connections are fixed before choosing $P$. They allow an extension to follow the original loop and, when necessary, cross a bridge to reach the vertex selected for a data attachment.

For each original code check $S$, mark every vertex whose check $A_q$ anticommutes with $S$ on its attached data qubit. The number of marked vertices is even because $S$ commutes with the full target $P$. Pair the corresponding vertices consecutively around the check's original loop and join each pair along the loop. Whenever a marked vertex is on the opposite side, continue the path across its bridge. Extend the original check along these paths using $Z$ on $X$-side and bridge ancillas and $X$ on $Z$-side ancillas. Each path endpoint then has one data and one ancilla anticommutation, while each interior vertex has two ancilla anticommutations, so the extended check commutes with every vertex check. The fixed loops and bridges provide the required paths for every canonical logical Pauli product; only the choice of paths depends on $P$.

To complete the construction of the full extractor, we add any further edges needed to satisfy the expansion condition of Appendix~\ref{subsec:distance_preserving_surgery}, and complete the cycle checks to a full cycle basis of the joined graph. These cycle checks use $Z$ on $X$-side and bridge ancillas and $X$ on $Z$-side ancillas. Table~\ref{tab:full_extractor_overhead} reports the overheads of the completed extractor.

\begin{remark}[Generalization to extractors of LP codes of general shape]
\label{rem:extractor_general_shape}

For any LP code satisfying the square invertibility condition, the same construction uses the blocks supporting the canonical logical basis. The $X$-side skeleton uses only qubits in the union of the canonical $X$ supports, giving at most $f_Ac_2|G|$ vertex checks attached within row blocks $\alpha\le f_A$ of the first block of qubits, and at most $f_Ar_2|G|$ ancillas extending the $Z$-checks acting on them. The $Z$-side skeleton similarly has at most $c_1f_B|G|$ vertex checks attached within column blocks $\beta\le f_B$, and at most $r_1f_B|G|$ ancillas extending the corresponding $X$-checks. The conversion to a graph, addition of edges and cycle checks, bridging, and expansion certification proceed as above.
\end{remark}

\begin{table}[t]
\centering
\begin{tabular}{l c c c c}
\toprule
Code $\llbracket n,k,d \rrbracket $ &  $+$anc & $+X$-chk & $+Z$-chk & $w_{\mathrm{max}}$\\
\midrule
$\llbracket 150,30,10 \rrbracket $              & 390  & 186  & 186  & 13\\
$\llbracket 200,40,12 \rrbracket $              & 522  & 256  & 244  & 13 \\
$\llbracket 300,60,14 \rrbracket $ & 780  & 382  & 372  & 13 \\
$\llbracket 500,100,16 \rrbracket $         & 1370 & 691  & 649  & 13 \\
$\llbracket 540,108,18 \rrbracket $             & 1477 & 719  & 724  & 13 \\
$\llbracket 630,126,\leq20 \rrbracket $              & 1730 & 849  & 843  & 13  \\
$\llbracket 780,156,\leq22 \rrbracket $             & 2086 & 1046 & 998  &  14\\
$\llbracket 975,195,\leq24 \rrbracket $             & 2606 & 1278 & 1282 & 14 \\
\bottomrule
\end{tabular}
\caption[Overhead of the full extractors for mitten processor codes.]{Overhead of the full extractors for mitten processor codes.
$+$anc is the number of total added ancilla qubits in the full extractors, $+X$-chk and $+Z$-chk count the added $X$- and $Z$-checks, respectively, and $w_{\mathrm{max}}$ is the maximum merged-code check weight after attaching the extractors to the original code blocks and measuring Pauli products of arbitrary type.}
\label{tab:full_extractor_overhead}
\end{table}

\subsection{Distance-preserving condition for the surgery and extractor gadgets}
\label{subsec:distance_preserving_surgery}

Throughout the construction of the surgery and extractor gadgets, we quantify fault tolerance by the distance of the merged code. Ideally the merged code preserves the distance of the original code, $d(\MergeCode)\ge d$. Achieving this while simultaneously minimizing the gadget overhead is the central difficulty as usually reducing the number of ancilla qubits conflicts with achieving the distance-preserving property. This subsection collects the sufficient condition we design towards, and the methods we use to certify it or to certify $d(\MergeCode)$ directly. Unlike the previous subsections, nothing here refers to the structure of the mitten codes; the conditions depend only on the attached gadget and apply to any CSS data code, hence to every base-matrix shape alike.

A sufficient condition for distance preservation is that the gadget, viewed through its Tanner graph, is a sufficiently good expander~\cite{Cohen_2022, Williamson_2026, Ide_2025, zheng2025highratesurgeryconstantoverheadlogical}. The mechanism is that the gadgets introduce no logical qubits of their own, so the distance can only drop if multiplying a logical operator of the original code by new stabilizers which are products of vertex checks over a vertex subset $T$ of the gadget lowers its weight. Such a product removes at most $|T|$ qubits from the data support while adding $|\partial_2 T|$ ancilla qubits, where $\partial_2 T$ is the mod-$2$ boundary of $T$ in the gadget (hyper)graph defined below. Boundaryless subsets only exchange the logical operator for an equivalent one of the original code. Hence, if $|\partial_2 T|\ge |T|$ for every $T$ of minimal weight modulo boundaryless subsets, no such multiplication can decrease the weight, and the merged code is distance preserving. We formalize this expansion condition using the Cheeger constant, defined below for hypergraphs to cover all of our gadget constructions.

\begin{definition}[Cheeger constant for hypergraphs]
\label{def:cheeger_hypergraph}
Let $H\in\mathbb{F}_2^{E\times V}$ be the nonzero incidence matrix of a hypergraph with vertex set $V$ and hyperedge set $E$, where $H_{e,v}=1$ exactly when $v$ belongs to $e$. The Cheeger constant of $H$ is
\[
  h(H)\;:=\;\min_{x\in\mathbb{F}_2^{V}\setminus\ker H}
  \frac{|Hx|}{\dist(x,\ker H)}\,,
  \qquad
  \dist(x,\ker H)\;:=\;\min_{y\in\ker H}\wt{x+y}\,.
\]
\end{definition}

\noindent For the indicator vector $x$ of a vertex subset, $Hx$ records its mod-$2$ boundary, and $\dist(x,\ker H)$ is the size of the smallest vertex subset with that same boundary. The sufficient distance-preserving condition described above is therefore $h(H)\ge 1$.

\begin{remark}[Cheeger constant for Graphs]
\label{rem:graph_cheeger}
Let $G=(V,E)$ be a connected graph on $n$ vertices with incidence matrix $H$, then every row of $H$ has weight $2$, so for any subgraph $S\subseteq G$, $|e\cap S|$ is odd iff $e$ has exactly one endpoint in $S$ and $\partial_2 S=\partial S$ is the usual edge boundary. So the Cheeger constant for the graph $G$ is:
\[
  h(G)=\min_{0<|S|\le n/2}\frac{|\partial S|}{|S|}.
\]
\end{remark}

\noindent Indeed, under Definition~\ref{def:cheeger_hypergraph} the identity in Remark~\ref{rem:graph_cheeger} holds because, for a connected graph, $\ker H=\{0,\mathbf{1}_V\}$: for the indicator $x$ of $S\subseteq V$ we have $\dist(x,\ker H)=\min(|S|,\,n-|S|)$ and $\partial_2 S=\partial_2(V\setminus S)$, so every ratio is attained by a set of size at most $n/2$, recovering the standard graph Cheeger constant $h(G)$.

\begin{theorem}[Cheeger inequality]
\label{thm:cheeger_bound}
Let $G=(V,E)$ be a connected graph on $n$ vertices with Laplacian $L$, and let $0=\lambda_1<\lambda_2\le\cdots\le\lambda_n$ be the
eigenvalues of $L$. Then
\[
  \frac{\lambda_2}{2}\;\le\;h(G).
\]
\end{theorem}

\begin{proof}
See \cite[Thm.~2.2]{chung1997spectral}.
\end{proof}

We certify the distance-preserving property with different methods for the different gadget families. For the graph surgery gadgets and the full extractors (Figure~\ref{fig:surgery_gadgets}(a), (b) and (d)) we can certify expansion directly: for the smaller seed gadgets (e.g.\ $\mathbf{S}_X$ and $\mathbf{S}_Z$ of the six smallest mitten codes in Table~\ref{tab:processor_graph_surgery}) we compute the Cheeger constant exactly, and for the larger gadgets we use the spectral lower bound of Theorem~\ref{thm:cheeger_bound} to certify the expansion properties of those surgery gadgets. For the parallel surgery gadgets (Figure~\ref{fig:surgery_gadgets}(c)) we instead certify the distance of the merged code directly with the fast distance estimator \textsf{sQetch} (Appendix~\ref{app:sqetch}) at every step of the construction, which is faster and therefore permits many more optimization rounds. For the extractors, direct estimation is not an option: a single extractor supports exponentially many patterns of activated and deactivated connections, each defining a different merged code, so estimating, let alone exactly computing, the distance of every merged code is intractable. Certifying the expansion of the fixed gadget once instead guarantees the distance-preserving property simultaneously for all activation patterns.

\section{Parallel magic state injection}
\label{app:parallel_magic}

Using the group structure of the canonical logical basis (Appendix~\ref{app:canonical-basis}) and the product structure of the lifted product codes, we can inject $k=|G|$ logical $\ket{\bar{T}}$ states in parallel into a block of the $\llbracket n,k,d\rrbracket$ processor code $\LP(A,B)$, through a magic factory and an ancillary lifted product code that we call the \emph{magic port}. The resulting $\ket{\bar{T}}^{\otimes|G|}$ blocks can then be fed into a distillation factory based on transversal CNOTs~\cite{xu2025fast, cain2026shorsalgorithmpossible10000}. For each processor block, the magic factory consists of $|G|$ copies of a distance-$d_{\mathrm{rep}}$ unrotated surface code, each prepared in a noisy $\ket{\bar{T}}$ state, and the magic port is the ancillary code $\LP(\RepCode(d_{\mathrm{rep}}),B)$, where $\RepCode(d_{\mathrm{rep}})$ denotes the distance-$d_{\mathrm{rep}}$ repetition code viewed over $\F_2[G]$ by replacing each $1$ in its check matrix with the identity $e\in G$.

\begin{figure}[t]
    \centering
    \includegraphics[width=0.75\linewidth]{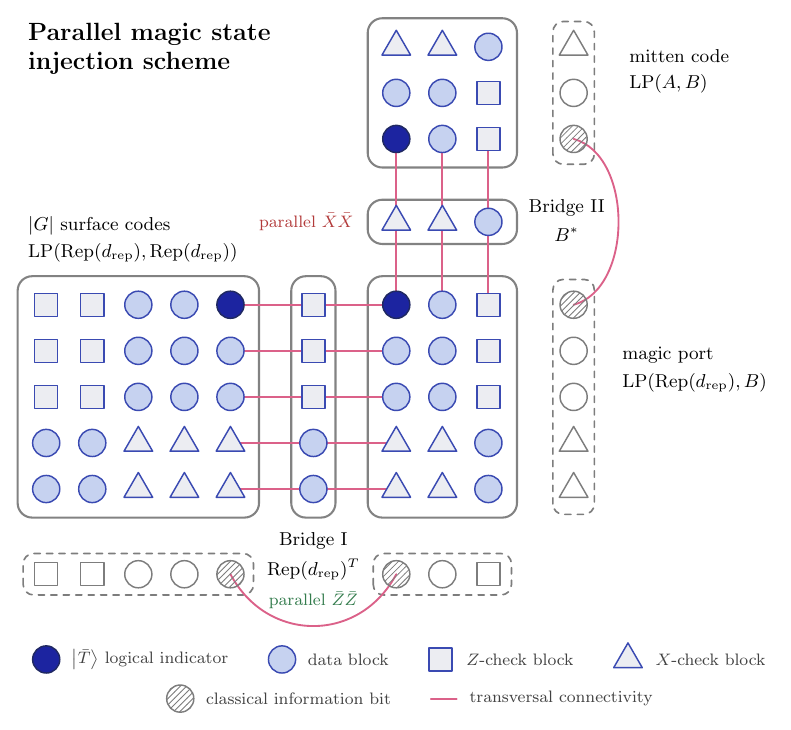}
    \caption{Parallel magic state injection scheme for mitten processor code $\LP(A,B)$, drawn for $d_{\mathrm{rep}}=3$. Every node represents one block of $|G|$ data qubits (circles), $X$-checks (triangles), or $Z$-checks (squares), and the dark blue nodes track the injected logicals: the noisy $\ket{\bar{T}}$ states are grown in the surface codes, moved through the magic port $\LP(\RepCode(d_{\mathrm{rep}}),B)$, and finally land in the mitten code. The $|G|$ distance-$d_{\mathrm{rep}}$ surface codes (left) together form the lifted product $\LP(\RepCode(d_{\mathrm{rep}}), \RepCode(d_{\mathrm{rep}}))$. Bridge~I, given by $|G|$ copies of the transposed repetition code $\RepCode(d_{\mathrm{rep}})^T$ with $d_{\mathrm{rep}}$ $Z$-checks and $d_{\mathrm{rep}}-1$ ancilla qubits per copy, attaches via a simple seam surgery gadget (pink lines, row by row) to the surface-code stack and to the $\RepCode(d_{\mathrm{rep}})$ component of the magic port $\LP(\RepCode(d_{\mathrm{rep}}),B)$, realizing the parallel $\bar{Z}\bar{Z}$ measurement of the first switching step. Bridge~II, the code $B^*$ over $\F_2[G]$ using the right-regular representation with two $X$-check blocks and one ancilla block, also attaches via a simple seam surgery gadget (pink lines, column by column) to the $B$ side of the magic port and to the $B$ side of the mitten code $\LP(A,B)$, realizing the parallel $\bar{X}\bar{X}$ measurement of the second step. These seam surgery gadgets are similar in spirit to how lattice surgery is performed between surface codes. The dashed strips depict the classical base codes underlying the three lifted product codes. In each strip, the hatched node marks the bit associated with the relevant logical operator, while the pink arcs show how the bridges merge these bits at the level of the classical base codes.
Internal check--qubit connections within each component are omitted.}
    \label{fig:parallel_magic_injection}
\end{figure}

As shown in Figure~\ref{fig:parallel_magic_injection}, the injection consists of two steps. First, a parallel $\bar{Z}\bar{Z}$ measurement, realized by a transposed repetition-code bridge (Bridge~I), switches between the surface-code stack and the magic port, moving all $|G|$ $\ket{\bar{T}}$ states into $\LP(\RepCode(d_{\mathrm{rep}}),B)$. Second, a parallel $\bar{X}\bar{X}$ measurement (Bridge~II) switches between the magic port and the processor mitten code, moving the $|G|$ states into the mitten processor code $\LP(A,B)$. Both switches can act on all $|G|$ logical qubits at once precisely because of the group structure of the canonical logical basis (Appendix~\ref{app:canonical-basis}). Bridge~I consists of $|G|$ copies of the transpose of the repetition code $\RepCode(d_{\mathrm{rep}})^T$ with $Z$-type checks, and it attaches to both sides in the same ``transversal'' way. On the surface-code side, each surface code is itself a product of two repetition codes. Pairing each of the $|G|$ surface codes with one copy of $\RepCode(d_{\mathrm{rep}})^T$, the bridge attaches to the repetition-code factor along the orientation hosting the $\bar{Z}$ logical of that surface code, so the total connectivity is $|G|$ copies of the identity map (Figure~\ref{fig:parallel_magic_injection}). On the magic-port side, the same connectivity is most naturally described at the ring level: viewing the $|G|$ copies of $\RepCode(d_{\mathrm{rep}})^T$ collectively as a single copy of $\RepCode(d_{\mathrm{rep}})^T$ over $\F_2[G]$, the bridge attaches to the $\RepCode(d_{\mathrm{rep}})$ factor of $\LP(\RepCode(d_{\mathrm{rep}}),B)$ by the identity map. This realizes the parallel $\bar{Z}\bar{Z}$ measurement. The second step follows the same idea where Bridge~II uses the code $B^*$ over $\F_2[G]$ with $X$-checks, it is attached along the $B$ factors of the magic port and of $\LP(A,B)$. This realizes the parallel $\bar{X}\bar{X}$ measurement. One can also view the two code-switching steps as instances of parallel surgery (Appendix~\ref{subsec:parallel_surgery}) in which, thanks to the matching product structures of the three codes, the surgery gadgets are simple and structured rather than obtained by randomized optimization as in Appendix~\ref{app:surgery}.

In the remainder of this section, we first prove in Lemma~\ref{lem:magic_port_dist} that the magic port code $\LP(\RepCode(d_{\mathrm{rep}}),B)$ has distance $\min(d_{\mathrm{rep}},d_B)$. We then bound the distances of the two merged codes $\MergeCode_I$ and $\MergeCode_{II}$, obtained by attaching Bridge~I and Bridge~II respectively. These bounds establish that the parallel magic state injection scheme of Figure~\ref{fig:parallel_magic_injection} is end-to-end distance preserving (Theorem~\ref{thm:parallel_magic_injection_dist}). Finally, we report the overhead of the scheme on the eight mitten codes for several values of $d_{\mathrm{rep}}$ in Table~\ref{tab:parallel_magic_complexity}.

\subsection{Distance of the magic port code}

We choose the magic factory to be $|G|$ copies of the distance-$d_{\mathrm{rep}}$ unrotated surface code, which together can be viewed as a single lifted product code $\LP(\RepCode(d_{\mathrm{rep}}),\RepCode(d_{\mathrm{rep}}))$ over $\F_2[G]$ whose check matrices are those of one surface code with each $1$ replaced by the identity $e\in G$. Using the structure of the lifted product and of the repetition code, we now present the distance of the magic port code exactly in the following lemma:

  \begin{lemma}[Distance of the magic port code]
\label{lem:magic_port_dist}
    For $B$ as in Definition~\ref{def:mitten_processor_code}, let $d_B=d(\Rrep{B})$. The magic port code satisfies
    \begin{equation}
    \label{eq:magic_port_sector_distances}
        d_X(\LP(\RepCode(d_{\mathrm{rep}}),B))=d_B,
        \qquad
        d_Z(\LP(\RepCode(d_{\mathrm{rep}}),B))=d_{\mathrm{rep}}.
    \end{equation}
    Consequently,
     \begin{equation}
        d(\LP(\RepCode(d_{\mathrm{rep}}),B)) = \min(d_{\mathrm{rep}},d_B).
    \end{equation}
\end{lemma}

\begin{proof}
    By Theorem~\ref{thm:dc_ub_dq}, $d_X(\LP(\RepCode(d_{\mathrm{rep}}),B))\leq d_B$ and $d_Z(\LP(\RepCode(d_{\mathrm{rep}}),B))\leq d_{\mathrm{rep}}$. We therefore only need to prove the matching lower bounds.

    The $X$ and $Z$ parity check matrices of the magic port are given by:

\begin{equation}
\label{eq:magic_port_HX}
\resizebox{0.92\linewidth}{!}{$\displaystyle
    H_X^{\LP(\RepCode(d_{\mathrm{rep}}),B)}=
    \magicportmatrix{
      & D_1 & D_2 & D_3 & D_4 & \cdots
      & D_{2d_{\mathrm{rep}}-3} & D_{2d_{\mathrm{rep}}-2}
      & D_{2d_{\mathrm{rep}}-1} & D_{2d_{\mathrm{rep}}}
      & D_{2d_{\mathrm{rep}}+1} & D_{2d_{\mathrm{rep}}+2}
      & \cdots & D_{3d_{\mathrm{rep}}-2} & D_{3d_{\mathrm{rep}}-1} \cr
      X_1 & I_{|G|} & 0 & I_{|G|} & 0 & \cdots
      & 0 & 0 & 0 & 0 & \Rrep{b_0}^{T} & 0 & \cdots & 0 & 0 \cr
      X_2 & 0 & I_{|G|} & 0 & I_{|G|} & \cdots
      & 0 & 0 & 0 & 0 & \Rrep{b_1}^{T} & 0 & \cdots & 0 & 0 \cr
      X_3 & 0 & 0 & I_{|G|} & 0 & \cdots
      & 0 & 0 & 0 & 0 & 0 & \Rrep{b_0}^{T} & \cdots & 0 & 0 \cr
      X_4 & 0 & 0 & 0 & I_{|G|} & \cdots
      & 0 & 0 & 0 & 0 & 0 & \Rrep{b_1}^{T} & \cdots & 0 & 0 \cr
      \vdots & \vdots & \vdots & \vdots & \vdots & \ddots
      & \vdots & \vdots & \vdots & \vdots & \vdots & \vdots
      & \ddots & \vdots & \vdots \cr
      X_{2d_{\mathrm{rep}}-3} & 0 & 0 & 0 & 0 & \cdots
      & I_{|G|} & 0 & I_{|G|} & 0 & 0 & 0 & \cdots & 0 & \Rrep{b_0}^{T} \cr
      X_{2d_{\mathrm{rep}}-2} & 0 & 0 & 0 & 0 & \cdots
      & 0 & I_{|G|} & 0 & I_{|G|} & 0 & 0 & \cdots & 0 & \Rrep{b_1}^{T} \cr
    }
    $}
\end{equation}
\begin{equation}
\label{eq:magic_port_HZ}
\resizebox{0.92\linewidth}{!}{$\displaystyle
    H_Z^{\LP(\RepCode(d_{\mathrm{rep}}),B)}=
    \magicportmatrix{
      & D_1 & D_2 & D_3 & D_4 & \cdots
      & D_{2d_{\mathrm{rep}}-3} & D_{2d_{\mathrm{rep}}-2}
      & D_{2d_{\mathrm{rep}}-1} & D_{2d_{\mathrm{rep}}}
      & D_{2d_{\mathrm{rep}}+1} & D_{2d_{\mathrm{rep}}+2}
      & \cdots & D_{3d_{\mathrm{rep}}-2} & D_{3d_{\mathrm{rep}}-1} \cr
      Z_1 & \Rrep{b_0} & \Rrep{b_1} & 0 & 0 & \cdots
      & 0 & 0 & 0 & 0 & I_{|G|} & 0 & \cdots & 0 & 0 \cr
      Z_2 & 0 & 0 & \Rrep{b_0} & \Rrep{b_1} & \cdots
      & 0 & 0 & 0 & 0 & I_{|G|} & I_{|G|} & \cdots & 0 & 0 \cr
      \vdots & \vdots & \vdots & \vdots & \vdots & \ddots
      & \vdots & \vdots & \vdots & \vdots & 0 & \ddots
      & \ddots & \vdots & 0 \cr
      Z_{d_{\mathrm{rep}}-1} & 0 & 0 & 0 & 0 & \cdots
      & \Rrep{b_0} & \Rrep{b_1} & 0 & 0 & 0 & 0 & \cdots & I_{|G|} & I_{|G|} \cr
      Z_{d_{\mathrm{rep}}} & 0 & 0 & 0 & 0 & \cdots
      & 0 & 0 & \Rrep{b_0} & \Rrep{b_1} & 0 & 0 & \cdots & 0 & I_{|G|} \cr
    }
    $}
\end{equation}
 where each $D_i$, $X_i$, and $Z_i$ labels a block of $|G|$ data qubits, $X$-checks, and $Z$-checks, respectively.

    We first consider a nontrivial logical $X$ operator $x=(x_1^T,\ldots,x_{3d_{\mathrm{rep}}-1}^T)^T$, where $x_i\in\F_2^{|G|}$. Commutation with the $Z$-checks gives
    \begin{equation}
        \left\{\begin{aligned}
            \Rrep{b_0}x_1+\Rrep{b_1}x_2+x_{2d_{\mathrm{rep}}+1}&=0,\\
            \Rrep{b_0}x_3+\Rrep{b_1}x_4+x_{2d_{\mathrm{rep}}+1}+x_{2d_{\mathrm{rep}}+2}&=0,\\
            &\ \vdots\\
            \Rrep{b_0}x_{2d_{\mathrm{rep}}-1}+\Rrep{b_1}x_{2d_{\mathrm{rep}}}+x_{3d_{\mathrm{rep}}-1}&=0.
        \end{aligned}\right.
    \end{equation} 
      Summing these equations cancels the last $d_{\mathrm{rep}}-1$ blocks and yields
    \begin{equation}
        \Rrep{B}\sigma(x)=0 \quad\text{where}\quad \sigma(x)=
        \begin{pmatrix}
            \displaystyle\sum_{j=1}^{d_{\mathrm{rep}}}x_{2j-1}\\[3pt]
            \displaystyle\sum_{j=1}^{d_{\mathrm{rep}}}x_{2j}
        \end{pmatrix}.
    \end{equation} 

\par

{
    We will show that $\sigma(x)\neq 0$ and is hence a non-trivial codeword of $\Rrep{B}$. Consequently, we have 
    \begin{equation}
        \label{eq:wtx_ge_wtsigma}
        \wt{x}\geq\wt{\sigma(x)}\geq d_B
    \end{equation} and hence \begin{equation}
        \label{eq:dX_ge_dB}
        d_X(\LP(\RepCode(d_{\mathrm{rep}}),B))\geq d_B.
    \end{equation}

    Suppose for contradiction that $\sigma(x)=0$. We will show that this implies $x$ is a stabilizer which contradicts the assumption that $x$ is a nontrivial logical operator. We do this by multiplying $x$ by $X$-checks until we obtain the identity.

    Start with the first pair of data blocks, $(x_1,x_2)$. For each nonzero entry of $x_1$, add the corresponding row of block $X_1$ to $x$, and do the same for $x_2$ using block $X_2$. These row additions are equivalent to multiplying $x$ by stabilizers. From the identity blocks in Eq.~\eqref{eq:magic_port_HX}, this adds $x_1$ to both $D_1$ and $D_3$, and $x_2$ to both $D_2$ and $D_4$. Thus the first pair cancels and its values are added to the second pair:
    \begin{equation}
        \bigl((x_1,x_2),(x_3,x_4)\bigr)
        \longmapsto
        \bigl((0,0),(x_1+x_3,x_2+x_4)\bigr).
    \end{equation}
    In the second step, use the updated values $x_1+x_3$ and $x_2+x_4$ to select checks from $X_3$ and $X_4$, respectively. This cancels $D_3,D_4$ and adds those vectors to $D_5,D_6$. Repeat until only the last pair remains, always using the current values of the pair being canceled. Later steps do not change the pairs already canceled. The last pair is then
    \[
        \left(\sum_{j=1}^{d_{\mathrm{rep}}}x_{2j-1},
              \sum_{j=1}^{d_{\mathrm{rep}}}x_{2j}\right)=(0,0),
    \]
    by the assumption $\sigma(x)=0$. Thus all of the first $2d_{\mathrm{rep}}$ blocks are zero.

    These additions also change the blocks $D_{2d_{\mathrm{rep}}+1},\ldots,D_{3d_{\mathrm{rep}}-1}$. Let $x'$ denote the resulting operator. Every added $X$-stabilizer commutes with every $Z$-check, so $x'$ still satisfies the $Z$-check equations. The first equation is
    \[
        \Rrep{b_0}x'_1+\Rrep{b_1}x'_2+x'_{2d_{\mathrm{rep}}+1}=0.
    \]
    Since $x'_1=x'_2=0$, this forces $x'_{2d_{\mathrm{rep}}+1}=0$, including any support added to that block. Similarly, after substituting the other canceled blocks, the next equations become
    \[
        x'_{2d_{\mathrm{rep}}+i-1}+x'_{2d_{\mathrm{rep}}+i}=0,
        \qquad 2\leq i\leq d_{\mathrm{rep}}-1.
    \]
    The second remaining block is zero because the first is zero, and so on through all the remaining blocks. Hence $x'=0$. Since we reduced $x$ to the identity by multiplying it by $X$-stabilizers, $x$ itself is a stabilizer, contradicting it being a nontrivial logical operator.
\par}

We now turn to the $Z$ distance. Consider a nontrivial logical $Z$ operator $z=(z_1^T,\ldots,z_{3d_{\mathrm{rep}}-1}^T)^T$, where $z_i\in\F_2^{|G|}$. We prove that every pair $(z_{2j-1},z_{2j})$, for $1\leq j\leq d_{\mathrm{rep}}$, is nonzero. Commutation with the $X$-checks gives, for $1\leq i\leq d_{\mathrm{rep}}-1$,
    \begin{equation}\label{eq:xcomm}
        \begin{aligned}
            z_{2i-1}+z_{2i+1}+\Rrep{b_0}^{T}z_{2d_{\mathrm{rep}}+i}&=0,\\
            z_{2i}+z_{2i+2}+\Rrep{b_1}^{T}z_{2d_{\mathrm{rep}}+i}&=0.
        \end{aligned}
    \end{equation}

\par

{
    Suppose, for a contradiction, that the pair $(z_{2j-1},z_{2j})=(0,0)$ for some $1\leq j\leq d_{\mathrm{rep}}$. We will show that this allows us to use the $Z$-stabilizers to cancel all of the other pairs leaving $z$ homologically equivalent to the identity.

    If $j<d_{\mathrm{rep}}$, the commutation equations in~\eqref{eq:xcomm} at $i=j$ give the values of the pair $(z_{2j+1},z_{2j+2})$ immediately to the right of the zero pair:
    \begin{equation}
        z_{2j+1}=\Rrep{b_0}^{T}z_{2d_{\mathrm{rep}}+j},\qquad
        z_{2j+2}=\Rrep{b_1}^{T}z_{2d_{\mathrm{rep}}+j}.
    \end{equation}
    Add the stabilizer obtained by summing the rows of $Z_{j+1}$ corresponding to the nonzero entries of $z_{2d_{\mathrm{rep}}+j}$. By~\eqref{eq:magic_port_HZ}, this cancels the pair to the right:
    \[
        \begin{aligned}
            z_{2j+1}&\longmapsto z_{2j+1}+\Rrep{b_0}^{T}z_{2d_{\mathrm{rep}}+j}=0,\\
            z_{2j+2}&\longmapsto z_{2j+2}+\Rrep{b_1}^{T}z_{2d_{\mathrm{rep}}+j}=0.
        \end{aligned}
    \]
    All other pairs among the first $2d_{\mathrm{rep}}$ blocks are unchanged, and the commutation equations still hold. We can therefore repeat this step with the updated block values until every pair to the right is zero.

    If $j>1$, the same argument cancels pair $j-1$ using the rows of $Z_{j-1}$ corresponding to the nonzero entries of $z_{2d_{\mathrm{rep}}+j-1}$. Repeating to the left with the updated block values cancels all pairs on that side. When $j=d_{\mathrm{rep}}$, only these leftward steps are needed; when $j=1$, only the rightward steps are needed. Thus every choice of $j$ leads to zero support on $D_1,\ldots,D_{2d_{\mathrm{rep}}}$ after adding the appropriate $Z$-stabilizers.

    The added checks also change the remaining blocks $D_{2d_{\mathrm{rep}}+1},\ldots,D_{3d_{\mathrm{rep}}-1}$. Let $z'$ be the resulting operator after adding the $Z$-stabilizers described above. It still commutes with all $X$-checks because we added only $Z$-stabilizers. Substituting its zero pairs into the second check equation gives
    \[
        \Rrep{b_1}^{T}z'_{2d_{\mathrm{rep}}+i}=0,
        \qquad 1\leq i\leq d_{\mathrm{rep}}-1.
    \]
    Since $\Rrep{b_1}$ is invertible, every remaining block is zero. Thus $z'=0$, so $z$ was a product of $Z$-stabilizers which is a contradiction. Every pair $(z_{2j-1},z_{2j})$ must therefore be nonzero. There are $d_{\mathrm{rep}}$ disjoint pairs, each contributing at least one to the weight, so $\wt{z}\geq d_{\mathrm{rep}}$. Hence,
    \begin{equation}
    \label{eq:dZ_ge_drep}
        d_Z(\LP(\RepCode(d_{\mathrm{rep}}),B))\geq d_{\mathrm{rep}}.
    \end{equation}
\par}

      Combining these lower bounds with the upper bounds proves both equalities in Eq.~\eqref{eq:magic_port_sector_distances}, and hence $d(\LP(\RepCode(d_{\mathrm{rep}}),B))=\min(d_{\mathrm{rep}},d_B)$.
 \end{proof}

\begin{remark}
    The distance equalities in Lemma~\ref{lem:magic_port_dist} also hold for $B\in\mathcal R^{r\times c}$ with $r<c$, satisfying the square invertibility condition. The proof extends by replacing the pairs of data blocks that we analyze with groups of $c$ blocks.
      Lemma~\ref{lem:magic_port_dist} and its proof also hold for other shapes of $B$, as long as the square invertibility condition (Definition~\ref{def:square_inv_condition}) is satisfied.
\end{remark}

\subsection{Distance-preserving property of the two bridges}
As illustrated in Figure~\ref{fig:parallel_magic_injection}, both merged codes $\MergeCode_I$ and $\MergeCode_{II}$ are lifted product codes. In the first step, the pre-merge system is the disjoint union
$\LP(\RepCode(d_{\mathrm{rep}}),\RepCode(d_{\mathrm{rep}}))\oplus\LP(\RepCode(d_{\mathrm{rep}}),B)
=\LP\bigl(\RepCode(d_{\mathrm{rep}}),\,\RepCode(d_{\mathrm{rep}})\oplus B\bigr)$,
and Bridge~I amounts to appending one row to the second factor: the first merged code is
\[
\MergeCode_I \;=\; \LP\left(\RepCode(d_{\mathrm{rep}}),\;
\begin{bmatrix}\RepCode(d_{\mathrm{rep}})\oplus B\\ \eta_I\end{bmatrix}\right),
\]
where $\eta_I\in\R^{1\times(d_{\mathrm{rep}}+2)}$ has entry $e$ in the two columns corresponding to the marked information bits of $\RepCode(d_{\mathrm{rep}})$ and $B$, and zeros elsewhere. These are the hatched ``classical information bits'' joined by the pink arc in Figure~\ref{fig:parallel_magic_injection}. Each identity entry expands to $I_{|G|}$, so the added row checks that matching bits in the two blocks have equal values. This exactly corresponds to Bridge~I with $d_{\mathrm{rep}}$ new blocks of $Z$-checks and $d_{\mathrm{rep}}-1$ new blocks of ancilla qubits, i.e.\ $|G|$ copies of $\RepCode(d_{\mathrm{rep}})^T$ with $Z$-type checks. Likewise, in the second step the pre-merge system is $\LP\bigl(\RepCode(d_{\mathrm{rep}})\oplus A,\,B\bigr)=\LP(\RepCode(d_{\mathrm{rep}}),B)\oplus\LP(A,B)$, and Bridge~II appends one row $\eta_{II}$ to the first factor, coupling the columns corresponding to the marked information bits of $\RepCode(d_{\mathrm{rep}})$ and $A$:
\[
\MergeCode_{II} \;=\; \LP\left(
\begin{bmatrix}\RepCode(d_{\mathrm{rep}})\oplus A\\ \eta_{II}\end{bmatrix},\;B\right),
\]
contributing $2$ new blocks of $X$-checks and $1$ new block of ancilla qubits, i.e.\ one copy of $B^*$ over $\F_2[G]$ with $X$-type checks.

The following lemma lower bounds the $X$ distance of $\MergeCode_I$ and the $Z$ distance of $\MergeCode_{II}$ by restricting logical operators to the original components.

\begin{lemma}[Bounds from restriction to the component codes]
\label{lem:magic_bridge_restriction}
    The two merged codes satisfy
    \begin{equation}
    \begin{aligned}
        d_X(\MergeCode_I)&\geq\min(d_{\mathrm{rep}},d_B),\\
        d_Z(\MergeCode_{II})&\geq\min\bigl(d_{\mathrm{rep}},d_Z(\LP(A,B))\bigr).
    \end{aligned}
    \end{equation}
\end{lemma}
\begin{proof}
    Bridge~I introduces only new $Z$ checks and the original $Z$ checks of the component $\LP$ codes are unchanged in the merged code. Therefore, the restriction of a logical $X$ operator of $\MergeCode_I$ to either component is either a stabilizer or a logical operator of the original component. If either restriction is a nontrivial component logical, its weight is at least
    \begin{equation}
        \min\bigl(d_X(\LP(\RepCode(d_{\mathrm{rep}}),\RepCode(d_{\mathrm{rep}}))),\,
        d_X(\LP(\RepCode(d_{\mathrm{rep}}),B))\bigr)
        =\min(d_{\mathrm{rep}},d_B),
    \end{equation}
    by Lemma~\ref{lem:magic_port_dist}. Otherwise, multiplying by the merged-code extensions of the two component stabilizers leaves an operator supported on the bridge ancillas alone. The new $Z$-checks enforce the constraints of $\RepCode(d_{\mathrm{rep}})^T$, whose kernel is trivial, so this remaining operator vanishes. Thus a nontrivial logical $X$ operator must have a nontrivial component restriction, proving the first bound.

    The same argument with the two check types exchanged applies to Bridge~II: the original component $X$-checks are unchanged, and the new bridge $X$-checks act on bridge ancillas through $B^*$, whose binary expansion has trivial kernel. Using $d_Z(\LP(\RepCode(d_{\mathrm{rep}}),B))=d_{\mathrm{rep}}$ gives the second bound.
\end{proof}

It remains to bound the $Z$ distance of $\MergeCode_I$ and the $X$ distance of $\MergeCode_{II}$. The next lemma moves logical support onto one component. The lemma following it then bounds the weight of every representative of that logical class.

\begin{lemma}[Cleaning logical support]
\label{lem:cleaning_merge_logical_support}
Every logical $X$ class of $\MergeCode_{II}$ contains a representative
$\bar{X}^{\star}$ supported only on the processor code $\LP(A,B)$. Similarly, every logical $Z$ class of $\MergeCode_I$ contains a
representative $\bar{Z}^{\star}$ supported only on the magic port
$\LP(\RepCode(d_{\mathrm{rep}}),B)$.
\end{lemma}

\begin{proof}
We first prove the claim for $\MergeCode_{II}$. Let $\bar{X}$ be an arbitrary logical $X$ operator of $\MergeCode_{II}$. The first step is to remove any support this operator has on the bridge ancilla qubits.  By looking at the structure of the $X$ parity-check matrix of $\MergeCode_{II}$, we see that one block of the bridge $X$-checks acts on the ancillas through the invertible block $\Rrep{b_1^*}$. Multiplying by a suitable product of these checks therefore gives an equivalent logical operator $\bar{X}'$ with no support on the bridge ancillas. 

The next step is to express the restriction of $\bar{X}'$ to the port $\LP(\RepCode(d_{\mathrm{rep}}),B)$ in terms of the canonical logical $X$ operators of the port. In the merged code, some port $Z$-checks acquire support on the bridge ancillas, while the others remain unchanged. Since $\bar{X}'$ has no support on the bridge ancillas and commutes with all these checks, its port restriction commutes with every original port $Z$-check. Hence, the restriction of $\bar{X}'$ to the port is either a port $X$-stabilizer or a logical operator of the port. If it is a stabilizer, we can multiply by this port stabilizer to remove the port support entirely and obtain an equivalent logical representative supported only on $\LP(A,B)$.

Otherwise, by Theorem~\ref*{thm:canonical-logicals}, the magic port $\LP(\RepCode(d_{\mathrm{rep}}),B)$ has a canonical logical $X$ basis supported on the first row of its data-qubit grid, exactly where the bridge attaches. Denote these canonical port representatives by $\bar{X}^{\mathrm{port}}_g$ and their canonical counterparts on $\LP(A,B)$ as $\bar{X}_g$. The port restriction of $\bar{X}'$ can be expressed as 
\begin{equation}
    \left.\bar{X}'\right|_{\mathrm{port}}
      = S_{\mathrm{port}}
        \prod_{g\in G}\bigl(\bar{X}^{\mathrm{port}}_g\bigr)^{c_g}
    \label{eq:port-decomposition}
\end{equation}
where $S_{\mathrm{port}}$ is a port $X$-stabilizer and the coefficients $c_g\in\F_2$ are not all zero. By the construction of the bridge, any joint operator $\bar{X}^{\mathrm{port}}_g\bar{X}_g$ can be written as the product of bridge $X$-checks. Hence, multiplying $\bar{X}'$ by $S_{\mathrm{port}}$ as well as the appropriate bridge $X$-checks corresponding to all $\bar{X}^{\mathrm{port}}_g\bar{X}_g$ with $c_g=1$ leaves an equivalent logical operator supported only on $\LP(A,B)$, as desired.

For $\MergeCode_I$, the same argument applies with $X$ and $Z$ exchanged,
transferring support from the surface-code stack to the magic port. The full
row rank of $\RepCode(d_{\mathrm{rep}})$ allows any initial support on the bridge ancillas to be removed. Surface-code $Z$-stabilizers and the bridge stabilizers corresponding to $\bar{Z}\bar{Z}$ measurements then cancel the surface-code restriction without adding bridge support, leaving an equivalent representative
supported only on $\LP(\RepCode(d_{\mathrm{rep}}),B)$.
 \end{proof}

\begin{lemma}[Weight bound after cleaning]
\label{lem:weight_preservation_from_single_side}
    In the setting of Lemma~\ref{lem:cleaning_merge_logical_support}, let $\bar{X}^{\star}$ be a nontrivial logical $X$ operator of $\MergeCode_{II}$ supported only on $\LP(A,B)$. Multiplying $\bar{X}^{\star}$ by $X$ stabilizers of $\MergeCode_{II}$ cannot decrease its weight below $d(\LP(A,B))$. Similarly, for a nontrivial logical $Z$ operator $\bar{Z}^{\star}$ of $\MergeCode_I$ supported only on $\LP(\RepCode(d_{\mathrm{rep}}),B)$, multiplying by $Z$ stabilizers of $\MergeCode_I$ cannot decrease its weight below $d_Z(\LP(\RepCode(d_{\mathrm{rep}}),B))=d_{\mathrm{rep}}$.
\end{lemma}
\begin{proof}
    Every $X$ stabilizer of $\MergeCode_{II}$ is a product
    \begin{equation}
        S_X=S_X^{\LP(A,B)}\,S_X^{\mathrm{Bridge}}\,S_X^{\LP(\RepCode(d_{\mathrm{rep}}),B)}
    \end{equation}
    of stabilizers generated by the processor, bridge, and port $X$ stabilizers. We bound the weight of $\bar{X}^{\star}$ after multiplying these factors in sequence:
    \begin{equation}
    \begin{aligned}
        \bar{X}_1&=\bar{X}^{\star}S_X^{\LP(A,B)},\\
        \bar{X}_2&=\bar{X}_1S_X^{\mathrm{Bridge}},\\
        \bar{X}&=\bar{X}_2S_X^{\LP(\RepCode(d_{\mathrm{rep}}),B)}.
    \end{aligned}
    \end{equation}

    \noindent\emph{Processor stabilizers.}
    Since $\bar{X}^{\star}$ has no support outside the processor code $\LP(A,B)$, its commutation with the merged-code $Z$-checks implies commutation with the processor $Z$-checks. It cannot be a processor $X$-stabilizer, since every such stabilizer is also a merged-code stabilizer. Thus it is a nontrivial logical $X$ operator of $\LP(A,B)$. Multiplying by processor $X$-stabilizers preserves its logical class, so
    \begin{equation}
        \wt{\bar{X}_1}\geq d(\LP(A,B)).
    \end{equation}

    \noindent\emph{Bridge stabilizers.}
    Each bridge $X$-check acts on one processor qubit, one port qubit, and bridge ancillas. Different bridge checks act on different processor qubits and different port qubits. Multiplying by any set of bridge checks therefore removes at most one qubit of processor support per check while adding exactly one qubit of previously empty port support per check. Hence
    \begin{equation}
        \wt{\bar{X}_2}\geq\wt{\bar{X}_1}.
    \end{equation}

    \noindent\emph{Magic port stabilizers.}
    Group the port qubit blocks $D_1,\ldots,D_{2d_{\mathrm{rep}}}$ into grid rows, with row $i$ consisting of $(D_{2i-1},D_{2i})$. The port support of $\bar{X}_2$ lies entirely in the first such row; denote it by $v\in\F_2^{2|G|}$. From Eq.~\eqref{eq:magic_port_HX}, the check block $X_{2i-1}$ has identity blocks on $D_{2i-1}$ and $D_{2i+1}$, while $X_{2i}$ has identity blocks on $D_{2i}$ and $D_{2i+2}$, with no other support on these grid qubits. Thus, for $1\leq i<d_{\mathrm{rep}}$, selecting checks from the pair $(X_{2i-1},X_{2i})$ adds the same vector $s_i\in\F_2^{2|G|}$ to grid rows $i$ and $i+1$. The identity blocks ensure that the affected coordinates in the two rows match. The blocks involving $\Rrep{b_0}^{T}$ and $\Rrep{b_1}^{T}$ act only on the remaining port qubits.

    Since $s_i$ is added to both rows $i$ and $i+1$, its contribution to the sum of all grid rows is $s_i+s_i=0$ over $\F_2$. Thus any product of port $X$-checks preserves this sum, whose initial value is $v$. Writing $u_i\in\F_2^{2|G|}$ for grid row $i$ of the final operator $\bar{X}$, we obtain
    \begin{equation}
        \sum_{i=1}^{d_{\mathrm{rep}}}u_i
        =v.
    \end{equation}
    Consequently,
    \begin{equation}
        \sum_{i=1}^{d_{\mathrm{rep}}}\wt{u_i}
        \geq\operatorname{wt}\left(\sum_{i=1}^{d_{\mathrm{rep}}}u_i\right)
        =\wt{v}.
    \end{equation}
    Port stabilizers leave the processor and bridge support unchanged. Within the port, the row-sum argument shows that the weight on $D_1,\ldots,D_{2d_{\mathrm{rep}}}$ cannot decrease. The remaining port qubits $D_{2d_{\mathrm{rep}}+1},\ldots,D_{3d_{\mathrm{rep}}-1}$ carry no support in $\bar{X}_2$, so any support introduced there can only increase the weight. Therefore $\wt{\bar{X}}\geq\wt{\bar{X}_2}$. Combining the three steps gives
    \begin{equation}
        \wt{\bar{X}}\geq\wt{\bar{X}_2}\geq\wt{\bar{X}_1}\geq d(\LP(A,B)).
    \end{equation}

For $\MergeCode_I$, the same argument applies with the two check types exchanged. Start with $\bar{Z}^{\star}$ on the port. By Lemma~\ref{lem:magic_port_dist}, its weight is bounded below by $d_{\mathrm{rep}}$. Then follow the same line of reasoning where we multiply first by port $Z$ stabilizers, then by bridge $Z$ stabilizers, and finally by surface code $Z$ stabilizers. Each bridge check adds a distinct qubit in the attaching column of the previously empty surface code grid. And surface code $Z$ stabilizers preserve the same sum over columns, so the same weight inequality applies.
\end{proof}

With these lemmas, we can now prove the main theorem of this section:
\begin{theorem}[End-to-end preservation of code distance]
\label{thm:parallel_magic_injection_dist}
    The parallel magic state injection scheme illustrated in Figure~\ref{fig:parallel_magic_injection} preserves code distance, with
    \begin{equation}
        \min(d(\MergeCode_I),d(\MergeCode_{II}))\geq\min(d_{\mathrm{rep}},d(\LP(A,B))).
    \end{equation}
    The surface-code stack, magic port, and processor code also satisfy this lower bound.
\end{theorem}
\begin{proof}
    Let $\delta=\min(d_{\mathrm{rep}},d(\LP(A,B)))$. The surface-code stack and processor have distance at least $\delta$ by definition. The magic port does as well by Lemma~\ref{lem:magic_port_dist}, since $d(\LP(A,B))\leq d_B$ by Theorem~\ref{thm:dc_ub_dq}.
    For the merged codes, combining the bounds of Lemmas~\ref{lem:magic_bridge_restriction} and~\ref{lem:weight_preservation_from_single_side} gives
    \begin{equation}
    \label{eq:magic_merge_distance_bounds}
        \min\bigl(d(\MergeCode_I),d(\MergeCode_{II})\bigr)\geq\delta.
    \end{equation}
\end{proof}

\subsection{Complexity of parallel magic state injection}
\label{subsec:magic_injection_complexity}

The total numbers of physical qubits, $X$-checks, and $Z$-checks involved in the parallel magic state injection are
\begin{equation}
\begin{split}
    n_{\mathrm{ParallelMagic}} &= \Bigl(\bigl(d_{\mathrm{rep}}^2+(d_{\mathrm{rep}}-1)^2\bigr)+(d_{\mathrm{rep}}-1)+\bigl(2d_{\mathrm{rep}}+(d_{\mathrm{rep}}-1)\bigr)+1+5\Bigr)\,|G|\\
    &= \bigl(2d_{\mathrm{rep}}^2+2d_{\mathrm{rep}}+5\bigr)\,|G|,\\
    \#X\text{-checks} &= \bigl(d_{\mathrm{rep}}(d_{\mathrm{rep}}-1)+2(d_{\mathrm{rep}}-1)+2+2\bigr)\,|G| = \bigl(d_{\mathrm{rep}}^2+d_{\mathrm{rep}}+2\bigr)\,|G|,\\
    \#Z\text{-checks} &= \bigl(d_{\mathrm{rep}}(d_{\mathrm{rep}}-1)+d_{\mathrm{rep}}+d_{\mathrm{rep}}+2\bigr)\,|G| = \bigl(d_{\mathrm{rep}}^2+d_{\mathrm{rep}}+2\bigr)\,|G|,
\end{split}
\end{equation}
where the successive terms of the qubit count are the surface-code stack, Bridge~I, the magic port, Bridge~II, and the mitten code. Per $\F_2[G]$-block count, the surface-code stack contributes $d_{\mathrm{rep}}(d_{\mathrm{rep}}-1)$ blocks of each check type, the magic port $2(d_{\mathrm{rep}}-1)$ $X$-check and $d_{\mathrm{rep}}$ $Z$-check blocks, Bridge~I contributes $d_{\mathrm{rep}}-1$ ancilla and $d_{\mathrm{rep}}$ $Z$-check blocks, Bridge~II contributes $1$ ancilla and $2$ $X$-check blocks, and the mitten code $2$ blocks of each check type (Figure~\ref{fig:parallel_magic_injection}). As a consistency check, $n_{\mathrm{ParallelMagic}}-\#X\text{-checks}-\#Z\text{-checks}=|G|$, matching the fact that, with both bridges attached, the full system carries exactly the $|G|$ injected logical qubits. Note that the space cost is dominated by the magic factory itself: the surface codes account for the leading $\bigl(2d_{\mathrm{rep}}^2-2d_{\mathrm{rep}}+1\bigr)|G|$ qubits. The remainder consists of the magic port with $(3d_{\mathrm{rep}}-1)|G|$ qubits and the two bridges which cost only $d_{\mathrm{rep}}|G|$ ancilla qubits, which is below $18\%$ of the total for $d_{\mathrm{rep}}=11$ or higher. The counts for the $8$ mitten codes as processors with different $d_{\mathrm{rep}}$ are presented in Table~\ref{tab:parallel_magic_complexity}.

\begin{table}[t]
\centering

\setlength{\tabcolsep}{3.5pt}
\renewcommand{\arraystretch}{1.1}
\begin{tabular}{c *{5}{ccc}}
\toprule
 & \multicolumn{3}{c}{$d_{\mathrm{rep}}=5$} & \multicolumn{3}{c}{$d_{\mathrm{rep}}=7$}
 & \multicolumn{3}{c}{$d_{\mathrm{rep}}=9$} & \multicolumn{3}{c}{$d_{\mathrm{rep}}=11$}
 \\
\cmidrule(lr){2-4}\cmidrule(lr){5-7}\cmidrule(lr){8-10}\cmidrule(lr){11-13}
Processor & $n$ & $X$ & $Z$ & $n$ & $X$ & $Z$ & $n$ & $X$ & $Z$ & $n$ & $X$ & $Z$  \\
\midrule
$\llbracket 150,30,10 \rrbracket $  & 1950  & 960   & 960   & 3510  & 1740  & 1740  & 5550  & 2760  & 2760  & 8070  & 4020  & 4020  \\
$\llbracket 200,40,12 \rrbracket $  & 2600  & 1280  & 1280  & 4680  & 2320  & 2320  & 7400  & 3680  & 3680  & 10760 & 5360  & 5360  \\
$\llbracket 300,60,14 \rrbracket $  & 3900  & 1920  & 1920  & 7020  & 3480  & 3480  & 11100 & 5520  & 5520  & 16140 & 8040  & 8040  \\
$\llbracket 500,100,16 \rrbracket $ & 6500  & 3200  & 3200  & 11700 & 5800  & 5800  & 18500 & 9200  & 9200  & 26900 & 13400 & 13400 \\
$\llbracket 540,108,18 \rrbracket $ & 7020  & 3456  & 3456  & 12636 & 6264  & 6264  & 19980 & 9936  & 9936  & 29052 & 14472 & 14472 \\
$\llbracket 630,126,\leq20 \rrbracket $ & 8190  & 4032  & 4032  & 14742 & 7308  & 7308  & 23310 & 11592 & 11592 & 33894 & 16884 & 16884 \\
$\llbracket 780,156,\leq22 \rrbracket $ & 10140 & 4992  & 4992  & 18252 & 9048  & 9048  & 28860 & 14352 & 14352 & 41964 & 20904 & 20904 \\
$\llbracket 975,195,\leq24 \rrbracket $ & 12675 & 6240  & 6240  & 22815 & 11310 & 11310 & 36075 & 17940 & 17940 & 52455 & 26130 & 26130 \\
\bottomrule
\end{tabular}
\caption{The number of qubit, $X$- and $Z$-check in the parallel magic state injection process. $n$ is the total number of physical qubits used, $X$ is the total number of $X$ checks, and $Z$ is the total number of $Z$ checks.}
\label{tab:parallel_magic_complexity}
\end{table}

\section{A qLDPC processor discovery pipeline}
\label{app:knitter_pipeline}

As illustrated in Fig.~\ref{fig:pipeline}, we present a discovery pipeline for qLDPC processors supporting all three quantum instruction sets ($\bquis$, $\hquis$, and $\fquis$) based on general lifted product codes. Its inputs are the specific properties required for a fault-tolerant qLDPC processor, effectively the target values for the three processor parameters of Definition~\ref{def:processor_params}, namely the target processing capacity, throughput, and cycle time, together with a physical hardware budget. The outputs of the qLDPC processor pipeline are concrete processor codes with low-weight addressable logical bases, and their gadgets realizing the instruction sets ($\bquis$, $\hquis$, $\fquis$) with atom movements and qubit layout matched to the target atom-array or superconducting devices. Because each requirement enters as an adjustable input rather than being fixed in advance, the same pipeline can be re-run to design processors for different platforms or computational tasks.

The pipeline starts from the code family that the processor is built upon. The shape of the two base matrices $A$ and $B$ sets a floor on the encoding rate by itself, so fixing the shape guarantees a rate floor before any code is built: as in Eq.~\eqref{eq:k_of_LP}, for $A\in\ringR^{r_1\times c_1}$ and $B\in\ringR^{r_2\times c_2}$, the rate of $\LP(A,B)$ satisfies
\begin{equation}
    \frac{k}{n} \geq \frac{(c_1-r_1)(c_2-r_2)}{r_1r_2+c_1c_2},
\end{equation}
and once the shape is fixed, the block size $n = (r_1r_2+c_1c_2)|G|$ depends only on the group order $|G|$. The $1\times 2$ shape we use most often for mitten codes guarantees a rate of at least $20\%$.

Simulating a full memory experiment for every candidate is too slow; therefore, the pipeline filters codes by their check weight, code distance, and circuit-level distance, and keeps only the survivors. The cheapest test is the check weight, which we read straight off the base matrices. The code distance is one of the most important parameters we filter on. As proved in Theorem~\ref{thm:dc_ub_dq}, under the square-invertibility condition of Definition~\ref{def:square_inv_condition}, the distance of a quantum code is upper bounded by the distances of its classical base codes, so we can discard weak candidates before ever constructing the quantum code (Appendix~\ref{app:distance_bounds}). To be more specific, we first check the square-invertibility condition, and only preserve those satisfying it; then we consider the classical distance bounds based on the pattern of the classical codes $A$ and $B$. For example, when $G$ is abelian, we can use distance bounds for classical QC-LDPC codes (Appendix~\ref{app:distance_bounds}), such as Theorems~\ref{thm:perm-poly} and~\ref{thm:perm-wmat}, as filters. When $G$ is non-abelian and $r_1=r_2=1$, we may use other distance bounds, such as the commutator-subgroup bound of Proposition~\ref{prop:comm-dist-bound}, as filters. We certify the survivors with \textsf{sQetch}, a fast GPU-based distance estimator developed for this purpose that runs roughly $10^{5}\times$ faster than previous quantum CSS code distance estimators (Appendix~\ref{app:sqetch}). Finally, we account for errors that spread through the syndrome-extraction circuit itself.
Only codes that clear all three tests reach a full memory-experiment simulation with our telescoping decoder.

A high-rate code is more useful as a processor if each of its logical qubits can be addressed individually with a low-weight operator, which is a prerequisite for all three instruction sets ($\bquis, \hquis, \fquis$) to be efficient. This is in general difficult for high-rate processors as the supports of the logical representatives overlap with each other. We enforce this in the third step of the pipeline by keeping only codes whose base matrices satisfy the square-invertibility condition (Definition~\ref{def:square_inv_condition}), which guarantees a group-structured canonical logical basis (Appendix~\ref{app:canonical-basis}). Every processor the pipeline returns therefore has individually addressable logical operators by construction. From this basis we assemble the logical machinery, namely the reusable surgery gadgets of $\bquis$, the high-rate surgery and parallel magic-state injection of $\hquis$, and the full extractor of $\fquis$ (Appendix~\ref{app:surgery} and Appendix~\ref{app:parallel_magic}), turning the memory into a universal fault-tolerant processor. We establish the fault tolerance of these gadgets with theoretical proofs and benchmark them with the telescoping decoder.

The final stage of the pipeline matches each surviving code to its target hardware. On atom arrays, the pipeline optimizes the placement, labeling, and movement of the atoms based on the group and product structure of the mitten codes. We optimize and benchmark the cost for both atom-array and superconducting layouts; this stage fixes the cycle time of the resulting quantum processors (Appendix~\ref{app:experimental_complexity}). Validation at any stage can feed back to the earlier ones, as shown by the dashed arrows in Fig.~\ref{fig:pipeline}, so a design that misses a distance, decoding, or hardware target\textemdash such as a shortfall in processing capacity, throughput, or cycle time\textemdash is retuned by changing the group, the base matrices, or the schedule.

Running the pipeline end to end yields the family of processors in Table~\ref{tab:1x2_processor_code_param}, with further families in Table~\ref{tab:code_param} below; the explicit construction data for the processor codes are given in Table~\ref{tab:code_construction_full} in Appendix~\ref{app:construction_data}.

\begin{table}[t]
\centering
\setlength{\tabcolsep}{4pt}
\renewcommand{\arraystretch}{1.1}
\begin{tabular}{@{}l c c@{}}
\toprule
\textbf{$\llbracket n,k,d \rrbracket $} & \textbf{code family} & $\wt{Lx}/\wt{Lz}$\\
\midrule
$\llbracket 150,30,10 \rrbracket $   & mitten     & 18/10\\
$\llbracket 200,40,12 \rrbracket $   & mitten     & 20/18\\
$\llbracket 300,60,14 \rrbracket $   & mitten     & 22/22\\
$\llbracket 500,100,16 \rrbracket $  & mitten    & 28/24\\
$\llbracket 540,108,18 \rrbracket $  & mitten     & 22/28\\
$\llbracket 630,126,\leq 20 \rrbracket $  & mitten    & 28/44\\
$\llbracket 780,156,\leq 22 \rrbracket $  & mitten    & 74/84\\
$\llbracket 975,195,\leq 24 \rrbracket $  & mitten   & 102/92\\
$\llbracket 560,112,\leq 14 \rrbracket $  & $2\times 4$ polynomial abelian LP   & $\{22,30\}/\{22,30\}$\\
$\llbracket 280,56,\leq 10 \rrbracket $   & structured mitten  & 30/34\\
$\llbracket 330,66,\leq 12 \rrbracket $   & structured mitten  & 20/30\\
$\llbracket 600,120,\leq 14 \rrbracket $  & structured mitten  & 28/40\\
$\llbracket 600,120,\leq 16 \rrbracket $  & structured mitten  & 40/40\\
$\llbracket 840,168,\leq 18 \rrbracket $  & structured mitten & 80/84\\
\bottomrule
\end{tabular}
\caption{Parameters of the processor code instances. $\wt{L_X}/\wt{L_Z}$ are the weights of the canonical $X$/$Z$ logical basis operators. For the $1\times 2$ LP families, each basis is weight-uniform (all $k$ $X$-logicals share one weight, and likewise the $Z$-logicals), since the canonical basis is a single $G$-orbit. For the $2 \times 4$ abelian LP codes, the canonical basis spans two block-columns and is not single-weight, so we list the weight set $\{\cdot,\cdot\}$. For distance estimation, the distances of the first $5$ codes are exactly computed; distances for the rest are estimated using at least $50$M rounds of \textsf{sQetch} with $50$k rounds of a BP+OSD-based distance estimator.}
\label{tab:code_param}
\end{table}

\begin{table}[t]
\centering
\setlength{\tabcolsep}{4pt}
\renewcommand{\arraystretch}{1.1}
\begin{tabular}{@{}l l l l l@{}}
\toprule
\textbf{$\llbracket n,k,d \rrbracket $} & \textbf{group structure} & \textbf{GAP ID} & \textbf{order relation} & \textbf{conjugation relation} \\
\midrule
$\llbracket 150,30,10 \rrbracket $   & $C_5 \times S_3$                  & $(30,\,1)$   & --- & --- \\
$\llbracket 200,40,12 \rrbracket $   & $C_4 \times D_{10}$               & $(40,\,5)$   & --- & --- \\
$\llbracket 300,60,14 \rrbracket $   & $C_{10} \times S_3$               & $(60,\,11)$  & --- & --- \\
$\llbracket 500,100,16 \rrbracket $  & $C_5 \rtimes C_{20}$              & $(100,\,9)$  & $\langle c, y \mid c^{5} = y^{20} = e \rangle$ & $\varphi_{y}(c) = y\,c\,y^{-1} = c^{3}$ \\
$\llbracket 540,108,18 \rrbracket $  & $C_9 \rtimes C_{12}$              & $(108,\,9)$  & $\langle c, y \mid c^{9} = y^{12} = e \rangle$ & $\varphi_{y}(c) = y\,c\,y^{-1} = c^{5}$ \\
$\llbracket 630,126, \leq 20 \rrbracket $  & $C_7 \rtimes C_{18}$              & $(126,\,1)$  & $\langle a, b \mid a^{7} = b^{18} = e \rangle$ & $\varphi_{b}(a) = b\,a\,b^{-1} = a^{5}$ \\
$\llbracket 780,156, \leq 22 \rrbracket $  & $C_{13} \times A_4$               & $(156,\,13)$ & --- & --- \\
$\llbracket 975,195,\leq 24 \rrbracket $  & $C_{13} \rtimes C_{15}$           & $(195,\,1)$  & $\langle c, w \mid c^{13} = w^{15} = e \rangle$ & $\varphi_{w}(c) = w\,c\,w^{-1} = c^{9}$ \\
\midrule
$\llbracket 330,66,\leq 12 \rrbracket $   & $C_{11} \times S_3$               & $(66,\,1)$   & --- & --- \\
$\llbracket 600,120,\leq14 \rrbracket $  & $C_5 \times S_4$                  & $(120,\,37)$ & --- & --- \\
$\llbracket 600,120,\leq 16 \rrbracket $  & $C_5 \times S_4$                  & $(120,\,37)$ & --- & --- \\
$\llbracket 840,168,\leq 18 \rrbracket $  & $C_7 \times S_4$                  & $(168,\,45)$ & --- & --- \\
\bottomrule
\end{tabular}
\caption{Explicit descriptions of the code groups. The GAP ID $(|G|, i)$ identifies the group $G$ as \texttt{SmallGroup($|G|$,\,$i$)} in GAP's small-group library \cite{GAP4, SmallGrp}. For the direct products of a cyclic group with a standard group ($S_3$, $D_{10}$, $A_4$) the factor names already determine the group, and no relations are needed. Generators and conjugation relations are given when the group involves a genuine semidirect product, because there the symbol $\rtimes$ alone does not determine the group. Different twists $\varphi$ can yield non-isomorphic groups with the same factors. The order relations column lists the generators, and the conjugation relations column gives the twist automorphisms $\varphi_b(a) = b \, a \, b^{-1}$ of Definition~\ref{def:semidirect}. Throughout the paper, when we write $G_1 \rtimes G_2$, we mean $G_1 \rtimes_{\varphi} G_2$, where $\varphi$ is defined in this table.}
\label{tab:semidirect}
\end{table}

\section{Distance bounds for lifted product codes}
\label{app:distance_bounds}
In this appendix we prove upper bounds on the distance of a lifted product code $\LP(A,B)$ in terms of algebraic properties of its base matrices $A\in\R^{r_1\times c_1}$ and $B\in\R^{r_2\times c_2}$ and of the group $G$. The starting point is Theorem~\ref{thm:dc_ub_dq}: whenever the binary matrices $\Lrep{A}$ and $\Rrep{B}$ have full row rank --- which holds in particular under the square invertibility condition of Definition~\ref{def:square_inv_condition}, and hence for every base matrix used in this paper (Remark~\ref{rem:applicability}) --- the distances of the two classical base codes upper-bound the quantum distance. Any construction of a low-weight classical codeword therefore caps the quantum distance, and the rest of the appendix gives three such constructions. The commutator subgroup bound builds a codeword from the sum over the commutator subgroup $[G,G]$ and is strongest when $G$ is close to abelian. The element-order bound applies whenever some entry of a base matrix has the form $1+g$, in which case the distance is at most the order of $g$. Finally, for abelian $G$ the base codes are classical quasi-cyclic LDPC codes, and we adapt the permanent-based distance bounds known for that family. These abelian bounds are tight: searching with our processor discovery pipeline and the distance estimation algorithm $\textsf{sQetch}$ (Appendix~\ref{app:sqetch}), we find a $\llbracket 560,112,\leq14 \rrbracket$ code that attains the distance-$14$ ceiling for $2\times4$ base matrices of check weight at most $9$.

\subsection{Classical and quantum code distance}
In this section we prove that the minimum distance of the two classical base codes $A$ and $B$ upper-bounds the distance of the lifted product code $\LP(A,B)$ provided $\Lrep{A}$ and $\Rrep{B}$ have full row rank. We also show that the full row rank condition cannot be
removed in general by exhibiting a counterexample where the classical base codes are not full row rank and have distance less than the lifted product code.
The precise statement is the following.

\begin{theorem}[Classical distance upper bound of quantum distance]
\label{thm:dc_ub_dq}
Let $\R=\F_2[G]$ for a finite group $G$, and let
$A\in\R^{r_1\times c_1}$ and $B\in\R^{r_2\times c_2}$ with $r_1<c_1$ and
$r_2<c_2$. Assume that the binary matrices
\[
    \Lrep{A}\in\F_2^{r_1|G|\times c_1|G|}
    \qquad\text{and}\qquad
    \Rrep{B}\in\F_2^{r_2|G|\times c_2|G|}
\]
have full row rank. Define the classical distances of the two
base codes,
\begin{equation}
    d_A=\min\bigl\{\wt{u}: 0\neq u\in\ker\Lrep{A}\bigr\},
    \qquad
    d_B=\min\bigl\{\wt{v}: 0\neq v\in\ker\Rrep{B}\bigr\},
\end{equation}
where $\wt{\cdot}$ is the Hamming weight of a binary vector. Let $H_X$
and $H_Z$ be the check matrices of $\LP(A,B)$
(Definition~\ref{def:LP-code}), and define
\begin{equation}
\begin{split}
    d_x &=\min\bigl\{\wt{x}:
        x\in\ker H_Z\setminus\rowspace(H_X)\bigr\},\\
    d_z &=\min\bigl\{\wt{z}:
        z\in\ker H_X\setminus\rowspace(H_Z)\bigr\},
\end{split}
\end{equation}
where $\rowspace(H)$ is the linear space spanned by the row of $H$, and $d_Q=\min(d_x,d_z)$. Then
\begin{equation}
    d_x\leq d_B,
    \qquad
    d_z\leq d_A,
\end{equation}
and hence
\begin{equation}
    d_Q\leq\min(d_A,d_B).
\end{equation}
\end{theorem}

\begin{remark}
\label{rem:applicability}
Note that the square invertibility condition of Definition~\ref{def:square_inv_condition} is a special case of the full row rank condition, and so Theorem~\ref{thm:dc_ub_dq} applies to all codes we search for in this paper.
\end{remark}

We will prove $d_x\leq d_B$ by exhibiting an explicit logical $X$ operator $X^\star$ of weight $d_B$; the bound $d_z\leq d_A$ follows from an almost identical argument. The physical qubits of a lifted product code come in two blocks, of sizes $c_1c_2|G|$ and $r_1r_2|G|$, and so the support of an $X$ or $Z$ type Pauli operator over these physical qubits is naturally described as a pair of matrices over the group algebra $\R=\F_2[G]$, of shapes $c_1\times c_2$ and $r_1\times r_2$ respectively (Lemma~\ref{lemma:ring-form} below makes this precise). In this
language the candidate operator $X^\star$ is easy to describe: take a minimum-weight codeword of the classical base code $B$ and place it in a single row of the $c_1\times c_2$ matrix, setting all other entries, and the entire $r_1\times r_2$ matrix to zero. The weight of $X^\star$ is then $d_B$ by construction, and it is straightforward to check that it commutes with all the $Z$ stabilizers.

The real content of the proof is showing that the candidate $X^\star$ is not simply a product of $X$ stabilizers. The remainder of this section develops the proof in two steps. We first recast $X$ and $Z$ type Pauli operators on the physical qubits as pairs of matrices over $\R$ (Lemma~\ref{lemma:ring-form}), so that we may work over the ring for the remainder of the proof. We then show that the full row rank condition allows us to always construct such a candidate that is not a product of $X$ stabilizers (Lemma~\ref{lem:avoid_stabs}).

Over the ring, the minimum-distance codewords of the two classical base codes come from the kernels of the base matrices,
\[
    K_A\coloneqq\{u\in\R^{c_1}:Au=0\},
    \qquad
    K_B\coloneqq\{v\in\R^{1\times c_2}:vB^{*}=0\}.
\]
In binary coordinates, Proposition~\ref{prop:ring_bin_compat}(i) gives
$\ker\Lrep{A}=\{\b{u}:u\in K_A\}$ and
$\ker\Rrep{B}=\{\b{v}:v\in K_B\}$, so $d_A$ and $d_B$ are the minimum
weights of $\b{u}$ and $\b{v}$ over nonzero $u\in K_A$ and
$v\in K_B$, respectively.

To describe how we will construct logical operators from these minimum-weight codewords, it will be helpful to express $X$ and $Z$ type Pauli operators on $\LP(A,B)$ as pairs of matrices over $\R$.
The physical qubits of $\LP(A,B)$ come in two blocks, of sizes
$c_1c_2|G|$ and $r_1r_2|G|$, so an $X$ or $Z$ type Pauli operator is described by a binary vector split across these two blocks. Since the checks are built from Kronecker products, they act most simply when each block of this vector is arranged as a matrix. Accordingly, we arrange the first block as $V\in\R^{c_1\times c_2}$ and the second as $U\in\R^{r_1\times r_2}$, so that the full binary vector is $\bigl(\bvec{V},\,\bvec{U}\bigr)$, where $\operatorname{vec}$ denotes the row-major flattening of~\eqref{eq:vec-def}. In this form, the checks act by left and right matrix multiplication. 

\begin{lemma}[Kernels and row spaces over the ring]
\label{lemma:ring-form}
Let $H_X$ and $H_Z$ be the check matrices of the lifted product code $\LP(A,B)$ as in Definition~\ref{def:LP-code}. Then the map $(V,U)\mapsto\bigl(\bvec{V},\,\bvec{U}\bigr)$ is an
$\F_2$-linear bijection from
$\R^{c_1\times c_2}\oplus\R^{r_1\times r_2}$ onto
$\F_2^{(c_1c_2+r_1r_2)|G|}$, under which:
\begin{enumerate}
    \item[(i)] $\begin{alignedat}[t]{2}
        &\ker H_X&&=\bigl\{\bigl(\bvec{V},\,\bvec{U}\bigr):
            AV+UB=0\bigr\},\\
        &\ker H_Z&&=\bigl\{\bigl(\bvec{V},\,\bvec{U}\bigr):
            VB^{*}+A^{*}U=0\bigr\};
    \end{alignedat}$
    \item[(ii)] $\begin{alignedat}[t]{2}
        &\rowspace(H_X)&&=\bigl\{\bigl(\bvec{A^{*}S},\,\bvec{SB^{*}}\bigr):
            S\in\R^{r_1\times c_2}\bigr\},\\
        &\rowspace(H_Z)&&=\bigl\{\bigl(\bvec{TB},\,\bvec{AT}\bigr):
            T\in\R^{c_1\times r_2}\bigr\}.
    \end{alignedat}$
\end{enumerate}
\end{lemma}

\begin{proof}
The map is a bijection since it is the composition of the $\operatorname{vec}$ operator with the entrywise binary expansion $\bin$, both of which are $\F_2$-linear bijections. We prove the two $H_X$ statements; the $H_Z$
statements follow from the same argument applied to the blocks of $H_Z$ in~\eqref{eq:LP-HZ}.

The two blocks of $H_X$ in Definition~\ref{def:LP-code} are tensor products with an identity factor, so
Proposition~\ref{prop:ring_bin_compat} yields
\[
    \Lrep{A\otimes I_{c_2}}\,\bvec{V}=\bvec{AV},
    \qquad
    \Rrep{I_{r_1}\otimes B^{*}}\,\bvec{U}=\bvec{UB},
\]
and adding the two blocks gives
\begin{equation}
\label{eq:check-actions}
    H_X\,\bigl(\bvec{V},\,\bvec{U}\bigr)=\bvec{AV+UB}.
\end{equation}
By \eqref{eq:check-actions} and the injectivity of
$\b{\operatorname{vec}(\cdot)}$ on the syndrome space
$\R^{r_1\times c_2}$, we have $H_X\,\bigl(\bvec{V},\,\bvec{U}\bigr)=0$ if and only if $AV+UB=0$; since every vector over the physical qubits is $\bigl(\bvec{V},\,\bvec{U}\bigr)$ for some pair $(V,U)$, this proves (i).

For (ii), the row space of $H_X$ is the image of $H_X^{T}$, and every binary vector in the $X$-syndrome space $\F_2^{r_1c_2|G|}$ equals $\bvec{S}$ for a unique $S\in\R^{r_1\times c_2}$; hence
\[
    \rowspace(H_X)
    =\bigl\{H_X^{T}\,\bvec{S}: S\in\R^{r_1\times c_2}\bigr\}.
\]
Using $\Lrep{M}^{T}=\Lrep{M^{*}}$ and $\Rrep{M}^{T}=\Rrep{M^{*}}$
followed by $(M\otimes I)^{*}=M^{*}\otimes I$ and
$(I\otimes N)^{*}=I\otimes N^{*}$,
\[
    H_X^{T}
    =\begin{bmatrix}
        \Lrep{A^{*}\otimes I_{c_2}}\\[2pt]
        \Rrep{I_{r_1}\otimes B}
    \end{bmatrix},
\]
so Proposition~\ref{prop:ring_bin_compat} gives
$H_X^{T}\,\bvec{S}=\bigl(\bvec{A^{*}S},\,\bvec{SB^{*}}\bigr)$, and (ii) follows.
\end{proof}

In light of Lemma~\ref{lemma:ring-form}, we identify the space of physical qubits with $\R^{c_1\times c_2}\oplus\R^{r_1\times r_2}$ and regard $\ker H_X$, $\ker H_Z$, $\rowspace(H_X)$, $\rowspace(H_Z)$ as ring-level sets. Working in this picture will allow us to construct our candidate logical operator and prove that it is indeed a logical operator and not a product of stabilizers.

Write $e_i$ for the $i$-th standard basis vector of $\R^{c_1}$ that has the identity element $e$ in the $i$-th position and zeros elsewhere. For $g\in G$ and a row vector $v_B\in\R^{1\times c_2}$, the product $e_ig\,v_B\in\R^{c_1\times c_2}$ is the matrix whose $i$-th row is $gv_B$ and whose remaining rows are zero; matrices of this form will correspond to the support of our candidate $X$ logical operators, which live entirely on the first block of physical qubits. By Lemma~\ref{lemma:ring-form}(ii), the products of $X$ stabilizers are exactly the pairs $(A^{*}S,\,SB^{*})$ with $S\in\R^{r_1\times c_2}$, so if the pair $(e_ig\,v_B,\,0)$ were a product of $X$ stabilizers, then $e_ig\,v_B=A^{*}S$ for some $S$. The next lemma shows that the row index $i$ and the group element $g$ can always be chosen so that no such $S$ exists. Part (ii) of the lemma is the corresponding statement for the candidate $Z$ logical operators. Writing $f_j$ for the analogous $j$-th standard basis row vector of $\R^{1\times c_2}$, the candidate $(u_A\,gf_j,\,0)$ carries the column $u_Ag$ in its $j$-th column, and the lemma chooses $j$ and $g$ so that it is not among the $Z$ stabilizer products $(TB,\,AT)$, $T\in\R^{c_1\times r_2}$.

\begin{lemma}[Avoiding the stabilizers]
\label{lem:avoid_stabs}
Suppose that $\Lrep{A}$ and $\Rrep{B}$ have full row rank, with
$r_1<c_1$ and $r_2<c_2$. Then:
\begin{enumerate}
    \item[(i)] for every nonzero row vector $v_B\in\R^{1\times c_2}$ there exists $i\in\{1,\dots,c_1\}$ and $g\in G$ such that $e_ig\,v_B\neq A^{*}S$ for any $S\in\R^{r_1\times c_2}$; \item[(ii)] for every nonzero column vector $u_A\in\R^{c_1}$ there exists $j\in\{1,\dots,c_2\}$ and $g\in G$ such that $u_A\,gf_j\neq TB$ for any $T\in\R^{c_1\times r_2}$.
\end{enumerate}
\end{lemma}

\begin{proof}
We prove (i), the statement for $\Lrep{A}$; the proof of (ii) is similar but with the roles of left and right multiplication exchanged.

Since $\Lrep{A}$ is full row rank, Proposition~\ref{prop:ring_bin_compat}(i) implies that the map $u\mapsto Au$ corresponding to left multiplication by $A$ is a surjective map from $\R^{c_1}$ onto $\R^{r_1}$. Let $e_k$ be the $k$-th standard basis vector of $\R^{r_1}$ that has the identity element $e$ in the $k$-th position and zeros elsewhere. By surjectivity we may choose preimages $s_k\in\R^{c_1}$ with $As_k=e_k$ for $k=1,\dots,r_1$; collecting them as the columns of $S_A\in\R^{c_1\times r_1}$ gives $AS_A=I_{r_1}$.

Since the involution reverses products, this implies
$S_A^{*}A^{*}=I_{r_1}$, so $P_A\coloneqq I_{c_1}-A^{*}S_A^{*}$ satisfies
\begin{equation}
\label{eq:splitting-A}
    P_AA^{*}=A^{*}-A^{*}(S_A^{*}A^{*})=0,
    \qquad
    I_{c_1}=A^{*}S_A^{*}+P_A.
\end{equation}
Since $P_AA^{*}=0$, the projector $P_A$ annihilates every matrix of the form $A^{*}S$. Hence, it suffices to find $i$ and $g$ such that $(P_A\,e_ig)\,v_B\neq0$, since then $e_ig\,v_B$ cannot equal $A^{*}S$ for any $S$.

Suppose then, for contradiction, that $(P_A\,e_ig)\,v_B=0$ for every $i$ and $g$. Form the orbit space
\[
    M_B\coloneqq \R v_B
    =\operatorname{span}_{\F_2}\{gv_B:g\in G\}
    \subseteq\R^{1\times c_2}.
\]
By associativity,
$0=(P_A\,e_ig)\,v_B=(P_Ae_i)(gv_B)$ for every $i$ and $g$, and since the vectors $gv_B$ span $M_B$,
\begin{equation}
\label{eq:selector-vanishing}
    (P_Ae_i)\,m=0
    \qquad\text{for every $i$ and every $m\in M_B$.}
\end{equation}

Now view $M_B^{\oplus c_1}$ as the space of $c_1\times c_2$ matrices whose rows lie in $M_B$. Any $V \in M_B^{\oplus c_1}$ can be written as $V=\sum_{i=1}^{c_1}e_i m_i$ with $m_i\in M_B$. Since $P_AV=0$ by~\eqref{eq:selector-vanishing}, substituting $I_{c_1}=A^{*}S_A^{*}+P_A$ from~\eqref{eq:splitting-A} gives
\[
    V=I_{c_1}V=(A^{*}S_A^{*}+P_A)V=A^{*}(S_A^{*}V).
\]
Every row of $S_A^{*}V$ is a left $\R$-linear combination of rows of $V$, so it still lies in $M_B$. Consequently, left multiplication by $A^{*}$ is a surjective map from $M_B^{\oplus r_1}$ onto $M_B^{\oplus c_1}$. But this is impossible since $r_1<c_1$. Hence, some pair $(i,g)$ satisfies $(P_A\,e_ig)\,v_B\neq0$, and the lemma follows.
\end{proof}

We can now assemble the proof of Theorem~\ref{thm:dc_ub_dq}.

\begin{proof}[Proof of Theorem~\ref{thm:dc_ub_dq}]
We work in the ring picture of Lemma~\ref{lemma:ring-form} and first prove $d_x\leq d_B$. Pick $v_B\in K_B$ such that $\wt{\b{v_B}}=d_B$. By Lemma~\ref{lem:avoid_stabs}(i) there exists $i\in\{1,\dots,c_1\}$ and $g\in G$ such that $e_ig\,v_B\neq A^{*}S$ for any $S\in\R^{r_1\times c_2}$. Set
\begin{equation}\label{eq:Xstar}
    X^{\star}\coloneqq(e_ig\,v_B,\;0).
\end{equation}
We claim that $X^{\star}$ is a nontrivial logical $X$ operator of weight $d_B$. First, $X^{\star}$ commutes with all the $Z$ stabilizers. Since $v_B\in K_B$,
\[
    (e_ig\,v_B)B^{*}+A^{*}\cdot 0
    =e_ig\,(v_BB^{*})
    =0,
\]
so $X^{\star}\in\ker H_Z$ by Lemma~\ref{lemma:ring-form}(i). Second, $X^{\star}$ is not a product of $X$ stabilizers i.e. $X^{\star}\notin\rowspace(H_X)$. If
it were, then by Lemma~\ref{lemma:ring-form}(ii)
its first component $e_i g v_B$ would equal $A^{*}S$ for some $S$, contradicting the choice of $i$ and $g$. Hence $X^{\star}$ is a nontrivial logical $X$ operator. Finally, the matrix $e_ig\,v_B$ has a single nonzero row, equal to $gv_B$, so in binary coordinates $\wt{X^{\star}}=\wt{\b{gv_B}}=\wt{\b{v_B}}=d_B$. Hence $d_x\leq d_B$.

The bound $d_z\leq d_A$ follows from a similar argument, using Lemma~\ref{lem:avoid_stabs}(ii) to build
\begin{equation}\label{eq:Zstar}
    Z^{\star}\coloneqq(u_A\,gf_j,\;0)
\end{equation}
from a minimum-weight
$u_A\in K_A$. Combining the two bounds gives
$d_Q=\min(d_x,d_z)\leq\min(d_A,d_B)$.
\end{proof}

\begin{remark}
One might hope to prove an unconditional statement about the classical distance upper bounding the quantum distance of the lifted product code. However, we present a counterexample that demonstrates this is not possible and that one cannot in general get rid of the full row rank condition in Theorem~\ref{thm:dc_ub_dq}. 

Let $G=C_4\times C_2=\langle x,y\mid x^4=y^2=e,\ xy=yx\rangle$, so that $|G|=8$, let $\R=\F_2[G]$, and consider the base matrices $A,B\in\R^{3\times4}$ given by
\begingroup
\setlength{\arraycolsep}{6.5pt}
\begin{align}
    A&=
    \begin{pmatrix}
        1+x+xy+x^3y & 1+y+xy+x^2+x^3 & 1+xy+x^2y & x+xy+x^2y+x^3\\[1mm]
        y+x^2+x^2y+x^3 & 1+x+x^2y & 1+x^2y+x^3 & 1+y+x+xy+x^3+x^3y\\[1mm]
        1+y+x^2y+x^3y & 1+y+x+x^3+x^3y & 1+x+x^2y+x^3y & x+x^3+x^3y
    \end{pmatrix},
    \nonumber\\[3mm]
    B&=
    \begin{pmatrix}
        xy+x^2y+x^3+x^3y & x+xy+x^2+x^2y+x^3 & xy+x^2y & xy+x^2y+x^3\\[1mm]
        y+x+x^2+x^3 & 1+x+x^3 & x+x^2+x^3+x^3y & x^2+x^3+x^3y\\[1mm]
        1+x+x^3y & y+x^3 & y+x+xy+x^3y & 1+xy+x^2+x^3
    \end{pmatrix}.
\end{align}
\endgroup
Both matrices have the strictly wide shape required by Theorem~\ref{thm:dc_ub_dq}, with $r_1=r_2=3<4=c_1=c_2$ but do not satisfy the full row rank condition. Indeed, $\operatorname{rank}_{\F_2}\Lrep{A}=22$ and $\operatorname{rank}_{\F_2}\Rrep{B}=23$, both strictly smaller than $r_1|G|=r_2|G|=24$. Constructing bases for $\ker\Lrep{A}$ and $\ker\Rrep{B}$ and enumerating all $\F_2$-linear combinations of the basis vectors, we find that the classical distances of the base codes are $d_A=d_B=8$. Using integer programming, we then verified that the distance of $\LP(A,B)$ is exactly $10$, so $\LP(A,B)$ is a $\llbracket200,13,10\rrbracket$ code with
\begin{equation}
    d_Q=10>8=d_A=d_B,
\end{equation}
and the conclusion of Theorem~\ref{thm:dc_ub_dq} fails without the full row rank condition.
\end{remark}

We obtained the above counterexample by a guided randomized search over rank-deficient $3\times 4$ matrices over $\mathbb{F}_2[G]$ where $G = C_4 \times C_2$. Let $\Omega=\sum_{g\in G}g$. Since $g\Omega=\Omega$ for every $g\in G$, multiplying $\Omega$ by any group-algebra element containing an even number of group elements gives zero; hence, if every entry in one column of a base matrix has even weight, the vector supported on that column with entry $\Omega$ is a codeword of binary weight $|G|=8$. In the matrices in the above counterexample this occurs in the first column of $A$ and the third column of $B$, giving the codewords $u_A=e_1\Omega=(\Omega,\,0,\,0,\,0)^T\in K_A$ and $v_B=\Omega f_3=(0,\,0,\,\Omega,\,0)\in K_B$. We sampled the remaining entries randomly and rejected a matrix unless exhaustive enumeration of its kernel showed that the planted codeword was its unique minimum-weight codeword. After pairing candidate matrices $A$ and $B$, we rejected the pair whenever any of the candidate operators~\eqref{eq:Xstar}--\eqref{eq:Zstar} was not a stabilizer.
Since $g\Omega=\Omega$, these candidate logical operators reduce to $X_i^\star=(e_i\Omega f_3,0)$ and $Z_j^\star=(e_1\Omega f_j,0)$, and we checked they were stabilizers by solving $A^*S=e_i\Omega f_3$ and $SB^*=0$, and $TB=e_1\Omega f_j$ and $AT=0$, for every $i$ and $j$. We constructed $H_X$ and $H_Z$ and used an exhaustive search to verify that every operator in $\ker H_X$ and $\ker H_Z$ of weight at most $D=\max\{d_A,d_B\}$ was a stabilizer, thereby certifying $d_Q>D$. Finally, we verified  the distance of $\LP(A,B)$ with integer programming.

    \subsection{The commutator subgroup bound}

    We begin by recalling the definition of a commutator subgroup and establishing a few key properties.

    \begin{definition}[Commutator and commutator subgroup]\label{def:commutator}
        Let $G$ be a finite group.  For elements $g,h\in G$, the \emph{commutator} of $g$ and $h$ is
        \begin{equation}
            [g,h] = g\,h\,g^{-1}\,h^{-1}.
        \end{equation}
        The \emph{commutator subgroup} (or \emph{derived subgroup}) of $G$ is the subgroup generated by all commutators:
        \begin{equation}
            [G,G] = \langle\, [g,h] : g,h\in G \,\rangle.
        \end{equation}
        The commutator subgroup $[G,G]$ is a normal subgroup of $G$, and the quotient $G/[G,G]$ is the abelianization of $G$.  The group $G$ is abelian if and only if $[G,G]=\{e\}$.
    \end{definition}

    \begin{definition}[Commutator subgroup sum]\label{def:comm-sum}
        Let $K=[G,G]$ denote the commutator subgroup of $G$.  Define the element
        \begin{equation}
            \sigma = \sum_{k\in K} k \;\in\; \F_2[G].
        \end{equation}
    \end{definition}

    The following proposition collects the key algebraic properties of $\sigma$ that will be used in the distance analysis.

    \begin{proposition}[Properties of the commutator subgroup sum]\label{prop:comm-sum-props}
        Let $G$ be a finite group with commutator subgroup $K=[G,G]$, and let $\sigma=\sum_{k\in K}k\in\F_2[G]$.  Then:
        \begin{enumerate}
            \item $\sigma$ is central in $\F_2[G]$, i.e.\ $g\sigma=\sigma g$ for all $g\in G$.
            \item For any $c\in K$, $(1+c)\,\sigma=0$.
            \item For all $a,b\in\F_2[G]$, $(ab+ba)\,\sigma=0$.
        \end{enumerate}
    \end{proposition}

    \begin{proof}
        \textit{(1)} Since $K$ is a normal subgroup of $G$, for any $g\in G$ we have $gKg^{-1}=K$, so $gK=Kg$ as sets.  Therefore
        \[
            g\sigma = \sum_{k\in K}gk = \sum_{k'\in gK}k' = \sum_{k'\in Kg}k' = \sum_{k\in K}kg = \sigma g.
        \]

        \textit{(2)} For $c\in K$, left multiplication by $c$ is a bijection on $K$, so $c\sigma=\sum_{k\in K}ck = \sum_{k'\in K}k' = \sigma$.  Hence $(1+c)\sigma = \sigma + c\sigma = \sigma + \sigma = 0$ over~$\F_2$.

        \textit{(3)} It suffices to verify the identity on group elements $g,h\in G$, since the general case follows by bilinearity.  The commutator identity $gh = [g,h]\,hg$ gives
        \[
            (gh + hg)\,\sigma = \bigl(1+[g,h]\bigr)\,hg\,\sigma.
        \]
        By part~(1), $hg\sigma = \sigma hg$, so
        \[
            \bigl(1+[g,h]\bigr)\,hg\,\sigma = \bigl(1+[g,h]\bigr)\,\sigma\,hg = 0,
        \]
        where the last equality uses part~(2) with $c=[g,h]\in K$. For general $a=\sum_i\alpha_i g_i$ and $b=\sum_j\beta_j h_j$ in $\F_2[G]$, bilinearity gives
        \[
            (ab+ba)\,\sigma = \sum_{i,j}\alpha_i\beta_j\,(g_i h_j + h_j g_i)\,\sigma = 0.
        \]
    \end{proof}

    We now apply this algebraic machinery to bound the minimum distance of classical codes defined over the group algebra.  Let $A\in\R^{r_1\times c_1}$ be a matrix over $\R=\F_2[G]$.  Define the \emph{lifted classical code} associated with $A$ as the binary linear code with parity-check matrix
    \begin{equation}
        A_L = \Lrep{A},
    \end{equation}
    and let $d_{A_L}$ denote its minimum distance.  That is, $d_{A_L}$ is the minimum Hamming weight of a nonzero vector $u\in\F_2^{c_1|G|}$ satisfying $\Lrep{A}\,u=0$.

    \begin{proposition}[Commutator distance bound]\label{prop:comm-dist-bound}
        Let $G$ be a finite group with commutator subgroup $K=[G,G]$, and let $\sigma=\sum_{k\in K}k$.  Let $A=[a_1,\ldots,a_m]\in\R^{1\times m}$ with $a_i=\sum_{g\in I_i}g$ for index sets $I_i\subseteq G$.  Then for any pair $1\leq i<j\leq m$ such that $a_i\,\sigma\neq 0$ or $a_j\,\sigma\neq 0$,
        \begin{equation}
            d_{A_L}\leq (|I_i|+|I_j|)\cdot|K|.
        \end{equation}
    \end{proposition}

    \begin{proof}
        Define $w\in\R^m$ by $w_i=a_j\,\sigma$, $w_j=a_i\,\sigma$, and $w_\ell=0$ for $\ell\notin\{i,j\}$.  Let $\tilde{w}\in\F_2^{m|G|}$ denote the binary vector obtained by expanding each entry $w_\ell\in\R$ into its coefficient vector in $\F_2^{|G|}$.  We claim that $\tilde{w}$ is a nonzero codeword of $A_L$. Since $A$ has a single row, $A_L\tilde{w}$ is the binary expansion of $Aw=\sum_\ell a_\ell w_\ell$.  Computing:
        \[
            \Lrep{A}\tilde{w} = \sum_{\ell=1}^m a_\ell w_\ell = a_i w_i + a_j w_j = a_i a_j\,\sigma + a_j a_i\,\sigma = (a_i a_j + a_j a_i)\,\sigma = 0,
        \]
        where the last equality follows from Proposition~\ref{prop:comm-sum-props}(3).

        All that remains is to bound the weight of $\tilde{w}$. Write $w_i=a_j\,\sigma=\sum_{g\in I_j}g\sigma=\sum_{g\in I_j}\sum_{k\in K}gk$.  The elements $gk$ for $k\in K$ form the coset $gK$, so $a_j\,\sigma=\sum_{g\in I_j}\mathbf{1}_{gK}$, where $\mathbf{1}_{gK}$ denotes the formal sum of elements in $gK$.  Over $\F_2$, group elements appearing an even number of times in the sum cancel, so $\operatorname{wt}_G(w_i)\leq|I_j|\cdot|K|$.  Similarly $\operatorname{wt}_G(w_j)\leq|I_i|\cdot|K|$.  Hence
        \[
            d_{A_L}\leq\operatorname{wt}(\tilde{w})=\operatorname{wt}_G(w_i)+\operatorname{wt}_G(w_j)\leq(|I_j|+|I_i|)\cdot|K|. \qedhere
        \]
    \end{proof}

    \begin{remark}
        For groups of the form $G=C_N\times G'$ where $G'$ is nonabelian, the commutator subgroup is $K=[G,G]=\{e\}\times[G',G']$, since $C_N$ is abelian.  In this case $|K|=|[G',G']|$, and the bound from Proposition~\ref{prop:comm-dist-bound} becomes $d_{A_L}\leq(|I_i|+|I_j|)\cdot|[G',G']|$.
    \end{remark}

    \subsection{The element order bound}
    A single weight-two entry in a base matrix already caps the distance by the order of a group element.
    \begin{definition}[Order of a group element]
    \label{def:element-order}
    Let $G$ be a group and $g \in G$. The \emph{order} of $g$, denoted
    $\operatorname{ord}(g)$, is the smallest positive integer $n$ such that
    \[
    g^n = e,
    \]
    where $e$ is the identity of $G$. If no such $n$ exists, we set
    $\operatorname{ord}(g) = \infty$.
    Equivalently, $\operatorname{ord}(g) = |\langle g \rangle|$, the order of
    the cyclic subgroup generated by $g$.
    \end{definition}

    \begin{theorem}[Element-order bound]\label{thm:element-order}
        Let $A\in\R^{1\times n_a}$ have an entry equal to $1+g$ for some $g\in G$. Then $\ker(\Lrep{A})$ contains a nonzero codeword of Hamming weight $\order{g}$, so the classical distance is at most $\order{g}$; by Theorem~\ref{thm:dc_ub_dq}, the quantum code $\LP(A,B)$ then satisfies $d_Q\le\order{g}$.
    \end{theorem}

    \begin{proof}
        Say the entry is $a_j=1+g$. Any $r\in\R$ with $(1+g)\,r=0$, placed in coordinate $j$ and padded with zeros, is a codeword of $\ker(\Lrep{A})$, since the single row evaluates to $a_j r=(1+g)r=0$. Writing $r=\sum_{h\in G}a_h\,h$,
        \begin{equation}
            (1+g)\,r=\sum_{h\in G}\bigl(a_h+a_{g^{-1}h}\bigr)\,h=0
            \quad\Longleftrightarrow\quad
            a_{gh}=a_h\ \ \text{for all }h\in G,
        \end{equation}
        so the coefficients of $r$ are constant on the orbits of left multiplication by $g$. These orbits are the right cosets of $\langle g\rangle$,
        \begin{equation}
            \orbit_h=\{\,h,\;gh,\;g^2h,\;\ldots,\;g^{\order{g}-1}h\,\},
        \end{equation}
        each of size $\order{g}$. Hence every solution is a sum of orbit indicators $\hat{\orbit}_i\coloneqq\sum_{k=0}^{\order{g}-1}g^k h_i$, and its weight lies in $\{0,\order{g},2\order{g},\ldots,|G|\}$. The lightest nonzero choice, a single orbit, has weight $\order{g}$, which bounds the distance.
    \end{proof}

    \subsection{Distance bounds on classical quasi-cyclic LDPC codes}

    By the fundamental theorem of finite abelian groups, every finite abelian group is a direct product of cyclic groups. In the univariate case the classical base codes are the well-studied quasi-cyclic LDPC codes, where $G=\Z/\ell\Z$ and $\F_2[G]\cong\F_2[x]/(x^\ell-1)$. Therefore, each entry of the $m_a\times n_a$ base matrix $A$ is a polynomial $A_{ij}(x)$ of degree less than $\ell$, and the classical code is the binary kernel $\mathcal{C}_A=\ker(A_{\mathrm{bin}})$ with $A_{\mathrm{bin}}=L[A]$; we write $d(A_{\mathrm{bin}})$ for its minimum distance. Since $G$ is abelian, the left and right regular representations coincide and we need not distinguish them. Under the identification $\F_2[G]^{n_a}\cong\F_2^{n_a\ell}$, ring multiplication is realized by $L[\cdot]$ and the ring kernel embeds isometrically into the binary kernel, so the Hamming weight of a lift $\tilde c$ equals the sum of the ring weights of its components.

    The rest of this subsection develops the classical distance bounds that let our pipeline discard unpromising base matrices before building any quantum code, exploiting the fact that the quantum distance never exceeds the classical one (Theorem~\ref{thm:dc_ub_dq}). The development proceeds in two parts. First, we attach to each base matrix its integer \emph{weight matrix} (Definition~\ref{def:weight-matrix}) and recall the permanent-based upper bounds on the binary distance $d(A_{\mathrm{bin}})$ from \cite{Smarandache_2012,mitchell2014quasicyclicldpccodesbased} (Theorems~\ref{thm:perm-poly} and~\ref{thm:perm-wmat}). These bounds expose a sharp dichotomy: monomial (weight-one) base matrices cap the distance at $(m_a+1)!$ (Corollary~\ref{cor:fact-bound}), whereas polynomial entries are exactly what break this ceiling (Remark~\ref{rem:poly-vs-mono}), which is why our search uses polynomial entries. We then formalize the \emph{check weight} of the lifted product code (Definition~\ref{def:check-weight}) and combine it with the permanent bound to prove our main result: every $2\times4$ base matrix of check weight at most $9$ yields an LP code of distance at most $14$ (Theorem~\ref{thm:2x4-w9-ceiling}).
        
\subsubsection{Weight-matrix bound}
    \begin{definition}[Weight matrix of a QC-LDPC code]\label{def:weight-matrix}
      The \emph{weight matrix} of $A$ is the integer matrix
    \begin{equation}
    W(A)\in\Z_{\ge 0}^{m_a\times n_a},\qquad W(A)_{ij}\;\triangleq\;\#\bigl\{\,g\in G\;:\;\text{coefficient of }g\text{ in }A_{ij}(x)\text{ equals }1\,\bigr\},
    \end{equation}
    i.e.\ each entry records the number of monomials in the corresponding polynomial. Following \cite{Smarandache_2012}, we call $A$ \emph{monomial} (or of \emph{type-$1$}) if $W(A)_{ij}\le 1$ for all $i,j$, and of \emph{type-$M$} more generally with $M=\max_{ij}W(A)_{ij}$.
    \end{definition}

    \begin{lemma}[Permanent codeword construction \cite{Smarandache_2012}]\label{lem:perm-codeword}
    For any column subset $S\subseteq[n_a]$ of size $|S|=m_a+1$, the ring vector $c\in(\F_2[G])^{n_a}$ defined by
    \begin{equation}\label{eq:perm-codeword}
    c_j\;=\;\begin{cases}\operatorname{perm}\!\bigl(A_{[m_a],\,S\setminus\{j\}}\bigr),&j\in S,\\[2pt]0,&j\notin S,\end{cases}
    \end{equation}
    where $A_{[m_a],T}$ denotes the $m_a\times|T|$ submatrix on columns $T$ and the permanent is taken over the commutative ring $\F_2[G]$, satisfies $Ac=0$. Consequently, its lift $\tilde c\in\F_2^{n_a\ell}$ is a binary codeword of $\mathcal{C}_A$ with Hamming weight
    \begin{equation}\label{eq:perm-weight}
    \mathrm{wt}_H(\tilde c)\;=\;\sum_{j\in S}\mathrm{wt}\!\bigl(\operatorname{perm}\!\bigl(A_{[m_a],S\setminus\{j\}}\bigr)\bigr).
    \end{equation}
    \end{lemma}
    
    \begin{proof}
    For each row $i\in[m_a]$, expanding by cofactors along the row $A_{i,S}$ gives
    \begin{equation}
    (Ac)_i\;=\;\sum_{j\in S}A_{ij}\operatorname{perm}\!\bigl(A_{[m_a],S\setminus\{j\}}\bigr)\;=\;\operatorname{perm}\!\begin{bmatrix}A_{i,S}\\ A_{[m_a],S}\end{bmatrix}\;=\;0,
    \end{equation}
    because in characteristic $2$ the permanent equals the determinant, and a determinant with a repeated row vanishes (row $i$ of $A_{[m_a],S}$ already appears below). The weight identity \eqref{eq:perm-weight} then follows from the isometry $\mathrm{wt}_H(\tilde c)=\sum_j\mathrm{wt}_{\F_2[G]}(c_j)$ recorded above.
    \end{proof}
    
    \begin{remark}
        This codeword construction relies only on the commutativity of the group. Although it is usually stated for QC-LDPC codes over a cyclic group, it therefore holds for any abelian group, including the direct product of two cyclic groups considered here, so every bound below that rests on it generalizes immediately to linear codes over abelian groups.
    \end{remark}
    
    Minimising the right-hand side of \eqref{eq:perm-weight} over admissible $S$ with $\min^{*}$ skipping the trivial all-zero codeword by returning the smallest \emph{nonzero} value of the sum yields the two upper bounds that will be used throughout.
    
    \begin{theorem}[Polynomial permanent bound {\cite[Thm.~7]{Smarandache_2012}}]\label{thm:perm-poly}
    The minimum distance of the QC-LDPC code $\mathcal{C}_A$ satisfies
    \begin{equation}\label{eq:thm-poly}
    d(A_{\mathrm{bin}})\;\le\;\min^{*}_{\substack{S\subseteq[n_a]\\ |S|=m_a+1}}\;\sum_{j\in S}\mathrm{wt}\!\bigl(\operatorname{perm}\!\bigl(A_{[m_a],S\setminus\{j\}}\bigr)\bigr),
    \end{equation}
    where the weight of each ring permanent reflects the $\F_2$-cancellations among the monomials in its support.
    \end{theorem}
    
    \begin{theorem}[Weight-matrix permanent bound {\cite[Thm.~8]{Smarandache_2012}}, {\cite[Thm.~2]{mitchell2014quasicyclicldpccodesbased}}]\label{thm:perm-wmat}
    With $W(A)$ the integer weight matrix of Definition~\ref{def:weight-matrix},
    \begin{equation}\label{eq:thm-wmat}
    d(A_{\mathrm{bin}})\;\le\;\min^{*}_{\substack{S\subseteq[n_a]\\ |S|=m_a+1}}\;\sum_{j\in S}\operatorname{perm}\!\bigl(W(A)_{[m_a],\,S\setminus\{j\}}\bigr),
    \end{equation}
    the right-hand permanents being ordinary integer permanents of $m_a\times m_a$ submatrices.
    \end{theorem}
    
    \begin{proof}
    The crude inequality $\mathrm{wt}_{\F_2[G]}(\operatorname{perm}(M))\le\operatorname{perm}(W(M))$ holds for any ring matrix $M$, since the right-hand side counts the total support of the permanent expansion before $\F_2$-cancellations. Applying this to each summand in \eqref{eq:thm-poly} gives \eqref{eq:thm-wmat}.
    \end{proof}
    
    \begin{remark}
        The bound of Theorem~\ref{thm:perm-poly} is tighter than that of Theorem~\ref{thm:perm-wmat}; in practice we use the weight-matrix bound of Theorem~\ref{thm:perm-wmat} to choose an ansatz and the polynomial bound of Theorem~\ref{thm:perm-poly} as a finer filter.
    \end{remark}
    
    \begin{corollary}[Factorial bound for monomial $A$ {\cite[Cor.~9]{Smarandache_2012}}, {\cite[Thm.~1]{mitchell2014quasicyclicldpccodesbased}}]\label{cor:fact-bound}
    If $A$ is monomial, then
    \begin{equation}\label{eq:fact-bound}
    d(A_{\mathrm{bin}})\;\le\;(m_a+1)!.
    \end{equation}
    \end{corollary}
    
    \begin{proof}
    When $A$ is monomial, every entry of $A_{[m_a],S\setminus\{j\}}$ is a single group element, so
    \begin{equation}
    \operatorname{perm}\!\bigl(A_{[m_a],S\setminus\{j\}}\bigr)\;=\;\sum_{\sigma\in\mathrm{Sym}(m_a)}\,\prod_{i=1}^{m_a}A_{i,\sigma(i)}
    \end{equation}
    is a sum of $m_a!$ monomials with $\F_2$-weight at most $m_a!$. Summing over the $m_a+1$ indices in $S$ in \eqref{eq:perm-weight} bounds the right-hand side by $(m_a+1)\cdot m_a!=(m_a+1)!$, and minimising over $S$ in Theorem~\ref{thm:perm-poly} yields \eqref{eq:fact-bound}.
    \end{proof}
    
    \begin{remark}[The monomial ceiling and the role of polynomial entries]\label{rem:poly-vs-mono}
    Corollary~\ref{cor:fact-bound} caps monomial QC-LDPC codes at $d\le 6$ for $m_a=2$, $d\le 24$ for $m_a=3$, and $d\le 120$ for $m_a=4$, irrespective of the lifting factor $\ell$. \emph{Polynomial entries are the lever that breaks this ceiling}: each weight-$w$ entry contributes $w$ rather than $1$ monomials to the permanent expansion, so the right-hand side of \eqref{eq:thm-wmat} can scale up to $(m_a+1)!\cdot\prod_{ij}W(A)_{ij}$ in the worst case, while \eqref{eq:thm-poly} can shrink it further only through $\F_2$-cancellations among those monomials. Concretely, \cite{Smarandache_2012} exhibits a type-$2$ $(3,4)$-regular QC-LDPC code with $d=32>24=(3+1)!$ and a type-$3$ one with $d=54$ at the same regularity. Our pipeline therefore samples over polynomial entries from the outset, restricted only by the desired check weight and by the girth obstructions of Theorem~\ref{thm:girth-bounds} below.
    \end{remark}
    
    \begin{theorem}[Girth bounds from weight-matrix substructure {\cite[Thm.~18]{Smarandache_2012}}]\label{thm:girth-bounds}
    Let $g(A_{\mathrm{bin}})$ denote the girth of the Tanner graph of $A_{\mathrm{bin}}$. Whenever $W(A)$ admits one of the substructures below (modulo row and column permutations and transposition), the girth is bounded accordingly:
    \begin{equation}\label{eq:girth-bounds}
    [\,3\,]\;\Rightarrow\;g\le 6,\qquad
    [\,2,\;2\,]\;\Rightarrow\;g\le 8,\qquad
    \bigl[\begin{smallmatrix}1&1\\ 1&1\end{smallmatrix}\bigr]\;\Rightarrow\;g\le 10,\qquad
    \bigl[\begin{smallmatrix}1&1&1\\ 1&1&1\end{smallmatrix}\bigr]\;\Rightarrow\;g\le 12.
    \end{equation}
    \end{theorem}

\begin{definition}[Check weight of the LP code]\label{def:check-weight}
For a weight matrix $W(A)$, write $R_i(A)\coloneqq\sum_j W(A)_{ij}$ and $C_j(A)\coloneqq\sum_i W(A)_{ij}$ for its row and column sums. The \emph{check weight} of the lifted product code built from $(A,B)$ is the largest Hamming weight of any stabilizer generator; using $\max_{i,j}\bigl(f(i)+g(j)\bigr)=\max_i f(i)+\max_j g(j)$, it equals
\begin{equation}\label{eq:check-weight}
w_{\mathrm{ch}}(A,B)\;=\;\max\!\Bigl\{\,\max_i R_i(A)+\max_j C_j(B),\;\;\max_j C_j(A)+\max_i R_i(B)\,\Bigr\}.
\end{equation}
In the abelian case $B=A$ this reduces to $w_{\mathrm{ch}}=R+C$, where $R\coloneqq\max_i R_i(A)$ and $C\coloneqq\max_j C_j(A)$.
\end{definition}

\begin{theorem}[Distance ceiling for LP codes at check weight $\le 9$ with $2\times 4$ classical base matrix]\label{thm:2x4-w9-ceiling}
Let $A$ be a $2\times 4$ ring matrix over a commutative group algebra $\F_2[G]$, paired with itself ($B=A$) in the LP construction. Write $W=W(A)$, $R_i\coloneqq \sum_jW_{ij}$, $C_j\coloneqq \sum_iW_{ij}$, and set
\[R\;\coloneqq \;\max_{i\in\{0,1\}}R_i,\qquad C\;\coloneqq \;\max_{j\in\{0,1,2,3\}}C_j.\]
If the check weight of the LP code satisfies
\begin{equation}\label{eq:2x4-w9-hyp}
w_{\mathrm{ch}}\;=\;\max_{(i_a,i_b)\in\{0,1\}\times\{0,1,2,3\}}\!\bigl(R_{i_a}+C_{i_b}\bigr)\;=\;R+C\;\le\;9,
\end{equation}
Then the distance of the LP code satisfies $d_q\leq 14$.
\end{theorem}

\begin{proof}
We prove it by showing that the classical QC-LDPC code $\mathcal C_A=\ker(A_{\mathrm{bin}})$ satisfies $d(A_{\mathrm{bin}})\le14$; since $d_q\le d(A_{\mathrm{bin}})$, this gives $d_q\le14$.

For each size-$3$ column subset $S\subset\{0,1,2,3\}$ write
\[\operatorname{perm}_S\coloneqq\sum_{j\in S}\operatorname{perm}\bigl(W_{\{0,1\},\,S\setminus\{j\}}\bigr).\]

Set $D\coloneqq \sum_iW_{0i}W_{1i}\ge 0$. Each unordered pair $\{i,j\}\subset\{0,1,2,3\}$ lies in exactly two of the $\binom{4}{3}=4$ size-$3$ subsets, and the $2\times2$ permanent of columns $\{i,j\}$ equals $W_{0i}W_{1j}+W_{1i}W_{0j}$, so
\begin{equation}\label{eq:2x4-avg-sum}
\sum_{S}\operatorname{perm}_S=2\sum_{i<j}\bigl(W_{0i}W_{1j}+W_{1i}W_{0j}\bigr)=2(R_0R_1-D),
\end{equation}
where the last equality uses $\sum_{i\ne j}W_{0i}W_{1j}=R_0R_1-D$. Hence
\begin{equation}\label{eq:2x4-min-le-avg}
\min^{*}_{S}\operatorname{perm}_S\;\le\;\frac{1}{4}\sum_S\operatorname{perm}_S=\frac{R_0R_1-D}{2}.
\end{equation}

\textbf{Case 1: $R_0R_1-D\le 28$.} Inequality \eqref{eq:2x4-min-le-avg} gives $\min^{*}_S\operatorname{perm}_S\le 14$, and Theorem~\ref{thm:perm-wmat} yields $d(A_{\mathrm{bin}})\le 14$.

\textbf{Case 2: $R_0R_1-D\ge 30$.} We show that the configurations satisfying this inequality together with \eqref{eq:2x4-w9-hyp} are extremely restricted, and in each one either $\min^{*}_S\operatorname{perm}_S\le 14$ holds anyway or $A$ admits a size-$2$ codeword of Hamming weight $\le 6$.

In the following, we will focus on discussing the second case:

Since $D\ge 0$, $R_0R_1\ge 30$, and $R_0,R_1\le R$ implies $R^2\ge 30$, so $R\ge 6$. The total-weight identity gives $R_0+R_1=\sum_jC_j\le 4C$. Therefore,
\[
R_0R_1\le\tfrac{(R_0+R_1)^2}{4}\le 4C^2,\]
hence $4C^2\ge 30$, so $C\ge 3$. As check-weight $w_{\mathrm{ch}}=R+C\le 9$ and $R_0R_1\ge 29$, the only possible configuration is $R=6$ and $C=3$.

Since $R_0+R_1\ge 2\sqrt{R_0R_1}\ge 2\sqrt{30}>10$ and $R_0+R_1\le 4C=12$, the only possibilities are $R_0+R_1\in\{11,12\}$, which we treat in turn.

For $R_0+R_1 = 11$, without loss of generality, we assume $R_0=5$, $R_1=6$, so $R_0R_1=30$ and the constraint $R_0R_1-D\ge 30$ gives $D= 0$.  The only column multiset consistent with $R_0=5$, $R_1=6$, $D=0$, and $\max_jC_j=3$ is $\{(3,0)^\top,(2,0)^\top,(0,3)^\top,(0,3)^\top\}$. Direct computation on this multiset gives the triplet-sum multiset $\{12,15,15,18\}$, so $\min^{*}_S\operatorname{perm}_S=12$ and Theorem~\ref{thm:perm-wmat} yields $d(A_{\mathrm{bin}})\le 12<14$.

For $R_0+R_1 = 12$, as each row sum $\le R=6$, then $R_0=R_1=6$. Also since $\sum_jC_j=12=4C$ each column sum is exactly $3$. Let $w_i\coloneqq W_{0i}\in\{0,1,2,3\}$, we have $W_{1i}=3-w_i$ and $\sum_iw_i=6$. The partitions of $6$ into four ordered parts in $\{0,1,2,3\}$ are, up to permutation, the five listed below where in each case $D=18-\sum_iw_i^2$, and the multi-set of the four triplet sums $\{\operatorname{perm}_S\}_{|S|=3}$ is invariant under permutations of $W$'s columns, so a direct computation on a single representative gives:
\begin{center}\renewcommand{\arraystretch}{1.25}
\begin{tabular}{c|c|c|c}
$(w_i)$ partition & $D$ & $\{\operatorname{perm}_S\}_{|S|=3}$ & $\min^{*}_S\operatorname{perm}_S$ \\\hline
$(2,2,1,1)$ & $8$ & $\{14,14,14,14\}$ & $14$ \\
$(3,2,1,0)$ & $4$ & $\{14,14,18,18\}$ & $14$ \\
$(2,2,2,0)$ & $6$ & $\{12,16,16,16\}$ & $12$ \\
$(3,1,1,1)$ & $6$ & $\{12,16,16,16\}$ & $12$ \\
$(3,3,0,0)$ & $0$ & $\{18,18,18,18\}$ & $18$ \\
\end{tabular}
\end{center}
For the first four partitions, $\min^{*}_S\operatorname{perm}_S\le 14$, so Theorem~\ref{thm:perm-wmat} directly gives $d(A_{\mathrm{bin}})\le 14$. The remaining partition $(3,3,0,0)$ has $D=0$, which forces $W_{0i}W_{1i}=0$ for every $i$: each column of $W$ has at most one nonzero entry, and with $R_0=R_1=6$ and column sums $3$, the matrix $W$ is, up to column permutation,
\begin{equation}\label{eq:abelian-2x4-pattern-III-proof}
W=\begin{pmatrix}3&3&0&0\\ 0&0&3&3\end{pmatrix}.
\end{equation}
The Tanner graph of a code with this $W$ splits into at least two connected components, and a codeword supported on a single component, which corresponds to a $1\times 2$ check matrix, already bounds the distance of the whole code. As each connected component corresponds to a check matrix $A^{sub}$ with weight matrix
\begin{equation}
    W_{sub} = \begin{pmatrix}3&3\end{pmatrix},
\end{equation}
applying the distance bound on it gives
\begin{equation}
    d(A_{\mathrm{bin}})\leq d(A^{sub}_{\mathrm{bin}})\leq 3+3 = 6 <14.
\end{equation}
Combining all these cases, we have $d(A_{\mathrm{bin}})\le 14$. Thus,
\begin{equation}
    d_q \leq d(A_{\mathrm{bin}}) \leq 14.
\end{equation}
This completes the proof.
\end{proof}

\begin{remark}
    We have found $\llbracket 560,112,\leq14 \rrbracket $ codes that reach this $d=14$ theoretical upper bound as shown in Table~\ref{tab:code_param}.
\end{remark}

\renewcommand{\NSMxskcpreff}{\tilde{N}_{x,\mathrm{RREF}}'}
\newcounter{sqetchalgorithm}

\section{\texorpdfstring{\textsf{sQetch}}{sQetch}: A fast GPU-based  CSS-code distance estimator}
\label{app:sqetch}

We introduce \textsf{sQetch}~\cite{knitkitgithub}, a very fast GPU-based algorithm for estimating the minimum distance of a quantum CSS code. Like the QDistRnd estimator~\cite{Pryadko_2022}, \textsf{sQetch} is a form of random information-set decoding (ISD). Its key ingredient is that each trial operates on a random low-dimensional \emph{sketch} of a check-matrix null space rather than on the full null space. This both reduces the cost per trial and maps naturally to GPU shared memory. On a single NVIDIA RTX 5090, \textsf{sQetch} achieves a speedup of roughly $100{,}000\times$ over the CPU-based QDistRnd baseline, making it feasible to explore the vast design space of high-rate fault-tolerant processors, in which viable candidates are extremely sparse. The same estimator also offers significant acceleration in the search over surgery gadgets and the estimation of circuit-level distance.

\subsection{Notation}
\label{app:sq:setup}

We consider a quantum CSS code with parameters $\llbracket n,k\rrbracket$ and check matrices $\Hx,\Hz$ satisfying $\Hx\Hz^\top=0$ (Appendix~\ref{app:prelim}). Its distance is $d=\min(\dx,\dz)$. Here, $\dx$ is the minimum Hamming weight of a nontrivial logical $X$ operator, namely a vector in $\ker(\Hz)\setminus\rowspan(\Hx)$, whereas $\dz$ is the minimum Hamming weight of a nontrivial logical $Z$ operator, namely a vector in $\ker(\Hx)\setminus\rowspan(\Hz)$, as defined in Appendix~\ref{app:distance_bounds}.
Let the null spaces be:
\begin{equation}
    \NSx \coloneqq \ker(\Hx),
    \qquad
    \NSz \coloneqq \ker(\Hz),
\end{equation}
with dimensions
\begin{equation}
    \kerx \coloneqq \dim(\NSx)=n-\rank(\Hx),
    \qquad
    \kerz \coloneqq \dim(\NSz)=n-\rank(\Hz),
\end{equation}
respectively. After fixing a basis for each null space, let $\NSMx\in\F_2^{\kerx\times n}$ and $\NSMz\in\F_2^{\kerz\times n}$ be the matrices whose rows are these basis vectors. Thus,
\begin{equation}
    \rowspan(\NSMx)=\NSx,
    \qquad
    \rowspan(\NSMz)=\NSz.
\end{equation}
Because $\ker(\NSMz)=\rowspan(\Hz)$, a vector $v\in\NSx$ represents a nontrivial logical $Z$ operator exactly when $\NSMz v\neq0$. Similarly, a vector $u\in\NSz$ represents a nontrivial logical $X$ operator exactly when $\NSMx u\neq0$. Therefore, \textsf{sQetch} searches over
\begin{equation}
    \label{eq:dx_dz_def2}
    \begin{split}
        \dx &= \min\bigl\{\wt{u}:u\in\NSz,\ \NSMx u\neq0\bigr\},\\
        \dz &= \min\bigl\{\wt{v}:v\in\NSx,\ \NSMz v\neq0\bigr\}.
    \end{split}
\end{equation}
Recall that
\begin{equation}
    k=n-\rank(\Hx)-\rank(\Hz).
\end{equation}
Since $\kerx=k+\rank(\Hz)$, and symmetrically $\kerz=k+\rank(\Hx)$, both null-space dimensions satisfy $\kerx,\kerz\geq k$.

\begin{remark}
    If $\Hz=0$, estimating $\dz$ reduces to estimating the distance of the classical code with parity-check matrix $\Hx$.
\end{remark}

\subsection{The \textsf{sQetch} algorithm}

Computing the exact distance of a general linear code is NP-hard~\cite{641542,Kapshikar_2023}. Instead, \textsf{sQetch} performs many random trials biased toward low-weight codewords. Every nontrivial logical operator found in this way gives an upper bound on the distance, and the running estimate decreases whenever a lower-weight operator is found. The heuristic is co-designed with the GPU. Without loss of generality, we describe the estimation of $\dz$; the same procedure estimates $\dx$ after exchanging $\Hx$ and $\Hz$.

Our goal is to find the lowest-weight vector in $\ker(\Hx)$ that is not in $\operatorname{rowspace}(\Hz)$. Starting from a matrix whose rows form a basis of $\ker(\Hx)$, we randomly shuffle its columns and put it in reduced row echelon form (RREF), which biases the rows toward low weight vectors by forcing zeros in the pivot columns. \textsf{sQetch} follows the same idea but samples only a few random basis vectors in each trial, so the working matrix fits cheaply in GPU shared memory.

Concretely, the null-space matrix $\NSMx\in\F_2^{\kerx\times n}$, satisfying $\rowspan(\NSMx)=\NSx$, is computed once. Each trial forms a random $\kappa$-row \emph{sketch} $\NSMxsk$ by sampling row indices independently and uniformly with replacement:
\begin{equation}
    i_1,\ldots,i_\kappa
    \stackrel{\mathrm{i.i.d.}}{\sim}
    \mathrm{Unif}\bigl(\{0,\ldots,\kerx-1\}\bigr),
    \qquad
    \NSMxsk[s,:]=\NSMx[i_s,:],
    \quad
    s=1,\ldots,\kappa.
\end{equation}
The GPU constructs $\gamma$ such sketches in parallel, each from an independent random seed. For each sketch $\NSMxsk^{(j)}$, it samples a uniformly random permutation $\pi^{(j)}\in S_n$, permutes the columns to obtain $\NSMxskcp^{(j)}$, and applies Gauss-Jordan elimination to obtain the row-reduced echelon form $\NSMxskcpreff^{(j)}$. Let $P_\pi^{(j)}\subseteq\{0,\ldots,n-1\}$ denote its pivot columns and let
\begin{equation}
    F_\pi^{(j)}
    \coloneqq
    \{0,\ldots,n-1\}\setminus P_\pi^{(j)}
\end{equation}
denote its free columns. Every nonzero row $r$ of the row-reduced matrix has exactly one nonzero entry in $P_\pi^{(j)}$; all its remaining nonzero entries lie in $F_\pi^{(j)}$. Hence,
\begin{equation}
    \wt{r}
    =
    1+\bigl|\supp(r)\cap F_\pi^{(j)}\bigr|.
\end{equation}
Each nonzero row $r_\ell^{(j)}$ represents a vector in $\NSx$ and is a nontrivial logical $Z$ operator precisely when $\NSMz(r_\ell^{(j)})^\top\neq0$. Whenever such a row has weight below the current best value, \textsf{sQetch} lowers its estimate of $\dz$. Algorithm~\ref{alg:sqetch} summarizes one trial.

\begin{figure}[t]
\noindent\rule{\linewidth}{0.6pt}\\[-2pt]
\refstepcounter{sqetchalgorithm}\label{alg:sqetch}
\textbf{Algorithm~\thesqetchalgorithm.}\quad
\textsc{\textsf{sQetch} trial}: one independent sketched-ISD trial for estimating $\dz$.
\\[-4pt]\rule{\linewidth}{0.3pt}\\[1pt]
\textbf{Input:}\;\;
   $\NSMx\in\F_2^{\kerx\times n}$ with $\rowspan(\NSMx)=\NSx=\ker(\Hx)$;\;
   $\NSMz\in\F_2^{\kerz\times n}$ with $\rowspan(\NSMz)=\NSz=\ker(\Hz)$;\;
   sketch size $\kappa\leq\kerx$;\;
   current target weight $d^\star$.\\
\textbf{Output:}\;\;
   the lowest-weight nontrivial logical candidate $v^\star\in\NSx$ found in this trial and its weight $\hat{\dz}$, or $(\bot,+\infty)$ if no such candidate is found; the algorithm also signals \textsc{Found} when $\hat{\dz}<d^\star$.
\\[-4pt]\rule{\linewidth}{0.3pt}\\[2pt]
\begingroup
\setlength{\tabcolsep}{0pt}
\renewcommand{\arraystretch}{1.05}
\noindent
\begin{tabular}{@{}r@{\;\;}l@{}}
\textbf{1:}  & $\pi\gets\textsc{Fisher--Yates}(\{0,\ldots,n-1\})$
                \quad\textit{// uniform random column permutation}\\
\textbf{2:}  & $i_1,\ldots,i_\kappa
                \stackrel{\mathrm{i.i.d.}}{\sim}
                \mathrm{Unif}\bigl(\{0,\ldots,\kerx-1\}\bigr)$\\
\textbf{3:}  & $\NSMxsk\gets
                \bigl[\NSMx[i_s,:]\bigr]_{s=1}^{\kappa}$
                \quad\textit{// row sketch of $\NSMx$}\\
\textbf{4:}  & $\NSMxskcp[:,c]\gets\NSMxsk[:,\pi(c)]$
                for $c=0,\ldots,n-1$
                \quad\textit{// permute columns}\\
\textbf{5:}  & $\mathit{pr}\gets0$\\
\textbf{6:}  & \textbf{for} $c=0,1,\ldots,n-1$ \textbf{do}
                \quad\textit{// column-ordered Gauss--Jordan elimination}\\
\textbf{7:}  & \quad scan rows
                $r\in\{\mathit{pr},\ldots,\kappa-1\}$
                for a $1$ in column $c$ of $\NSMxskcp$\\
\textbf{8:}  & \quad \textbf{if} one is found in row $r^\star$ \textbf{then}\\
\textbf{9:}  & \quad\quad swap rows
                $r^\star\leftrightarrow\mathit{pr}$ of $\NSMxskcp$\\
\textbf{10:} & \quad\quad \textbf{for}
                $r\in\{0,\ldots,\kappa-1\}$ with
                $r\neq\mathit{pr}$ \textbf{do}\\
\textbf{11:} & \quad\quad\quad \textbf{if}
                $\NSMxskcp[r,c]=1$ \textbf{then}
                $\NSMxskcp[r,:]\gets
                \NSMxskcp[r,:]\oplus\NSMxskcp[\mathit{pr},:]$\\
\textbf{12:} & \quad\quad \textbf{end for}\\
\textbf{13:} & \quad\quad $\mathit{pr}\gets\mathit{pr}+1$\\
\textbf{14:} & \quad \textbf{end if}\\
\textbf{15:} & \quad \textbf{if} $\mathit{pr}=\kappa$ \textbf{then break}\\
\textbf{16:} & \textbf{end for}\\
\textbf{17:} & $\NSMxskcpreff\gets\NSMxskcp$\\
\textbf{18:} & $\hat{\dz}\gets+\infty;\;\;v^\star\gets\bot$\\
\textbf{19:} & \textbf{for} each nonzero row $r'$ of $\NSMxskcpreff$ \textbf{do}
                \quad\textit{// undo permutation, test logical nontriviality, and compute weight}\\
\textbf{20:} & \quad define $r\in\F_2^n$ by
                $r[\pi(c)]\gets r'[c]$ for $c=0,\ldots,n-1$\\
\textbf{21:} & \quad \textbf{if}
                $\NSMz r^\top\neq0$ \textbf{and}
                $\wt{r}<\hat{\dz}$ \textbf{then}
                $\hat{\dz}\gets\wt{r};\;\;v^\star\gets r$\\
\textbf{22:} & \textbf{end for}\\
\textbf{23:} & \textbf{if} $\hat{\dz}<d^\star$ \textbf{then}
                signal \textsc{Found};
                \textbf{return} $v^\star,\hat{\dz}$\\
\textbf{24:} & \textbf{return} $v^\star,\hat{\dz}$
\end{tabular}
\endgroup
\\[-3pt]\rule{\linewidth}{0.6pt}
\end{figure}

\subsection{Hit-probability analysis}

Consider a constant-rate CSS code, for which $\kerx=\Theta(n)$, and assume $\dz\ll n$, as is typical in practice. Fix a minimum-weight nontrivial logical $Z$ operator $c^\star$ with $\wt{c^\star}=\dz$. Because the rows of $\NSMx$ form a basis, $c^\star$ has a unique expansion in these rows. Let $j^\star$ be the number of basis rows with nonzero coefficients in this expansion. The span of a $\kappa$-row sketch contains $c^\star$ exactly when all $j^\star$ required basis rows are sampled at least once. Because the rows are sampled independently with replacement, inclusion--exclusion gives
\begin{equation}
    P_{\operatorname{sketch}}
    =
    \sum_{t=0}^{j^\star}
    (-1)^t
    \binom{j^\star}{t}
    \left(1-\frac{t}{\kerx}\right)^\kappa.
\end{equation}
When $j^\star\ll\kappa\ll\kerx$, the leading approximation is
\begin{equation}
    P_{\operatorname{sketch}}
    \approx
    \left(\frac{\kappa}{\kerx}\right)^{j^\star}.
\end{equation}
Conditioned on the sketch span containing $c^\star$, the random column permutation exposes $c^\star$ as a row of $\NSMxskcpreff$ exactly when one element of $\supp(c^\star)$ lies in the pivot set. Under the genericity assumptions that the sketch has full row rank and that its pivot set behaves as a uniformly random $\kappa$-subset of the $n$ columns, the conditional hit probability is
\begin{equation}
    P\bigl(\mathrm{hit}\mid\operatorname{sketch}\bigr)
    \approx
    \frac{\dz\binom{n-\dz}{\kappa-1}}{\binom{n}{\kappa}}.
\end{equation}
For $n\gg\dz$, this becomes
\begin{equation}
    P\bigl(\mathrm{hit}\mid\operatorname{sketch}\bigr)
    \approx
    \frac{\dz\kappa}{n}
    \left(1-\frac{\dz}{n}\right)^{\kappa-1}.
\end{equation}
Therefore,
\begin{equation}
    P_{\mathrm{hit}}
    =
    P\bigl(\mathrm{hit}\mid\operatorname{sketch}\bigr)
    P_{\operatorname{sketch}}
    \approx
    \left(\frac{\kappa}{\kerx}\right)^{j^\star}
    \frac{\dz\kappa}{n}
    \left(1-\frac{\dz}{n}\right)^{\kappa-1}.
\end{equation}

\subsection{Benchmark}

We benchmark \textsf{sQetch} against QDistRnd~\cite{Pryadko_2022}, a state-of-the-art CPU distance estimator, on the workload relevant to our high-rate qLDPC processor discovery pipeline. Figure~\ref{fig:sqetch_runtime} reports the wall-clock time required to screen approximately $10^7$ candidate codes using approximately $10^5$ random-ISD trials per code, or about $10^{12}$ trials in total. The benchmark compares a single GPU with a high-end CPU as a function of code size.

\begin{figure}[t]
    \centering
    \includegraphics[width=0.85\linewidth]{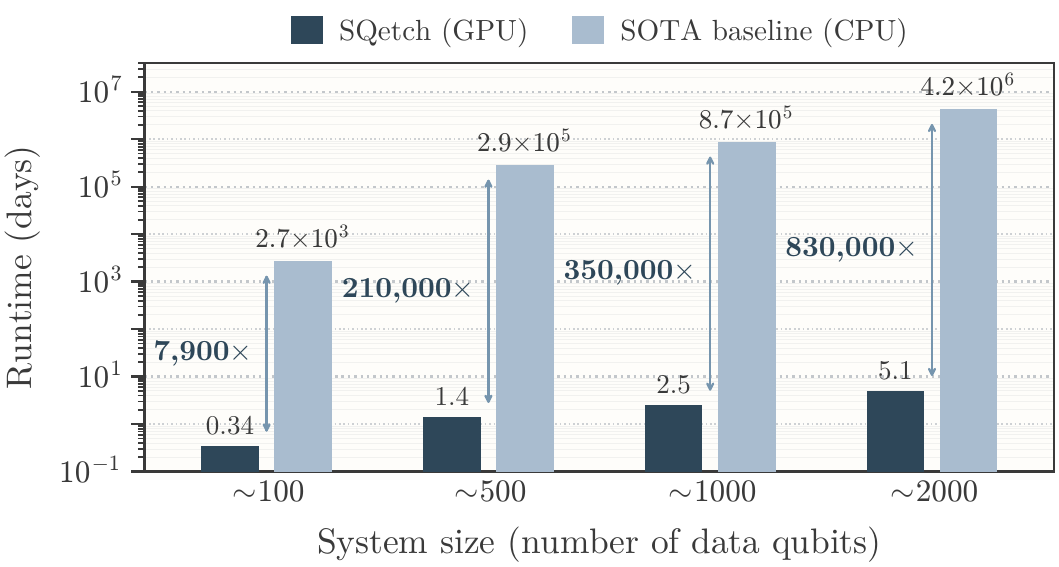}
    \caption{Wall-clock time, in days on a logarithmic scale, required to screen approximately $10^7$ candidate qLDPC codes using approximately $10^5$ distance-estimation trials per code, or approximately $10^{12}$ trials in total, as a function of code size. Dark bars show \textsf{sQetch} running on a single GPU; light bars show the state-of-the-art CPU baseline QDistRnd~\cite{Pryadko_2022}. \textsf{sQetch} completes the full workload in under one week at every size shown, whereas QDistRnd requires years at approximately $100$ data qubits and millennia at approximately $2000$ data qubits. The resulting speedup is of order $10^5\times$ and increases with system size.}
    \label{fig:sqetch_runtime}
\end{figure}

The workloads in Figure~\ref{fig:sqetch_runtime} reflect a realistic search. Good CSS codes are extremely sparse in their natural design spaces. For the lifted-product family with a fixed group, base-matrix shape, and check weight, only a fraction between $10^{-5}$ and $10^{-3}$ of the candidates that survive inexpensive classical prefilters have a useful distance. One must therefore sample on the order of $10^6$--$10^7$ candidates to find good codes reliably. For each candidate, the probability that a single random-ISD trial finds a minimum-weight logical operator is of a comparable order, motivating a budget of approximately $10^5$ trials per code. The resulting total of approximately $10^{12}$ trials is therefore the natural scale of one discovery run. At this scale, \textsf{sQetch} reduces a computation that would require years to millennia of QDistRnd CPU time to a few GPU-days, making the automated pipeline of Appendix~\ref{app:knitter_pipeline} practical.

\subsection{Hook-error-free syndrome extraction}
\label{app:hookfree}

A fixed pair of check matrices $H_X$ and $H_Z$ can be implemented by many orderings of the two-qubit gates in a syndrome-extraction (SE) cycle. The ordering matters because a fault on an ancilla can propagate through subsequent gates to several data qubits, producing a \emph{hook error}. If the resulting spacetime fault is undetected and flips a logical observable, it can reduce the circuit-level distance.

In a mitten code, a ring entry such as $a=g_1+g_2+g_3$ in the check matrices of \cref{eq:HxHz} corresponds to three separate group-element layers. On an atom-array platform, each layer associated with a group element $g_i$ is implemented by a permutation (\cref{app:connectivity-atomarrays}). We therefore preserve the group-element layers instead of applying an arbitrary Tanner-graph coloring, and search only over their order in the SE cycle. Each of the $X$- and $Z$-check halves contains $12$ layers: six left-regular layers acting on the data blocks $D_1,\ldots,D_4$ and six right-regular layers acting on $D_5$. To minimize the SE-cycle time, we keep the left- and right-regular sections separate and permute only the six group-element layers within each section.

For each candidate ordering, we compile a two-round memory circuit with Stim~\cite{Gidney_2021} and extract its detector-error model. We convert this model into a spacetime check/logical pair $(H,L)$: columns correspond to circuit fault mechanisms, rows of $H$ encode detector parities, and rows of $L$ indicate which logical observables each mechanism flips. A fault pattern $e$ is undetected when $He=0$ and is harmful when, in addition, $Le\neq0$. We use \textsf{sQetch} to estimate the distance of this spacetime code, which is the circuit-level distance of the memory experiment, and discard schedules with low estimated distance. Because the spacelike logical operators of the spacetime code include the logical operators of the underlying quantum code, the circuit-level distance is upper-bounded by the code distance. We find distance-preserving schedules for all mitten codes except the $\llbracket150,30,10\rrbracket$ code. For that code, the best schedule contains weight-$8$ fault patterns in both the $X$ and $Z$ blocks, giving $d_{\mathrm{circ}}=8<10$. The exact schedules and corresponding SE-cycle visualizations are available in Ref.~\cite{knitkitgithub}. Further details of the schedule search are given in Table~\ref{tab:hook-free-perms}.

\begin{table}[t]
\centering

\setlength{\tabcolsep}{4pt}
\renewcommand{\arraystretch}{1.1}
\begin{tabular}{@{}l c c c c r@{}}
\toprule
\textbf{$\llbracket n,k,d \rrbracket $} & \textbf{group structure} & \textbf{circuit-level distance} & \textbf{x-block} & \textbf{z-block} & \textbf{trials/side} \\
\midrule
$\llbracket 150,30,10 \rrbracket $   & $C_5 \times S_3$                            & \textbf{8}  & 8  & 8  & 60M  \\
$\llbracket 200,40,12 \rrbracket $   & $C_4 \times D_{10}$                         & \textbf{12} & 12 & 12 & 60M  \\
$\llbracket 300,60,14 \rrbracket $   & $C_{10} \times S_3$                         & \textbf{14} & 14 & 14 & 60M  \\
$\llbracket 500,100,16 \rrbracket $   & $C_5 \rtimes C_{20}$           & \textbf{16} & 19 & 16 & 150M  \\
$\llbracket 540,108,18 \rrbracket $  & $C_9 \rtimes C_{12}$             & \textbf{18} & 28 & 18 & 60M  \\
$\llbracket 630,126,\leq 20 \rrbracket $  & $C_7 \rtimes C_{18}$         & \textbf{20} & 27 & 20 & 150M \\
$\llbracket 780,156,\leq 22 \rrbracket $  & $C_{13} \times A_4$ & \textbf{22} & 38 & 22 & 150M \\
$\llbracket 975,195,\leq 24 \rrbracket $  & $C_{13} \rtimes C_{15}$                 & \textbf{24} & 44 & 28 & 150M \\
\bottomrule
\end{tabular}
\caption{Verification of circuit-level distance for our syndrome-extraction schedules for our mitten codes. The reported circuit-level distance is the upper bound on the spacetime distance of the two-round memory circuit, estimated under single-qubit Pauli faults and defined as the minimum of the distance of the $X$ and $Z$ blocks. These results indicate that the syndrome-extraction schedules preserve the code distance in all cases except for the $\llbracket 150,30,10\rrbracket$ code, whose circuit-level distance is two less than the code distance.
}
\label{tab:hook-free-perms}
\end{table}

\section{Decoding}
\label{app:decoding}
In this section, we describe our telescoping decoder and provide additional details on how we performed our memory and surgery simulations. We also compare the performance of our decoder to the current state-of-the-art decoders for qLDPC codes and perform a worst-case analysis to demonstrate the capability of our decoder to support sub-millisecond real-time decoding when ported over to FPGAs.

\subsection{Telescoping decoder construction}
Our decoding infrastructure is designed to meet two requirements. First, we want to understand the best logical error rates mitten codes can achieve down to the low logical error rate (LER) regime. This requires many Monte Carlo samples and hence a decoding pipeline with high throughput. Second, we want these estimates to reflect a realistic decoding stack that in principle can be optimized to enable real-time decoding. To satisfy both, we adopt a telescoping design \cite{delfosse2020hierarchicaldecodingreducehardware, zhao2026towards, toshio2025decoderswitchingbreakingspeedaccuracy, ravi2022betterworstcasedecodingquantum} in which fast, high-throughput decoders handle the vast majority of shots, and progressively more expensive decoders are invoked only on the small residual of hard shots that earlier tiers defer.

Concretely, our telescoping decoder consists of four stages. The first two run on the GPU, using custom CUDA kernels~\cite{cuda} that implement layered variants of belief propagation (BP) and Relay BP~\cite{mueller2025improvedbeliefpropagationsufficient}; by decoding large batches of shots concurrently, these stages absorb the vast majority of the Monte Carlo samples at high throughput. The third stage runs on the CPU, using custom C kernels that implement serial variants of BP and Relay-BP, and is applied only to the harder shots the GPU tiers defer. The final stage solves the most-likely-error (MLE) decoding problem through integer programming with Gurobi~\cite{gurobi}, and is reserved for the hardest shots that all belief-propagation stages leave undecided. Across all of our decoding experiments we construct the full correlated $XYZ$ detector-error model (DEM), but each stage decodes the representation of this model best suited to its role in our telescoping design. One representation restricts the decoding graph to the initialization-basis detectors, giving a much smaller graph and hence higher throughput. The other is the GARI transform \cite{gari_paper}, which adds auxiliary error and check nodes to remove the short cycles that $Y$ errors induce. We find this improves BP convergence and lowers the LER, and although GARI increases the number of nodes, it reduces the number of edges in the decoding graph by roughly a factor of three for our codes. This selective choice of decoding graph per stage is a key feature of our design that lets us trade off between throughput and accuracy without altering the decoders themselves. We now describe each of the stages in more detail. 

\textit{Stage 1 (GPU BP).}
Stage 1 $(S_1)$ decodes the initialization-basis restriction of the detector-error model. We tune it for maximum convergence without introducing logical errors. Depending on the code and the decoding experiment, between $80\%$ and $99.6\%$ of all shots converge in this stage for physical error rates less than or equal to $p=0.2\%.$ Every shot receives the same fixed number of BP iterations, after which we check for convergence. The first few iterations use the exact sum-product update and the remainder use the cheaper min-sum approximation. We find that this improves convergence and is worth the extra cost of needing to invoke the GPU's special function unit for the sum-product updates. Within each iteration the message updates are layered. BP passes messages between the check nodes (detectors) and the error nodes of the decoding graph. We partition the check nodes into groups, and each group, together with the error nodes its checks touch, forms a layer. All checks within a layer update in parallel from the same current beliefs. After the layer, the error nodes it touched refresh their beliefs, so the next layer acts on updated information. A fully serial schedule, which updates one check at a time, converges marginally better on our codes but cannot exploit the parallelism of the GPU as effectively as the layered schedule. For circuit-level memory experiments on the $\llbracket 150, 30, 10 \rrbracket$ mitten code and the $\llbracket 144, 12, 12 \rrbracket$ gross code~\cite{bravyi2024high} we are able to process over a million rounds of syndrome extraction per second on a single NVIDIA H100 GPU in this stage.

\textit{Stage 2 (GPU Relay-BP).}
Stage 2 $(S_2)$ decodes the GARI transform \cite{gari_paper} of the detector-error model with a layered version of Relay-BP~\cite{mueller2025improvedbeliefpropagationsufficient}. Relay-BP augments BP with a memory term that blends each error node's fixed channel prior with its earlier beliefs. We run many BP legs in sequence. Each leg warm-starts this memory from the beliefs the previous leg ended on, but it restarts message passing from scratch and draws a fresh memory strength for every error node. The strength is signed so it may pull a node's belief back toward the warm start or push it away. This randomness helps convergence; a shot that stalls in one leg can often converge in a later one. To minimize the logical error rate, we operate this stage with a quorum. A converged leg votes for the logical class corresponding to its candidate correction and we only accept the decoding result when the first $q$ converged legs cast the same vote. Otherwise, we defer to the later CPU stages. The quorum helps protect against BP convergences to the wrong logical class.

\textit{Stage 3 (CPU serial BP).}
Stage 3 $(S_3)$ receives the small residual of hard shots that the GPU stages defer. With few shots left, throughput matters less than convergence, so this stage runs on CPUs and uses the fully serial BP schedule. Checks update one at a time, and each update sees the beliefs refreshed by all the updates before it. The stage is split into three substages of increasing depth, all decoding the GARI transform~\cite{gari_paper} of the full correlated $XYZ$ detector-error model. A shot proceeds to the next substage only if the previous one fails to decode it.

$S_{3A}$ runs plain serial BP over a small ensemble of variants that perturb the channel priors, scaling them up or down or adding per-node random noise. The first variant to converge is accepted.

$S_{3B}$ runs a serial version of Relay-BP over various choices of perturbed priors. As in $S_2$, we use a quorum acceptance rule: we accept a class once $q$ converged legs vote for it. However, unlike $S_2$, where the first $q$ collected legs must be unanimous, the legs of all variants here accumulate in one shared tally, and the first class to reach $q$ votes is accepted even if other converged legs dissent.

$S_{3C}$ runs serial BP for many more iterations than any previous stage over a much wider array of prior perturbations. We accept the first converged correction since the goal of this stage is to converge as many of the remaining shots in order to minimize the number that make it to the integer programming stage ($S_4$). At low physical error rates, we never observe a logical error at this stage, with most logical errors coming from $S_2$ or $S_{3B}$.

\textit{Stage 4 (Integer Programming)} Stage 4 $(S_4)$ decodes the few shots that survive every BP stage by solving the most likely error decoding problem exactly through integer programming using Gurobi~\cite{gurobi}. To minimize the size of the problem, we initially only use the initialization basis detectors. We give each shot one CPU hour running on a single CPU core and if the solver times out before solving to optimality, we accept its solution as long as the relative mixed integer programming optimality gap (MIPGap~\cite{gurobi_mipgap_parameter}) is less than $0.1$. At high physical error rates, a non-negligible amount of shots reach the timeout with a gap greater than $0.1$ and we simply count these as logical errors. At low physical error rates, a negligible fraction of shots reach the timeout (at $p=0.1\%$ on the $\llbracket 300, 60, 14 \rrbracket$ mitten code only $3$ out of $7.522$ billion reach the timeout) and for these we employ the following strategies to quickly solve them after the timeout:
\begin{enumerate}
    \item For the memory experiments at the lowest physical error rate probed for each code, we re-decode each timed-out shot on a sub-DEM before elevating to the full correlated $XYZ$ detector error model. The sub-DEM restricts the full $XYZ$ model to the neighborhood of the shot's triggered detectors. It keeps every error mechanism that touches a triggered detector, and every detector those mechanisms touch. We solve this much smaller problem with a timeout of a few hundred seconds and accept its solution if the MIPGap is below $0.2$. Otherwise, we promote the shot to the full $XYZ$ model for the remaining hour, accepting at a MIPGap below $0.1$. All of these solves run on a single CPU core and this strategy is sufficient to solve all remaining shots.

    \item For the surgery experiments, we simply re-solve the same problem on the initialization-basis detectors with Gurobi parallelized across multiple CPU cores. This dramatically reduces the time to solve to optimality, with all previously timed-out shots converging in tens of seconds to a few minutes at most.
\end{enumerate}

\subsection{Comparison with prior decoders}

\begin{figure}
    \centering
    \includegraphics[width=0.5\linewidth]{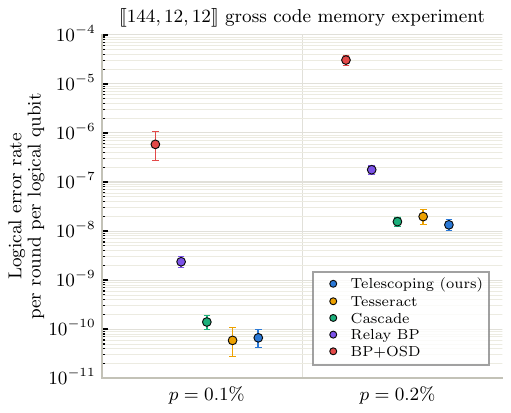}
    \caption{\textbf{Decoder comparison on the $\llbracket 144, 12, 12 \rrbracket$ gross code.}
    We plot the logical error rate per round per logical qubit for memory experiments under uniform depolarizing circuit-level noise at physical error rates $p=0.1\%$ and $p=0.2\%$, using the detector error models of Ref.~\cite{relayBicycleBivariateTestdata}. For our
    telescoping decoder, we averaged the logical error rate over $X$ and $Z$ basis experiments with an equal number of shots in each basis (Table~\ref{tab:gross144_gari}), and error bars are two-sided $95\%$ confidence intervals. The BP+OSD, Relay BP, Tesseract,
    and Cascade results are those of Ref.~\cite{gu2026scalableneuraldecoderspractical} which we extracted from the vector graphics of Figs.~1(b) and 2(b) of that reference using PyMuPDF~\cite{pymupdf}. At $p=0.1\%$ our telescoping decoder reaches a lower logical error rate than Cascade and matches Tesseract within error bars. It does so at nearly twice the throughput of Cascade, and at orders of magnitude higher throughput than Tesseract, which approximates most-likely error decoding (see Table~\ref{tab:gross144_gari}).}
    \label{fig:bb144_dec_comp}
\end{figure}
We compare our telescoping decoder to current state-of-the-art decoders by benchmarking it on memory experiments on the $\llbracket 144, 12, 12 \rrbracket$ gross code using the uniform depolarizing noise detector error models from Ref.~\cite{mueller2025improvedbeliefpropagationsufficient}. These detector error models can be downloaded from Ref.~\cite{relayBicycleBivariateTestdata} and are the same~\cite{gu2026private} as the ones used in Ref.~\cite{gu2026scalableneuraldecoderspractical}. Our results are reported in Fig.~\ref{fig:bb144_dec_comp} and Table~\ref{tab:gross144_gari}. Averaged over $X$ and $Z$ basis memory experiments at a physical error rate of $p=0.1\%$ we achieve a logical error error rate per round per logical qubit\footnote{In this section and in Table~\ref{tab:gross144_gari}, the superscripts and subscripts denote the $95\%$ confidence interval, unlike in the main text where they denote the more standard $68.27\%$ confidence interval.} of {$6.7^{+3.3}_{-2.4}\times10^{-11}$}  with a decoding throughput of $160{,}000$--$180{,}000$ syndrome extraction cycles per second or equivalently, in reciprocal units, $5.5$--$6.2$ $\mu$s per syndrome extraction cycle. At a physical error rate of $p=0.2\%$, we achieve a logical error rate per round per logical qubit of {$1.3^{+0.4}_{-0.3}\times10^{-8}$}. These logical error rates match or improve upon the logical error rates of all current state-of-the-art decoders benchmarked in Ref.~\cite{gu2026scalableneuraldecoderspractical} including BP+OSD, Relay BP, Tesseract, Cascade, and other neural decoders\footnote{It would also be interesting to compare to the beam search decoder of Ref.~\cite{ye2025beamsearchdecoderquantum} using the same detector error models; we leave that for future exploration.}~\cite{Blue_2026}. Additionally, at low physical error rates such as $p=0.1\%$, the throughput of our telescoping decoder achieves nearly double the throughput of the comparable best neural decoder in Ref.~\cite{gu2026scalableneuraldecoderspractical} that is closest in logical error rate to our decoder\footnote{The throughput in Ref.~\cite{gu2026scalableneuraldecoderspractical} is benchmarked using an NVIDIA H200 whereas our telescoping decoder uses an NVIDIA H100 and a single CPU core. Using an H200 would further increase the throughput of our decoder.} This demonstrates that a properly tuned telescoping decoder that relies only on variants of belief propagation can achieve state-of-the-art logical error rates and throughput even when compared to neural decoders that require the additional overhead of training. 

Finally, we remark that the throughput of our telescoping decoder decreases with increasing physical error rate since fewer shots are able to converge in the first few stages. By contrast, neural decoders typically maintain constant throughput across physical error rates. However, since the low logical error rate regime is precisely the one that is most costly to probe, maximizing the throughput in this regime is most important. Additionally, while the throughput drops at higher physical error rates, the number of shots needed to estimate the logical error rate also drops. We find that this more than compensates the decrease in throughput.

\begin{table}[t]
\centering
\begin{tabular}{l cc cc}
\toprule
 \textit{Convergence and throughput} & \multicolumn{2}{c}{$X$ basis} & \multicolumn{2}{c}{$Z$ basis} \\
\cmidrule(lr){2-3}\cmidrule(lr){4-5}
Stage & \shortstack{Fraction \\reaching} & \shortstack{Reciprocal throughput\\($\mu$s/cycle)}
      & \shortstack{Fraction \\reaching} & \shortstack{Reciprocal throughput\\($\mu$s/cycle)} \\
\midrule
$S_1$          & $1$                & $0.77$            & $1$                & $0.83$ \\
$S_2$          & $6.5\times10^{-3}$ & $8.3\times10^{2}$ & $6.0\times10^{-3}$ & $7.9\times10^{2}$ \\
$S_{3A}$ & $9.8\times10^{-8}$ & $4.2\times10^{4}$ & $1.6\times10^{-7}$ & $9.8\times10^{4}$ \\
$S_{3B}$ & $5.7\times10^{-8}$ & $4.4\times10^{5}$ & $7.3\times10^{-8}$ & $3.3\times10^{5}$ \\
$S_{3C}$ & $1.6\times10^{-8}$ & $9.6\times10^{5}$ & $1.8\times10^{-8}$ & $1.5\times10^{6}$ \\
\cmidrule(lr){1-5}
Average reciprocal throughput & \multicolumn{2}{c}{$6.2$ $\mu$s/cycle} & \multicolumn{2}{c}{$5.5$ $\mu$s/cycle} \\
\addlinespace[3pt]
\toprule
\addlinespace[3pt]
 \textit{Logical error rates}                                & \multicolumn{2}{c}{$p=0.1\%$}                          & \multicolumn{2}{c}{$p=0.2\%$} \\
\cmidrule(lr){2-3}\cmidrule(lr){4-5}
Logical errors / shots           & \multicolumn{2}{c}{$23 \, / \, 2.4 \text{ billion}$}    & \multicolumn{2}{c}{$54 \, / \, 28 \text{ million}$} \\
\addlinespace[3pt]
Block LER                        & \multicolumn{2}{c}{$9.6^{+4.8}_{-3.5}\times10^{-9}$}   & \multicolumn{2}{c}{$1.9^{+0.6}_{-0.5}\times10^{-6}$} \\
\addlinespace[3pt]
LER per round                    & \multicolumn{2}{c}{$8.0^{+4.0}_{-2.9}\times10^{-10}$}   & \multicolumn{2}{c}{$1.6^{+0.5}_{-0.4}\times10^{-7}$} \\
\addlinespace[3pt]
LER per round per logical qubit & \multicolumn{2}{c}{$6.7^{+3.3}_{-2.4}\times10^{-11}$}  & \multicolumn{2}{c}{$1.3^{+0.4}_{-0.3}\times10^{-8}$} \\
\bottomrule
\end{tabular}
\caption{\textbf{Benchmark of our decoding pipeline on memory experiments on the $\llbracket 144, 12, 12 \rrbracket$ gross code}. We decode $1.2$ billion shots in each basis at a physical error rate of $p = 0.1\%$ and $14$ million shots in each basis at a physical error rate of $p=0.2\%$ using the uniform depolarizing noise circuit-level detector error models of Ref.~\cite{relayBicycleBivariateTestdata}. In the top panel, we track the fraction of shots that reach each stage as well as the throughput of each stage at $p=0.1\%$. We report the decoding throughput in units of microseconds per syndrome extraction cycle for ease of comparison with the results in Ref.~\cite{gu2026scalableneuraldecoderspractical}. The throughput numbers for $S_1$ and $S_2$ are benchmarked using a single NVIDIA H100 GPU. For all $S_3$ stages we use a single CPU core. Since so few shots reach $S_3$, its latency can be completely hidden by running the CPU core in parallel with the GPU so it does not contribute to the average throughput. Remarkably, belief propagation decodes nearly every shot. We suspect this is due to the single basis Tanner graph of the gross code having girth $6$ compared to girth $4$ for mitten codes. At $p=0.1\%$, only two shots reach $S_4$, the integer programming stage, and only four reach it at $p=0.2\%$. We conservatively count every shot that reaches $S_4$ as a logical error, even though both $S_4$ shots at $p=0.1\%$ are decoded correctly after solving to optimality in under $5$ seconds on a single CPU core. }
\label{tab:gross144_gari}
\end{table}

\subsection{Real-time decoding} \label{subsec:real_time_decoding}

A real-time decoder faces two distinct requirements. The first is a \emph{rate} requirement: to avoid the backlog problem~\cite{Terhal_2015}, the decoder must decode syndromes at least as fast as the processor produces syndrome data. Our decoding stack as described already meets this requirement on any hardware platform, since parallel window decoding~\cite{Skoric_2023} allows adding compute to each stage of the pipeline until the amortized decoding time per window falls below the syndrome extraction cycle time. The second is a \emph{latency} requirement, which emerges when the circuit contains feedforward operations. An operation conditioned on a decoded logical measurement outcome cannot be scheduled until one specific window is decoded. What matters here is the \emph{reaction time}, the time to decode that single window. 

In this section, we perform a worst-case analysis of both requirements on neutral atom quantum processors, assuming a syndrome extraction cycle time of $T_{\mathrm{cyc}} = \SI{1}{\milli\second}$. We show that a single dedicated decoder per stage sustains even the most demanding commitment schedule and that the average reaction time is nevertheless below $T_{\mathrm{cyc}}$, with its full distribution across stages given in Table~\ref{tab:rt-load}. We consider an implementation that processes windows one at a time, with the belief propagation (BP) stages running on FPGAs and the integer programming stage parallelized across multiple CPU cores. The FPGA timing is extrapolated from the BP implementations of Refs.~\cite{bascones2026scalablefpgaarchitecturerealtime, maurer2025realtimedecodinggrosscode}.

We assume a sliding-window decoder~\cite{Skoric_2023} that processes a window of $W = \mathcal{O}(d)$ rounds of syndrome extraction at a time. In the most demanding mode, the window advances by only one round each time, so the decoder must \emph{complete one full window decode every syndrome extraction cycle}. This is the worst-case commitment schedule: it maximizes the required decode rate, and committing a wider region less often would only reduce this rate. The reaction time, by contrast, is the duration of a single window decode, set by the decode times of the stages that window traverses, and is unaffected by the commitment schedule. We adopt this worst-case commitment assumption throughout and hence never amortize the cost of a window decode over its $\mathcal{O}(d)$ rounds.

Our telescoping decoder maps onto this requirement stage by stage. We assume each stage has its own hardware so that windows that reach later stages do not hold up incoming windows. Let $f_i$ be the fraction of syndrome rounds whose window reaches stage $i$, with $f_{S_1} = 1$, and let $t_i$ be the worst-case time of one stage-$i$ window decode.  A dedicated decoder for stage $i$ keeps up with the incoming syndromes when the utilization ratio
\begin{equation}
  \rho_i \;=\; \frac{f_i\, t_i}{T_{\mathrm{cyc}}} \;<\; 1.
  \label{eq:rt-utilization}
\end{equation}

We now extrapolate what the per-iteration times of the BP stages would be on an FPGA for memory and surgery on mitten codes. We focus on the $\llbracket 200, 40, 12 \rrbracket$ and $\llbracket 300, 60, 14 \rrbracket$ codes for memory and on the $\llbracket 540, 108, 18 \rrbracket$ code for $XX$ surgery, all at $p=0.1\%$ physical error rate. We base our extrapolation on the numbers measured in the FPGA implementations of Refs.~\cite{bascones2026scalablefpgaarchitecturerealtime, maurer2025realtimedecodinggrosscode} which focus on the $\llbracket 144,12,12 \rrbracket$ gross code memory. The implementation of Ref.~\cite{maurer2025realtimedecodinggrosscode} lays out the entire decoding graph on the FPGA. Every check node and error node has its own processing unit, and the edges of the graph are physical wires between them. All nodes compute at once in a flooding BP manner, with one flooding BP iteration taking $\SI{24}{\nano\second}$ when using detectors only from one basis. The decoder of Ref.~\cite{bascones2026scalablefpgaarchitecturerealtime} instead focuses on the GARI decoding graph. It operates in a serial BP manner, processing one detector check per clock cycle while small parallel units handle the auxiliary consistency checks alongside the serial sweep. One full iteration, where every check is updated once, takes $\SI{6.3}{\micro\second}$ with each individual check update taking $\SI{3.647}{\nano\second}$.

We time each of our stages with the FPGA implementation that processes the same kind of decoding graph. Ref.~\cite{maurer2025realtimedecodinggrosscode} decodes 
the $X$ and $Z$ detector families separately, and its $\SI{24}{\nano\second}$ figure is measured on those single-family decoding graphs. These are the same kind of decoding graphs that we use in $S_1$ of the memory experiments and in every BP stage of the $\llbracket 540, 108, 18 \rrbracket$ surgery experiment. The serial timing applies to the stages that decode the GARI transform, $S_2$ and $S_3$ of the memory experiments. The final stage $S_4$ is not extrapolated at all; we benchmark its timing directly on a 32 CPU core machine.

The gross-code graph is smaller than ours, so we must scale the $\SI{24}{\nano\second}$ figure up. To do so, we keep the circuit size fixed and stream the larger graph through it in several passes, each pass processing one chunk of gross-code size. The per-iteration time is set by the number of passes an iteration needs:
\begin{equation}
  t_{\mathrm{iter}} \;=\; \SI{24}{\nano\second} \times N_{\mathrm{pass}},
  \label{eq:rt-F}
\end{equation}
where $N_{\mathrm{pass}}$ depends on the size of our decoding graph and on the BP schedule. Our stages do not run flooding BP but rather layered or fully serial schedules. For the layered stages, a pass cannot span two layers, because the beliefs must refresh between them, so each layer needs a whole number of passes and an iteration whose layer $k$ carries $E_k$ edges takes
\begin{equation}
  N_{\mathrm{pass}} \;=\; \sum_{k=1}^{K} \left\lceil \frac{E_k}{E_{\mathrm{gross}}} \right\rceil
  \;\le\; \sum_{k=1}^{K} \left( \frac{E_k}{E_{\mathrm{gross}}} + 1 \right)
  \;=\; F + K
  \label{eq:rt-passes}
\end{equation}
passes where $K$ is the total number of layers and $F = E / E_{\mathrm{gross}}$, where $E$ is the edge count of our decoding graph and $E_{\mathrm{gross}} = 30{,}672$ is the edge count of the gross-code single-family graph on which the $\SI{24}{\nano\second}$ was measured.\footnote{Ref.~\cite{maurer2025realtimedecodinggrosscode} does not report the number of edges in the decoding graph, so we rebuilt it from the depth-7 syndrome extraction circuits of Ref.~\cite{bravyi2024high}.} In all of our decoding experiments, we have $E_k < 0.3 \cdot E_{\mathrm{gross}}$ so we may take $N_{\mathrm{pass}} = K$.  

Streaming a graph through a fixed circuit works because the structure of the graph can live in memory instead of wiring: the beliefs and messages sit in on-chip memory banks, and address tables tell each processing unit what to read on each pass. This approach pays additional overheads the hard-wired measurement does not: address logic, memory access, and possibly a slower clock. We do not model these; instead we rely on having sufficient margin in our worst case analysis. Most stages extrapolated from the $\SI{24}{\nano\second}$ number sit at utilization ratios below $0.05$ and would survive a twentyfold overhead; the busiest, the surgery $S_2$ stage at $\rho = 0.24$, still survives a fourfold overhead (see Table~\ref{tab:rt-load}).

For the serial BP stages, the per-iteration time is the number of detector checks in the window times the $\SI{3.647}{\nano\second}$ clock. The $\llbracket 200,40,12 \rrbracket$ and $\llbracket 300,60,14 \rrbracket$ windows have $1920$ and $3360$ detector checks, giving $\SI{7.0}{\micro\second}$ and $\SI{12.3}{\micro\second}$ per iteration. Even though the $\llbracket 540, 108, 18 \rrbracket$ surgery experiment only uses one basis of detectors, we adopt the same $\SI{3.647}{\nano\second}$ per check update timing. It has $5112$ detector checks which gives $\SI{18.6}{\micro\second}$ per iteration.

We use worst-case iteration counts throughout. $S_1$ always runs its full fixed budget: $10$ iterations for the memory experiments and $30$ for the surgery experiment. For $S_2$, we include the warm-up leg and all $40$ relay legs with no early stopping, for $2520$ iterations in total ($2480$ for surgery). For the $S_3$ substages, we include all iterations of every variant: $500$ iterations for $S_{3A}$, $24{,}480$ for $S_{3B}$, and $79{,}500$ for $S_{3C}$. Most shots do not use the maximum allotted iterations in every stage, so extrapolations based on these totals are conservative upper bounds.

Finally, for the integer-programming stage, we directly benchmark rather than extrapolate. We collected the actual shots that reached $S_4$ and re-solved them with Gurobi parallelized across $32$ CPU cores, requiring only for the MIPGap~\cite{gurobi_mipgap_parameter} to fall below $0.1$. Accepted solves are fast, with mean times of $5$ to $\SI{8}{\second}$ and a worst accepted solve of about a minute, with all shots we tested giving the correct logical correction.

\begin{table}[tb]
\centering
\footnotesize
\setlength{\tabcolsep}{4pt}
\begin{tabular}{@{}l c c c @{\hskip 1em} c c c @{\hskip 1em} c c c @{}}
\toprule
 & \multicolumn{3}{c}{$\llbracket 200,40,12 \rrbracket$ memory} & \multicolumn{3}{c}{$\llbracket 300,60,14 \rrbracket$ memory} & \multicolumn{3}{c}{$\llbracket 540,108,18 \rrbracket$ $XX$ surgery} \\
\cmidrule(lr){2-4}\cmidrule(lr){5-7}\cmidrule(lr){8-10}
Stage & $f$ & $t$ & $\rho$ & $f$ & $t$ & $\rho$ & $f$ & $t$ & $\rho$ \\
\midrule
$S_1$ & $1$ & \SI{7.68}{\micro\second} & $7.7\times10^{-3}$ & $1$ & \SI{30.7}{\micro\second} & $3.1\times10^{-2}$ & $1$ & \SI{46.1}{\micro\second} & $4.6\times10^{-2}$ \\
$S_2$ & $2.7\times10^{-2}$ & \SI{17.6}{\milli\second} & $0.47$ & $2.1\times10^{-2}$ & \SI{30.9}{\milli\second} & $0.64$ & $6.3\times10^{-2}$ & \SI{3.81}{\milli\second} & $0.24$ \\
$S_{3A}$ & $7.8\times10^{-6}$ & \SI{3.50}{\milli\second} & $2.7\times10^{-5}$ & $1.3\times10^{-5}$ & \SI{6.13}{\milli\second} & $7.8\times10^{-5}$ & $6.1\times10^{-4}$ & \SI{9.32}{\milli\second} & $5.7\times10^{-3}$ \\
$S_{3B}$ & $3.7\times10^{-6}$ & \SI{171}{\milli\second} & $6.3\times10^{-4}$ & $7.4\times10^{-6}$ & \SI{300}{\milli\second} & $2.2\times10^{-3}$ & $4.8\times10^{-4}$ & \SI{456}{\milli\second} & $0.22$ \\
$S_{3C}$ & $1.3\times10^{-6}$ & \SI{557}{\milli\second} & $7.2\times10^{-4}$ & $2.3\times10^{-6}$ & \SI{974}{\milli\second} & $2.2\times10^{-3}$ & $1.8\times10^{-5}$ & \SI{1.48}{\second} & $2.7\times10^{-2}$ \\
\midrule
$S_4$ mean & $8.5\times10^{-8}$ & \SI{7.8}{\second} & $6.6\times10^{-4}$ & $1.9\times10^{-7}$ & \SI{5.0}{\second} & $9.6\times10^{-4}$ & $8.3\times10^{-6}$ & \SI{6.9}{\second} & $5.7\times10^{-2}$ \\
$S_4$ worst & & \SI{60}{\second} & $5\times10^{-3}$ & & \SI{14}{\second} & $3\times10^{-3}$ & & \SI{17}{\second} & $0.14$ \\
\midrule
$\bar{t} =  \sum_i f_i t_i$ & \multicolumn{3}{c}{\SI{483}{\micro\second}} & \multicolumn{3}{c}{\SI{678}{\micro\second}} & \multicolumn{3}{c}{\SI{597}{\micro\second}} \\
\bottomrule
\end{tabular}
\caption{Worst-case analysis of real-time decoding, assuming FPGA implementations of all BP stages extrapolated from Refs.~\cite{maurer2025realtimedecodinggrosscode, bascones2026scalablefpgaarchitecturerealtime}. For each experiment at $p = 0.1\%$: $f$ is the fraction of syndrome rounds whose window reaches the stage, pooled over the $X$- and $Z$-basis runs ($f_{S_1} = 1$); $t$ is the worst-case time to decode one window at that stage; and $\rho = f\,t/T_{\mathrm{cyc}}$ is the utilization ratio at $T_{\mathrm{cyc}} = \SI{1}{\milli\second}$. The $S_2$ and $S_3$ rows of the memory experiments and the $S_3$ rows of the surgery experiment are timed from the serial decoder of Ref.~\cite{bascones2026scalablefpgaarchitecturerealtime}; the remaining BP rows are timed from the spatially parallel decoder of Ref.~\cite{maurer2025realtimedecodinggrosscode} via Eq.~\eqref{eq:rt-F}. The $S_4$ times are not extrapolated but measured directly on a $32$-core CPU, giving the mean and the worst accepted solve over the hardest production shots (same $f$). $\bar{t} = \sum_i f_i t_i$ is the average decoding latency, using the mean $S_4$ time.}
\label{tab:rt-load}

\end{table}

Our results are summarized in Table~\ref{tab:rt-load}, and both real-time requirements are met. For the rate requirement, every stage of every experiment has a utilization ratio below one, so a single decoding unit per stage sustains even the worst-case commitment schedule with no backlog. For the latency requirement, the average decoding latency $\bar{t} = \sum_i f_i t_i$ is comfortably below $T_{\mathrm{cyc}} = \SI{1}{\milli\second}$ in all three experiments: the typical window resolves in $S_1$ within tens of microseconds, the $2$--$6\%$ of windows that defer to $S_2$ react within a few tens of milliseconds at worst, and later stages are reached too rarely ($f \lesssim 10^{-3}$) to affect the average. This holds despite uniformly worst-case assumptions: we commit only one round per window decode, every stage pays its full iteration budget, and we even consider the case where all $S_4$ windows take their maximum measured time. 

We expect real-time performance can be pushed well past this worst case. Our pipeline is tuned to maximize Monte Carlo throughput on GPUs and CPUs; a stack designed for real-time decoding could likely use fewer stages. By relaxing our worst-case assumptions, further optimizing the stages, and using a fast neural decoder~\cite{zhao2026towards, gu2026scalableneuraldecoderspractical} in $S_4$ instead of integer programming, we expect that it should also be possible to push towards the real-time decoding requirements at microsecond time-scales needed for superconducting quantum processors. 

Finally, we caution that our noise models do not account for atom loss, a major error source in current neutral-atom processors~\cite{bluvstein2026fault}. Loss is more challenging to decode than Pauli noise because it is typically detected only at a later readout, with the exact moment the atom was lost unknown. A loss-aware decoder must therefore condition on each shot's detected loss pattern, either by reconstructing the decoding hypergraph shot by shot~\cite{Baranes_2026} or by reweighting the priors of a fixed decoding graph, as in the recently introduced Pauli-envelope framework~\cite{liu2026achieving}. However, these works focus primarily on the surface code and understanding the achievable performance of qLDPC codes more broadly under realistic atom loss remains an important direction for future work. On the overhead side there is room for optimism. Replacing worst-case iteration budgets with typical iteration counts shortens the window decodes themselves---and hence the reaction time---by roughly an order of magnitude, while committing $\mathcal{O}(d)$ rounds per window decode rather than a single round reduces the required decode rate by another. A loss-aware decoder an order of magnitude more expensive than the Pauli-only decoding benchmarked here could therefore still achieve reaction times comparable to those reported above, and even at two orders of magnitude the stack would still meet the rate requirement, with average reaction times growing to the order of several syndrome extraction cycles.

\subsection{Details of decoding experiments}\label{app:decdetails}
All of our mitten code decoding experiments use circuit-level depolarizing noise with a single strength $p$. Every CNOT gate is followed by two-qubit depolarizing noise of strength $p$. Every qubit receives single-qubit depolarizing noise of strength $p$ after it is initialized and again before it is measured except in our memory experiments where we omit state preparation errors on our data qubits. Motivated by the long coherence times of neutral atoms, no noise is applied to idling qubits.

For mitten code memory experiments, we use the hook-error free syndrome extraction schedules described in Appendix~\ref{app:hookfree}. For the surgery experiments, we use a coloration circuit for syndrome extraction which we verify using \textsf{sQetch} that the circuit level distance is likely preserved. 

In a surgery experiment, we fault-tolerantly measure one or more logical Pauli products using the gadgets of Appendix~\ref{app:surgery}, which merge the code with a set of ancillary gadget qubits and checks. We describe the measurement of $X$-type products; the $Z$-type case is the transpose. The gadget qubits are prepared in the $Z$ basis, the data qubits in the $X$ or $Z$ basis, and the merged code then undergoes $t_s$ rounds of syndrome extraction with an edge-coloring schedule. Finally, the gadget qubits are read out in the $Z$ basis and the data qubits in their preparation basis. The value of each measured product is reconstructed from the first round of gadget-check outcomes, and the later rounds protect this value. Detectors follow the same convention as in the memory experiments. The observables are the logical operators that are deterministic in the preparation basis. In the $X$-basis run these include the measured products themselves, which the data preparation pins to $+1$, so their decoded values give a direct end-to-end check of the surgery. The $Z$-basis run instead checks how much the surgery disturbs the surviving $Z$ logicals. Our experiments measure a single logical operator, a joint product on two logical qubits, and up to ten joint products simultaneously in the high-rate setting. 

\section{Hardware Complexity}
\label{app:experimental_complexity}
In this section, we describe and compute various metrics which characterize the complexity of implementing mitten codes experimentally both in atom arrays and superconducting qubits.
\subsection{Atom arrays}
Our first set of metrics represent the complexity of implementing a syndrome extraction (SE) cycle with current hardware, in which atoms are transported by acousto-optic deflectors (AODs) under rigid movement constraints. The second metric targets hardware that may be available within the next few years, particularly a fast spatial light modulator (SLM) where each atom can be steered independently on its own path.
For both models, we report the SE cycle time, as well as other experimentally relevant metrics.
Before defining these metrics, we will first describe the connectivity required and atom movements that realize a SE cycle.

\subsubsection{Connectivity needed for a syndrome extraction cycle}
\label{app:connectivity-atomarrays}
The connectivity between check qubits and data qubits is defined by $H_X$ and $H_Z$, which for our codes (Appendix~\ref{sec:app_structured-mittencodes}), are
\begin{equation}
\begin{split}
    H_X &=
\bordermatrix{
    & D_1    & D_2    & D_3    & D_4    & D_5 \cr
X_0 & L(a_0) & 0      & L(a_1) & 0      & R(b_0^*) \cr
X_1 & 0      & L(a_0) & 0      & L(a_1) & R(b_1^*) \cr
} \\
H_Z &=
\bordermatrix{
    &   &  &   &   & \cr
Z_0 & R(b_0) & R(b_1) & 0      & 0      & L(a_0^*) \cr
Z_1 & 0      & 0      & R(b_0) & R(b_1) & L(a_1^*) \cr
}.
\end{split}
\label{eq:HxHz}
\end{equation}

In all our mitten code instances, all ring elements $a_0, a_1, b_0, b_1 \in \F_2[G]$ have support over 3 group elements $g_i \in G$ as $a_0 = g_1 + g_2 + g_3, a_1 = g_4 + g_5 + g_6$, and similarly for $b_0, b_1$.

\textit{Check matrix in binary.}
When a check matrix $H$ is written in binary, each row represents a check, each column represents a data qubit, and a nonzero entry $H_{ij} = 1$ means that check qubit $i$ interacts with data qubit $j$ during syndrome extraction. Written compactly in Eq.~\eqref{eq:HxHz}, each ring entry expands to a binary matrix through the left and right regular representations. Recall from Definition~\ref{def:group-ring-and-regular-reps} that, for $g \in G$, $\Lrep{g}$ and $\Rrep{g}$ are the binary matrix representations of left (and right) multiplication by the group element $g$ (and $g^{-1}$) respectively. More precisely, they act on the standard basis vectors $\{\b{h} :h\in G\}$ as 
\[\Lrep{g}\b{h}=\b{gh}, \qquad \Rrep{g}\b{h}=\b{h g^{-1}}.\]
Here, for $h \in G, \, \b{h} \in \F_2^{|G|}$ is the binary vector of length $|G|$ with a single $1$ in the index corresponding to $h$, and $0$'s elsewhere. The maps $h \mapsto gh$ (and $h \mapsto hg^{-1}$) are bijections of $G$ to itself, as they are invertible by left multiplication of $g^{-1}$ (and right multiplication by $g$) respectively. Therefore $\Lrep{g}$ and $\Rrep{g}$ permute the standard basis vectors of $\mathbb F_2^{|G|}$, and hence are $|G|\times |G|$ binary permutation matrices.

\textit{Each atom is a group element.}
Because $\Lrep{g}$ and $\Rrep{g}$ are $|G| \times |G|$ permutation matrices, every block of Eq.~\eqref{eq:HxHz} expands to a binary block with $|G|$ check indices and $|G|$ data indices. Each data block ($D_1,\dots,D_5$) and each check block ($X_0,X_1,Z_0,Z_1$) therefore holds $|G|$ physical atoms, each corresponding to one group element of $G$.

\textit{The SE cycle is a sequence of gate layers.}
In our construction, each ring element is composed of 3 group elements $a_0 = g_1 + g_2 + g_3$, and so $\Lrep{a_0} = \Lrep{g_1} + \Lrep{g_2} + \Lrep{g_3}$ (likewise for $\Rrep{\cdot}$ entries). During syndrome extraction, a check block must perform an entangling pulse with each of the single-element layers $\Lrep{g_i}$ (or $\Rrep{g_i}$) in its row of~\eqref{eq:HxHz}. The order in which these layers are performed is determined by the hook-error-free schedules of Appendix~\ref{app:hookfree}. In our implementation, we keep the data atoms fixed and move the check atoms to perform these subsequent layers.

\textit{The movement between layers is multiplication by a group element.}
It remains to describe what exact movements of the check atoms are required for one such move. Layer $\Lrep{g_i}$ requires check atom $k$ to sit adjacent to the data atom $h = g_i^{-1}k$ (since $\Lrep{g_i}_{k,h} = 1 \Longleftrightarrow k = g_ih.$) To then perform layer $\Lrep{g_j}$, every check atom must move from data atom $h = g_i^{-1}k$ to data atom $h' = g_j^{-1} k$. This rearrangement of the check atom is exactly performing the movement that corresponds to left multiplication
\begin{equation}
    \Lrep{g_j^{-1}g_i}:\, g_i^{-1} k \longmapsto (g_j^{-1}g_i) (g_i^{-1} k ) = g_j^{-1} k
    \label{eq:atom-rel-perm}
\end{equation} on the check qubits in that block. 

The same reasoning applies to right regular layers $\Rrep{g_i}$, but we must be careful about the inverse in our convention for $\Rrep{\cdot}$ ($\Rrep{\cdot}$ is not just literal right multiplication but rather right multiplication by the inverse $\Rrep{g} h = hg^{-1}$, \cref{def:group-ring-and-regular-reps}).
Recall that $\Rrep{g_i}_{k,h} = 1 \Longleftrightarrow k = hg_i^{-1}$.
Thus, during layer $\Rrep{g_i}$, check atom $k$ sits next to data atom $h = k g_i$, and must move to $h' = k g_j$ for the next layer $\Rrep{g_j}$.
This transition is therefore literal right multiplication by $g_i^{-1} g_j$, which under our convention for the right regular representation, is the map \begin{equation}
    \Rrep{g_j^{-1} g_i} :\, kg_i \longmapsto (k g_i) (g_i^{-1} g_j) = kg_j.
    \label{eq:atom-rel-perm-R}
\end{equation}

\textit{Atom layout on a 2D grid for a direct product group.}
Because an AOD only allows separable row and column movement, we place the $|G|$ atoms of each block on a 2D grid such that group multiplication is equivalent to row and column permutations. 
For direct products $G = G_1 \times G_2$, we can place one subgroup along each axis of the 2D grid. The atom at grid coordinate $(a,b)$ corresponds to group element $(a,b) \in G_1 \times G_2$. Because multiplication in a direct product is componentwise $(a_1, b_1) \cdot (a_2,b_2) = (a_1 a_2, b_1 b_2),$ the required movements from~\cref{eq:atom-rel-perm} factor into independent permutations of the two coordinates. Namely, if $g_j^{-1} g_i = (a,b)$, then left multiplication by $g_j ^{-1} g_i$ sends each check atom $k' = (a', b')$ to \[(a', b') \longmapsto (a \cdot a', b \cdot b')\] under left multiplication by $(a,b)$. Physically, this means that the movement can be decomposed into a row permutation corresponding to left multiplication by $a$ and column permutation for left multiplication by $b$.

\begin{figure}[t]
    \centering
    \includegraphics[width=\linewidth]{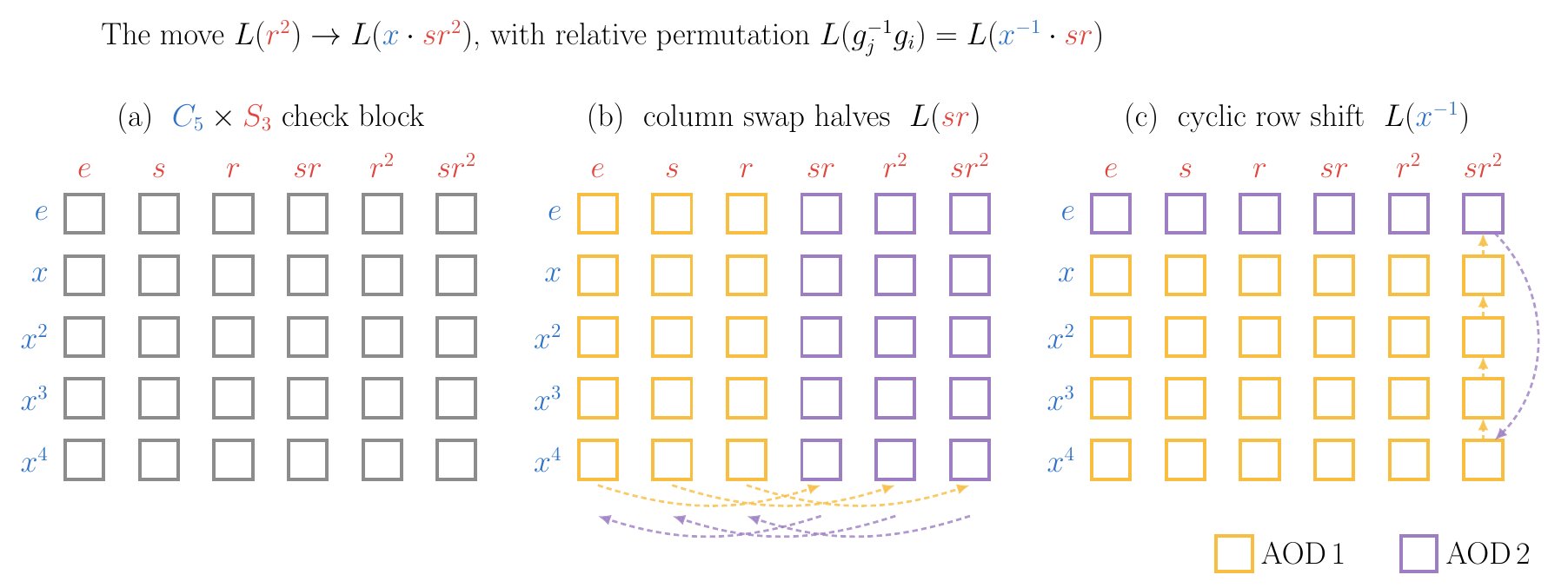}
    \caption{Example atom movement in the SE cycle for the $\llbracket  150, 30, 10 \rrbracket$ code. (a) Each of the $|G| = 30$ atoms are labeled by a group element $(x^a, \sigma) \in C_5 \times S_3$, the row giving the $x^a \in C_5$ element (blue) and the column the $\sigma \in S_3$ element (red). The specific move is the layer transition $\Lrep{g_i} = \Lrep{r^2}$ to $\Lrep{g_j} = \Lrep{x \cdot sr^2}$ performed by $X_0$ and $X_1$. Thus, the relative permutation to perform is $\Lrep{g_j^{-1} g_i} = \Lrep{x^{-1} \cdot sr}.$ This factorizes into (b) the $S_3$ column permutation $\Lrep{sr}$ and (c) the $C_5$ row shift $\Lrep{x^{-1}}$. $\Lrep{sr}$ is simply swapping the 3 left and right columns, and $\Lrep{x^{-1}}$ is just a cyclic shift of the rows by 1, both moves which can be done in 1 swift motion with 2 AODs. Each check atom is outlined by the AOD pair that carries it, those in AOD 1 are in yellow and those in AOD 2 are in purple. We draw them as squares to emphasize the check qubits are the ones moving. Note too, that here, we leave the coordinate labeling fixed while the qubits permute, as opposed to the main figure.}
    \label{fig:c150-move}
\end{figure}
\textit{Example: $\llbracket  150, 30, 10\rrbracket$.}
For the $\llbracket 150, 30, 10 \rrbracket$ code, the group is $G = C_5 \times S_3$. We place the $C_5$ coordinate along the vertical axis, and $S_3$ along the horizontal axis. Thus, each of the $|G| = 30$ atoms are labeled by a group element $(x^a, \sigma) \in C_5 \times S_3$ where $x$ is the generator of $C_5$, $a \in \{0,1,2,3,4\}$, and $\sigma \in S_3$. This same group-element labelling is used for each of the five data blocks ($D_1, \dots, D_5$) and four check blocks ($X_0, X_1, Z_0, Z_1$). Now, it is clear that multiplication by an element of $C_5$ is a cyclic shift of the rows, while multiplication by an element of $S_3$ applies that permutation to the columns. An explicit movement from the hook error free SE schedule determined in Appendix \ref{app:hookfree} is depicted in \cref{fig:c150-move}.

\textit{Semidirect products.} For semidirect products $G=G_1\rtimes G_2$, we again place $G_1$ along the vertical axis and $G_2$ along the horizontal axis, but the multiplication rule is now different. We label the check qubit in row $a' \in G_1$ and  column $b' \in G_2$ by the group element $a' b'$, exactly as specified in \cref{def:semidirect}.
Let the required permutation from~\cref{eq:atom-rel-perm} be $g_j^{-1} g_i = a b$. Then left multiplication by $g_j ^{-1} g_i$ sends each check atom $a' b'$ to
\[(a b)\cdot(a' b')=(a \,\varphi_{b}(a')) (b b'),\] by \cref{def:semidirect} of the group multiplication of a semidirect product. Thus, left multiplication by a fixed element $(a, b)$ sends the $B$ coordinate to $b' \mapsto bb',$ a fixed cyclic shift of the columns when $B$ is cyclic. On the other hand, right multiplication by $(ab)$ acts on check qubit $(a'b')$ as \[(a' b')\cdot(a b)=(a' \varphi_{b'}(a)) (b'b).\] The column coordinate $b'$ is once again a simple fixed cyclic shift. But the $A$ coordinate is transformed as $a' \mapsto a' \, \varphi_{b'} (a),$ where the row shift depends on the column $b'$. Unlike the direct product case, the semidirect product twist appears as column-dependent row shifts rather than a uniform one.

\begin{figure}[t]
    \centering
    \includegraphics[width=\linewidth]{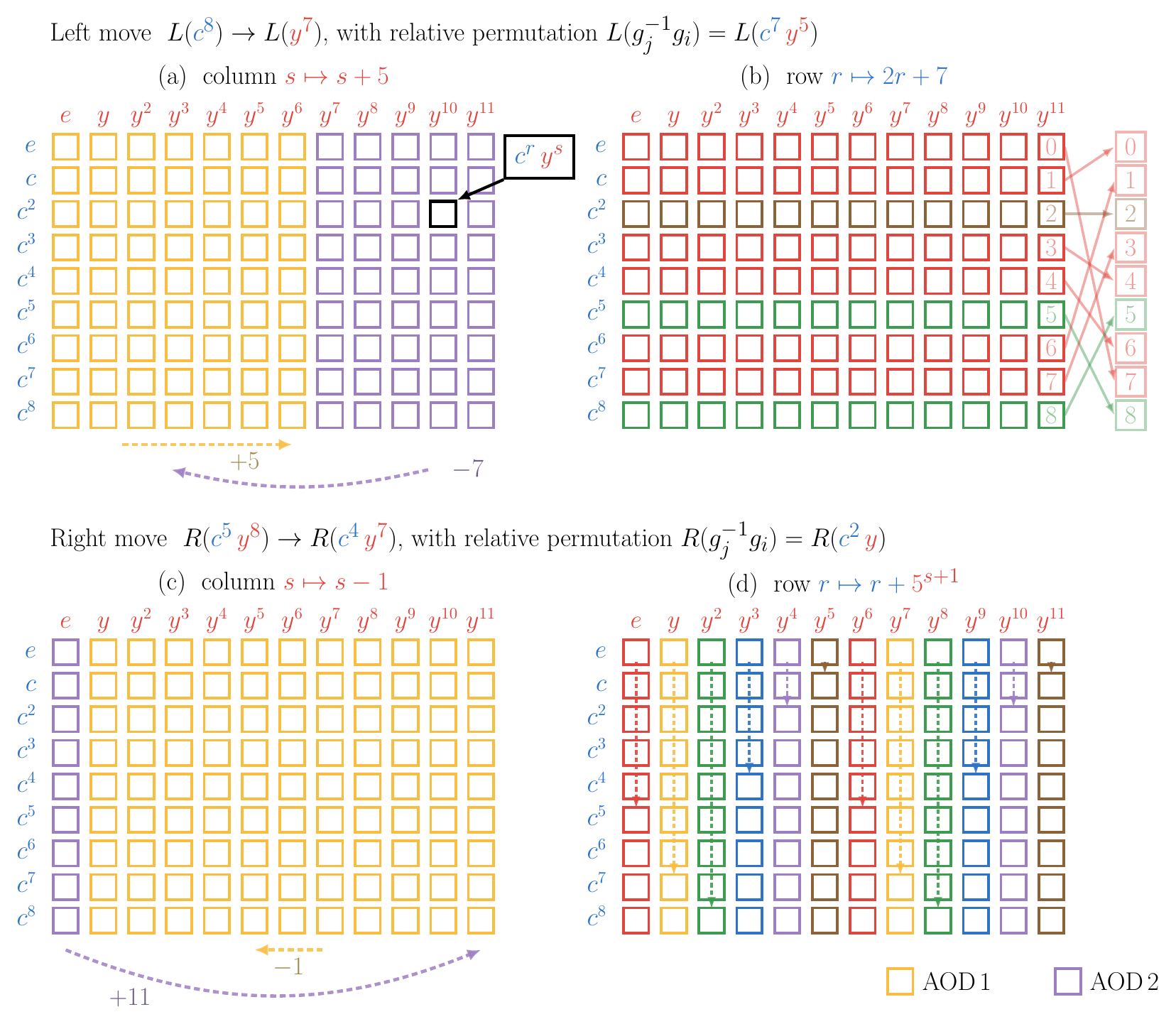}
    \caption{Example atom movement in the SE cycle for the $\llbracket 540, 108, 18 \rrbracket$ code with $G = C_9 \rtimes C_{12}$. Each atom in a block is labeled by a group element $c^r y^s$. The column index $s\in\{0,\dots,11\}$ is the $C_{12}$ complement power $y^s$ (red), the row index $r\in\{0,\dots,8\}$ is the normal $C_9$ power $c^r$ (blue). (a,b) The left regular transition $\Lrep{c^8}\to\Lrep{y^7}$ has relative permutation $\Lrep{g_j^{-1}g_i} = \Lrep{c^7y^5}$. Using the group multiplication defined in \cref{tab:semidirect}, it acts on coordinates as $(s,r)\mapsto(s+5,\,2r+7)$. The column motion is a uniform cyclic shift, while the row motion is a single permutation of the rows applied identically in every column. (c,d) The right regular transition $\Rrep{c^5y^8}\to\Rrep{c^4y^7}$ on data block $D_5$ has relative element $\Rrep{g_j^{-1} g_i} = \Rrep{c^2y}$ and acts as $(s,r)\mapsto(s-1,\,r+5^{s+1})$. Once again, the column motion is a uniform cyclic shift, but the row shift depends on the column.}
    \label{fig:c540-move}
\end{figure}

\textit{Example: $\llbracket 540, 108, 18 \rrbracket$.} For the $\llbracket 540, 108, 18 \rrbracket$ code, the group is $G = C_9 \rtimes C_{12}$ with $\varphi_y(c) =ycy^{-1}=c^5$ for $c$ is the generator of $C_9$ and $y$ is the generator of $C_{12}$ (see \cref{tab:semidirect}). We place $C_9$ along the vertical axis and $C_{12}$ along the horizontal axis, so the atom in row $r$ and column $s$ is labeled by $c^r y^s$. 
For left multiplication by $c^p y^q$, we have
\[ c^p y^q \cdot c^r y^s = c^{p+5^q r} y^{q+s}.\]
The check qubits at coordinates transform as \[ (s,r)\longmapsto (s+q,\, 5^q r+ p). \]
In \cref{fig:c540-move}(a,b), we outline a left multiplication permutation from the hook error free SE cycle found in \cref{app:hookfree}. The relative permutation from
$\Lrep{c^8} \to \Lrep{y^7}$ is $\Lrep{g_j^{-1} g_i} = \Lrep{y^5 c^8} = \Lrep{c^7y^5}$. The check qubits at col $s$, row $r$ shift as
\[ (s,r)\longmapsto (s+5,\, 2r+7).\] The column move is still a uniform cyclic shift by 5, but the row move is not just a shift: it is the map $r\mapsto 2r+7$, which can be implemented in 3 cyclic shifts. 

Right multiplication is a little less uniform.  
For the move shown in \cref{fig:c540-move}(c,d), the transition
$\Rrep{c^5y^8} \to \Rrep{c^4y^7}$ has relative element $\Rrep{g_j^{-1}g_i} = \Rrep{(c^4 y^7)^{-1} (c^5y^8)} = 
\Rrep{c^2y}$ (\cref{eq:atom-rel-perm-R}). 
Since $\Rrep{c^2y}$ acts by literal right multiplication by
$(c^2y)^{-1}=c^5y^{11}$, and literal right multiplication by $c^p y^q$ gives 
\[ c^r y^s \cdot c^p y^q = c^{r+5^s p}y^{s+q},\] then \[
c^r y^s\longmapsto c^r y^s(c^5y^{11})
 =c^{r+5^s\cdot5}y^{s+11}
 =c^{r+5^{s+1}}y^{s-1}.
\]
Thus, the check qubits at coordinates transform as \[ (s,r)\longmapsto (s-1,\, r + 5^{s+1}). \] The column shift is again just a cyclic shift by 1. However, the row shift now depends on the current column $s$, which decomposes into $6$ cyclic shifts done sequentially with 2 AODs.

\textit{Semidirect decomposition of $G$.} There are different ways to decompose groups, and in reality, the final decomposition is optimized over SE cycle time as well. Each such decomposition provides a possible coordinate grid, and thus different row and column permutations. 
For example, the group used for the
$\llbracket540,108,18\rrbracket$ code admits the descriptions
\[ G\cong C_9\rtimes C_{12}
 \cong (C_9\rtimes C_3)\rtimes C_4
 \cong (C_9\rtimes C_4)\rtimes C_3, \]
with the appropriate conjugation action understood in each semidirect product. These descriptions yield 2D grids of size $9\times12$, $27\times4$, and $36\times3$, respectively. We use the $C_9\rtimes C_{12}$ description, with $12$ columns and $9$ rows (shown in \cref{fig:c540-move}), because it gives a faster SE-cycle time. 

\textit{Optimized data qubit layouts.} The coordinate grids in \cref{fig:c150-move,fig:c540-move} label the data qubit sites by group elements $x \in G$. In deriving the check atom movements above, we assumed that data qubit $x$ occupies the site with coordinate $x$. To reduce transport time, these data qubit layouts can be optimized within each data block $D$. Optimization for a block $D$ assigns a data qubit $x$ a new coordinate $P_D(x)$, where $P_D: G \to G$ is a permutation. Now, the site with coordinate $s$ hosts data qubit $P_D^{-1}(s)$.

Let $\pi:G\to G$ denote one of the check atom site permutations for the unoptimized data qubit assignment. Specifically,
$\pi=\Lrep{g_j^{-1}g_i}$ or $\Rrep{g_j^{-1}g_i}$ as derived in \cref{eq:atom-rel-perm,eq:atom-rel-perm-R}. Suppose a transition takes a check atom from data block $D$ to data block $D'$. Under the unoptimized layouts, it moves from coordinate $x$ to coordinate $\pi(x)$. Under the optimized layouts, the source and destination coordinates are now $P_D(x)$ and $P_{D'}(\pi(x))$. Thus, the physical movement that the check qubits perform is \[\pi_{\text{phys}} = P_{D'} \circ \pi \circ P_{D}^{-1}.\]

\textit{Figures and videos.} The example moves described for $\llbracket150,30,10\rrbracket$ and $\llbracket540,108,18\rrbracket$ are depicted using the chosen data qubit labeling in \cref{fig:c150-move} and \cref{fig:c540-move}. The accompanying videos in \cite{knitkitgithub} animate the complete SE cycles using the optimized data qubit assignments $P_D$. A particular transition is unchanged under the optimized layout when \[ \pi_{\mathrm{phys}}=\pi. \]
The transitions selected for both figures have this property under their optimized layouts. So, the exact movement shown in \cref{fig:c150-move} appears at 00:32--00:34 of the $\llbracket150,30,10\rrbracket$ 2-AOD SE-cycle video, and the $D_5$ transition $\Rrep{c^5y^8}\to\Rrep{c^4y^7}$ shown in \cref{fig:c540-move} appears at 00:59--01:05 of the corresponding $\llbracket540,108,18\rrbracket$ video.

In summary, an SE cycle reduces to a sequence of relative permutations $\Lrep{r_{ij}}$ and $\Rrep{\cdot}$ applied to each check block. The SE cycle videos animate the resulting optimized physical permutations $P_{D'}\pi P_D^{-1}$ over the complete schedule.
The metrics reported in \cref{table:atomarray-metrics} measure the costs of executing exactly these permutations under the two hardware options: AODs available today and idealized fast SLMs.

\begin{figure}[t]
    \centering
    \includegraphics[width=\linewidth]{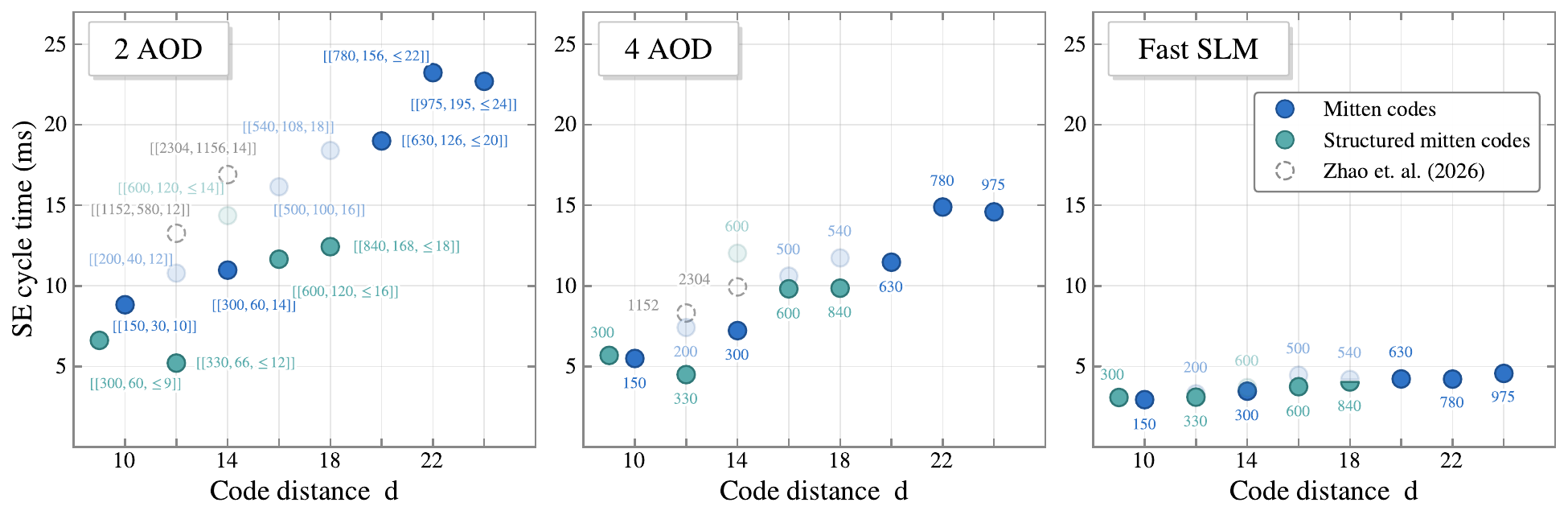}
    \caption{SE cycle time for mitten codes and structured mitten codes. To represent current day hardware, we estimate SE cycle times given 2 (left) or 4 (center) pairs of crossed AODs. With AODs, atom movement is restricted to a strict product grid set of atoms and rigid translations. To represent idealized, futuristic hardware, a fast SLM (right) can steer atoms independently and quickly. 
    For comparison, we assume the same physical layout and acceleration assumptions as the AOD estimates.
    All SE time estimates assume $5500 \mathrm{\, m/s^2}$ for acceleration and deceleration, following prior theoretical SE cycle estimates \cite{zhao2026towards}, though this parameter is tuneable. 
    The exact values for all plotted points are listed in \cref{table:atomarray-metrics}. When multiple codes have the same distance, the fastest is shown in full color and the others are faded for readability. 
    Blue points are instances of the mitten codes detailed throughout this paper. Green points are structured mitten codes of \cref{sec:app_structured-mittencodes}, which are designed to reduce SE cycle time by requiring the $D_5$ polynomials to coincide so the corresponding moves execute simultaneously. Grey dashed points are the codes analyzed in Ref.~\cite{zhao2026towards} from which we pasted their SE cycle times from. 
    In the 2 AOD panel, we label the $\llbracket n, k, d \rrbracket$ parameters, and for the subsequent plots, we just keep the $n$ value.
    }
    \label{fig:atom_SE_cycle_times}
\end{figure}

\subsubsection{Hardware Option 1: AODs}
Having fixed the required check-atom permutations and a physical 2D layout, we now estimate the cost of implementing a SE cycle using present day hardware.

\begin{definition}[Crossed AOD pair]
A crossed pair of AODs consists of two acousto-optic deflectors with orthogonal deflection axes. One AOD is driven by a set of RF tones that determines the horizontal coordinates $X(t)=\{x_1(t),\ldots,x_{n_x}(t)\}$, while the other determines the vertical coordinates $Y(t)=\{y_1(t),\ldots,y_{n_y}(t)\}$. Together, they produce the product grid of tweezer positions \[ X(t)\times Y(t) = \{(x_i(t),y_j(t)):1\leq i\leq n_x,\ 1\leq j\leq n_y\}. \]

Changing $x_i(t)$ moves the entire corresponding column, while changing $y_j(t)$ moves the entire corresponding row. Consequently, atoms sharing a row or column coordinate cannot follow completely independent trajectories.
\end{definition}

In this work, \emph{$q$-AODs} means $q$ independently controlled crossed AOD pairs. Thus, ``2 AOD'' and ``4 AOD'' refer to two and four crossed pairs, respectively, rather than to two and four individual one-axis deflectors.
Using this convention, we estimate the time for a SE cycle with 2 or 4 pairs of crossed AODs (Table~\ref{table:atomarray-metrics}). 

For the spatial layout and transport timing, we follow the same assumptions listed in points 1-7 of Appendix E.4 of Ref.~\cite{zhao2026towards}. In particular, data qubits are spaced by $12~\mu\mathrm{m}$, and an ancilla qubit must be brought within $2~\mu\mathrm{m}$ of a data qubit to perform an entangling gate. We also use the same transport model, that atoms accelerate and decelerate at a fixed $5500~\mathrm{m/s^2}$, a value motivated by previous experimental data \cite{bluvstein2022quantum}.

\begin{table}[t]
\centering
\setlength{\tabcolsep}{6pt}
\renewcommand{\arraystretch}{1.15}
\begin{tabular}{@{}l l c c c c c@{}}
\toprule
 & & \multicolumn{4}{c}{\textbf{AOD}} & \textbf{Fast SLM} \\
\cmidrule(lr){3-6}\cmidrule(lr){7-7}
Code $\llbracket n,k,d \rrbracket$ & Group $G$ & \makecell{2 AOD SE time\\(ms)} & \makecell{4 AOD SE time\\(ms)}
 & \makecell{Avg.\\transfers} & \makecell{Max speed\\(m/s)} & \makecell{SE time\\(ms)} \\
\midrule
\multicolumn{7}{@{}l}{\emph{Mitten codes}}\\
\addlinespace[2pt]
$\llbracket 150,30,10 \rrbracket$ & $C_5 \times S_3$ & 8.83 & 5.49 & 6.0 & 0.81 & 2.94 \\
$\llbracket 200,40,12 \rrbracket$        & $C_4 \times D_{10}$     & 10.79 &  7.41 &  6.2 & 1.12 & 3.31 \\
$\llbracket 300,60,14 \rrbracket$        & $C_{10} \times S_3$     & 10.98 &  7.22 &  6.4 & 0.96 & 3.47 \\
$\llbracket 500,100,16 \rrbracket$       & $C_5 \rtimes C_{20}$    & 16.15 & 10.61 &  6.8 & 1.15 & 4.45 \\
$\llbracket 540,108,18 \rrbracket$       & $C_9 \rtimes C_{12}$    & 18.41 & 11.74 &  7.6 & 1.21 & 4.19 \\
$\llbracket 630,126,\leq20 \rrbracket$   & $C_7 \rtimes C_{18}$    & 19.00 & 11.47 &  6.0 & 1.33 & 4.22 \\
$\llbracket 780,156,\leq22 \rrbracket$   & $C_{13} \times A_4$     & 23.24 & 14.90 &  9.1 & 1.26 & 4.21 \\
$\llbracket 975,195,\leq24 \rrbracket$   & $C_{13} \rtimes C_{15}$ & 22.70 & 14.60 &  5.4 & 1.52 & 4.57 \\
\midrule
\multicolumn{7}{@{}l}{\emph{Structured mitten codes}}\\
\addlinespace[2pt]
$\llbracket 300,60,\leq9 \rrbracket$     & $C_{10} \times S_3$     &  6.62 &  5.68 &  7.5 & 1.09 & 3.07 \\
$\llbracket 330,66,\leq12 \rrbracket$    & $C_{11} \times S_3$     &  5.21 &  4.49 &  5.8 & 1.09 & 3.09 \\
$\llbracket 600,120,\leq 14 \rrbracket$  & $C_5 \times S_4$        & 14.36 & 12.03 & 12.7 & 1.41 & 3.69 \\
$\llbracket 600,120,\leq 16 \rrbracket$  & $C_5 \times S_4$        & 11.66 &  9.82 & 10.7 & 1.54 & 3.74 \\
$\llbracket 840,168, \leq 18 \rrbracket$ & $C_7 \times S_4$        & 12.43 &  9.86 & 11.0 & 1.54 & 4.04 \\
\bottomrule
\end{tabular}
\caption{SE metrics for implementing mitten and structured mitten codes in atom arrays. All the SE time estimates for either 2 or 4 pairs of crossed AODs, and the fast SLM, are plotted in \cref{fig:atom_SE_cycle_times}. In addition, we report two experimental considerations relevant for AOD implementations. First, the number of trap handoffs experienced by a check qubit across the entire SE cycle, averaged over all check qubits. Second, the maximum speed reached by any atom throughout the SE cycle.} 
\label{table:atomarray-metrics}
\end{table}

These acceleration and distance assumptions are the two main spatial inputs to the timing model. The other points in \cite{zhao2026towards} assume that atoms travel along straight lines and that entangling pulses take $1~\mu \mathrm{s}$, which is negligible compared to the timing of the movements. Thus, we only time the movement of the check qubits (data qubits remain fixed). We do not include Hadamards and measurement time because those operations can be parallelized.

We emphasize that this is not an optimized model over all possible SE strategies. Faster implementations may be possible by also moving data qubits, using Shor-style syndrome extraction with GHZ states for ancillas, or even using different transport methods. 
We optimize the data qubit layout for 2 AOD schedule. We then use this same layout for estimating the 4 AODs and fast SLM SE cycle times, and other metrics. For the 4 AOD estimates, we pipeline successive SE rounds, and report the interval between successive rounds, rather than the time to just do 1 isolated round.

\subsubsection{Hardware Option 2: Fast SLM} 

The second metric assumes possible hardware that may be available in the next coming years. Static tweezers of arbitrary patterns are currently generated with spatial light modulators (SLMs), and the refresh rate of SLMs available today is too slow to transport atoms quickly. However, there is ongoing progress towards developing faster SLMs, with refresh rates ranging from a few kHz \cite{meadowlark_slm_1536}, to 100s of kHz \cite{SLM_DPM}, to even the MHz-scale \cite{wei202610,deters2026ultrafast,bytyqi2026device}. A fast SLM would remove the restriction of AODs to allow arbitrary movements of atoms. The trap pattern could be updated so that each atom follows its own path to the next interaction location. We model this idealized device by keeping the same qubit spacing and acceleration value $5500~\mathrm{m/s^2}$ used for the AOD estimates. Because this hardware model is forward-looking, the raw SE cycle times are not to be emphasized, but rather the scaling is.

\subsubsection{Discussion on the transport parameter assumptions}
\label{sec:transport_assumption_discussion}
For consistency with the hardware model of Ref.~\cite{zhao2026towards}, we evaluate all layouts using the same acceleration $a=5500~\mathrm{m/s^2}$ during atom transport.
We want to emphasize that this acceleration is a tunable parameter based on the trap depth of each tweezer. However, deeper traps entail other trade offs, such as higher laser power.

For some classical intuition, one upper bound on the acceleration follows from the finite tweezer depth. Assume we are in the regime where transport is limited by the finite tweezer trap depth, rather than by AOD technicalities. In the frame of a tweezer accelerating at $a$, the atom experiences an inertial force of $m a$. As the tweezer pulls the atom, the atom lags behind the trap center, and the gradient of the tweezer potential produces a restoring force \[F_{\mathrm{rest}}(x)=-\frac{dU(x)}{dx}\] that pulls the atom back toward the center. To keep the atom trapped, the inertial force must be lower than this maximum restoring force, $ma < F_{\mathrm{rest, \, max}}$.
For a Gaussian tweezer potential $U(x) = -U_0 e^{-2x^2/w_0^2}$ of fixed waist $w_0$, the restoring force is maximal at $|x| = w_0/2$, giving \[F_{\mathrm{rest, \, max}} =\frac{2e^{-1/2}}{w_0}U_0.\] Then, the maximum acceleration is therefore
\begin{equation*}
    a_{\mathrm{conf}}
    =\frac{F_{\mathrm{rest,max}}}{m}
    =\frac{2e^{-1/2}U_0}{mw_0}.
\end{equation*}
At fixed tweezer waist $w_0$ and detuning, $U_0\propto I$, so this classical bound scales as $a_{\mathrm{max}}\propto I$.
If this bound alone determined the transport rate, then a move over fixed distance $d$ which starts and ends at rest, takes \[T_{\mathrm{move}} =  2\sqrt{d/a} \propto I^{-1/2}.\]

In practice, high-fidelity transport is generally limited before this classical escape threshold is reached. This scaling is quite a bit higher than how fast atoms accelerate in practice, and is usually limited by transitions to adjacent vibrational states. Ref~\cite{bluvstein2022quantum} show that the resulting excitation depends on the Fourier component of the acceleration profile at the trap frequency $\omega_0$. For the constant-jerk trajectory analyzed there, maintaining fixed motional excitation gives
\[
    T_{\mathrm{move}}\propto\omega_0^{-3/4}.
\]
To simplify these nuances in exact scaling, we keep this approximate quadratic scaling as a comparison with atom number in the main text.

\textit{Speed limit with AODs.} In current AOD systems, the maximum transport speed is limited by AOD-induced optical aberrations, including cylindrical lensing effects, which can limit transport speed \cite{dickson1972optical}.
In our schedules however, the maximum speed reached is $1.54~\mathrm{m/s}$ (\cref{table:atomarray-metrics}). Peak speeds \textit{around this scale} have been achieved using deeper traps, ie. during Ref.~\cite{manetsch2025tweezer}'s $270 \mu m$ transport in $400 \mu s$ using an adiabatic sine trajectory. Thus, our modeled speed is on the scale of experimentally demonstrated fast transport. Note, however, that Ref.~\cite{manetsch2025tweezer} reported a lower transport fidelity for the deeper-trap move (approximately 99.85\% compared to the shallower trap benchmark 99.95\%), so this comparison establishes experimental accessibility of the speed rather than equivalent high-fidelity transport.

\textit{Trap transfers.} Consecutive movement events may require a check atom to be carried by different crossed-AOD pairs. We therefore report the mean number of trap transfers per check atom in an SE cycle. The number is a valid consideration when performing SE cycles, as recent rigorous experimental characterization of AOD transfer and movement errors \cite{manetsch2025tweezer} indicates that, for the approximate number of transfers and total travel distance required per SE cycle, transfer errors can be more limiting than the motion itself.
In \cref{table:atomarray-metrics}, the average transfers column reports the AOD-to-AOD and AOD-to-SLM trap handoffs that an atom experiences. The number of handoffs is averaged over all check qubits because the data qubits remain stationary in our implementation. Check qubits are handed off to SLM traps if they are stationary for the upcoming move.

All in all, we want to emphasize that SE cycle time is only one of several hardware considerations. While experiments have moved faster with deeper traps, they also saw lower coherence compared to their shallower traps \cite{manetsch2025tweezer}. More laser power also means less qubits which can be moved at a time. Other considerations include number of trap transfers and the maximum speed limit in current AODs. Thus, a holistic analysis is therefore needed when selecting transport speeds and assessing SE cycle times and fidelity.

\subsection{Superconducting qubits}
\label{app:superconducting_hw}

The simplest metric for classifying the hardware complexity of implementing a quantum error correcting code on a superconducting qubits platform is the thickness of the code's Tanner graph~\cite{bravyi2024high,mutzel1998thickness,tremblay2022constant}. 

\begin{definition}[Planar thickness] \label{def:thickness}
The \emph{planar thickness}  $\theta(G)$ of a graph $G = (V,E)$ is the minimum integer $t$ such that the set of edges of $G$ can be decomposed as 
\[
E = E_1 \sqcup E_2 \sqcup \cdots \sqcup E_t
\]
where each graph $G_i=(V,E_i)$ is planar.
\end{definition}
A graph is planar if there exists some placement of its vertices in the plane such that its edges may be routed so that they only intersect at vertices. At first glance, this makes the thickness of the Tanner graph seem unrelated to the complexity of laying out a code on a superconducting qubits chip; the vertices of the Tanner graph correspond to physical data and ancilla qubits whose locations must remain fixed on the chip whereas in Definition~\ref{def:thickness} each subgraph is free to choose the locations of its vertices in order to arrange for planarity. 

Fortunately, this issue is resolved by a theorem due to Pach and Wenger~\cite{Pach2001}. 
\begin{theorem}[Planar graph embedding {\cite[Theorem~1]{Pach2001}}]
    Any planar graph $G = (V,E)$ can be embedded with its vertices placed at any prescribed set of $|V|$ distinct points in the plane, provided the edges are allowed to be drawn as polygonal curves. Furthermore, such an embedding can be computed in $O(|V|^2)$ time.
\end{theorem}
Hence, the planar thickness of the Tanner graph is still a justified proxy for the minimum number of planar gate layers needed to realize the code.

\begin{proposition}[Tanner graph thickness lower bound]
Let $\mathcal T$ be the Tanner graph of an $\llbracket n,k,d \rrbracket $ CSS stabilizer code with $m$ stabilizer generators. Let $\bar w$ be the average check weight of the code defined as $\bar w \coloneqq \frac{1}{m}\sum_{\alpha} |S_\alpha|$ where $|S_\alpha|$ denotes the number of qubits involved in the stabilizer check $S_\alpha$. The planar thickness $\theta(\mathcal T )$ satisfies the lower bound
\begin{equation}
    \theta(\mathcal T) \geq \left\lceil \frac{\bar w m} {2(n+m-2)} \right\rceil
\end{equation}
In particular, if the stabilizer generators are independent, so that $m=n-k$, and every stabilizer has weight $w$, then
\begin{equation}
    \theta(\mathcal T)
    \geq
    \left\lceil
    \frac{w(1-r)}{2\left(2 - r - \frac{2}{n}\right)}
    \right\rceil \label{eq:thickness_lb_simp}
\end{equation}
where $r = k/n$ is the encoding rate.
\end{proposition}

\begin{proof}
We assume $\mathcal T$ is connected. Since $\mathcal T $ is bipartite its girth is lower bounded by four. In fact, the Tanner graph of any non-trivial CSS code with distance greater than one must have girth exactly equal to four. By Euler's formula the number of vertices, edges, and faces of a planar subgraph must satisfy 
\begin{equation}
    |V| - |E| + |F| = 2.
\end{equation}
Each face must have at least four edges (since $\mathcal{T}$ has girth 4) and each edge participates in at most two faces, so 
\[4|F| \leq 2|E|.\] Plugging this into Euler's formula, we obtain 
\begin{equation}
    |E| \leq 2(|V| - 2),
\end{equation}
so a planar layer of $\mathcal{T}$ can have at most $2(|V| - 2)$ edges. Since $\mathcal{T}$ has $n + m$ vertices and $\bar w m$ edges and each planar layer can only support at most $2(|V| - 2)$ edges, we obtain the lower bound 
\begin{equation*}
        \theta(\mathcal T) \geq \left\lceil \frac{\bar w m} {2(n+m-2)} \right\rceil
\end{equation*}
which simplifies to \eqref{eq:thickness_lb_simp} after taking $m = n - k$, $\bar w = w$, and expressing everything in terms of $r$.
\end{proof} 

All of the mitten codes we consider in this work have block size $n \geq 150$, encoding rate $r = 20\%$, and  check weight $w = 9$. Plugging these numbers into \eqref{eq:thickness_lb_simp} produces the following immediate corollary.
\begin{corollary} \label{cor:thickness_lb}
    The thickness $\theta$ of a mitten code with check weight $w = 9$ and encoding rate $r = 20\%$ is lower bounded by 3.
\end{corollary}

\begin{figure}
    \includegraphics[width=\linewidth]{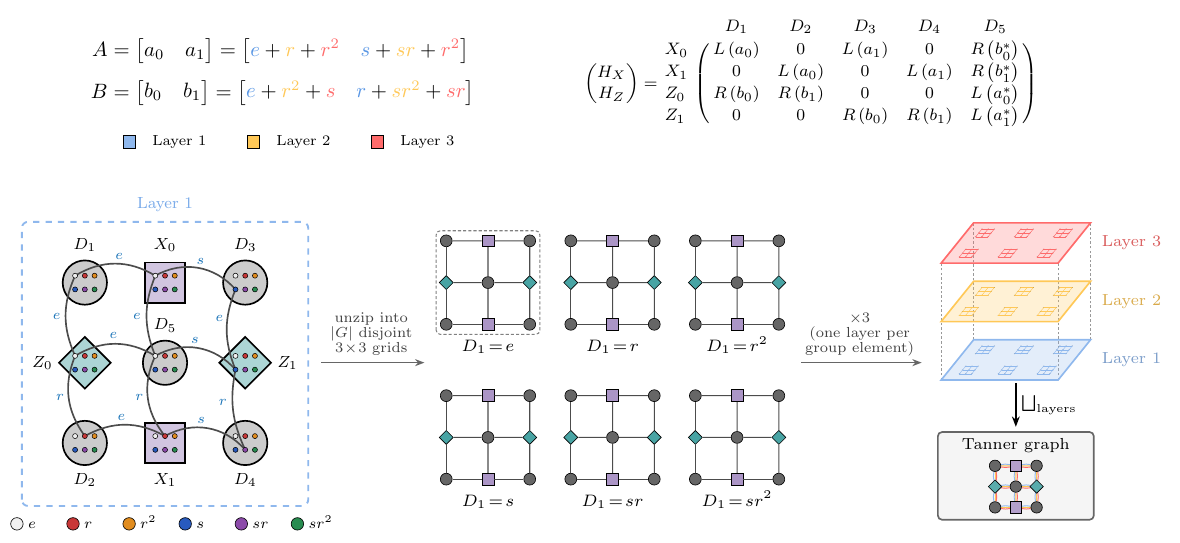}
    \caption{Thickness-3 decomposition of the Tanner graph of a mitten code, illustrated for the nonabelian group $G = S_3 = \langle r, s \rangle$, with $r$ a rotation and $s$ a reflection.
    \textbf{(top)}~Base matrices $A = [a_0\ a_1]$ and $B = [b_0\ b_1]$ over the group algebra $\mathbb{F}_2[S_3]$, together with the resulting parity-check matrices $H_X$ and $H_Z$. Each ring entry is a sum of three group-element monomials, color-coded by the layer they contribute to.
    \textbf{(bottom left)}~Layer~1 corresponds to the first group element in each ring entry (the ones labeled in blue). The data qubits in each block ($D_1, \dots, D_5$) and the check qubits in each block ($X_0, X_1, Z_0, Z_1$) are labeled by group elements: each block holds $|G|=6$ qubits arranged in a $3 \times 3$ grid as shown. Fixing the first $D_1$ qubit to the identity $e$, we extract one $3 \times 3$ grid by following the left and right multiplication---all horizontal edges from a data qubit block to a check qubit block act via the left regular representation, and all vertical edge from a data qubit block to a check qubit block act via the right regular representation.
    \textbf{(bottom center)}~Choosing a different starting qubit in $D_1$, labeled by some other group element, yields a different but isomorphic $3 \times 3$ grid; ranging the starting qubit over all of $G$ produces $|G|$ disjoint $3 \times 3$ grids, one per group element.
    \textbf{(bottom right)}~Repeating this construction for each of the three monomials in every ring entry, we obtain three copies of this $|G|$-disjoint-$3 \times 3$-grids graph, one per group element in each ring entry of the base matrices. The disjoint union of these layers reconstructs the Tanner graph $\mathcal{T}$ of the code.}
    \label{fig:thickness3_layout}
\end{figure}

We now show that this lower bound is tight by presenting an explicit decomposition of the Tanner graph of our mitten codes into three planar layers. We present the decomposition in a fully self-contained manner in Fig.~\ref{fig:thickness3_layout}. In what follows, we simply focus on proving the correctness of the decomposition.

\begin{theorem}\label{thm:thickness}
    The thickness $\theta$ of a mitten code with check weight $w = 9$ and encoding rate $r = 20\%$ is exactly 3. 
\end{theorem}

\begin{proof}
Let $\mathcal{T} = (V, E)$ be the Tanner graph of a mitten code. As illustrated in Fig.~\ref{fig:thickness3_layout}, for each layer $c \in \{1,2,3\}$ we keep only the edges of $\mathcal{T}$ that arise from the $c$-th group element in each ring entry of $A$ and $B$. Since each ring entry contributes one group element per layer, this partitions the edges of $\mathcal{T}$ into three sets $E_0, E_1,$ and $E_2$, and it suffices to show that each layer $(V, E_c)$ is planar.

We do this by exhibiting layer $c$ as $|G|$ disjoint copies of $3 \times 3$ grids as shown in Fig.~\ref{fig:thickness3_layout}. Fix $g \in G$ and anchor the qubit $(D_1, g)$. We label every other data and check qubit block $B \in \{D_2, \ldots, D_5, X_0, X_1, Z_0, Z_1\}$ by traversing the graph from $D_1$ to $B$ and accumulating group elements along the way. The rule is: a horizontal data-to-check step left-multiplies by the layer-$c$ group element on that edge, and a vertical data-to-check step right-multiplies by it; the corresponding check-to-data steps apply the inverse group elements. These labels tell us which qubits from each block participate in the $3 \times 3$ grid corresponding to choosing qubit $g$ in the $D_1$ block. All that remains to show is that this is well-defined i.e. that the label we assign to a block $B$ does not depend on which path we walked from $D_1$ to $B$, so that we really can traverse the graph in a self-consistent way and extract a single $3 \times 3$ grid.

Suppose $\gamma_1$ and $\gamma_2$ are two paths from $D_1$ to $B$. Walking
$\gamma_2$ in reverse undoes the multiplications $\gamma_2$ applied, so the
cycle $\gamma_2^{-1} \circ \gamma_1$, traversed starting from label $g$ at
$D_1$, returns to $D_1$ with label $g$ exactly when $\gamma_1$ and $\gamma_2$
assigned the same label at $B$. Hence, it suffices to verify that every cycle
in the $3 \times 3$ grid leaves the starting label fixed.

The cycle space of the $3 \times 3$ grid is generated by its four unit faces. We check only the $D_1 \to X_0 \to D_5 \to Z_0 \to D_1$ face since the argument is the same for the remaining faces. As demonstrated in
Fig.~\ref{fig:thickness3_layout}, in layer $c$ the $D_1 \to X_0$ edge
corresponds to left multiplication by $a_0^{(c)}$ and the $X_0 \to D_5$ edge corresponds to left multiplication by $\bigl(a_0^{(c)}\bigr)^{-1}$, where $a_0^{(c)}$ denotes the $c$-th group element in the ring entry $a_0$ of $A$. Similarly, the vertical edges $D_5 \to Z_0$ and $Z_0 \to D_1$ correspond to right multiplication by $b_0^{(c)}$ and $\bigl(b_0^{(c)}\bigr)^{-1}$ respectively. Since the left and right representation commute, the net effect of traversing the face is therefore left multiplication by $a_0^{(c)} \cdot \bigl(a_0^{(c)}\bigr)^{-1} = e$ and right multiplication by $b_0^{(c)} \cdot \bigl(b_0^{(c)}\bigr)^{-1} = e$, which is the identity. Hence, this face leaves the starting label fixed, and by the same argument so does every other unit face.

Therefore, the labeling is well-defined for each $g$, and varying $g$ over $G$ produces $|G|$ vertex-disjoint $3 \times 3$ grids $\mathcal{G}_g$ that together satisfy
\[
  (V,\, E_c) \;\cong\; \bigsqcup_{g \in G} \mathcal{G}_g.
\]
Hence, each subgraph $(V, E_c)$ is planar. Combined with the lower bound from Corollary~\ref{cor:thickness_lb}, this gives $\theta(\mathcal{T}) = 3$.
\end{proof}

The thickness of mitten codes is one higher than the gross codes, which have thickness 2~\cite{bravyi2024high}. However, the thickness is still only a coarse proxy of the complexity of laying out a code on multi-layer superconducting hardware and so thickness-3 need not preclude practical implementation. In practice, even if a code's thickness is small, there are many other factors to take into account when designing the layout. Ideally one would like to minimize the length, the number of bump bond transitions, and the number of through-silicon vias (TSVs) per coupler, and this design goal may conflict with the minimum thickness construction of the code \cite{HAL_paper}.   

As a first step toward evaluating the hardware feasibility of mitten codes beyond their thickness, we use HAL~\cite{HAL_paper}\textemdash a recently introduced heuristic algorithm that automates the placement and routing of arbitrary qLDPC codes on multilayer superconducting hardware. At a high level, HAL works by extracting a large planar subgraph of the Tanner graph to layout all the qubits on the first tier. Then it uses a modified version of $A^*$ to route all the remaining edges that cannot be placed on the first tier through higher tiers without collision \cite{HAL_paper, astar}. While HAL can automatically extract a large planar subgraph to fix the qubit layout on the first tier, it also has the nice feature of being able to pass a custom layout of qubits on the first tier. We find that this is crucial to achieving good hardware complexity with our codes. 

In order to achieve low hardware complexity, we lay out all the qubits on the first tier as a super-grid of $|G|$ $3 \times 3$ grids $\mathcal{G}_g$ as shown in the bottom center panel of Fig.~\ref{fig:thickness3_layout}. We refer to each $\mathcal{G}_g$ as a module. The edges in each module only correspond to the connections from the first group element in each ring entry of the base matrices used to construct the mitten code. To realize the full connectivity of the Tanner graph, we require additional inter-module connections between qubits in different $\mathcal{G}_g$. As a first approximation to the hardware complexity, we attempt to minimize the length of these connections by appropriately choosing where we place each $\mathcal{G}_g$ on the super-grid.

\begin{table}[t]
\centering
\begin{tabular}{lcccc@{\hspace{0.4cm}}c}
\toprule
Mitten Code & Tiers & Length & Bumps & TSV & $C_{\mathrm{hw}}$ \\
\midrule
$\llbracket 150,30,10 \rrbracket $ & 6 & 6.49 & 4.12 & 3.58 & 2.02 \\
$\llbracket 200,40,12 \rrbracket $ & 6 & 8.38 & 4.12 & 3.77 & 2.09 \\
$\llbracket 300,60,14 \rrbracket $ & 7 & 9.40 & 4.69 & 4.57 & 2.28 \\
$\llbracket 500,100,16 \rrbracket $ & 9 & 12.17 & 4.75 & 5.87 & 2.60 \\
$\llbracket 540,108,18 \rrbracket $ & 10 & 11.95 & 5.00 & 6.43 & 2.72 \\
$\llbracket 630,126,\leq20 \rrbracket $ & 11 & 15.58 & 5.80 & 7.08 & 2.98 \\
$\llbracket 780,156,\leq22 \rrbracket $ & 12 & 15.41 & 5.30 & 7.92 & 3.08 \\
$\llbracket 975,195,\leq24 \rrbracket $ & 14 & 18.91 & 5.53 & 8.90 & 3.40 \\
\midrule
\multicolumn{6}{l}{\textit{Gross codes}~\cite{HAL_paper}} \\
\midrule
$\llbracket 144,12,12 \rrbracket $ & 5 & 11.08 & 5.06 & 3.27 & 2.12 \\
$\llbracket 288,12,18 \rrbracket $ & 5 & 13.94 & 5.13 & 3.75 & 2.24 \\
\bottomrule
\end{tabular}
\caption{HAL hardware complexities for mitten codes. \emph{Tiers} is the number of vertical routing tiers needed by
HAL. \emph{Length} is the mean routed coupler length (in units of the nearest-neighbor coupler). \emph{Bumps} is the maximum over all tiers of the average number of bump-bond transitions per coupler in each tier, and \emph{TSV} is the average number of through-silicon vias per coupler. The composite hardware complexity $C_{\mathrm{hw}}$ is a weighted average of these four quantities, rescaled so that a surface code layout yields $C_{\mathrm{hw}}=1$ and a layout saturating the optimistic hardware targets of Ref.~\cite{HAL_paper} (5 tiers, coupler length $10\times$ nearest-neighbor, 4 bumps per coupler, 3 TSVs per coupler) yields $C_{\mathrm{hw}}=2$. Gross code values from Ref.~\cite{HAL_paper} are shown for comparison.}
\label{tab:paper_nonab_hw}
\end{table}

\smallskip
\noindent\textit{Placement as a Quadratic Assignment Problem.}
Each module $\mathcal{G}_g$ occupies one cell of a super-grid with $|G|$ sites on the chip. With module orientations held fixed and the intra-module (first-group-element) edges fixed on the first tier by construction, the only free choice affecting the higher-tier edge lengths is the assignment $\pi : G \to \{1,\ldots,|G|\}$ that maps each module $\mathcal{G}_g$ to a slot on the super-grid. Given $\pi$, the total Manhattan length of the inter-module edges is
\begin{equation}
    C(\pi) \;=\; \sum_{(u,v)\in E_{2}\cup E_{3}} \bigl\lVert \mathrm{pos}_\pi(u) - \mathrm{pos}_\pi(v)\bigr\rVert_1,
    \label{eq:qap_cost}
\end{equation}
where $E_2,E_3$ are the layer-2 and layer-3 edge sets of the thickness-3 decomposition in Fig.~\ref{fig:thickness3_layout} and $\mathrm{pos}_\pi(\cdot)$ is the chip coordinate implied by $\pi$. Minimizing $C(\pi)$ over all $|G|!$ permutations is an instance of the quadratic assignment problem, which is NP-hard in general \cite{burkard2013, SahniGonzalez1976PCompleteApproximation}. We take two heuristic approaches to optimize $C(\pi)$: local search with simulated annealing and spectral placement. For both approaches, once we obtain an optimized module placement, we use greedy search to further reduce coupler lengths by optimizing each module's orientation over the eight rotations and reflections in the square's symmetry group $D_4$. The resulting layout is then passed to HAL.

\smallskip
\noindent\textit{Heuristic Search}
The first is a multi-restart local search combined with simulated annealing. For $n_\mathrm{iter}$ trials, we sample a uniformly-random permutation and then repeatedly apply the pairwise swap $(i,j)$ that decreases $C(\pi)$ the most among all $\binom{|G|}{2}$ candidate swaps, until no improving swap exists (i.e. the current permutation is a local minimum under transpositions). Cost evaluations are vectorized so a full pass through the $\binom{|G|}{2}$ candidates takes about $10$~ms at $|G|=30$ and about $1$~s at $|G|=195$; a complete greedy descent to a local minimum needs $5$–$20$ such passes. Different random starting permutations land in different basins of attraction; we keep the best local optimum found across all initial permutations. We then run simulated annealing warm-started from this local optimum which lets the search escape shallow basins. Finally, we optimize the module orientations and pass the resulting configuration to HAL. We repeat this process for a range of different super-grid aspect ratios and random seeds and take the best hardware layout found by HAL. While we find that lower cost $C(\pi)$ is correlated with lower hardware complexity, the best hardware complexities for our codes do not always come from the configuration with the minimum cost. This is because coupler lengths are not the only thing that dictate hardware complexity and configurations with lower $C(\pi)$ can sometimes require more couplers to be promoted to higher layers to avoid crossing. 

\smallskip
\noindent\textit{Spectral Placement} As an alternative to local search we also tried a deterministic approach that turns the combinatorial layout problem into a linear algebra one. Let $W_{gh}$ be the number of inter-module edges between $\mathcal{G}_g$ and $\mathcal{G}_h$, $D$ the diagonal degree matrix with $D_{gg}=\sum_h W_{gh}$, and $L = D - W$ the graph Laplacian of the module-connectivity graph. If we assign each module $\mathcal{G}_g$ a continuous 2D position $x_g \in \mathbb{R}^2$, the sum of squared edge lengths weighted by connectivity is
\begin{equation}
    \sum_{g < h} W_{gh}\,\lVert x_g - x_h\rVert_2^2 \;=\; \sum_{d\in\{1,2\}} v^{(d)\top} L\, v^{(d)},
    \label{eq:spectral_obj}
\end{equation}
where $v^{(d)}$ is the vector of the $d$-th coordinates of all $x_g$, and we have fixed a lexicographic ordering of the group elements of $G$. We want to minimize~\eqref{eq:spectral_obj}. However, because the objective splits into two identical quadratic forms $v^{(d)\top} L\, v^{(d)}$, minimizing it naively is degenerate in two distinct ways. To resolve this, we impose two constraints. First, we require each $v^{(d)}$ to be unit-norm and orthogonal to the all-ones vector $\mathbf{1}$: unit-norm rules out the trivial collapse $x_g = 0$, while orthogonality to $\mathbf{1}$ excludes the solution where all modules coincide or lie on a 1D line. Second, we require $v^{(1)}$ to be orthogonal to $v^{(2)}$. This constraint is the less obvious one: since both summands are the same function of their argument, without it both axes would independently select the single minimizer of $v^\top L v$, placing every module on a diagonal line and collapsing the 2D layout to 1D.

Since $L$ is symmetric and positive semi-definite by standard spectral theory arguments, the closed form solution to minimizing~\eqref{eq:spectral_obj} subject to the aforementioned constraints is given by choosing $v^{(1)}$ and $v^{(2)}$ to be the eigenvectors of $L$ associated with its second and third smallest eigenvalues $\lambda_2 \leq \lambda_3$, respectively. To see why, expand any feasible $v^{(d)}$ in an orthonormal
eigenbasis $\{u_1, u_2, \ldots, u_{|G|}\}$ of $L$ with eigenvalues
$0 = \lambda_1 \leq \lambda_2 \leq \cdots \leq \lambda_{|G|}$; note that $u_1 = \mathbf{1}/\sqrt{|G|}$ since our underlying graph is connected. 
The constraint $v^{(d)} \perp \mathbf{1}$ removes the $u_1$ component, so the quadratic form $v^{(d)\top} L\, v^{(d)}$ is a convex combination of $\lambda_2, \ldots, \lambda_{|G|}$ and is minimized by $v^{(1)} = u_2$, achieving the value $\lambda_2$. The mutual orthogonality constraint then forces $v^{(2)}$ into the orthogonal complement of $\operatorname{span}\{\mathbf{1}, u_2\}$, over which the minimizer is $v^{(2)} = u_3$ with value $\lambda_3$. The optimal value
of~\eqref{eq:spectral_obj} is therefore $\lambda_2 + \lambda_3$, and the resulting layout assigns module $\mathcal{G}_g$ the position $x_g = \big(u_2(g),\, u_3(g)\big)$. 

However, these positions are continuous, while the modules must occupy the $|G|$ integer slots of the super-grid. To resolve this, we first rescale each axis of the continuous embedding independently so that its range spans the bounding box of the super-grid. Let $\tilde{x}_g$ denote the resulting rescaled position of module $g$. We then snap the spectral placement onto the grid by solving the classical \emph{assignment problem}: given the chip coordinates $y_1, \ldots, y_{|G|}$ of the super-grid slots, we find the bijection $\pi : G \to \{1, \ldots, |G|\}$ that minimizes the total squared displacement
\begin{equation}
\sum_g \bigl\lVert \tilde{x}_g - y_{\pi(g)}\bigr\rVert_2^2 .
\label{eq:assignment}
\end{equation}
We solve the assignment problem in $O(|G|^3)$ time with the Hungarian algorithm~\cite{kuhn1955hungarian}. Since the eigenvectors $u_2, u_3$ are only unique up to sign and axis labeling, we run the matching across all eight sign flips and axis swaps of $(u_2, u_3)$ and keep the assignment with the lowest cost. We then optimize the module orientations and pass the resulting layout to HAL. 

Spectral placement produced the best hardware complexity we found for $\llbracket 300,60,14 \rrbracket $, whereas for all other codes the preprocessing from heuristic search gave the lowest hardware complexity. Table~\ref{tab:paper_nonab_hw} summarizes our results. 

\section{Code construction data}
\label{app:construction_data}
In Table~\ref{tab:code_construction_full} we present the base matrices used to construct all the codes used in this work based on their group element index in GAP~\cite{GAP4}. 
\begin{table}[H]
\centering\small
\setlength{\tabcolsep}{5pt}\renewcommand{\arraystretch}{1.3}
\begin{tabular}{@{}l c c c c@{}}
\toprule
\textbf{$\llbracket n,k,d\rrbracket$} & \textbf{group} & \textbf{GAP ID} & $A=[a_0\,\|\,a_1]$ & $B=[b_0\,\|\,b_1]$\\
\midrule
$\llbracket150,30,10\rrbracket$   & $C_5\times S_3$        & $(30,1)$   & $\{0,14,23\}\,\|\,\{0,2,11\}$      & $\{7,20,24\}\,\|\,\{0,2,29\}$\\
$\llbracket200,40,12\rrbracket$   & $C_4\times D_{10}$     & $(40,5)$   & $\{10,21,29\}\,\|\,\{0,17,18\}$    & $\{2,27,38\}\,\|\,\{0,19,21\}$\\
$\llbracket300,60,14\rrbracket$   & $C_{10}\times S_3$     & $(60,11)$  & $\{38,51,54\}\,\|\,\{0,6,45\}$     & $\{25,33,48\}\,\|\,\{0,16,58\}$\\
$\llbracket500,100,16\rrbracket$  & $C_5\rtimes C_{20}$    & $(100,9)$  & $\{19,84,87\}\,\|\,\{0,75,78\}$    & $\{39,45,71\}\,\|\,\{0,7,77\}$\\
$\llbracket540,108,18\rrbracket$  & $C_9\rtimes C_{12}$    & $(108,9)$  & $\{20,35,52\}\,\|\,\{0,36,39\}$    & $\{38,63,104\}\,\|\,\{0,35,94\}$\\
$\llbracket630,126, \leq 20\rrbracket$  & $C_7\rtimes C_{18}$    & $(126,1)$  & $\{50,117,123\}\,\|\,\{0,62,104\}$ & $\{4,39,82\}\,\|\,\{0,67,87\}$\\
$\llbracket780,156, \leq 22\rrbracket$  & $C_{13}\times A_4$     & $(156,13)$ & $\{38,46,88\}\,\|\,\{0,8,59\}$     & $\{13,40,131\}\,\|\,\{0,38,133\}$\\
$\llbracket975,195, \leq 24\rrbracket$  & $C_{13}\rtimes C_{15}$ & $(195,1)$  & $\{112,123,135\}\,\|\,\{0,104,185\}$ & $\{52,56,132\}\,\|\,\{0,62,75\}$\\
\midrule
$\llbracket 300,60, \leq 9\rrbracket$   & $C_{10} \times S_3$  & $(60,11)$  & $\{21,35,46\}\,\|\,\{0,11,59\}$    & $\{9,35,58\}\,\|\,\{0,23,47\}$\\
$\llbracket330,66, \leq12\rrbracket$   & $C_{11}\times S_3$     & $(66,1)$   & $\{0,19,49\}\,\|\,\{0,8,62\}$      & $\{0,14,56\}\,\|\,\{0,35,41\}$\\
$\llbracket600,120, \leq 14\rrbracket$  & $C_5\times S_4$        & $(120,37)$ & $\{1,28,99\}\,\|\,\{0,57,90\}$     & $\{1,60,80\}\,\|\,\{0,33,114\}$\\
$\llbracket600,120, \leq 16\rrbracket$  & $C_5\times S_4$        & $(120,37)$ & $\{16,43,107\}\,\|\,\{0,46,113\}$  & $\{2,46,113\}\,\|\,\{0,61,82\}$\\
$\llbracket840,168, \leq 18\rrbracket$  & $C_7\times S_4$        & $(168,45)$ & $\{1,33,162\}\,\|\,\{0,65,142\}$   & $\{1,57,138\}\,\|\,\{0,41,166\}$\\
\midrule
$\llbracket560,112, \leq 14\rrbracket$  & $C_{28}$ & $(28,2)$ & \multicolumn{2}{c}{$A=B=\left[\begin{smallmatrix}x^{10}{+}x^{11} & x^{26} & x^{2}{+}x^{19} & x^{6}\\[1pt] x^{13} & x^{15}{+}x^{27} & x^{15} & 1{+}x^{10}\end{smallmatrix}\right]$}\\
\bottomrule
\end{tabular}
\caption{Base matrices used to construct the mitten, structured mitten, and other lifted product codes in this work. Groups are identified by their
GAP SmallGroup ID $(\lvert G\rvert,i)$~\cite{GAP4}. For the $1\times2$ mitten codes,
$A,B$ are $1 \times 2$ matrices over $\mathbb{F}_2[G]$. For each group element in the ring entry, the listed number corresponds to the index of the group element in the list returned by \texttt{Elements}$(G)$. The abelian $2\times4$ case with $A=B$ is given as polynomial base matrices in terms of a single cyclic generator $x$.}
\label{tab:code_construction_full}
\end{table}

\end{document}